%% file: main.tex
\documentclass[mnsc,nonblindrev]{informs4ssrn}
\usepackage{dirtytalk}

\OneAndAHalfSpacedXI 
\usepackage{endnotes,mathtools}
\let\enotesize=\normalsize
\def\notesname{Endnotes}%
\def\makeenmark{$^{\theenmark}$}
\def\enoteformat{\rightskip0pt\leftskip0pt\parindent=1.75em
	\leavevmode\llap{\theenmark.\enskip}}
\usepackage{pgfplots}
\pgfplotsset{compat=1.17}
\usepackage{tikz}
\usepgfplotslibrary{fillbetween}
\usetikzlibrary{patterns}
\usepackage{graphicx}
\usepackage{color}
\usepackage[ruled]{algorithm2e}
\usepackage{fancybox}
\usepackage{tcolorbox}
\usepackage{mdframed}

\definecolor{royalazure}{rgb}{0.0, 0.22, 0.66}
\usepackage{hyperref}
\hypersetup{ colorlinks = true, linkcolor=royalazure,	citecolor=royalazure}
\usepackage{enumitem}
\usepackage{mathrsfs}
\usepackage{booktabs}

\usepackage{verbatim}  
\newif\ifshowtikz
\showtikztrue
\newcommand{\revcolor}[1]{{\color{black}#1}}

\let\oldtikzpicture\tikzpicture
\let\oldendtikzpicture\endtikzpicture

\renewenvironment{tikzpicture}{%
    \ifshowtikz\expandafter\oldtikzpicture%
    \else\comment%
    \fi
}{%
    \ifshowtikz\oldendtikzpicture%
    \else\endcomment%
    \fi
}

\usepackage[capitalize,noabbrev]{cleveref}
\usepackage{multirow}
\usepackage{hhline}
\usepackage{mathtools}
\usepackage[small, margin=1cm]{caption}
\usepackage{appendix}
\usepackage{color}
\definecolor{strcolor}{rgb}{0.6, 0.2, 0.6}
\definecolor{commentcolor}{rgb}{0.3125, 0.5, 0.3125}
\definecolor{keycol}{rgb}{0, 0, 1}

\usepackage{amssymb}
\usepackage{amsmath}
\usepackage{bbm}

\usepackage{listings}
\usepackage{url}

\newtheorem{rem}{Remark}
\newtheorem{lem}{Lemma}
\newtheorem{thm}{Theorem}
\newtheorem{prop}{Proposition}
\newtheorem{cor}{Corollary}
\newtheorem{exm}{Example}
\newtheorem{Defi}{Definition}

\hypersetup{
    colorlinks=true,
    linkcolor=blue,
    filecolor=magenta,      
    urlcolor=cyan,
    pdftitle={A Stochastic Growth Model for Online Platforms},
    pdfpagemode=FullScreen,
    }

\newcommand{\size}{S}

\newcommand{\reals}{\mathbb{R}}
\newcommand{\EX}{\mathbb{E}}

\newcommand{\tpi}{\tilde{\pi}}

\newcommand{\tempS}{\ell}
\newcommand{\D}{{\rm d}} 
\newcommand{\imm}{\Gamma}
\newcommand{\e}{\mathbb{E}}

\newcommand{\Indi}{\mathbbm{1}}

\newcommand{\stationary}[2]{\Psi_{#1}(#2)}
\newcommand{\demand}[1]{\lambda_{#1}}
\newcommand{\liability}[1]{L_{#1}}
\newcommand{\sOne}{s_1}

\newcommand{\sTwo}{s_2}

\def\blot{\quad \mbox{$\vcenter{ \vbox{ \hrule height.4pt
				\hbox{\vrule width.4pt height.9ex \kern.9ex \vrule width.4pt}
				\hrule height.4pt}}$}}

\usepackage[authoryear,round]{natbib}
 \bibpunct[, ]{(}{)}{,}{a}{}{,}%
 \def\bibfont{\small}%
 \def\bibsep{\smallskipamount}%

\TheoremsNumberedThrough
\ECRepeatTheorems
\EquationsNumberedThrough 
\gdef\AQ#1{}
\gdef\CQ#1{}
\begin{document}
\def\COPYRIGHTHOLDER{INFORMS}%
\def\COPYRIGHTYEAR{2024}%
\def\DOI{\fontsize{7.5}{9.5}\selectfont\sf\bfseries\noindent https://doi.org/10.1287/opre.2017.1714\CQ{Word count = 9740}}

\definecolor{OliveGreen}{cmyk}{0.64,0,0.95,0.40}
\definecolor{BrickRed}{cmyk}{0,0.89,0.94,0.28}

\newcommand{\rene}[1]{\textcolor{BrickRed}{\bf[Rene: #1]}}
\newcommand{\farbod}[1]{\textcolor{OliveGreen}{\bf[Farbod: #1]}}


\RUNAUTHOR{Ekbatani and Caldentey} %

	\RUNTITLE{A stochastic growth model for online platforms}

\TITLE{A Stochastic Growth Model for Online Platforms}

	\ARTICLEAUTHORS{

\AUTHOR{Farbod Ekbatani, René Caldentey}
\AFF{University of Chicago, Booth School of Business}

}
	\ABSTRACT{We investigate the growth trajectory of online platforms operating as market makers in two-sided matching markets. Their expansion, often referred to as the ``chicken-and-egg dilemma'' in the literature, poses an inherent challenge\revcolor{:} service providers find value in joining the platform \revcolor{because of} the presence of customers, while customers are drawn to the platform by the availability of \revcolor{service providers}. This endogenous network effect influences the growth trajectory and \revcolor{can propel} emerging platforms \revcolor{toward ``get big fast'' strategies aimed at} rapidly achieving critical mass on both sides.

Much of the existing literature addresses this expansion challenge through \revcolor{either} a static market-clearing lens \revcolor{or aggregate dynamic models}, emphasizing how network externalities shape a platform's growth trajectory and equilibrium size. \revcolor{These approaches do not fully capture how gradual growth must be implemented through server-level admission, assignment, and compensation decisions.} Embedded in this growth process are \revcolor{intertemporal participation constraints} that the platform must satisfy to \revcolor{attract and retain forward-looking servers}.

To explicitly account for the stochastic dynamics of a platform's gradual growth, we model the system as a Markov process. \revcolor{Potential servers and customers arrive according to Poisson processes, and customer demand is non-decreasing in the number of active servers.} Under mild conditions, we characterize \revcolor{the platform's optimal admission, assignment, and compensation policies}. \revcolor{The platform admits servers at the maximum feasible rate until it reaches an endogenous target size, which can be smaller than the static market-clearing optimum. In the baseline model, seniority-based compensation acts as an operational substitute for financially subsidized growth: instead of using upfront subsidies to induce early server participation, the platform uses future cash flows to reward early servers for joining when demand is limited and for helping the platform reach later, higher-demand states. This mechanism embeds an intertemporal notion of fairness by linking later rewards to early participation.} Furthermore, we show that straightforward compensation schemes---such as \revcolor{fixed fees and} uniform payments---\revcolor{can slow growth and reduce platform value}.}
	
\maketitle
\newpage
\input{Tex/Intro}
\input{Tex/Literature}

\input{Tex/Model}
\input{Tex/Static}

\input{Tex/Homogenous}
\input{Tex/MicroLevel}
\input{Tex/numerics}
\input{Tex/Benchmarks}
\vspace{-0.4cm}
\input{Tex/Extensions}
\vspace{-.7cm}
\input{Tex/Conclusion}

\begingroup
\renewcommand{\bibfont}{\fontsize{11}{10}\selectfont}
\setlength{\bibsep}{5pt plus 0.5ex}
\bibliographystyle{plainnat}
\bibliography{References_Platform}
\endgroup

\newpage
\pagenumbering{arabic}
\renewcommand{\thepage}{A\arabic{page}}
\begin{appendices}
 \renewcommand{\thesection}{Appendix \Alph{section}}
 \renewcommand{\thesubsection}{\Alph{section}.\arabic{subsection}} 

 \numberwithin{equation}{section}
\renewcommand{\theequation}{\Alph{section}\arabic{equation}}
 \input{Tex/proofs}
\input{Tex/Refinments}

\input{Tex/Congestion}
\input{Tex/Extras}

\end{appendices}
\end{document}

%% file: Tex/Intro.tex
\section{Introduction}\label{sec:intro}
Two-sided platforms play a crucial role in connecting customers and service providers in various industries, ranging from ride-sharing and food delivery to online marketplaces. The success of these platforms depends on their ability to simultaneously attract both sides of the market, a challenge often referred to as the \say{Chicken and Egg} problem (\citealp{rochet2003platform}). This issue arises because the value that the platform offers to one group of users depends on the active participation of the other group. Customers (buyers) are reluctant to join platforms with limited service providers (sellers), while service providers are hesitant to invest time and resources in platforms with insufficient demand. Without careful design and strategic interventions, platforms risk falling into a cycle of stagnation, where neither side is incentivized to participate, preventing the platform from achieving critical mass and network effects. Addressing this challenge requires strategic design of matching policies, compensation schemes, and growth strategies that ensure sustained participation from both customers and service providers.

On many two-sided platforms, service providers are usually individual contractors rather than regular employees. This distinction introduces unique challenges, as these contractors are compensated based on the volume of jobs they perform, their participation is highly sensitive to the platform’s operational dynamics. 
As \cite{farrell2017online} emphasize, the growth of online platforms relies heavily on either attracting new participants or strengthening the engagement of existing ones. Unlike traditional employment models, platforms cannot directly mandate participation or performance levels; instead, they must design incentive structures that attract and retain contractors.
\revcolor{Balancing these individual-level incentives with broader platform growth objectives becomes even more complex under changing demand patterns that accompany the platform’s expansion.}

In this paper, we develop a mathematical model to analyze the growth dynamics of two-sided platforms that serve as intermediaries between customers and service providers (hereafter servers). Servers constitute \revcolor{a} homogeneous pool of individuals who operate as \revcolor{individual} contractors. Our model captures the stochastic nature of platform growth by incorporating network effects inherent in these systems 
. Through this approach, we aim to provide insights into designing revenue-maximizing policies that account for the complexities of real-world scenarios. By focusing on the impact of platform policies at the level of individual contractors, our work offers a way of thinking on how to align the servers’ incentives with sustainable and profitable platform growth.

\revcolor{Much of the existing literature on two-sided platform growth studies subsidy-led strategies, in which the platform temporarily stimulates one or both sides of the market so that network effects can develop and the market eventually becomes self-sustaining \citep{lian2021optimal}. Such strategies require the platform to finance early losses, typically through access to external capital and investors willing to bear substantial growth-stage risk. This assumption is restrictive. Although this playbook was feasible for some well-funded platforms during periods of abundant venture capital---with Uber often cited as a canonical example of capital-intensive marketplace growth \citep{rothman2016steroid,rothman2016uber}---it is not generally available. Recent funding data also suggest a tighter environment for startups, with global startup funding in 2023 falling to its lowest level in five years and capital becoming more selective \citep{forbesafrica2026seriesa,damani2025fundingwinter,teare2024globalStartupFunding}. As a result, subsidy-financed growth may be difficult for platforms operating in developing economies, outside major venture-capital hubs, or during periods of macroeconomic stress. This motivates our focus on an alternative growth mechanism based on intertemporal compensation design rather than deep-pocketed investors or sustained subsidies.}

\revcolor{A key contribution of this paper is to introduce a seniority-based compensation scheme as such an alternative operational mechanism for addressing the chicken-and-egg dilemma. The central idea is to reward early adopters---servers who join the platform in its initial stages---for bearing low early demand and helping the platform reach a more mature, higher-demand state. The platform does so by prioritizing senior servers' access to demand and by allowing their expected compensation to increase as the platform grows. These seniority-based benefits make early participation attractive even when current demand is limited, thereby supporting growth through long-term incentive design rather than upfront subsidies.}

\revcolor{This mechanism is most relevant for platforms with a stable core of repeat servers who regard platform participation as economically meaningful, rather than for settings in which all supply-side participants are purely transitory or myopic. Recent evidence suggests that such a core is nontrivial. Branch and Marqeta report that the share of surveyed gig workers relying on gig work as their primary income rose from 21\% in 2021 to 61\% in 2023; McKinsey estimates that independent workers increased from 27\% to 36\% of employed U.S. respondents between 2016 and 2022; and Pew reports that 31\% of current or recent online gig-platform workers viewed platform work as their main job over the prior year \citep{branchMarqeta2023gig,duaEtAl2022freelance,andersonEtAl2021gig}. The Federal Reserve's SHED survey similarly finds that 31\% of gig-activity participants would have difficulty making ends meet without this work \citep{federalReserve2025economicWellBeing}. While many gig workers remain occasional participants, these trends support the perspective we adopt in this paper: a meaningful segment of platform supply treats platform work as economically central and may therefore evaluate participation and compensation over a longer horizon.}

\revcolor{This perspective also distinguishes our work from existing approaches such as \cite{lian2021optimal} and \cite{he2022profit}, which focus primarily on macro-level platform growth dynamics. In contrast, our model examines how growth affects individual servers' incentives, contributions, and compensation over time. This server-level view is necessary for studying intertemporal participation constraints, the allocation of future demand across cohorts, and the role of seniority in making early participation credible and valuable. It also allows us to characterize when seniority-based compensation can serve as a viable substitute for subsidy-financed growth and how such a policy shapes both platform growth and server compensation.}

\revcolor{A broader set of platform practices illustrates that seniority-like rewards are not merely a theoretical construct. Creator and labor platforms often protect early or long-tenured supply-side participants through grandfathered fees, relationship-based take-rate reductions, or enhanced revenue-sharing tiers.  Take Patreon as an example. Patreon is a membership platform that allows creators to earn recurring revenue directly from subscribers, or ``patrons,'' in exchange for content, community access, or other benefits. At Patreon, creators who joined before the introduction of newer pricing plans were protected through legacy ``Founders'' pricing, allowing them to retain lower platform fees than later cohorts.\footnote{\revcolor{Specifically, when Patreon scaled and shifted to a tiered pricing model, it permanently locked in a 5\% platform fee for its foundational users, whereas newer arrivals are charged as much as 12\% for similar core functionality. For further details, see \url{https://support.patreon.com/hc/en-us/articles/11111747095181-Creator-fees-overview}.}} Similarly, Upwork historically used a relationship-tenure logic by reducing the freelancer service fee as cumulative billings with the same client increased. This structure allowed longer platform relationships to generate higher effective compensation for freelancers.\footnote{\revcolor{For Upwork's historical relationship-based fee schedule, see \url{https://wise.com/us/blog/upwork-fees}.}} These real-world strategies effectively grant long-term structural premiums to early adopters, directly illustrating how a seniority-driven framework can reward a foundational cohort for helping an emerging platform mature.}\footnote{\revcolor{Twitch and eBay provide related examples based on sustained platform commitment and performance: Twitch offers enhanced subscription revenue shares to streamers who maintain high subscriber volume over time, while eBay grants Top Rated Plus sellers a fee discount and additional visibility benefits. These practices do not replicate our model literally, since they are often based on entry cohort, relationship tenure, or performance history rather than a fully specified intertemporal contract. Nevertheless, they show that major platforms already use persistent differences in fees, visibility, and revenue shares to reward supply-side participants who join early, remain active, or help sustain platform growth. For Twitch's Plus Program, see \url{https://help.twitch.tv/s/article/plus-program}. For eBay's Top Rated Plus seller benefits, see \url{https://www.ebay.com/sellercenter/protections/top-rated-program}.}}

\revcolor{Taken together, these examples suggest that seniority-like rewards are most plausible when supply-side participants form a continuing relationship with the platform and when the platform can credibly maintain persistent, observable rules. Accordingly, our proposed mechanism is not intended to describe platforms in which participation is entirely transitory, anonymous, or myopic. Rather, it is most relevant when the platform can identify repeat providers, publicly commit to persistent rules, and make seniority-based benefits observable and auditable. Examples include creator platforms, professional-service marketplaces, and freelancer platforms, such as those discussed above, where a meaningful subset of providers builds an ongoing relationship with the platform. In such settings, seniority-based priority can be interpreted as a transparent rule for allocating future demand or future fee reductions to early cohorts. By contrast, when providers treat each interaction as a one-shot opportunity and do not trust future platform commitments, our optimal policy should be interpreted as a benchmark that quantifies the value lost when such intertemporal contracting is unavailable.}

\revcolor{\subsection{Results and Contributions.}\label{sec:results}
Before reviewing the existing literature and presenting a precise mathematical formulation of the platform's growth problem, we summarize our main results and contributions. Our approach yields several insights that contrast with, and complement, the existing research on platform growth.}

\begin{enumerate}
\item \revcolor{\textbf{Forward-looking supply enables growth without financial subsidies.}
In the context of our baseline model in \cref{sec:Main}, we show that when servers are sufficiently forward-looking, the platform can support early-stage growth through intertemporal compensation rather than financial subsidies. The seniority-based scheme rewards early adopters with priority access to future demand and higher expected long-run earnings, compensating them over time for joining when current demand is still limited. In this sense, forward-looking server behavior provides an internal foundation for growth: the platform can use promises of future compensation to attract early supply without relying on outside capital, early operating losses, or temporary wage subsidies. Thus, our findings in \cref{sec:Main} identify a sustainable alternative to the subsidy-heavy growth strategies commonly emphasized in the platform literature, while Section~7 shows that limited financing becomes valuable when servers are more impatient than the platform.}

\item \revcolor{\textbf{Gradual recruitment reshapes the platform's optimal size.} We show that when servers join gradually, the platform's optimal growth target need not coincide with the static market-clearing size.\footnote{\revcolor{By a static market-clearing size, we refer to the benchmark in which the relevant pools of potential servers and customers are effectively present at time zero, so that network effects are assumed to reach their equilibrium level instantaneously rather than through a gradual recruitment process.}} Because the benefits of reaching a larger server pool are discounted while the incentive costs of recruiting servers are incurred along the growth path, it can be optimal for the platform to stop expansion at a smaller local optimum. This result highlights an intertemporal friction that is absent from instantaneous market-clearing frameworks: the long-run number of active servers that the platform optimally supports is shaped not only by the eventual profitability of that server pool, but also by the dynamic cost of getting there.}
    
\item \revcolor{\textbf{Seniority-based compensation preserves intertemporal fairness.} We demonstrate that prioritizing early entrants is essential for maintaining structural fairness across server cohorts over time. Without seniority-based priority, early adopters face disproportionately low returns relative to later arrivals who enter a mature, high-demand marketplace. By explicitly rewarding servers who support the platform during its less profitable inception phases, this scheme corrects intertemporal imbalances. Consequently, intertemporal fairness emerges not just as a normative ideal, but as a critical operational mechanism for sustaining participation among forward-looking agents.}

\item \revcolor{\textbf{Platform fee flexibility is central to optimal growth.} We show that rigid compensation structures, such as uniform-compensation or fixed-fee policies, can substantially constrain platform growth. These restrictions may induce pauses in recruitment and reduce the platform's ultimate size relative to an optimal flexible policy. The underlying mechanism is that these restrictions prevent the platform from reallocating compensation across states and server cohorts in response to evolving demand. Our findings highlight intertemporal fee flexibility as a critical instrument for accelerating organic growth while preserving servers' participation incentives.}
\end{enumerate}



\revcolor{
\subsection{Roadmap of Analysis and Results} \label{sec:roadmap}

\cref{sec:literature} reviews the related literature, situating our contribution within the broader research streams on two-sided platform growth, dynamic incentive design, and intermediation markets. We highlight how prior work has largely relied on either static frameworks or macro-level dynamic models in which agents are assumed to be myopic and platform growth is primarily driven through subsidization. These models often focus on aggregate market dynamics while abstracting away from individual-level incentive compatibility. This discussion clarifies the gap our framework fills by incorporating forward-looking servers, intertemporal incentive constraints, and individualized compensation structures, allowing us to examine growth mechanisms beyond the standard subsidy-based approaches.

\cref{sec:model} introduces the formal model, including the platform’s operational structure, the arrival processes for customers and servers, and the incentive compatibility constraints that shape server participation over time. We define the platform’s objective function, characterize the compensation rules available to the platform, and introduce the central notions of Static Individual Rationality (SIR) and Dynamic Individual Rationality (DIR). This section establishes the mathematical foundation for the subsequent analysis and highlights the key mechanisms through which compensation, growth, and server retention interact in our framework.

In \cref{sec:market-clearing}, we analyze the limiting case in which the platform can recruit new servers arbitrarily quickly, so the market effectively clears instantaneously at each step of the growth process.} While this assumption is unrealistic in most practical situations, it represents an idealized scenario where all potential servers are available in the market at time $t = 0$. In this setting, the platform implements a compensation scheme that instantaneously clears the market and determines an equilibrium number of active servers joining the platform. Consequently, the platform's optimization problem simplifies to a one-dimensional optimization problem over the market size $\size$. As one should expect, the optimal revenue in this case provides an upper bound on the actual revenue that the platform can achieve.

\revcolor{In \cref{sec:Main}, we examine the complementary case in which growth unfolds gradually, with the platform expanding its server base step by step rather than instantaneously.} In this setting, the platform's problem requires the careful design of a compensation scheme that satisfies servers' individual inter-temporal incentive compatibility constraints. Notably, we demonstrate that the platform can consolidate all \eqref{eq:dynamicparticipation} constraints into a single aggregated constraint. This reduction significantly simplifies the platform's problem while also providing insights into the structure of an optimal operating policy and cumulative compensation scheme. \revcolor{In particular, it reveals that the platform should expand as quickly as possible until reaching an optimal size and, in the baseline model, that allowing the platform unlimited access to subsidies or external funding does not improve its optimal value when seniority-based assignment is employed.} Moreover, while the aggregate formulation yields a unique solution in terms of the optimal growth rate and market size, it permits a range of intermediation fees and server compensation structures that implement the optimal solution. We leverage this flexibility by imposing secondary optimality criteria, such as smoothing the platform’s cash flows and ensuring fair compensation for servers.

\cref{sec:benchmark} is dedicated to comparing the performance of the optimal growth policy against various benchmarks. While each of these alternatives has its own practical appeal, we show that, in general, they result in suboptimal performance. \revcolor{To understand the mechanics behind these inefficiencies, we study each baseline case individually and provide structural intuition into their respective short-comings.}

\revcolor{Beyond the baseline model, in \Cref{sec:extension} we extend our analysis to incorporate two features: platform access to liquidity and agent impatience. Specifically, we allow the platform to access a limited financial loan to fund operations, and we introduce asymmetric discounting where servers are more impatient and less forward-looking than the platform. This extension demonstrates that the value of seniority-based compensation reaches far beyond idealized settings. When liquidity is scarce and market participants are impatient, the platform cannot afford a simple "pay-as-you-go" strategy to match servers' outside options state-by-state. It must instead rely on future rewards, which carry a steep premium under asymmetric discounting. We show that in such environments, a seniority-driven framework provides the necessary temporal and distributional flexibility required to manage these expensive liabilities.}

Finally, in \ref{sec:congestion}, we extend our analysis to consider the case where service times are non-negligible, introducing congestion dynamics into the problem. Under mild approximations, we derive a tractable polyhedral representation of the feasible set of sustainable server compensations. This structural characterization allows us to recast the platform's optimization problem into a computationally efficient formulation.

%% file: Tex/Literature.tex
\section{Literature Review}\label{sec:literature}

Our work is related to several research streams on the operations of two-sided platforms, which we review in tandem.

\noindent \textbf{Platform Growth.} Addressing the chicken-and-egg problem has been a prominent research focus in two-sided markets. \citet{eisenmann2006internet} examines firms' initial investments and their relationship to market leadership within a specific sector, showing that pioneers tend to spend significantly more on IPO-year marketing than late movers. The question of how these funds should be allocated has been widely studied. Since the seminal work of \citet{caillaud2003chicken}, the ``divide-and-conquer'' strategy has gained traction as a key growth approach for online platforms. This strategy entails subsidizing one side of the market to attract users, thereby creating value and enabling the platform to generate profits from the other side. Related equilibrium models of two-sided platforms include \citet{rochet2003platform} and \citet{armstrong2006competition}, who incorporate cross-side network effects and user heterogeneity, and \citet{hagiu2006pricing}, who studies settings in which one side joins before the other and the platform can commit to future pricing policies.

\revcolor{These papers clarify how network effects shape platform participation and pricing, but they do not model the operational path through which a platform gradually reaches a viable scale. In practice, the key challenge is not only determining a desirable equilibrium size, but also sustaining participation while the platform is still small and demand is limited. This distinction is important for our paper because the feasibility of growth depends on whether early providers can be incentivized to join and remain active before the platform has reached maturity.}

\revcolor{The closest work to our paper is \citet{lian2021optimal}, which develops an aggregate dynamic model of two-sided-market growth. In their framework, the platform controls prices and subsidies to shape the growth of aggregate supply and demand stocks, and rapid subsidized growth is optimal in several regimes. \citet{zhu2024optimal} extend this aggregate growth perspective by incorporating heterogeneous agents. Our paper is complementary, but the unit of analysis is different. We take the gradual arrival of individual servers as primitive and ask how the platform should jointly design admission, matching, and compensation policies subject to each server's dynamic participation constraint. This distinction matters because individual-level participation creates cohort-specific obligations: early servers may need to be compensated later for joining when current demand is limited, and the platform must decide how future demand and cash flows are allocated across server cohorts.}

\revcolor{This micro-level view leads to questions that are absent from aggregate subsidy models such as \citet{lian2021optimal,zhu2024optimal}. When can future platform cash flows substitute for upfront subsidies? How does gradual recruitment affect the platform's optimal target size relative to a static market-clearing benchmark? What compensation rule implements the optimal growth path while preserving intertemporal fairness across server cohorts? Our stochastic formulation is used to make the timing of server arrivals, customer arrivals, and cohort-specific participation constraints explicit; it is these timing and incentive constraints that determine which compensation promises are feasible.}

\revcolor{Other dynamic platform-growth models address related but distinct questions. \citet{he2022profit} study order allocation in a hybrid workforce, where a firm can divert demand to an affiliated platform to influence market thickness; their growth lever is allocated demand between organizational units rather than individualized long-term compensation for arriving servers. \citet{kumar2010evolution} analyze the evolution and equilibrium of multi-platform markets through attachment curves, treating side-specific market values as exogenous. \citet{cabral2019towards} studies platform dynamics at the market level, whereas our focus is the design of a single platform's compensation and assignment mechanism for sustaining participation by individual contractors over time. Thus, relative to this stream, our contribution is to characterize how intertemporal compensation and seniority-based assignment can support organic growth when subsidy-financed expansion is limited or unavailable.}

\vspace{0.2cm}

\noindent \textbf{\revcolor{Incentives and Compensation for Servers.}} There is a growing body of literature on self-scheduling workers (\citealp{hu2022implications}, \citealp{dong2020managing}, and \citealp{castro2022supply}) that examines how the gig economy affects both the labor force and platform operations. As noted in \cite{gurvich2019operations}, traditional employees follow predetermined work schedules, whereas self-scheduling contractors adjust their participation in response to monetary incentives. A key distinction in our model is that servers are compensated only when they serve a customer, unlike traditional employees who receive a fixed periodic wage. This feature also differentiates \cite{gurvich2019operations} from \cite{taylor2018demand} and \cite{bai2019coordinating}, which primarily analyze the impact of servers' availability and utilization on consumer demand.

In the context of a ride-hailing platform, \cite{angrist2021uber} examined drivers' planning behavior and demonstrated that drivers benefit from their status as independent contractors, as they pay only a fraction of their service fee to the platform rather than a fixed monthly or weekly payment. On the platform side, \cite{lobel2024frontiers} highlight that a key advantage of using contractors is that platforms avoid overpaying workers when they are not actively engaged. However, a critical challenge persists: contractors must be consistently matched with enough customers to remain incentivized and continue working on the platform.

This challenge has been gaining increasing attention in recent years. \cite{besbes2024workforce} study mechanisms that ensure workers receive sufficient job assignments to sustain a high effective wage. They consider 
scenarios where agents self-select their working schedules and design these mechanisms to incentivize enough participation by guaranteeing that workers earn at least their reservation wage. \cite{iyengar2024online} study an online matching problem where the platform must ensure a minimum allocation guarantee—at any given time, the expected number of jobs a server receives in the future should exceed a certain threshold. They propose the \emph{Tracker} algorithm, which balances between a greedy approach that selects the best immediate option and a {\em Balance} algorithm that prioritizes the server with the fewest assignments. A similar trade-off appears in our model, where the platform seeks to maximize revenue while ensuring that individual contractors remain engaged. 

In a recent study, \cite{lian2022labor} identify the optimal pay for gig economy workers in a market with competing platforms. They demonstrate that smaller platforms can take advantage of larger ones by offering workers the minimum wage, while larger firms pay higher wages to ensure the availability of a shared worker pool. Additionally,  \cite{asadpour2023minimum} and \cite{benjaafar2022labor} study the welfare of individual contractors, particularly in scenarios where a minimum wage is regulated or the pool of contractors expands. Furthermore, \cite{jacobs2025chicago} analyze the earnings of rideshare drivers in Chicago, evaluating how a proposed minimum pay ordinance would impact driver compensation relative to current earnings and employee-equivalent wages.

Finally, an expanding body of literature (e.g., \citealp{siddiq2022ride}, \citealp{castro2024autonomous}, \citealp{freund2024supply}, \citealp{lian2024capturing},  \citealp{castro2024getting}) examines the challenges ride-hailing platforms face when deciding between serving customer demand with individual contractors (ICs or HVs) or autonomous vehicles (AVs). These studies highlight the strategic difficulties of employing individual workers, particularly due to their individual rationality constraints in the presence of outside options.

To the best of our knowledge, we are the first to study the problem of designing optimal intertemporal servers' compensations in  two-sided markets, where supply and demand exhibit evolving network effects, and servers' engagement is subject to dynamic participation constraints.
\vspace{0.2cm}

\noindent \textbf{\revcolor{Intertemporal Compensation.}} \revcolor{The server's intertemporal compensation connects our mechanism to a broader economics literature in which long-term relationships serve as an instrument for reallocating resources across time when current cash flows are constrained. In labor markets, deferred-compensation and implicit-contract models show how workers may accept lower current compensation in exchange for future rents, thereby supporting current participation and effort through the promise of future surplus \citep{baily1974wages,azariadis1975implicit,lazear1979mandatory}. More directly, \citet{michelacci2009financial} show that financially constrained firms can use back-loaded wage profiles to relax financing constraints, while \citet{guiso2013credit} provide empirical evidence of ``credit within the firm,'' whereby wage-tenure profiles operate as an implicit lending channel from workers to firms. Our mechanism is analogous, but it operates through platform's allocation of demand rather than employment wages: early contractors accept limited current demand or compensation because seniority-based priority gives them a claim on future platform cash flows. As in relational contracting models, such a mechanism is plausible only when future rewards are sufficiently credible and valuable to sustain current participation \citep{levin2003relational,baker2002relational}. Accordingly, our analysis characterizes the settings in which seniority-based priority can make future platform cash flows a meaningful instrument for supporting participation during early, demand-constrained stages.

}
\vspace{0.2cm}

\noindent \textbf{Matching in Two-Sided Markets.} There is a growing body of literature exploring the design of optimal dynamic matching policies in the context of two-sided platforms (e.g., \citealp{dynamic_matching_ride_sharing}, \citealp{feng2024two}, \citealp{aouad2023online}, \citealp{shi2022optimal}). A primary consideration in this stream of work is the heterogeneous nature of servers and customers, which influences the value and/or quality generated by a match.  
While we do not explicitly model heterogeneity in either servers or customers, our dynamic model implicitly creates it through the `seniority' of each active server. As a result, our proposed  matching policies are also sensitive to this specific form of heterogeneity, as senior servers are most costly to retain.

The relationship between dynamic matching policies and their impact on system congestion and waiting time delays in two-sided markets has received significant attention in recent research (see \citealp{Afecheetal2022}, \citealp{caldentey2024designing}, \citealp{aveklouris2024matching}, \citealp{kohlenberg2024cost}, \citealp{kerimov2024dynamic}, and references therein).  In our setting, we assume that customers are impatient and do not wait for service. As a result, service delays and congestion affect both the effective revenue streams generated by the platform and the availability of servers, ultimately influencing the maximum compensation they can receive.

Another important theme in this research area that intersects with our work is the impact of matching policies on the fairness of server compensation (e.g., \citealp{armony2010fair}, \citealp{esmaeili2023rawlsian}, \citealp{ma2020group}). In our setting, we explore this issue and propose an inter-temporal Rawlsian fairness criterion as a tie-breaking rule that platforms can use to determine server compensation while still pursuing a payoff-maximizing growth strategy.  

Finally, we acknowledge a growing literature on related topics such as information design (\citealp{besbes2023signaling}, \citealp{bimpikis2024information}), matching recommendations, and platform design (\citealp{kanoria2017facilitating}, \citealp{shi2024optimal}, \citealp{papanastasiou2024designing}) in two-sided markets. Readers may refer to \cite{hu2019sharing} for a comprehensive review of related literature on the sharing economy with an operations management perspective.

%% file: Tex/Model.tex
\section{Model Description}\label{sec:model}

In this section, we propose a stylized mathematical model that we will use to investigate the stochastic growth trajectory of a  platform acting as an intermediary in a two-sided market. One side consists of customers seeking a specific type of service, while the other comprises a homogeneous pool of servers capable of delivering the service.

On the demand side, customers arrive according to a point process ${\mathcal{A}}_t$ with time-$t$ intensity $\demand{{\mathcal S}_{t}}$, where ${\mathcal S}_t \in \mathbb{N}_0:=\mathbb{N} \cup \{0\}$ denotes the  number of active servers operating on the platform at time $t \geq 0$. We assume that the demand function $\demand{k}$ ($k \in \mathbb{N}_0$) is a bounded and non-decreasing function with $\demand{k} \in [0, \Lambda]$ for some $\Lambda >0$.
The monotonicity of $\demand{k}$ reflects a positive network effect in the platform's operations; namely, platforms with a larger pool of active servers are more visible and attract more customers.  We allow for the possibility that $\demand{k} = 0$ for  small values of $k$ to capture the scenario where the platform requires a minimum number of active servers before it can attract customers. We assume that arriving customers pay a fixed service fee $p > 0$ and are instantly matched by the platform to one of the active servers based on a matching rule discussed below. 
\revcolor{For simplicity, the base model (see \Cref{sec:Main}) assumes negligible service times. In \ref{sec:congestion}, we relax this assumption and analyze how non-negligible service times introduce congestion and server-availability constraints into the platform's growth problem.}


\revcolor{On the supply side, the state of the system at time $t$ is given by the number of active servers ${\mathcal S}_t \in \mathbb{N}_0$.}
New potential servers arrive intending to join the platform according to a point process ${\mathcal N}_t$ with time-$t$ intensity $m_{{\mathcal S}_t}$, for some non-negative function $m_k$ ($k \in \mathbb{N}_0$). Given our interest in studying the platform's growth trajectory, we assume that initially ${\mathcal S}_0=0$. The admission process for new servers arriving at time $t$ involves two steps, both contingent on the state ${\mathcal S}_t$. 

\revcolor{First, if a new server arrives when the system is in some state $s'$ with $s'=s-1$, then with probability $u(s')$, the platform offers the incoming server a compensation scheme, which we informally refer to as a `{\em contract}' in what follows, to join the system. Conversely, with probability $1-u(s')$, the platform denies the server admission, in which case the server leaves the market. 
Second, if the incoming server is offered a contract, the contract specifies a payment scheme that depends on the current and future number of active servers. This scheme is defined by an assignment probability vector $q_s = (q_s(k) : k\geq s)$ and a compensation vector $f_s = (f_s(k) : k\geq s)$. Here, $q_s(k)$ denotes the probability that server $s$ is assigned to the next arriving customer when the state of the system is $k$, and $f_s(k)$ represents the (potentially random) compensation that server $s$ receives for completing the service.


The platform acts as a middleman between customers and servers, collecting an intermediation fee each time a customer is matched with a server. This intermediation fee depends on the assignment policy and also the compensation fee offered to each server. In particular, $p-f_s({\mathcal S}_t) \in [0, p]$ is the platform's intermediation fee when the customer is served by server $s$. Thus, the expected intermediation fee that the platform collects in state ${\mathcal S}_t$ is given by
 $F({\mathcal S}_t) \coloneqq p-\sum_{s=1}^{\mathcal{S}_t} f_s({\mathcal S}_t)q_s({\mathcal S}_t)$.
The special case of a fixed intermediation fee, where $F({\mathcal S}_t) = F$ for some constant $F \in [0, p]$, is particularly relevant in practice and will be discussed in detail in \Cref{sec:benchmark}.}


Incoming servers act as risk-neutral profit-maximizing agents when deciding whether to accept or reject  the platform's proposed contract  upon arrival. Specifically, we assume that each server has an outside option that pays a fixed discounted payoff $w$. The contracts $(f_s,q_s)$ must guarantee this minimum expected payoff to ensure the servers are willing to join the platform. If we let $\tau_s$ denote the  time at which the $s^{\mbox{\tiny th}}$ server joins the platform, then the contract $(f_s,q_s)$ must satisfy the {\it static individual rationality} (SIR) constraint:
\revcolor{\begin{equation}\label{eq:staticparticipation}\tag{SIR}
\EX\left[\int_{\tau_s}^\infty e^{-r\,(t-\tau_s)}\, f_s({\mathcal S}_{t-})\,q_s({\mathcal S}_{t-})\, \D {\mathcal A}_t \big| {\mathcal S}_{\tau_s}=s\right] \geq w, \text{\footnotemark}
\end{equation}}\footnotetext{\revcolor{Here, ${\mathcal S}_{t-}$ denotes the left-limit of the state process, $\lim_{u \uparrow t} {\mathcal S}_u$, representing the system's state prior to any jump triggered by the arrival process ${\mathcal A}_t$ at time $t$.}}where $r>0$ is a fixed discount factor, assumed to be common to all servers and the platform. We refer to the previous condition as a `static' IR constraint on the contract $(f_s,q_s)$ because it only promises a server enough economic incentive to join the platform at the time of arrival. In order to ensure that a server who is willing to join the platform has no incentive to quit it at a later time, the contract $(f_s,q_s)$ must satisfy the stronger {\it dynamic individual rationality} (DIR) constraint:
\revcolor{
\begin{equation}\label{eq:dynamicparticipation}\tag{DIR}\EX\left[\int_{\tau}^\infty e^{-r\,(t-\tau)}f_s({\mathcal S}_{t-})q_s({\mathcal S}_{t-})\D {\mathcal A}_t \big| {\mathcal S}_{\tau}=k\right] \geq w,\quad \mbox{for all } k \geq s.\end{equation}}
In other words, \revcolor{at any moment in time $\tau>0$}, all active servers are promised sufficient economic incentives to continue working on the platform. Since we have assumed that the demand rate is bounded above by $\Lambda$ and each arriving customers pays a fixed service fee $p$, there is a maximum number of active servers, $\size$, that the platform can sustain under a DIR constraint on the compensation contracts. In fact, setting $f_s = p$ in \eqref{eq:dynamicparticipation}, we obtain the upper bound $\size \leq \Lambda\,p /(r\,w)$. 

In what follows, we will mainly focus on the case in which the platform must ensure a DIR constraint, which is consistent with settings in which servers have no formal contractual agreements that force them to stay active on the platform after joining and are free to leave at any time. 
A feasible operating policy for the platform is defined by the tuple $\pi=(u,f,q,S)$ satisfying conditions \eqref{eq:dynamicparticipation}. We will denote by $\Pi$ the class of {\em admissible} operating policies.
\begin{rem}\label{rem:size} For an admissible operating policy $\pi=(u,f,q,S)$, the value of $S$ is interpreted as the maximum size that the platform can reach under $\pi$. Given that we have assumed an initial condition of ${\mathcal S}_0=0$, it follows that \revcolor{$S=\min\{k \geq 0 \colon u(k)=0 \}$}. \end{rem}
\revcolor{Finally,} the expected discounted payoff for the platform under a policy $\pi \in \Pi$ is then given by
\revcolor{
\begin{align}\label{eq:platformBE}
J^\pi= \e\left[\int_0^\infty e^{-r\,t}\,F({\mathcal S}_{t-}) \D {\mathcal A}_t\big| {\mathcal S}_0=0\right].
\end{align}  }
Hence, the platform's objective is to find a feasible operating policy that maximizes its expected discounted payoff by solving the optimization problem:
\begin{equation}\label{eq:platopt}J^*=\sup_{\pi \in \Pi} J^\pi.\end{equation}
\vspace{0.1cm}
\noindent {\bf Model Discussion:} Let us provide some further clarifying remarks about our model and assumptions. First, for simplicity, we have restricted the platform's operating policies to the class of Markovian strategies that depend solely on the state ${\mathcal S}_t$, which is defined by the number of active servers. Considering more complex policies that take the entire history into account appears mathematically intractable and, \revcolor{arguably,} unrealistic for practical purposes.

\revcolor{Another simplifying assumption is that the service fee $p$ remains constant, meaning the platform does not engage in dynamic pricing. This assumption should be interpreted as isolating the supply-side intertemporal incentive problem. If price affects only contemporaneous demand, the model can incorporate state-dependent pricing by replacing $\lambda_k$ with a price-dependent demand rate $\hat\lambda_k(p)$ and optimizing $p\hat\lambda_k(p)$ state by state. This extension preserves the Markovian structure of the system because the platform's state remains fully summarized by the number of active servers $k$. A richer model in which temporary price discounts build a persistent stock of future customers would require adding a customer-side adoption state variable. Such a model would combine our server-level dynamic participation constraints with customer-side adoption dynamics and is outside the scope of the present paper.} 

Finally, regarding information, we assume that both the platform and the servers have full knowledge of market dynamics, including the demand function $\lambda_s$, the supply function $m_s$, and the number of active servers. Additionally, they can monitor the platform's assignment policies to ensure that their compensation agreements are upheld. \revcolor{We impose this full-information and monitoring assumption primarily for mathematical tractability: it keeps continuation values and participation constraints Markovian, allowing them to be written as functions of the current platform state rather than of each server's private history, beliefs, or learning process.} \revcolor{Equivalently, one can interpret the model as a setting in which the platform has the relevant information and can credibly commit to the long-term compensation and assignment rules, so servers do not need to infer these policies from their own histories.} An alternative approach, which we do not explore in this paper, would be to assume that servers can only monitor their own individual assignments and compensation over time. This could lead to different strategic behaviors and operational implications, making it an interesting direction for future research.
\vspace{0.3cm}




%% file: Tex/Static.tex
\section{A Market-Clearing Solution}\label{sec:market-clearing}

As mentioned in \cref{sec:literature}, much of the existing literature investigating network effects  on a platform's equilibrium size adopts a static market-clearing approach. Namely, the platform's size is characterized as an equilibrium —a fixed point that balances supply and demand— as if market forces act instantaneously, with little emphasis on understanding the growth process that leads to such an equilibrium. In other words, it is implicitly assumed that both sides of the market are present at time $t=0$ and the market clears immediately.

Within our model, we can recover this market-clearing solution as a limiting case where $m_k \to \infty$ for all $k \in \mathbb{N}_0$. That is, the rate at which new servers arrive to the platform grows unboundedly, effectively eliminating any gradual growth dynamics. In this asymptotic regime, the platform can effectively select the value of $\size$ at time $t=0$ to maximize its profit rate, $R(\size)$, net of the servers' compensation. With $\size$ active servers, this profit rate is given by \revcolor{$R(\size)=\lambda_\size p-r\,w\,\size,$
}
and the platform's optimal static market-clearing solution solves
\begin{equation}\label{eq:OptClearing} {\mathcal S}^*=\argmax_{\size \in \mathbb{N}_0} \{R(\size)\}.\end{equation}
Computing ${\mathcal S}^*$ requires specifying information about the demand function $\lambda_s$. For illustration purposes, throughout the paper, we will consider the following family of demand functions.

\begin{Defi}\label{def:MSL} {\sc Multi-Segment Logistic Demand Function (MSL):}  
We say that the market demand admits an MSL representation \footnote{Any increasing function can be approximated by a sum of step functions (see \cite{stein2009real}), which, in turn, are limits of logistic functions as $\xi  \rightarrow \infty$. Thus, MSL functions can approximate any demand function with arbitrary precision.} if there exist $N$ segments of consumers such that segment $n \in [N]$ is characterized by a non-negative triplet $(\Lambda_n, s_n, \xi_n)$. When the platform has $s$ active servers, customers in segment $n$ arrive requesting service according to a Poisson process with rate $\Lambda_n\,{\sf Logistic}(\xi_n\,(s-s_n))$, where $ {\sf Logistic}(x) = \exp(x)/(1+\exp(x))$ is the logistic function. Consequently, the platform's cumulative demand, as a function of its size $s$, follows a Poisson process with rate  
$$
\lambda_s = \sum_{n=1}^N \Lambda_n\,\frac{\exp(\xi_n\,(s-s_n))}{1+\exp(\xi_n\,(s-s_n))}.$$
\end{Defi}
\cref{fig:MSL_Demand} depicts the demand rate, $\lambda_s$ (left panel), and the profit rate, $R(s)$ (right panel), for a multi-segment logistic demand function with four segments. In this example, the platform's optimal static market-clearing size is ${\mathcal S}^* = 1038$ servers. Note that $R(s)$ is negative for small values of $s$, reflecting a scenario in which the platform must reach a minimum size to operate profitably.

\begin{figure}[h!]
\centering
    \input{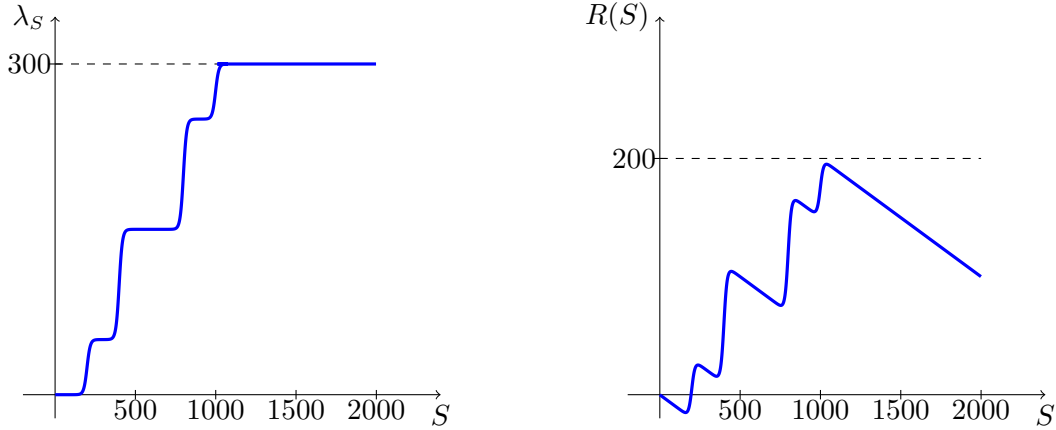}
   \caption{\footnotesize \sf Example of demand rate, $\lambda_s$ (left panel) and profit rate $R(s)$ (right panel) for a multi-segment logistic demand function as a function of number of active servers $S$. \\ \sc{Data}: $\Lambda_n=50\,(1, 2, 2, 1)$, $s_n=200\,(1, 2, 4, 5)$, $\xi_n=0.1\,(1, 1, 1, 1)$, $r=0.01$, $p=1$, and $w=10\,p$.}
   \label{fig:MSL_Demand}
\end{figure}

Later, we will compare the platform’s optimal growth level ${\mathcal S}^*$ in \eqref{eq:OptClearing}, derived under a static market-clearing equilibrium perspective, with a setting in which the platform can only gradually expand its supply side over time. As we will see, in this dynamic context, growing to the level ${\mathcal S}^*$ is generally suboptimal, and the platform instead prefers to operate at a smaller size (see \cref{cor:maxS} and the discussion that follows it).

%% file: Tex/Homogenous.tex
\section{Optimal Growth Trajectory}\label{sec:Main} 
\revcolor{
In this section, we consider the alternative regime in which platform growth is gradual and the time required to expand the server base is non-negligible. In this setting, the platform’s policy $\pi=(u,f,q,S)$ reduces to: (i) a {\em server admission policy}, $u=\{u_k \colon k \leq \size\}$ which regulates its growth rate and implicitly determines its maximum size, $\size=\min\{k \geq 0 \colon u_\size=0\}$; (ii) a {\em server compensation scheme}, $f=\{f_s(k) \colon s \leq k \leq \size\}$, which defines the rewards offered to servers; and  (iii) a {\em server assignment policy}, $q=\{q_s(k) \colon s \leq k \leq \size\}$, which controls how servers are assigned to arriving customers. }

We can also simplify the platform's optimization problem in \eqref{eq:platformBE} using the fact that ${\mathcal A}_t$ is a point-process with intensity $\demand{{\mathcal S}_t}$ (see \citealp[Chapter 2]{bremaud1981point}).
\begin{align}\label{eq:platformBENoCongestion}
J^\pi=\e\left[\int_0^\infty e^{-r\,t}\,F({\mathcal S}_{t-})\,\D {\mathcal A}_t\big| {\mathcal S}_0=0\right]=\e\left[\int_0^\infty e^{-r\,t}\,F({\mathcal S}_t)\,\demand{{\mathcal S}_t}\,\D t\big| {\mathcal S}_0=0\right]
\end{align}
Given an admissible policy $\pi \in \Pi$, we can compute the value of $J^\pi$ using backward induction on the number of active servers. Let us denote by $J^\pi_{k}$ the platform's {\em value function}, this is the  expected discounted payoff-to-go when there are $k$ active \revcolor{servers}. From \eqref{eq:platformBENoCongestion}, and the Markovian nature of ${\mathcal S}_t$ under an admissible policy, it follows that $J^\pi_{k}$ satisfies the following recursion:
\begin{align}\label{eq:Bellman}
    J^\pi_{k} & = \e\left[\int_{\tau_k}^{\infty} e^{-r\,(t-\tau_k)}\,F({\mathcal S}_t)\, \lambda_{{\mathcal S}_t}\,\D t\right] \nonumber \\
    & = \e\left[\int_{\tau_k}^{\tau_{k+1}} e^{-r\,(t-\tau_k)}\, F(k) \, \lambda_k\, \D t + \int_{\tau_{k+1}}^\infty e^{-r\,(t-\tau_k)}\,F({\mathcal S}_t)\, \lambda_{{\mathcal S}_t}\,\D t\right] \nonumber\\
    & = \e\left[1-e^{-r(\tau_{k+1}-\tau_k)}\right]\,{F(k)\,\lambda_k \over r}\,  + \e\left[e^{-r\,(\tau_{k+1}-\tau_k)}\int_{\tau_{k+1}}^\infty e^{-r\,(t-\tau_{k+1})}\,F({\mathcal S}_t)\, \lambda_{{\mathcal S}_t}\,\D t\right]\nonumber\\
& = \e\left[1-e^{-r(\tau_{k+1}-\tau_k)}\right]\,{F(k)\,\lambda_k \over r} +  \e\left[e^{-r\,(\tau_{k+1}-\tau_k)}\right] \, \e\left[\int_{\tau_{k+1}}^\infty e^{-r\,(t-\tau_{k+1})}\,F({\mathcal S}_t)\, \lambda_{{\mathcal S}_t}\,\D t\right] \nonumber \\
& = {r \over r+u_k\, m_k} \cdot {F(k)\, \lambda_k \over r} +{u_k\, m_k\over r +u_k\, m_k}\cdot J^\pi_{k+1},
\end{align}
where in the last equality we have used the fact that $\tau_{k+1}-\tau_k$ is exponentially distributed with rate $u_k\, m_k$, so $\e[\exp(-r(\tau_{k+1}-\tau_k)\revcolor{)}]=u_k\, m_k/(r+u_k\, m_k)$. The boundary condition is given by $ J^\pi_{\size+1}=0$.

To capture the dynamic individual rationality constraints of the servers, 
let $v^{\pi}_s(k)$ denote the value function of server $s$. This is the expected discounted payoff-to-go of server $s$ when there are $k \geq s$ active servers in the platform. Using a similar derivation as in \eqref{eq:Bellman}, we get that $v^{\pi}_s(k)$ satisfies the recursion:
\begin{align}\label{individualBE}
    v^{\pi}_s(k) =& {r \over r+u_k\, m_k} \cdot \frac{f_s(k)q_s(k)  \, \lambda_k}{r} + \frac{u_k\, m_k}{r+u_k\, m_k}\cdot v^{\pi}_s(k+1), \qquad k=s, \dots, {\size},
\end{align}
with boundary condition $v^{\pi}_s(\size+1)=0$ for all $s \in [\size]$. A policy $\pi$ satisfies dynamic individual rationality as described  in 
\eqref{eq:dynamicparticipation} if $ v^{\pi}_s(k) \geq w$ for all  \mbox{$s \leq k \leq \size$}.

\revcolor{To simplify the exposition, we introduce the following auxiliary variables:}
\begin{align}\label{Gamma definition}
    \mu_{k}=u_k\, m_k, \qquad \imm(k) = {F(k)\, \lambda_k \over r} \qquad \mbox{and} \qquad \gamma_s(k) = \frac{f_s(k)q_s(k)  \, \lambda_k}{r}.
\end{align}
We refer to $\mu= \{\mu_{k} \colon k \leq \size \}$ as the {\em server admission policy}. Under this policy, the number of active servers, $\mathcal{S}_t$, evolves as a point process with time-$t$ intensity $\mu_{\mathcal{S}_t}$. On the other hand, $\imm(k)$ and $\gamma_s(k)$ represent the expected discounted payoff-to-go for the platform and server $s$, respectively, if the platform had a size of $k$ active servers and did not admit any new servers, i.e., $\size = k$. 
We note that $\imm(k)$ and $\gamma_s(k)$  satisfy
$$ \imm(k) + \sum_{s=1}^k \gamma_s(k) = \frac{p\demand{k}}{r}.$$
\revcolor{In the base model, we restrict attention to platforms that cannot finance growth through subsidies or external fundraising. Accordingly, we impose an instantaneous budget-balance constraint requiring the platform's net intermediation revenue to remain non-negative at every platform size:
\begin{align}\label{eq:positiveCommision}
    \Gamma(k) \geq 0.
\end{align}
In \Cref{sec:extension}, we relax this requirement by allowing the platform to access limited external financing, thereby accommodating intermediate cases in which some temporary subsidization of growth is feasible.}

\revcolor{Combining all of these}, we derive the formulation of the platform's optimization problem in \eqref{eq:platopt}:
\begin{align*}\label{mainProgram}
\tag{$\mathcal{P}_{\texttt{Main}}$}
    \;&\; J^*=\sup_{\pi \in \Pi} \;J^{\pi}_0\\
    \text{subject to}\;&\; J^\pi_{k} = \frac{r}{r+\mu_{k}} \imm (k)+ \frac{\mu_{k}}{r+\mu_{k}}\,J^\pi_{k+1} &&\mbox{for}\quad k \leq \size \quad\mbox{and}\quad   J^\pi_{\size+1}=0 \\
    \;&\; v^{\pi}_s(k) = \frac{r}{r+\mu_{k}} \gamma_s (k) + \frac{\mu_{k}}{r+\mu_{k}}\,v^{\pi}_s(k+1) &&\mbox{for}\quad k \leq \size \quad\mbox{and}\quad v^{\pi}_s(\size+1)=0 \\
    \;&\; \imm(k) + \sum_{s=1}^k \gamma_s(k) = \frac{p\demand{k}}{r} &&\mbox{for} \quad  k \leq \size\\
    \;&\; v^{\pi}_s(k) \geq w &&\mbox{for} \quad  s \leq k \leq \size\\
    \:&\; 0 \leq \mu_{k} \leq m_k &&\mbox{for} \quad  0 \leq k \leq \size\\
    \;&\; \imm(k) ,\gamma_s(k) \geq 0 &&\mbox{for} \quad  s \leq k \leq \size.
\end{align*}

We tackle the solution to \eqref{mainProgram} by first deriving a reformulation that replaces all the dynamic individual rationality constraints of the servers by a single {\em dynamic collective rationality} (DCR) constraint. In this reformulation, an {\em aggregated policy} is represented by the tuple $\tilde{\pi} = (\mu, \imm, S)$, where $\imm = \{\imm(k) \colon k \in [\size]\}$, with $\imm(k)$ defined as in \eqref{Gamma definition}. Note that this policy specifies only an aggregate compensation scheme for all servers, according to the relationship $\sum_{s=1}^k f_s(k)\,q_s(k) = p - \imm(k)\, \frac{r}{\lambda_k}$.

We denote by $\widetilde{\Pi}$ the class of aggregated policies $\tilde{\pi} = (\mu, \imm, S)$. For each  policy $\pi = (\mu, f, q, S) \in \Pi$, there exists a corresponding aggregated policy $\tilde{\pi} = (\mu, \imm, S) \in \widetilde{\Pi}$, where $\imm(k) = \frac{\lambda_k}{r}\left(p - \sum_{s=1}^k f_s(k)q_s(k)\right)$ for all $k \in [\size]$. In this case, we say that $\pi$ generates $\tilde{\pi}$ or that $\tilde{\pi}$ is associated with $\pi$. 

In order to derive this alternative formulation, let us define servers' aggregate value function $V^{\tpi}_k $, which is the cumulative expected discounted compensation to all servers when the platform has $k$ active servers, 
\begin{align}\label{aggregateBE}
    V^{\tpi}_k = \frac{r}{r+\mu_{k}}\left(\frac{p\, \lambda_k}{r} - \imm (k) \right) + \frac{\mu_{k}}{r+\mu_{k}}\,V^{\tpi}_{k+1}, \qquad k\in[\size],
\end{align}
with boundary condition $V^{\pi}_{\size+1}=0$. \cref{lem:aggregate identity} establishes a relationship between the individual and cumulative expected discounted server compensations. The statement of this result makes use of the following definition:
\begin{equation}\label{eq:defalpha}
        \alpha_0=1 \qquad \mbox{and}\qquad \alpha_{k+1}={\mu_{k} \over r+\mu_{k}}\,\alpha_{k}= \prod_{\ell=0}^{k} \frac{\mu_{\ell}}{r+\mu_{\ell}} \qquad  0 \leq k \leq \size.
\end{equation}
We refer to $\{\alpha_k\colon k \in [S]\}$ as  {\em effective discount factors}. \revcolor{Under the admission rates $\{\mu_\ell\}$, $\alpha_k$ is the expected discount factor accumulated before the platform first reaches size $k$; equivalently, if $\tau_k$ denotes that hitting time, then $\alpha_k=\mathbb{E}[e^{-r\tau_k}]$.} Note that $\alpha_{\size+1} = 0$ because, by definition, $\mu_{\size} = 0$. Due to the correspondence between $\mu=\{\mu_{k}\colon k \in [S]\}$ and $\alpha=\{\alpha_k\colon k \in [S]\}$, it will be convenient to treat the effective discount factors as decision variables and denote by $\tilde{\pi} = (\alpha, \imm, S)$ an aggregate operating policy. 

\begin{lem}\label{lem:aggregate identity} Consider an operating policy  $\pi \in \Pi$ that generates an aggregated policy  $\tpi \in \widetilde{\Pi}$. Then, the value functions $v^\pi_s(k)$ and $V^{\tpi}_k$ satisfy:
 \begin{align*}
        V^{\tpi}_k = \sum_{s=1}^k v^{\pi}_s(k) + \sum_{s=k+1}^{\size} \frac{\alpha_s}{\alpha_k}\,v^{\pi}_s(s) \qquad  1 \leq k \leq \size.
    \end{align*}
\end{lem}
The proof of this and other results are relegated to \ref{apx:proofs}. It is worth noting that $V^{\tpi}_k > \sum_{s=1}^{k} v_s^{\pi}(k)$ for $k < \size$. This follows from the fact that $V^{\tpi}_k$ not only accounts for the compensation to the $k$ active servers but also for future compensations that will be allocated to servers that have not yet joined the platform.  The next result establishes a lower bound on the servers' cumulative expected discounted compensation $V^{\tpi}_k$ that guarantees the existence of a feasible a policy $\pi \in \Pi$ that generates $\tpi$. 

\begin{prop}\label{thm:N&S} Let $\pi\in \Pi$ be a policy that satisfies the dynamic individual rationality constraints  $v^{\pi}_s(k) \geq w$ for all $s \leq k \leq \size$ then its associated aggregated policy $\tpi \in \widetilde{\Pi}$
satisfies the following dynamic collective rationality constraint
    \begin{align}\label{iffCondition}\tag{DCR}
        V^{\tpi}_k \geq w \Bigg(k + \sum_{s=k+1}^{\size} \frac{\alpha_s}{\alpha_k}\Bigg)\quad \mbox{for all}\quad k \in [\size].
    \end{align}
Conversely, if $\tpi \in \widetilde{\Pi}$ is an aggregated policy that satisfies \eqref{iffCondition} then there exists a policy $\pi\in \Pi$ that generates $\tpi$ and satisfies the dynamic individual rationality constraints  $v^{\pi}_s(k) \geq w$ for all $s \leq k \leq \size$. 
\end{prop}
According to \cref{thm:N&S}, as long as the platform chooses a policy that, in aggregate, ensures the pool of active servers receives compensation $V^{\tpi}_k$ that satisfies \eqref{iffCondition}, then there exists a feasible allocation for each individual server that satisfies their individual rationality constraint and vice versa. 

The final step to obtain the aforementioned aggregated reformulation is to note that two operating policies $\pi_1, \pi_2 \in \Pi$ that generate the same aggregated policy $\tpi\in \widetilde{\Pi}$ are revenue equivalent for the platform, that is, $J^{\pi_1}=J^{\pi_2}$. As a result, it is sufficient for the platform to optimize over the class $\widetilde{\Pi}$ of aggregated operating policies by solving:
\begin{align*}\label{mainProgramAgg}
\tag{$\mathcal{P}_{\texttt{Agg}}$}
    \;&\; J^*= \sup_{
\tilde{\pi} \in \widetilde{\Pi}} \;J^{\tpi}_0\\
    \text{subject to}\;&\; J^{\tpi}_k = \left(1-\frac{\alpha_{k+1}}{\alpha_k}\right)\imm (k) + \frac{\alpha_{k+1}}{\alpha_k}\,J^{\tpi}_{k+1} && \mbox{for}\quad k \leq \size \quad\mbox{and}\quad   J^{\tpi}_{\size+1}=0 \\
    \:&\; V^{\tpi}_k = \left(1-\frac{\alpha_{k+1}}{\alpha_k}\right) \left(\frac{p\lambda_k}{r}-\imm (k) \right) + \frac{\alpha_{k+1}}{\alpha_k}\,V^{\tpi}_{k+1} && \mbox{for}\quad k \leq \size \quad\mbox{and}\quad \revcolor{V^{\tpi}_{\size+1}}=0\\
    \;&\;  V^{\tpi}_k \geq w \Big(k + \sum_{s=k+1}^{\size} \frac{\alpha_s}{\alpha_k}\Big) && \mbox{for}\quad k  \leq \size\\
    \;&\;  0 \leq \alpha_{k+1} \leq \frac{m_k}{r+m_k}\alpha_k \quad \mbox{and} \quad 0 \leq \imm(k) \leq \frac{p\lambda_k}{r} && \mbox{for}\quad k \leq \size.
\end{align*}

\subsection{Optimal Policy}\label{sec:optPolicy}
We tackle the solution to \eqref{mainProgramAgg} in two steps. First, in \cref{lem:relax}, we consider a relaxed version by removing the non-negativity constraint $\imm(k) \geq 0$ in \eqref{eq:positiveCommision} and reduce it to a one-dimensional optimization problem. \revcolor{Removing this constraint can be interpreted as allowing the platform to finance negative intermediation revenue at intermediate platform sizes through unrestricted subsidies or external funds. We use this financially unconstrained formulation as an upper bound for the base model. Surprisingly,} in \cref{thm:original&relaxed}, we show that there exists a feasible solution to the original problem \eqref{mainProgramAgg} that attains the same objective value as the relaxed formulation. As a result, this feasible solution is optimal for \eqref{mainProgramAgg}.

\begin{lem}\label{lem:relax} {\rm \bf (Relaxed Program)} Consider the program \eqref{mainProgramAgg} without the constraint $\imm(k) \geq 0$. This relaxed problem admits an optimal solution
$\hat{\pi}=(\hat{\alpha},\hat{\imm},\hat{\size})$ such that: 
\begin{itemize}
\item[\rm (i)] $\hat{\alpha}_0=1$,  $\hat{\alpha}_{k+1}={m_k \over r+m_k}\,\hat{\alpha}_k$ for $0 \leq k \leq \hat{\size}-1$, and $\hat{\alpha}_{\hat{\size}+1}=0$;
\item[\rm (ii)] $\hat{\imm}(k)=\Big(\frac{p\lambda_k}{r}-wk\Big)$  for $0 \leq k \leq \hat{\size}$. 
\item[\rm (iii)] The value of $\hat{\size}$ solves the optimization problem:
\begin{equation} \label{mainProgramAgg3}
\tag{$\mathcal{P}_{\texttt{relax}}$}
\widehat{J} :=J^{\hat{\pi}}_0= \sup_{\size} \; \sum_{\ell=1}^{\size} (\hat{\alpha}_{\ell}-\hat{\alpha}_{\ell+1})\hat{\imm}(\ell).
\end{equation}
\end{itemize}
\end{lem}
It is worth noting that condition (i) in \cref{lem:relax} implies that, in this relaxed formulation, there exists an optimal strategy under which the platform grows at its maximum rate (i.e., recruiting every arriving server) until it reaches an optimal size, at which point it stops accepting new servers. This maximum size is determined by solving the one-dimensional problem \eqref{mainProgramAgg3}.

\revcolor{As mentioned earlier, the optimal} strategy $\hat{\pi}=(\hat{\alpha}, \hat{\imm},\hat{\size})$ from \cref{lem:relax} may not be feasible for the original problem \eqref{mainProgramAgg}, as it is possible that $\hat{\imm}(k) < 0$ for some $k \in [\hat{\size}]$. Consequently, we have $J^* \leq \widehat{J}$. \revcolor{Perhaps unexpectedly,} however, \cref{thm:original&relaxed} demonstrates that, in fact, $J^* = \widehat{J}$.
\begin{thm}\label{thm:original&relaxed}
The programs \eqref{mainProgramAgg} and \eqref{mainProgramAgg3} have the same optimal value, $J^* = \widehat{J}$. Moreover, there exists an optimal solution $\pi^*=(\alpha^*, \imm^*,\size^*) \in \widetilde{\Pi}$ for \eqref{mainProgramAgg} such that $\size^*=\hat{\size}$ and $\alpha^*_k = \hat{\alpha}_k$ for $k \in [\size^*]$, where $\hat{\alpha}_k$ and $\hat{\size}$ are the solutions to the relaxed program in \cref{lem:relax}.
\end{thm}   
 


\revcolor{This equivalence constitutes one of the central takeaways of this paper, revealing a striking structural property of optimal platform design. This result demonstrates that under a seniority-based assignment policy, unlimited access to external capital yields no marginal benefit, allowing a highly restricted platform to fully recover the unconstrained optimal value. In other words, even though the relaxed program removes the non-negativity constraint on the cumulative intermediation revenue $\imm(k)$, effectively granting the platform an infinite credit line to run deficits and fund early-stage growth, the unconstrained optimal solution can be perfectly implemented without external subsidy. The platform is capable of financing its entire growth trajectory organically through its own generated intermediation revenue, provided it designs the joint admission and compensation policies optimally.

The previous result also reveals another important point: rather than being a constraint that limits the platform's growth, the forward-looking nature of servers is an essential element of the mechanism that sustains the optimal growth path. If participation had to be supported entirely through contemporaneous compensation, then early low-demand states would generally require immediate payments that the platform cannot finance under the instantaneous budget-balance constraint. Theorem~1 shows that, under the baseline assumptions, the platform can instead rely on servers' valuation of credible future payments. By shifting part of early servers' compensation to later, higher-demand states, seniority-based assignment converts future platform cash flows into current participation incentives. This intertemporal channel, sustained by servers' forward-looking behavior, allows the platform to attain the same value as in the financially unconstrained relaxation while still satisfying the instantaneous budget-balance constraint.

It is also worth noting that while this equivalence holds under symmetric discounting between the platform and the servers, the benefits of seniority-based compensation are not confined to this setting. In \Cref{sec:extension}, we extend the model to accommodate unequal discount rates. We demonstrate that the structural benefits of a seniority-based policy persist even when servers possess higher discount rates, whether driven by a higher intrinsic time-value of money, liquidity constraints, or perceived platform risk factors. By showing that internal financing can still thrive despite a discounting asymmetry that inherently works against long-term contracts, we further establish the robustness of the seniority mechanism as a tool for platform growth.

Also, it} follows from the previous result that the platform should aim to grow as rapidly as possible until reaching the optimal size, $\size^*$, determined by the solution to \eqref{mainProgramAgg3}. Moreover, this optimal size can be identified within the set of local maxima of the platform's instantaneous profit rate function.

\begin{cor}\label{cor:maxS}  The platform's optimal size $\size^*$ is a local maximum of the function $R(s) = p \, \lambda_s - r \, w \, s$ and satisfies $\size^* \leq {\mathcal S}^{*}=\arg\max\limits_{s \in \mathbb{N}} \{ R(s) \}$. If $R(s)$ has a unique local maximum, then $\size^* ={\mathcal S}^{*}$.
\end{cor}

To build some intuition, note that $p \lambda_s$ represents the platform's instantaneous revenue rate when there are $s$ active servers, while $r w s$ is the minimum cumulative rate of payment required to keep these $s$ servers active on the platform. Therefore, $R(s)=p \lambda_s - r w s$ corresponds to the maximum possible profit rate that the platform can achieve with $s$ active servers. It follows that if the platform were able to instantaneously attract a desired number of servers at time $t=0$, it would aim to recruit exactly ${\mathcal S}^*=\arg\max\limits_{s \in \mathbb{N}} \{ R(s) \}$ servers (as in \cref{sec:market-clearing}). \revcolor{However, because scaling the network is a gradual process, the platform is forced to spend time at smaller, less profitable sizes before reaching maturity. Due to the time value of money, the long-term gains from expanding to the static market-clearing optimum ${\mathcal S}^*$ are discounted, while the platform must still bear the costs related with passing through lower-revenue intermediate states. As a result, it may be optimal for the platform to limit its target size to $\size^* \leq {\mathcal S}^*$.}

\cref{cor:maxS} establishes that $\size^*$ is always a local maximum for $R(s)$ and that the optimal size ${\mathcal S}^*$ can be achieved if this local maximum is unique, for example, when $\lambda_s$ is strictly concave. In general, however, it is possible that the platform optimal size $\size^*< {\mathcal S}^*$. In such cases, the demand and supply functions $\lambda_s$ and $m_{s}$ restrict the platform's growth to a suboptimal level. Let us illustrate this point with a concrete example.
\begin{exm}\label{ex:rapid} Suppose $m_s=m$ for all $s \geq 0$ while $\lambda_s$ is a piecewise constant MSL function ($\xi_1=\xi_2=\infty$) with two segments: $\lambda_s=\sum_{i=0}^2 \Lambda_i\, \Indi(s_i \leq s <s_{i+1})$ with $0=s_0 < s_1< s_2$ and $ 0=\Lambda_0<\Lambda_1 <\Lambda_2$. In this setting, $S^*$ can take one of three possible values: $S^* \in \{0,s_1,s_2\}$.
The platform's expected discounted payoff in these three cases equals:
$$J^*_0=0, \quad J^*_{s_1}=p\,\Lambda_1\,\alpha_1^{s_1}-r\,w\,{\alpha_1-\alpha_1^{s_1+1} \over 1-\alpha_1}\quad \mbox{and}\quad J^*_{s_2}=J^*_{s_1}+p\,(\Lambda_2-\Lambda_1)\,\alpha_1^{s_2}-r\,w\,{\alpha_1^{s_1+1}-\alpha_1^{s_2+1} \over 1-\alpha_1},$$
with $\alpha_1=m/(r+m)$. It follows that  when
$${p\,\Lambda_1 \over r\,w} \geq {\alpha_1-\alpha_1^{s_1+1} \over \alpha_1^{s_1}\,(1-\alpha_1)} \quad \mbox{and} \quad s_2-s_1 < {p\,(\Lambda_2-\Lambda_1) \over r\,w} < {\alpha_1^{s_1+1}-\alpha_1^{s_2+1} \over \alpha_1^{s_2}\,(1-\alpha_1)}.$$
$J^*_{s_1} \geq \{J^*_0,J^*_{s_2}\}$ and $p\,\Lambda_2-r\,w\,s_2>p\,\Lambda_1-r\,w\,s_1$. Thus, the platform's optimal size 
is $S^*=s_1$ despite the fact that the instantaneous profit rate $p \lambda_s - r w s$ is maximized at $S=s_2$. \Cref{Fig:example} depicts the optimal growth size $S^*$ as function of $\Lambda_1$ and $\Lambda_2$ (left panel) or  $s_1$ and $s_2$ (right panel). $\Box$
\end{exm}

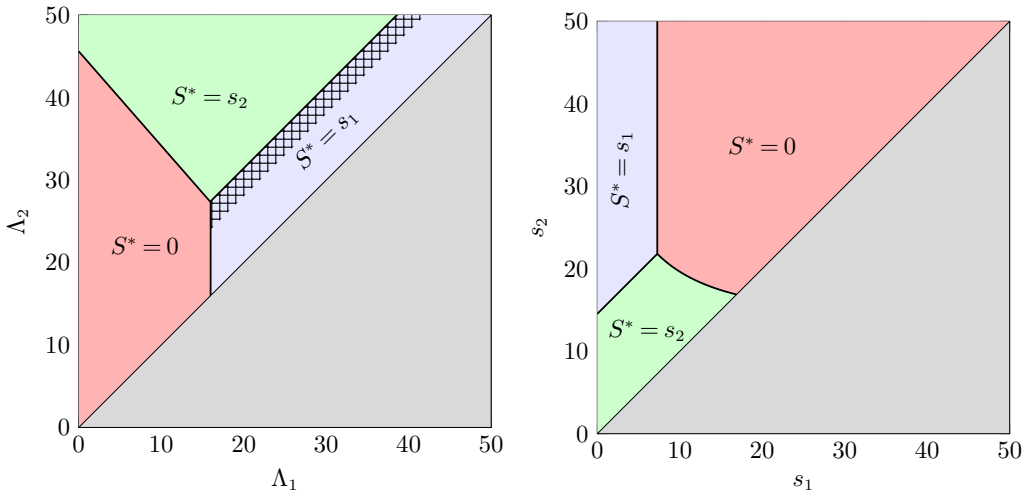
\begin{figure}[h!]
\centering
\scalebox{0.85}{
    \input{Figures/Lambda1Lambda2}
    \input{Figures/s1s2}
}
\caption{\footnotesize \sf Optimal size across different parameters. The blue shaded region (left) represents the scenario where the immediate reward of state $s_2$ is higher, yet stopping growth at $s_1$ remains optimal due to growth costs.  {\sc Data}: $r=0.1,$ $p=1,w=10, m =1$, and $\xi_1=\xi_2 = \infty$. In the left panel, $s_1 = 10$ and $s_2 = 18$, while in the right panel, $\Lambda_1 = 10$ and $\Lambda_2 = 40$. }
\label{Fig:example}
\end{figure}

It is important to note that while \cref{thm:original&relaxed} shows that the platform’s optimal growth strategy is to expand as quickly as possible under an optimal compensation scheme, this result may not hold if compensation is restricted to a particular class. In \cref{sec:benchmark}, we examine the case where the platform operates under a fixed commission fee $F$ and show that, in this setting, the platform may need to temporarily halt the recruitment of new servers before reaching its maximum size. 

Intuitively, the optimality conditions guarantee that the platform does not overcompensate any server upon joining, i.e., $v^*_s(s) = w$. Moreover, as demonstrated in the following proposition, enforcing DIR constraints does not reduce the platform’s payoff. In fact, the highest payoff the platform can achieve by removing the DIR constraints while retaining only the SIR constraints is no greater than \eqref{mainProgramAgg}.

\begin{prop}\label{prop:SIRvsDIR}
    The optimal objective the platform can achieve by relaxing the DIR constraints is no more than $J^*$ in \eqref{mainProgramAgg}.
\end{prop}

It should be emphasized that this proposition does not imply that the feasible policies of the two problems are identical. While the feasible region expands when DIR constraints are removed, the proposition establishes that the optimal objective remains unchanged. Specifically, our proof demonstrates that the optimal growth strategy under only the SIR constraints allows for a redistribution of compensations in a way that also satisfies the DIR constraints, without hurting the platform’s value $J^*$.

Finally, \Cref{thm:original&relaxed} establishes that, for an optimal policy $\tilde{\pi}^* = (\alpha^*, \imm^*, S^*)$ with a corresponding expected payoff $J^*$ for the platform, the values of $\alpha^*$, $\size^*$, and $J^*$ are  uniquely determined based on the result in \Cref{lem:relax}. However, the proof of \Cref{thm:original&relaxed} reveals that, in general, there are infinitely many possible values for the platform's optimal intermediation fee, $\imm^*$, as long as they satisfy the inter-temporal incentive constraints. Specifically, the set of optimal values for $\imm^*$ is characterized by the following polytope:
\begin{align*}
\Omega = \bigg\{\imm = \left(\imm(k) \in \left[0,\frac{p\lambda_k}{r}\right] \colon k \in [S^*]\right) \;\colon \;  
& \sum_{\ell = k}^{\size^*} (\alpha^*_{\ell} - \alpha^*_{\ell+1})\left(\frac{p\lambda_{\ell}}{r}-\imm (\ell) \right) \geq w \Big(k\alpha^*_k + \sum_{\ell = k+1}^{\size^*} \alpha^*_{\ell}\Big), \quad  k \in [S^*]\\
& \mbox{and}\quad J^* =\sum_{\ell = 0}^{\size^*} (\alpha^*_{\ell} - \alpha^*_{\ell+1})\imm(\ell) \bigg\}.
\end{align*}

This provides the platform with additional flexibility to choose a compensation scheme from $\Omega$ based on a secondary criteria beyond merely maximizing its expected payoff, $J^*$. For example, during its growth phase, the platform might prioritize stabilizing its cash flows or it could focus on minimizing its liabilities. In \ref{sec:Policy Refinements}, we show how these policy refinements influence the platform’s optimal strategies.

%% file: Figures/Lambda1Lambda2.tex
\begin{tikzpicture}
  \begin{axis}[
      width=8cm,
      height=8cm,
      xlabel={$ \Lambda_1 $},
      ylabel={$ \Lambda_2 $},
      legend pos=south east,
      xmin=0, xmax=50,
      ymin=0, ymax=50,
      domain=0:50,
      samples=500,
    ]

    \def\m{1}
    \def\r{0.1}
    \def\p{1}
    \def\w{10}
    \def\sone{10}
    \def\stwo{18}
    \def\alpha{\m / (\m + \r)}
    \def\xStar{\r*\w/\p * (\alpha-(\alpha)^(\sone+1))/(\alpha)^\sone/(1-\alpha)}
    \def\yStar{\xStar + (\r * \w / \p) * ((\alpha)^(\sone + 1) - (\alpha)^(\stwo + 1)) / ((\alpha)^\stwo * (1-\alpha))}
    \pgfmathsetmacro{\xStarNumber}{\xStar-0.05}

    \addplot [
      thick,
      opacity=0,
      domain=0:50,
      samples=500,
      name path=A,
    ]coordinates {(\xStar, -100) (\xStar-0.05, 100)};
    \addplot [
      thick,
      opacity=1,
      domain=0:50,
      samples=500,
      name path=APrime,
    ]coordinates {(\xStar, -100) (\xStar+0.05, 27.5)};

    \addplot [
      thick,
      opacity=0,
      domain=0:50,
      samples=500,
      name path=fifty,
    ]{50};
    \addplot [
      thick,
      opacity=0,
      domain=0:50,
      samples=500,
      name path=zero,
    ]{0};
    

    \addplot [
      thick,
      opacity=1,
      domain=\xStar:50,
      samples=500,
      name path=B,
    ] {x + (\r * \w / \p) * ((\alpha)^(\sone + 1) - (\alpha)^(\stwo + 1)) / ((\alpha)^\stwo * (1-\alpha))};

    \addplot [
      thick,
      opacity=1,
      domain=0:\xStar,
      samples=500,
      name path=C,
    ] {x * (1 - (\alpha)^(\sone - \stwo)) + (\r * \w / \p) * (\alpha - (\alpha)^(\stwo + 1)) / ((\alpha)^\stwo * (1-\alpha))};

\addplot [
      thick,
      opacity=0,
      domain=0:50,
      samples=500,
      name path=D,
    ] {x + (\r * \w / \p) * (\stwo-\sone};

\addplot [
    fill=blue!10
    ] fill between [
      of=D and B,
      soft clip={domain=\xStarNumber:50},
    ];


\addplot [
    pattern=grid,
    pattern color=black
    ] fill between [
      of=D and B,
      soft clip={domain=\xStarNumber:50},
    ];

\addplot [
      fill=blue!10
    ] fill between [
      of=D and zero,
      soft clip={domain=\xStarNumber:50},
    ];

\addplot [
      fill=green!20
    ] fill between [
      of=C and fifty,
      soft clip={domain=0:\xStarNumber+0.1},
    ];

\addplot [
      fill=green!20
    ] fill between [
      of=B and fifty,
      soft clip={domain=\xStarNumber:50},
    ];

    \addplot [fill=gray!30, domain=0:50] {x} \closedcycle;

\addplot [
      fill=red!30
    ] fill between [
      of=C and zero,
      soft clip={domain=0:\xStarNumber+0.1},
    ];

    \node[rotate=45] at (axis cs:30,35) {$S^*=s_1$}; 
    \node at (axis cs:\xStarNumber,40) {$S^*=s_2$}; 
    \node at (axis cs:8,22) {$S^*=0$};

  \end{axis}
\end{tikzpicture}

%% file: Figures/s1s2.tex
\begin{tikzpicture}
    \begin{axis}[
      width=8cm,
      height=8cm,
      xlabel={$ s_1 $},
      ylabel={$ s_2 $},
      legend pos=south east,
      xmin=0, xmax=50,
      ymin=0, ymax=50,
      domain=0:10,
      samples=500,
    ]

    \def\m{1}
    \def\r{0.1}
    \def\p{1}
    \def\w{10}
    \def\lOne{10}
    \def\lTwo{40}
    \def\alpha{\m / (\m + \r)}
    \def\xStar{((ln(\r * \w * \alpha/(1-\alpha)) - ln(\p * \lOne + \r * \w * \alpha / (1-\alpha)))/ln(\alpha)}
    
    \pgfmathsetmacro{\xStarNumber}{\xStar-0.05}

    \addplot [fill=gray!30, domain=0:50] {x} \closedcycle;

    \addplot [
      thick,
      opacity=1,
      domain=-0:50,
      samples=500,
      name path=A,
    ]coordinates {(\xStar,21.7) (\xStar,50)};

    \addplot [
      thick,
      opacity=1,
      domain=0:\xStar,
      samples=500,
      name path=B,
    ] {x+1-ln(\alpha + \p /\r / \w * (\lTwo - \lOne)*(1-\alpha))/ln(\alpha)};

    \addplot [
      thick,
      opacity=1,
      domain=\xStar:16.85,
      samples=500,
      name path=C,
    ] {(ln(\r * \w * \alpha /(1-\alpha) - \p * \lOne * (\alpha)^x)-ln(\p*(\lTwo-\lOne) + \alpha * \r * \w /(1-\alpha)))/ln(\alpha)};

    \addplot [
      thick,
      opacity=0,
      domain=0:50,
      samples=500,
      name path=fifty,
    ] {50};

    \addplot [
      thick,
      opacity=0,
      domain=0:50,
      samples=500,
      name path=zero,
    ] {0};

    \addplot [
      fill=red!30
    ] fill between [
      of=C and fifty,
      soft clip={domain=\xStarNumber:50},
    ];

    \addplot [
      fill=blue!10
    ] fill between [
      of=B and fifty,
      soft clip={domain=0:\xStarNumber},
    ];

    \addplot [
      fill=green!20
    ] fill between [
      of=B and {zero},
      soft clip={domain=0:\xStarNumber+0.1},
    ];

    \addplot [
      fill=green!20
    ] fill between [
      of=C and zero,
      soft clip={domain=\xStarNumber:50},
    ];

    \node[rotate = 90] at (axis cs:3,32) {$S^*=s_1$};
    \node at (axis cs:6,12.5) {$S^*=s_2$};
    \node at (axis cs:20,35) {$S^*=0$};

    \end{axis}
\end{tikzpicture}

%% file: Tex/MicroLevel.tex
\subsection{Seniority-based Servers' Compensation}
\label{sec:fairness} \Cref{thm:N&S} enabled us to analyze the platform's problem at an aggregate level, abstracting away the details of individual server compensation. In this section, we turn our attention to the design of an optimal server compensation scheme, $\gamma_s^*(k) = \{\gamma^*_s : s \in [\size]\}$ as there may exist multiple valid configurations for them as well as for $(f^*,q^*)$. Here, we aim to characterize the \say{fairest} server compensation scheme. Our notion of fairness is rooted on the work of John \revcolor{Rawls} on the theory of justice (\citealp{rawls1971atheory}), which has received significant attention, particularly in the field of algorithmic fairness. Rawlsian fairness, prioritizes improving the well-being of the least advantaged individuals in a system, ensuring that no participant is disproportionately disadvantaged. In mathematical terms, it means maximizing the minimum payoff among all individuals.

It is worth noting that in our formulation, particularly the \eqref{eq:dynamicparticipation} constraint, already inherently imposes a degree of fairness. Specifically, it guarantees a minimum (Rawlsian) fairness level of $w$; an assignment in which one server handles many customers while another remains idle—thus receiving less than $w$ in compensation—would violate the \eqref{eq:dynamicparticipation} constraint for the idle server.



As shown in \Cref{lem:relax}, under the optimal aggregate policy $\tpi^*$, DCR constraint is tight for $k=1$ which  implies that for all policies $\pi$ generating $\tpi^*$, we must have $v_s^{\pi}(s) = w$ for all $s$. This further implies:
\begin{align}\label{piNecCondi}
    \sum_{s=1}^k v_s^{\pi}(k) =
    V_k^{\tpi^*} - \sum_{\ell=k+1}^\size \frac{\alpha_{\ell}^*}{\alpha_{k}^*}v_{\ell}^{\pi}(\ell) = V_k^{\tpi^*} - w\sum_{\ell=k+1}^\size \frac{\alpha_{\ell}^*}{\alpha_{k}^*}
\end{align}
Thus, the total future compensation $\sum_{s=1}^k v_s(k)$ is determined by the aggregate policy $\tpi^*$ and the policy $\pi$ is simply distributing the fixed aggregate pool of discounted compensations among the servers. 
Consequently, our notion of fairness involves determining a distribution of the aggregate pool among the servers given the feasibility constraints.

\begin{Defi}{\sc Leximin fairness:}\label{def: fairness}
    A policy $\pi$, which generates the aggregate policy $\tilde{\pi}^*$, is said to be \emph{$k$-leximin-fair} if, for any other policy $\pi'$ that also generates $\tilde{\pi}^*$, there exists an index $s$ such that:
    \begin{align*}
        v_{\rho_{\ell}}^{\pi}(k) = v_{\rho_{\ell}'}^{\pi'}(k) \quad \forall \ell < s, \qquad \text{and} \qquad  v_{\rho_{s}}^{\pi}(k) > v_{\rho_{s}'}^{\pi'}(k),
    \end{align*}
    Here, $\rho$ and $\rho'$ represent the orderings of the values such that: $$v_{\rho_{k}}^{\pi}(k) \geq \dots \geq v_{\rho_{1}}^{\pi}(k) \quad \text{and} \quad v_{\rho_{k}'}^{\pi'}(k) \geq \dots \geq v_{\rho_{1}'}^{\pi'}(k)$$
\end{Defi}

Leximin fairness (\citealp{sen2017collective}) is a generalization of Rawlsian fairness and follows the lexicographic max-min principle, which prioritizes improving the compensation of the worst-off server first. If multiple allocations achieve the same minimum compensation, the next priority is to maximize the second-lowest compensation, then the third-lowest, and so on, until no further improvement is possible. This ensures a strong form of fairness by iteratively balancing compensation in favor of the least advantaged servers while maintaining feasibility. Although \Cref{def: fairness} implies a partial ordering among all policies, we will demonstrate that a single policy exists that is the \emph{fairest}—it is $k$-leximin-fair for all $k$ within the class of policies that share the same aggregate compensation $\tpi^*$.

For each $k\in[S]$, let $v_s^*(k)$ and $v_s^*(k+1)$ be the payoffs the $k$-leximin-fair and $(k+1)$-leximin-fair policy implement. In order for a single policy $\pi^*$ to be capable of implementing these payoffs, it must set $$\gamma_s^{\pi^*}(k) = \frac{\alpha_k^*}{\alpha_k^* - \alpha_{k+1}^*}v_s^*(k)-\frac{\alpha_{k+1}^*}{\alpha_k^* - \alpha_{k+1}^*}v_s^*(k+1).$$ Therefore, to have a feasible solution, we need to make sure $\gamma_s^{\pi^*}(k)$ are all positive or equivalently for all $s$, $v_s^*(k) \geq \frac{\alpha_{k+1}^*}{\alpha_k^*}v^*_s(k+1)$. We will show this in the following lemma.

\begin{lem}\label{lemma: microAssignment}
    The $k$-leximin-fair and $(k+1)$-leximin-fair assignments have the property $$v_s^*(k) \geq \frac{\alpha_{k+1}}{\alpha_k}v^*_s(k+1).$$ 
\end{lem}
\noindent Finally, to get the policy $\pi^*$ we can set:
\begin{align*}
    f_s(k) = \frac{\frac{p\lambda_k}{r} - \imm(k)}{\frac{p\lambda_k}{r}}\,p , \qquad q_s(k)= 
    \frac{\gamma_s^{\pi^*}(k)}{\frac{p\lambda_k}{r} - \imm(k)}.
\end{align*}
As shown in the proof of \Cref{lemma: microAssignment}, the solution to the fairness problem is $v_s^*(k) = \min\left(C_k,\frac{\alpha^*_s}{\alpha^*_k}w\right)$
for a fixed $C_k$. This exhibits a simple structure: for each $k$, there exists a threshold $s_k \leq k$ such that for all $s < s_k$, $v_s^*(k) = C_k$, and for all $s \geq s_k$, $\revcolor{v_s^*(k)} = \frac{\alpha_s^*}{\alpha_k^*}w$. This assignment has a clear and meaningful interpretation: it distributes compensation equally among all servers, except possibly the junior ones --those who join the platform toward the end of its growth process. This is because paying a junior server $s$ beyond $\frac{\alpha_s}{\alpha_k}w$ would lead to compensating them above $w$ upon their arrival, which would violate the optimality of the solution. 

This observation becomes more evident in scenarios when the platform experiences a sudden increase in demand, and distributing it equally among all servers is neither optimal nor fair, as senior servers should be compensated for their contributions during periods of low demand. Our analysis reveals that achieving efficient and fair compensation among servers inherently requires the platform to apply a form of \say{discrimination} based on seniority. While the term discrimination often carries a negative sentiment, in this context, it refers to a structured differentiation that ensures fairness in the long run. By prioritizing senior servers, the platform acknowledges the potential undercompensation these servers experienced during the platform's early stages, when market size and demand were limited. Rather than violating fairness principles, this approach reinforces them by addressing disparities that would otherwise favor servers who joined during the platform's more mature and stable phases. It ensures that early adopters are recognized and rewarded for their contributions during the platform's initial growth phase, fostering a policy that values their past and ongoing participation.

\revcolor{The following example illustrates the performance and outcomes of the seniority-based assignment policy. For a complete algorithmic representation of this mechanism, see \ref{sec:taskALG}.}

\begin{exm}
Consider a market with a maximum size of $\size=3$ and parameters $p = 1, w = 10,r = 0.1$. The demand rates are given as $\lambda_1 = 1$, $\lambda_2 = 2$, and $\lambda_3 = 4$, while the aggregate policy is set such that $\alpha_1 = 1,\alpha_2 = \frac{1}{2},\alpha_3 = \frac{1}{4}$ and $\imm(1) = \$1.25, \imm(2) = \$2.5, \imm(3) = \$ 5$, i.e. $F = \frac{1}{8}$. With these parameter choices, all DCR constraints are satisfied, including at \revcolor{state} 1, where the constraint is tight.

To satisfy the SIR of the third server in the final state, the policy assigns $\gamma_3(3) = \$10$ and divides the remaining amount equally between the first two servers, setting $\gamma_1(3) = \gamma_2(3) = \$12.5$. In the second state, to satisfy the SIR of the second server, the policy assigns $\gamma_2(2) = \$7.5$, leaving the remainder for the first server, $\gamma_1(2) = \$10$. Finally, in the first state, all of the allocation goes to the first server, with $\gamma_1(1) = \$8.75$. Refer to \Cref{Fig:seniority} (left panel) for a graphical representation of the output of the assignment policy. 

\revcolor{Additionally, the assignment outcomes for alternative aggregate policies (see \ref{sec:Policy Refinements}) are depicted alongside. Specifically, the middle panel displays the smooth-payoff policy (the solution to \eqref{eq:SmoothJ}), where the platform fixes the cumulative intermediation revenue at $\imm = \$2.5$ to smooth out returns. The right panel then depicts the minimum-liability policy (the solution to \eqref{eq:SmoothV}), a scenario where early payments are fully distributed to the servers, leaving the platform to retain revenue exclusively during its mature phase.}

\end{exm}


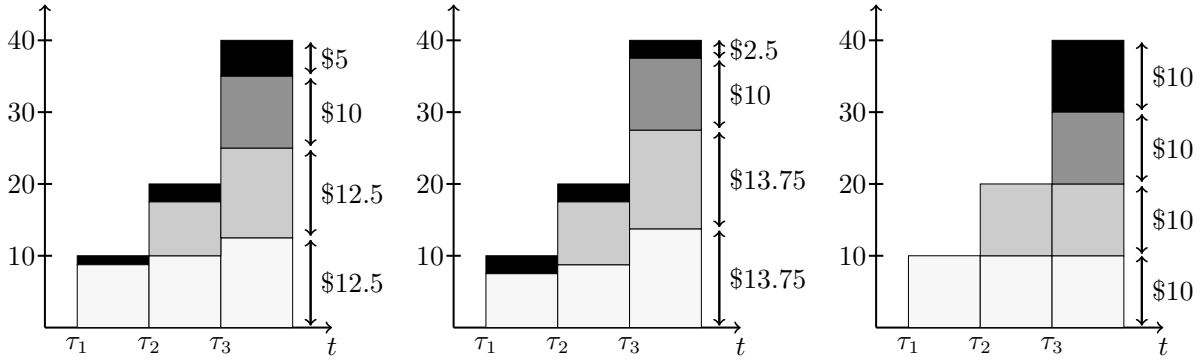
\begin{figure}[h!]
\centering
\scalebox{0.95}{
\input{Figures/seniorityBasedGraph}
\input{Figures/seniroityBasedSmooth}
\input{Figures/seniorityBasedLiability}
}
\caption{\footnotesize \sf Task allocations under the seniority-based policy for $\size = 3$: The y-axis measures the cumulative discounted payoff-to-go. \revcolor{Box heights represent individual payoffs: black boxes denote the platform's share, while gray boxes denote servers, with lighter shades indicating greater seniority.}
The left panel illustrates the case with the given $\imm$. The middle and right panels show the optimized $\imm$ for maximizing platform payoff smoothness and minimizing its liability, respectively (see \ref{sec:Policy Refinements}).
}  
\label{Fig:seniority}
\end{figure}

%% file: Figures/seniorityBasedGraph.tex
\begin{tikzpicture}

\def\yScale{1}
\def\xshift{0.55}
\draw[thick,->] (0,0) -- (4,0) node[below] {$t$};
\draw[thick,->] (0,0) -- (0,4.5*\yScale) node[left] {};

\node[below] at (1-\xshift,0) {$\tau_1$};
\node[below] at (2-\xshift,0) {$\tau_2$};
\node[below] at (3-\xshift,0) {$\tau_3$};

\draw[thick] (-0.1,\yScale) -- (0.1,\yScale);
\node[left] at (0,\yScale) {$10$};
\draw[thick] (-0.1,2*\yScale) -- (0.1,2*\yScale);
\node[left] at (0,2*\yScale) {$20$};
\draw[thick] (-0.1,3*\yScale) -- (0.1,3*\yScale);
\node[left] at (0,3*\yScale) {$30$};
\draw[thick] (-0.1,4*\yScale) -- (0.1,4*\yScale);
\node[left] at (0,4*\yScale) {$40$};

\filldraw[black, draw=black] (1-\xshift,7/8*\yScale) rectangle (2-\xshift,1*\yScale); 
\filldraw[black, draw=black] (2-\xshift,1.75*\yScale) rectangle (3-\xshift,2*\yScale); 
\filldraw[black, draw=black] (3-\xshift,3.5*\yScale) rectangle (4-\xshift,4*\yScale); 

\filldraw[black!3, draw=black] (1-\xshift,0) rectangle (2-\xshift,7/8*\yScale);   
\filldraw[black!3, draw=black] (2-\xshift,0) rectangle (3-\xshift,\yScale);   
\filldraw[black!3, draw=black] (3-\xshift,0) rectangle (4-\xshift,1.25*\yScale);   

\filldraw[black!20, draw=black] (2-\xshift,\yScale) rectangle (3-\xshift,1.75*\yScale);
\filldraw[black!20, draw=black] (3-\xshift,1.25*\yScale) rectangle (4-\xshift,2.5*\yScale);

\filldraw[black!43, draw=black] (3-\xshift,2.5*\yScale) rectangle (4-\xshift,3.5*\yScale);

\draw[thick,->] (4.25-\xshift, 0.03*\yScale) -- (4.25-\xshift, 1.22*\yScale) node[right, midway] {$\$12.5$}; 
\draw[thick,->] (4.25-\xshift, 1.22*\yScale) -- (4.25-\xshift, 0.03*\yScale);

\draw[thick,->] (4.25-\xshift, 1.28*\yScale) -- (4.25-\xshift, 2.47*\yScale) node[right, midway] {$\$12.5$}; 
\draw[thick,->] (4.25-\xshift, 2.47*\yScale) -- (4.25-\xshift, 1.28*\yScale);

\draw[thick,->] (4.25-\xshift, 2.53*\yScale) -- (4.25-\xshift, 3.47*\yScale) node[right, midway] {$\$10$}; 
\draw[thick,->] (4.25-\xshift, 3.47*\yScale) -- (4.25-\xshift, 2.53*\yScale);

\draw[thick,->] (4.25-\xshift, 3.53*\yScale) -- (4.25-\xshift, 3.97*\yScale) node[right, midway] {$\$5$}; 
\draw[thick,->] (4.25-\xshift, 3.97*\yScale) -- (4.25-\xshift, 3.53*\yScale);

\end{tikzpicture}

%% file: Figures/seniroityBasedSmooth.tex
\begin{tikzpicture}

\def\yScale{1}
\def\xshift{0.55}
\draw[thick,->] (0,0) -- (4,0) node[below] {$t$};
\draw[thick,->] (0,0) -- (0,4.5*\yScale) node[left] {};

\node[below] at (1-\xshift,0) {$\tau_1$};
\node[below] at (2-\xshift,0) {$\tau_2$};
\node[below] at (3-\xshift,0) {$\tau_3$};

\draw[thick] (-0.1,\yScale) -- (0.1,\yScale);
\node[left] at (0,\yScale) {$10$};
\draw[thick] (-0.1,2*\yScale) -- (0.1,2*\yScale);
\node[left] at (0,2*\yScale) {$20$};
\draw[thick] (-0.1,3*\yScale) -- (0.1,3*\yScale);
\node[left] at (0,3*\yScale) {$30$};
\draw[thick] (-0.1,4*\yScale) -- (0.1,4*\yScale);
\node[left] at (0,4*\yScale) {$40$};

\filldraw[black, draw=black] (1-\xshift,3/4*\yScale) rectangle (2-\xshift,1*\yScale); 
\filldraw[black, draw=black] (2-\xshift,7/4*\yScale) rectangle (3-\xshift,2*\yScale); 
\filldraw[black, draw=black] (3-\xshift,15/4*\yScale) rectangle (4-\xshift,4*\yScale); 

\filldraw[black!3, draw=black] (1-\xshift,0) rectangle (2-\xshift,6/8*\yScale);   
\filldraw[black!3, draw=black] (2-\xshift,0) rectangle (3-\xshift,0.875*\yScale);   
\filldraw[black!3, draw=black] (3-\xshift,0) rectangle (4-\xshift,1.375*\yScale);   

\filldraw[black!20, draw=black] (2-\xshift,0.875*\yScale) rectangle (3-\xshift,1.75*\yScale);
\filldraw[black!20, draw=black] (3-\xshift,1.375*\yScale) rectangle (4-\xshift,2.75*\yScale);

\filldraw[black!43, draw=black] (3-\xshift,2.75*\yScale) rectangle (4-\xshift,3.75*\yScale);

\draw[thick,->] (4.25-\xshift, 0.03*\yScale) -- (4.25-\xshift, 1.35*\yScale) node[right, midway] {$\$13.75$}; 
\draw[thick,->] (4.25-\xshift, 1.35*\yScale) -- (4.25-\xshift, 0.03*\yScale);

\draw[thick,->] (4.25-\xshift, 1.4*\yScale) -- (4.25-\xshift, 2.72*\yScale) node[right, midway] {$\$13.75$}; 
\draw[thick,->] (4.25-\xshift, 2.72*\yScale) -- (4.25-\xshift, 1.4*\yScale);

\draw[thick,->] (4.25-\xshift, 2.78*\yScale) -- (4.25-\xshift, 3.72*\yScale) node[right, midway] {$\$10$}; 
\draw[thick,->] (4.25-\xshift, 3.72*\yScale) -- (4.25-\xshift, 2.78*\yScale);

\draw[thick,->] (4.25-\xshift, 3.75*\yScale) -- (4.25-\xshift, 4*\yScale) node[right, midway] {$\$2.5$}; 
\draw[thick,->] (4.25-\xshift, 4*\yScale) -- (4.25-\xshift, 3.75*\yScale);

\end{tikzpicture}

%% file: Figures/seniorityBasedLiability.tex
\begin{tikzpicture}

\def\yScale{1}
\def\xshift{0.55}
\draw[thick,->] (0,0) -- (4,0) node[below] {$t$};
\draw[thick,->] (0,0) -- (0,4.5*\yScale) node[left] {};

\node[below] at (1-\xshift,0) {$\tau_1$};
\node[below] at (2-\xshift,0) {$\tau_2$};
\node[below] at (3-\xshift,0) {$\tau_3$};

\draw[thick] (-0.1,\yScale) -- (0.1,\yScale);
\node[left] at (0,\yScale) {$10$};
\draw[thick] (-0.1,2*\yScale) -- (0.1,2*\yScale);
\node[left] at (0,2*\yScale) {$20$};
\draw[thick] (-0.1,3*\yScale) -- (0.1,3*\yScale);
\node[left] at (0,3*\yScale) {$30$};
\draw[thick] (-0.1,4*\yScale) -- (0.1,4*\yScale);
\node[left] at (0,4*\yScale) {$40$};

\filldraw[black, draw=black] (3-\xshift,3*\yScale) rectangle (4-\xshift,4*\yScale); 

\filldraw[black!3, draw=black] (1-\xshift,0) rectangle (2-\xshift,\yScale);   
\filldraw[black!3, draw=black] (2-\xshift,0) rectangle (3-\xshift,\yScale);   
\filldraw[black!3, draw=black] (3-\xshift,0) rectangle (4-\xshift,\yScale);   

\filldraw[black!20, draw=black] (2-\xshift,\yScale) rectangle (3-\xshift,2*\yScale);
\filldraw[black!20, draw=black] (3-\xshift,\yScale) rectangle (4-\xshift,2*\yScale);

\filldraw[black!43, draw=black] (3-\xshift,2*\yScale) rectangle (4-\xshift,3*\yScale);

\draw[thick,->] (4.25-\xshift, 0.03*\yScale) -- (4.25-\xshift, 1*\yScale) node[right, midway] {$\$10$}; 
\draw[thick,->] (4.25-\xshift, 0.97*\yScale) -- (4.25-\xshift, 0.03*\yScale);

\draw[thick,->] (4.25-\xshift, 1.03*\yScale) -- (4.25-\xshift, 1.97*\yScale) node[right, midway] {$\$10$}; 
\draw[thick,->] (4.25-\xshift, 1.97*\yScale) -- (4.25-\xshift, 1.03*\yScale);

\draw[thick,->] (4.25-\xshift, 2.03*\yScale) -- (4.25-\xshift, 2.97*\yScale) node[right, midway] {$\$10$}; 
\draw[thick,->] (4.25-\xshift, 2.97*\yScale) -- (4.25-\xshift, 2.03*\yScale);

\draw[thick,->] (4.25-\xshift, 3.03*\yScale) -- (4.25-\xshift, 3.97*\yScale) node[right, midway] {$\$10$}; 
\draw[thick,->] (4.25-\xshift, 3.97*\yScale) -- (4.25-\xshift, 3.03*\yScale);

\end{tikzpicture}

%% file: Tex/numerics.tex
\revcolor{\subsection{Numerical Experiment: Sensitivity Analysis}\label{subsec:Sensitivity}
To complement our theoretical findings, we conduct a sensitivity analysis to evaluate how the platform's optimal profit and size vary under key demand and supply parameters. Throughout this numerical exploration, the service fee is normalized to $p=1$ and the discount rate is fixed at $r=0.01$.}

\revcolor{First, to isolate demand-side network effects, we fix $m(s)=m=1$, $w=100$ and consider the class of demand functions $\lambda(s)=\Lambda\min\{(s/\bar S)^{\xi},1\}$, where $\Lambda$ is the maximum demand rate and $\bar S$ is the supply size at which this maximum is reached. The parameter $\xi\ge0$ controls the marginal demand impact of additional servers: $\xi<1$ yields diminishing marginal impact, $\xi=1$ a constant marginal impact, and $\xi>1$ increasing marginal impact before demand reaches $\Lambda$. We vary $\xi$ to study its effect on $S^*$ and $J^*$. (For completeness, we also study an alternative class MSL demand functions characterized by a single step-discontinuity in \ref{apx:numericalMSL}.)}
\begin{figure}[h!]
\centering
\input{Figures/numericsfig2}
\caption{\footnotesize \sf \revcolor{Effect of demand curvature on optimal growth. Left: $\lambda(s)=\Lambda\min\{(s/\bar S)^{\xi},1\}$ for $\Lambda=150$, $\bar S=100$, and $\xi\in\{1/4,1,4\}$. Middle/right: optimal size $S^*$ and payoff $J^*$ as functions of $\xi$ on a logarithmic scale over $[0.01,10]$.}} 
\label{Fig:sensitivity2}
\end{figure}
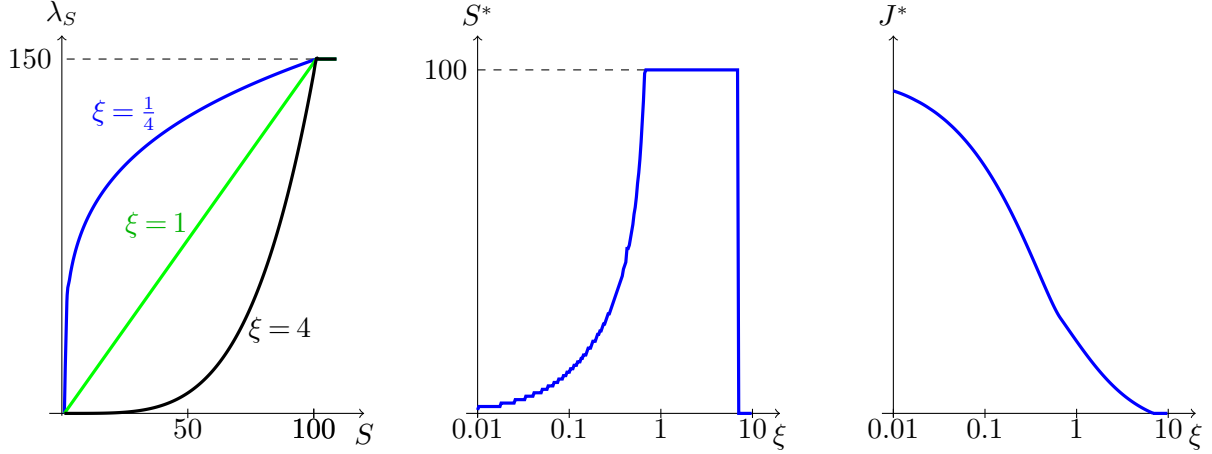

\revcolor{As \Cref{Fig:sensitivity2} illustrates, increasing $\xi$ changes both the level and curvature of demand before saturation. For any fixed $s<\bar S$, $\lambda(s)$ decreases with $\xi$ increasing; however, $\xi<1$ yields concave demand with large early marginal gains, whereas $\xi>1$ yields convex demand with weak adoption at small scales. Economically, convexity captures a more hesitant demand side, where additional servers generate meaningful demand only after the platform reaches sufficient scale. In this numerical experiment, $J^*$ decreases with $\xi$, reflecting the increasing difficulty of financing early growth; for sufficiently large $\xi$ (roughly $\xi\ge 8$), no profitable expansion is sustainable, so $S^*=0$ and $J^*=0$.}

\revcolor{The middle panel offers two takeaways about the optimal size that highlight the nuances of the platform's growth decision. First, although demand is pointwise lower as $\xi$ increases, $S^*$ does not necessarily decrease: until the maximum-size regime is reached, greater convexity pushes the platform to build a larger supply base before high demand materializes. Second, at the highly concave end (small $\xi$), most attainable demand is generated at small scale, so the platform optimally stops before full expansion. Thus, maximizing network size is not always optimal: late-stage demand must be weighed against the long-term compensation liabilities required to retain additional servers.}

\revcolor{We next conduct a sensitivity analysis with respect to the supply-side parameters. We fix demand as an MSL function with $\Lambda_n=(25,50,60)$, $s_n=(10,50,100)$, and $\xi_n=0.3(1,1,0.5)$,\footnote{\revcolor{We use the same demand specification in the subsequent sensitivity analyses.}} and vary the platform's growth rate $m$ and the servers' outside payoff $w$.}
\begin{figure}[h]
\centering
\begin{minipage}{0.25\textwidth}
    \centering
    \input{Figures/numericsfig3}
\end{minipage}
\hfill
\begin{minipage}{0.7\textwidth}
    \centering
    \includegraphics[width=\linewidth]{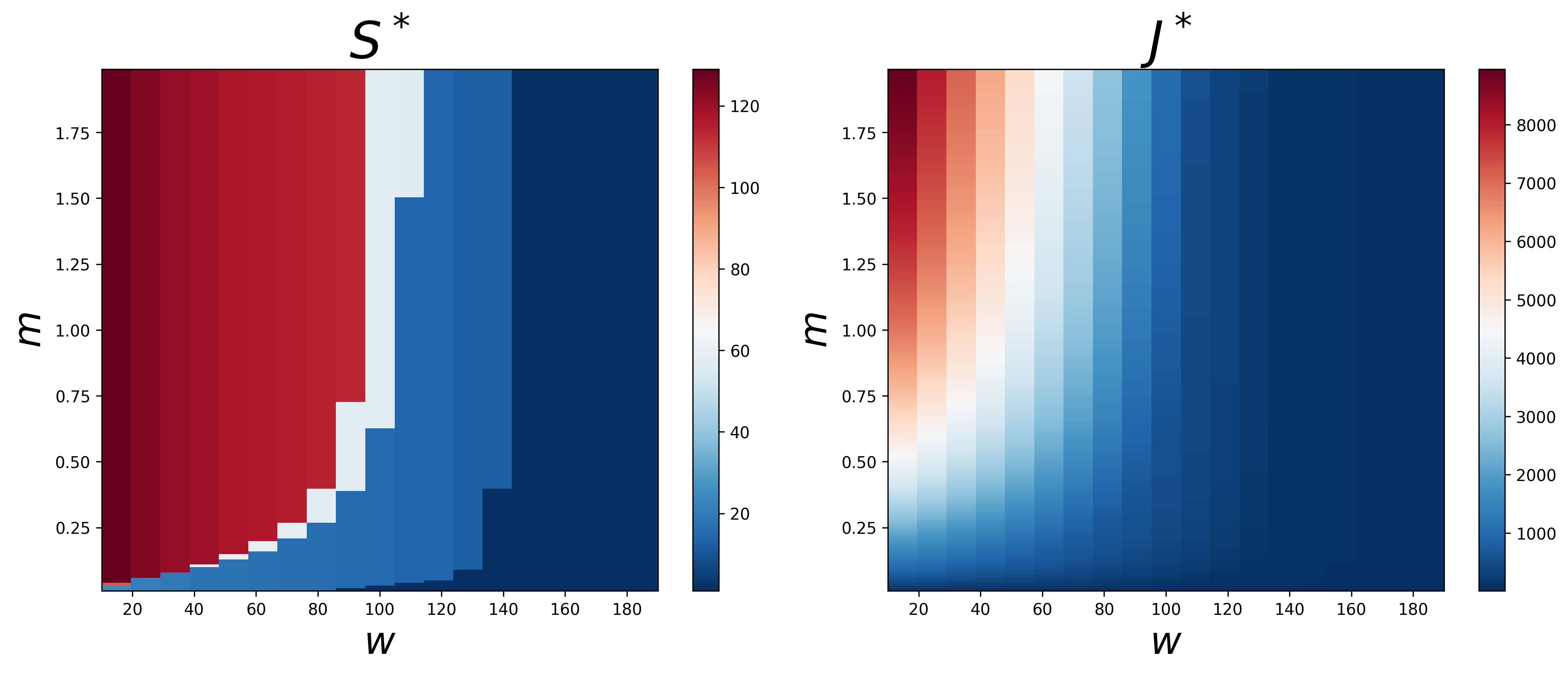}
\end{minipage}
\caption{\footnotesize\sf \revcolor{Sensitivity to supply-side parameters. The left panel shows the demand function. The middle and right panels show heatmaps of the optimal platform size $S^*$ and expected payoff $J^*$, respectively, as the maximum growth rate $m$ and servers' outside option $w$ vary.}}
\label{Fig:sensitivity3}
\end{figure}

\revcolor{As \Cref{Fig:sensitivity3} shows, the optimal payoff decreases with the servers' outside option $w$ and increases with the maximum growth rate $m$. A larger $m$ shortens the transition phase, allowing the platform to reach profitable scales sooner and reduce discounting costs. The optimal size $S^*$ also decreases with $w$: as the effective labor cost rises, the compensation required to satisfy servers' participation constraints eventually outweighs the incremental revenue from further expansion. More importantly, $m$ affects not only the speed of growth but also the optimal final size. This reinforces a central insight of our paper: gradual growth is a structural friction. When expansion is slow, revenues from a larger future scale arrive too late to justify the additional compensation obligations incurred along the growth path, making a smaller final size optimal.}

%% file: Figures/numericsfig2.tex
\begin{tikzpicture}
\begin{scope}[shift={(0, 0)},yscale = 1/32,xscale = 4/120]
        \draw[->] (-5,0) -- (120,0) node[below] {$S$};
        \draw[->] (0,-1) -- (0,160) node[above] {$\lambda_S$};
    
        \draw[dashed] (100, 150) -- (0,150) node[left] {$150$};
    
        \draw (50,-2) -- (50,2) node[below] {$50$};
        \draw (100,-2) -- (100,2) node[below] {$100$};
        \draw (100,-2) -- (100,2) node[below] {$100$};
        
        \draw[blue, very thick, smooth] plot coordinates {
(1,0.0)
(2,47.434164902525694)
(3,56.409046396295906)
(4,62.42687175431726)
(5,67.08203932499369)
(6,70.93062067523819)
(7,74.23848005759648)
(8,77.15530085415601)
(9,79.77443845417483)
(10,82.15838362577492)
(11,84.35119877855236)
(12,86.3852173597986)
(13,88.28492869148137)
(14,90.06936421333269)
(15,91.75363273593787)
(16,93.34994659327175)
(17,94.86832980505139)
(18,96.31712027277257)
(19,97.70333436489459)
(20,99.03293706119452)
(21,100.31104574646331)
(22,101.54208641568597)
(23,102.72991511532906)
(24,103.87791356902399)
(25,104.9890653474175)
(26,106.06601717798213)
(27,107.1111287662667)
(28,108.12651363606393)
(29,109.1140728769239)
(30,110.07552323769637)
(31,111.01242067384278)
(32,111.9261802098488)
(33,112.81809279259181)
(34,113.68933967048346)
(35,114.54100472495396)
(36,115.3740850970188)
(37,116.1895003862225)
(38,116.98810064776244)
(39,117.7706733727892)
(40,118.53794960432037)
(41,119.2906093150576)
(42,120.02928615227164)
(43,120.75457163775492)
(44,121.46701889681619)
(45,122.16714597877491)
(46,122.85543882190801)
(47,123.53235390792103)
(48,124.19832064445295)
(49,124.85374350863452)
(50,125.49900398011134)
(51,126.13446228805718)
(52,126.76045899341565)
(53,127.37731642481698)
(54,127.98533998423869)
(55,128.58481933644597)
(56,129.17602949450506)
(57,129.75923181216334)
(58,130.3346748925979)
(59,130.9025954219159)
(60,131.46321893481982)
(61,132.016760519009)
(62,132.56342546415271)
(63,133.10340986062894)
(64,133.6369011526606)
(65,134.16407864998737)
(66,134.68511400177727)
(67,135.20017163610004)
(68,135.7094091679461)
(69,136.21297777847556)
(70,136.71102256791758)
(71,137.20368288430416)
(72,137.69109263001295)
(73,138.17338054790721)
(74,138.65067048869415)
(75,139.12308166097353)
(76,139.59072886531493)
(77,140.0537227135824)
(78,140.51216983461708)
(79,140.96617306729178)
(80,141.41583164186417)
(81,141.86124135047638)
(82,142.30249470757707)
(83,142.73968110097903)
(84,143.17288693420508)
(85,143.60219576072393)
(86,144.02768841062817)
(87,144.4494431102638)
(88,144.8675355952804)
(89,145.28203921753519)
(90,145.69302504625043)
(91,146.10056196379452)
(92,146.50471675642828)
(93,146.9055542003339)
(94,147.30313714321957)
(95,147.6975265817728)
(96,148.08878173521506)
(97,148.47696011519295)
(98,148.86211759222437)
(99,149.24430845890322)
(100,149.62358549005216)
(101,150.0)
(102,150)
(103,150)
(104,150)
(105,150)
(106,150)
(107,150)
(108,150)
(109,150)
};
\node[blue, right] at (8, 127) {$\xi=\tfrac{1}{4}$};

\draw[green, very thick, smooth] plot coordinates {
(1,0.0)
(2,1.5)
(3,3.0)
(4,4.5)
(5,6.0)
(6,7.5)
(7,9.0)
(8,10.500000000000002)
(9,12.0)
(10,13.5)
(11,15.0)
(12,16.5)
(13,18.0)
(14,19.5)
(15,21.000000000000004)
(16,22.5)
(17,24.0)
(18,25.500000000000004)
(19,27.0)
(20,28.5)
(21,30.0)
(22,31.5)
(23,33.0)
(24,34.5)
(25,36.0)
(26,37.5)
(27,39.0)
(28,40.5)
(29,42.00000000000001)
(30,43.5)
(31,45.0)
(32,46.5)
(33,48.0)
(34,49.5)
(35,51.00000000000001)
(36,52.5)
(37,54.0)
(38,55.5)
(39,57.0)
(40,58.5)
(41,60.0)
(42,61.49999999999999)
(43,63.0)
(44,64.5)
(45,66.0)
(46,67.5)
(47,69.0)
(48,70.5)
(49,72.0)
(50,73.5)
(51,75.0)
(52,76.5)
(53,78.0)
(54,79.5)
(55,81.0)
(56,82.5)
(57,84.00000000000001)
(58,85.49999999999999)
(59,87.0)
(60,88.5)
(61,90.0)
(62,91.5)
(63,93.0)
(64,94.5)
(65,96.0)
(66,97.5)
(67,99.0)
(68,100.5)
(69,102.00000000000001)
(70,103.49999999999999)
(71,105.0)
(72,106.5)
(73,108.0)
(74,109.5)
(75,111.0)
(76,112.5)
(77,114.0)
(78,115.5)
(79,117.0)
(80,118.5)
(81,120.0)
(82,121.50000000000001)
(83,122.99999999999999)
(84,124.5)
(85,126.0)
(86,127.5)
(87,129.0)
(88,130.5)
(89,132.0)
(90,133.5)
(91,135.0)
(92,136.5)
(93,138.0)
(94,139.5)
(95,141.0)
(96,142.5)
(97,144.0)
(98,145.5)
(99,147.0)
(100,148.5)
(101,150.0)
(102,150)
(103,150)
(104,150)
(105,150)
(106,150)
(107,150)
(108,150)
(109,150)
};
\node[green!70!black, right] at (21, 80) {$\xi=1$};

\draw[black, very thick, smooth] plot coordinates {
(1,0.0)
(2,1.5e-06)
(3,2.4e-05)
(4,0.00012149999999999999)
(5,0.000384)
(6,0.0009375000000000002)
(7,0.0019439999999999998)
(8,0.0036015000000000014)
(9,0.006144)
(10,0.009841499999999998)
(11,0.015000000000000003)
(12,0.021961500000000002)
(13,0.031103999999999996)
(14,0.042841500000000005)
(15,0.05762400000000002)
(16,0.07593749999999999)
(17,0.098304)
(18,0.12528150000000005)
(19,0.15746399999999997)
(20,0.1954815)
(21,0.24000000000000005)
(22,0.29172149999999997)
(23,0.35138400000000003)
(24,0.41976150000000007)
(25,0.49766399999999994)
(26,0.5859375)
(27,0.6854640000000001)
(28,0.7971615000000002)
(29,0.9219840000000004)
(30,1.0609214999999996)
(31,1.2149999999999999)
(32,1.3852814999999998)
(33,1.572864)
(34,1.7788815000000002)
(35,2.0045040000000007)
(36,2.2509374999999996)
(37,2.5194239999999994)
(38,2.8112415)
(39,3.127704)
(40,3.4701615000000006)
(41,3.8400000000000007)
(42,4.238641499999999)
(43,4.6675439999999995)
(44,5.128201499999999)
(45,5.6221440000000005)
(46,6.1509375)
(47,6.716184000000001)
(48,7.319521499999999)
(49,7.962623999999999)
(50,8.6472015)
(51,9.375)
(52,10.1478015)
(53,10.967424000000001)
(54,11.835721500000004)
(55,12.754584000000003)
(56,13.725937500000004)
(57,14.751744000000006)
(58,15.834001499999996)
(59,16.974743999999994)
(60,18.176041499999997)
(61,19.439999999999998)
(62,20.768761499999997)
(63,22.164503999999997)
(64,23.629441500000002)
(65,25.165824)
(66,26.775937500000005)
(67,28.462104000000004)
(68,30.22668150000001)
(69,32.07206400000001)
(70,34.00068149999999)
(71,36.01499999999999)
(72,38.117521499999995)
(73,40.31078399999999)
(74,42.59736149999999)
(75,44.979864)
(76,47.4609375)
(77,50.043264)
(78,52.7295615)
(79,55.52258400000001)
(80,58.42512150000001)
(81,61.44000000000001)
(82,64.57008150000001)
(83,67.81826399999999)
(84,71.18748149999999)
(85,74.68070399999999)
(86,78.30093749999999)
(87,82.05122399999999)
(88,85.9346415)
(89,89.95430400000001)
(90,94.1133615)
(91,98.415)
(92,102.86244150000002)
(93,107.45894400000002)
(94,112.20780150000003)
(95,117.11234399999998)
(96,122.17593749999997)
(97,127.40198399999998)
(98,132.79392149999998)
(99,138.355224)
(100,144.08940149999998)
(101,150.0)
(102,150)
(103,150)
(104,150)
(105,150)
(106,150)
(107,150)
(108,150)
(109,150)
};
\node[black, right] at (70, 35.5) {$\xi=4$};

\end{scope}
\begin{scope}[shift={(5.5, 0)}, yscale = 1/22, xscale = 4/3.3]
        \draw[->] (-1/8,0) -- (3*1.1,0) node[below] {$\xi$};
        \draw[->] (0,-1) -- (0,110) node[above] {$S^*$};

        \draw[dashed] (2, 100) -- (0,100) node[left] {$100$};

        \draw (3,-4/3) -- (3,4/3) node[below] {$10$};
        \draw (2,-4/3) -- (2,4/3) node[below] {$1$};
        
        \draw (1,-4/3) -- (1,4/3) node[below] {$0.1$};
    
        \draw (0,-4/3) -- (0,4/3) node[below] {$0.01$};

        \draw[blue, very thick] plot coordinates {
(0.0,1)
(0.012837224705172234,2)
(0.02567444941034447,2)
(0.0385116741155167,2)
(0.05134889882068894,2)
(0.06418612352586095,2)
(0.07702334823103318,2)
(0.08986057293620542,2)
(0.10269779764137765,2)
(0.11553502234654989,2)
(0.12837224705172212,2)
(0.14120947175689436,2)
(0.1540466964620666,2)
(0.16688392116723882,2)
(0.17972114587241106,2)
(0.1925583705775833,2)
(0.20539559528275553,2)
(0.21823281998792754,2)
(0.23107004469309977,2)
(0.243907269398272,2)
(0.25674449410344424,3)
(0.2695817188086165,3)
(0.2824189435137887,3)
(0.2952561682189607,3)
(0.30809339292413296,3)
(0.3209306176293052,3)
(0.3337678423344774,3)
(0.34660506703964966,3)
(0.3594422917448219,3)
(0.37227951644999413,3)
(0.38511674115516636,3)
(0.3979539658603386,3)
(0.41079119056551083,4)
(0.42362841527068307,4)
(0.4364656399758553,4)
(0.4493028646810273,4)
(0.46214008938619955,4)
(0.4749773140913718,4)
(0.487814538796544,4)
(0.5006517635017163,4)
(0.5134889882068885,4)
(0.5263262129120607,5)
(0.539163437617233,5)
(0.5520006623224052,5)
(0.5648378870275772,5)
(0.5776751117327494,5)
(0.5905123364379217,5)
(0.6033495611430939,5)
(0.6161867858482661,6)
(0.6290240105534384,6)
(0.6418612352586106,6)
(0.6546984599637828,6)
(0.6675356846689551,6)
(0.6803729093741273,6)
(0.6932101340792995,6)
(0.7060473587844718,7)
(0.7188845834896438,7)
(0.731721808194816,7)
(0.7445590328999883,7)
(0.7573962576051605,7)
(0.7702334823103327,8)
(0.783070707015505,8)
(0.7959079317206772,8)
(0.8087451564258494,8)
(0.8215823811310217,8)
(0.8344196058361937,9)
(0.8472568305413661,9)
(0.8600940552465381,9)
(0.8729312799517104,9)
(0.8857685046568826,10)
(0.8986057293620549,10)
(0.9114429540672271,10)
(0.9242801787723993,10)
(0.9371174034775716,11)
(0.9499546281827438,11)
(0.962791852887916,11)
(0.9756290775930883,12)
(0.9884663022982603,12)
(1.0013035270034325,12)
(1.0141407517086047,13)
(1.026977976413777,13)
(1.0398152011189492,13)
(1.0526524258241214,14)
(1.0654896505292935,14)
(1.078326875234466,14)
(1.091164099939638,15)
(1.1040013246448104,15)
(1.1168385493499824,15)
(1.1296757740551548,16)
(1.1425129987603269,16)
(1.155350223465499,17)
(1.1681874481706713,17)
(1.1810246728758436,17)
(1.1938618975810158,18)
(1.206699122286188,19)
(1.2195363469913603,19)
(1.2323735716965325,19)
(1.2452107964017047,20)
(1.258048021106877,21)
(1.2708852458120492,21)
(1.2837224705172212,21)
(1.2965596952223937,22)
(1.3093969199275657,23)
(1.3222341446327381,23)
(1.3350713693379102,24)
(1.3479085940430824,24)
(1.3607458187482546,25)
(1.373583043453427,26)
(1.386420268158599,26)
(1.3992574928637713,27)
(1.4120947175689436,28)
(1.4249319422741158,29)
(1.437769166979288,29)
(1.45060639168446,30)
(1.4634436163896325,31)
(1.4762808410948045,32)
(1.489118065799977,33)
(1.501955290505149,34)
(1.5147925152103212,35)
(1.5276297399154934,36)
(1.5404669646206657,37)
(1.553304189325838,38)
(1.56614141403101,39)
(1.5789786387361822,40)
(1.5918158634413544,42)
(1.6046530881465266,43)
(1.6174903128516989,44)
(1.630327537556871,48)
(1.6431647622620433,48)
(1.6560019869672153,49)
(1.6688392116723876,51)
(1.6816764363775598,53)
(1.694513661082732,55)
(1.7073508857879043,58)
(1.7201881104930765,60)
(1.7330253351982488,63)
(1.7458625599034208,67)
(1.7586997846085932,70)
(1.7715370093137652,74)
(1.7843742340189375,79)
(1.7972114587241097,85)
(1.810048683429282,91)
(1.8228859081344542,99)
(1.8357231328396264,100)
(1.8485603575447986,100)
(1.8613975822499709,100)
(1.874234806955143,100)
(1.8870720316603153,100)
(1.8999092563654876,100)
(1.9127464810706596,100)
(1.9255837057758318,100)
(1.938420930481004,100)
(1.9512581551861763,100)
(1.9640953798913485,100)
(1.9769326045965208,100)
(1.989769829301693,100)
(2.002607054006865,100)
(2.0154442787120375,100)
(2.0282815034172095,100)
(2.041118728122382,100)
(2.053955952827554,100)
(2.0667931775327264,100)
(2.0796304022378984,100)
(2.092467626943071,100)
(2.105304851648243,100)
(2.1181420763534153,100)
(2.1309793010585873,100)
(2.14381652576376,100)
(2.156653750468932,100)
(2.1694909751741043,100)
(2.1823281998792763,100)
(2.1951654245844487,100)
(2.2080026492896208,100)
(2.2208398739947928,100)
(2.2336770986999652,100)
(2.2465143234051372,100)
(2.2593515481103097,100)
(2.2721887728154817,100)
(2.285025997520654,100)
(2.297863222225826,100)
(2.3107004469309986,100)
(2.3235376716361706,100)
(2.336374896341343,100)
(2.349212121046515,100)
(2.3620493457516876,100)
(2.3748865704568596,100)
(2.3877237951620316,100)
(2.400561019867204,100)
(2.4133982445723765,100)
(2.4262354692775485,100)
(2.4390726939827205,100)
(2.451909918687893,100)
(2.464747143393065,100)
(2.4775843680982375,100)
(2.4904215928034095,100)
(2.503258817508582,100)
(2.516096042213754,100)
(2.5289332669189264,100)
(2.5417704916240984,100)
(2.554607716329271,100)
(2.567444941034443,100)
(2.5802821657396153,100)
(2.5931193904447873,100)
(2.6059566151499594,100)
(2.618793839855132,100)
(2.6316310645603043,100)
(2.6444682892654763,100)
(2.6573055139706483,100)
(2.6701427386758207,100)
(2.682979963380993,100)
(2.695817188086165,100)
(2.7086544127913372,100)
(2.7214916374965097,100)
(2.7343288622016817,100)
(2.747166086906854,100)
(2.760003311612026,100)
(2.7728405363171986,100)
(2.7856777610223706,100)
(2.798514985727543,100)
(2.811352210432715,100)
(2.8241894351378876,100)
(2.8370266598430596,100)
(2.849863884548232,0)
(2.862701109253404,0)
(2.875538333958576,0)
(2.8883755586637485,0)
(2.9012127833689205,0)
(2.914050008074093,0)
(2.926887232779265,0)
(2.9397244574844374,0)
(2.9525616821896095,0)
(2.965398906894782,0)
(2.978236131599954,0)
(2.9910733563051264,0)
};
\end{scope}
\begin{scope}[shift={(11, 0)}, yscale = 1/22, xscale = 4/3.3]
        \draw[->] (-1/8,0) -- (3*1.1,0) node[below] {$\xi$};
        \draw[->] (0,-1) -- (0,110) node[above] {$J^*$};
        
        \draw (3,-4/3) -- (3,4/3) node[below] {$10$};
        \draw (2,-4/3) -- (2,4/3) node[below] {$1$};
        
        \draw (1,-4/3) -- (1,4/3) node[below] {$0.1$};
    
        \draw (0,-4/3) -- (0,4/3) node[below] {$0.01$};

        \draw[blue, very thick, smooth] plot coordinates {
(0.0,93.89365538166034)
(0.012837224705172234,93.77942421736115)
(0.02567444941034447,93.6643438132927)
(0.0385116741155167,93.5459566382769)
(0.05134889882068894,93.42417216595766)
(0.06418612352586095,93.298897664995)
(0.07702334823103318,93.17003816242413)
(0.08986057293620542,93.03749640756794)
(0.10269779764137765,92.90117283660751)
(0.11553502234654989,92.76096553792239)
(0.12837224705172212,92.61677021832007)
(0.14120947175689436,92.4684801702822)
(0.1540466964620666,92.31598624036383)
(0.16688392116723882,92.15917679889031)
(0.17972114587241106,91.99793771110694)
(0.1925583705775833,91.83215230994486)
(0.20539559528275553,91.66170137057838)
(0.21823281998792754,91.48646308695841)
(0.23107004469309977,91.3063130505188)
(0.243907269398272,91.12112423126351)
(0.25674449410344424,90.9484283911173)
(0.2695817188086165,90.77133742458406)
(0.2824189435137887,90.58928698320769)
(0.2952561682189607,90.40214905820591)
(0.30809339292413296,90.20979299896595)
(0.3209306176293052,90.01208550078)
(0.3337678423344774,89.80889059576012)
(0.34660506703964966,89.6000696471991)
(0.3594422917448219,89.38548134765607)
(0.37227951644999413,89.16498172106158)
(0.38511674115516636,88.9384241291515)
(0.3979539658603386,88.70565928255604)
(0.41079119056551083,88.47873068482492)
(0.42362841527068307,88.25093393390908)
(0.4364656399758553,88.01690020873596)
(0.4493028646810273,87.77647711273711)
(0.46214008938619955,87.52950964408217)
(0.4749773140913718,87.27584022420376)
(0.487814538796544,87.0153087325322)
(0.5006517635017163,86.74775254785531)
(0.5134889882068885,86.47300659673573)
(0.5263262129120607,86.19922646345383)
(0.539163437617233,85.92680542539969)
(0.5520006623224052,85.64708248293562)
(0.5648378870275772,85.35988868116017)
(0.5776751117327494,85.06505278687206)
(0.5905123364379217,84.76240136722312)
(0.6033495611430939,84.45175887770918)
(0.6161867858482661,84.13438282273837)
(0.6290240105534384,83.8237773327984)
(0.6418612352586106,83.50500622223366)
(0.6546984599637828,83.17789020447547)
(0.6675356846689551,82.84224821504763)
(0.6803729093741273,82.49789754274028)
(0.6932101340792995,82.14465397300215)
(0.7060473587844718,81.79636203134051)
(0.7188845834896438,81.44112358589929)
(0.731721808194816,81.07676885571989)
(0.7445590328999883,80.70311112170246)
(0.7573962576051605,80.31996272641003)
(0.7702334823103327,79.93508980568049)
(0.783070707015505,79.54816630261234)
(0.7959079317206772,79.15149442780377)
(0.8087451564258494,78.74488547431498)
(0.8215823811310217,78.32815056397423)
(0.8344196058361937,77.9130383226597)
(0.8472568305413661,77.49090305518466)
(0.8600940552465381,77.05836035036559)
(0.8729312799517104,76.61522254301647)
(0.8857685046568826,76.16956912203712)
(0.8986057293620549,75.7196222082527)
(0.9114429540672271,75.25878381892052)
(0.9242801787723993,74.78686977506713)
(0.9371174034775716,74.31427411505545)
(0.9499546281827438,73.83425216419946)
(0.962791852887916,73.34285217383807)
(0.9756290775930883,72.8432493816079)
(0.9884663022982603,72.34266301197944)
(1.0013035270034325,71.83039368703774)
(1.0141407517086047,71.30656835304015)
(1.026977976413777,70.78412520936868)
(1.0398152011189492,70.24969789928679)
(1.0526524258241214,69.70365794951311)
(1.0654896505292935,69.1581498409185)
(1.078326875234466,68.60036660693726)
(1.091164099939638,68.03357396038005)
(1.1040013246448104,67.46387926597247)
(1.1168385493499824,66.88163573549522)
(1.1296757740551548,66.29516674986805)
(1.1425129987603269,65.70025803366921)
(1.155350223465499,65.09535812485316)
(1.1681874481706713,64.4872135145945)
(1.1810246728758436,63.86616936308003)
(1.1938618975810158,63.243415649188584)
(1.206699122286188,62.608530734508314)
(1.2195363469913603,61.9695787356498)
(1.2323735716965325,61.32040790169639)
(1.2452107964017047,60.66627796131195)
(1.258048021106877,60.00243297433711)
(1.2708852458120492,59.3339965068241)
(1.2837224705172212,58.65514614763192)
(1.2965596952223937,57.97316341059171)
(1.3093969199275657,57.280591058479814)
(1.3222341446327381,56.584184712059624)
(1.3350713693379102,55.87954100941556)
(1.3479085940430824,55.167469600060286)
(1.3607458187482546,54.45130502031997)
(1.373583043453427,53.72689938610591)
(1.386420268158599,52.99629844861642)
(1.3992574928637713,52.26150266191506)
(1.4120947175689436,51.519750712964154)
(1.4249319422741158,50.77166386631302)
(1.437769166979288,50.019186420450886)
(1.45060639168446,49.26222223986582)
(1.4634436163896325,48.49995397380269)
(1.4762808410948045,47.73288345529205)
(1.489118065799977,46.961480435250046)
(1.501955290505149,46.18618799663277)
(1.5147925152103212,45.40742720570073)
(1.5276297399154934,44.625601128133134)
(1.5404669646206657,43.841098313088295)
(1.553304189325838,43.05429582894384)
(1.56614141403101,42.26556191914019)
(1.5789786387361822,41.4752583343727)
(1.5918158634413544,40.68440002987481)
(1.6046530881465266,39.89421038076358)
(1.6174903128516989,39.10374125462401)
(1.630327537556871,38.315692685839124)
(1.6431647622620433,37.5291366864294)
(1.6560019869672153,36.746561844200535)
(1.6688392116723876,35.96844506282453)
(1.6816764363775598,35.19592325718276)
(1.694513661082732,34.43041048320724)
(1.7073508857879043,33.674327557364784)
(1.7201881104930765,32.92850273812951)
(1.7330253351982488,32.19623651180182)
(1.7458625599034208,31.479756324475886)
(1.7586997846085932,30.782330863060427)
(1.7715370093137652,30.107766426013605)
(1.7843742340189375,29.461215642868975)
(1.7972114587241097,28.84821471690776)
(1.810048683429282,28.27618239441826)
(1.8228859081344542,27.75459047244217)
(1.8357231328396264,27.26613809191499)
(1.8485603575447986,26.776783353747554)
(1.8613975822499709,26.28660513581072)
(1.874234806955143,25.795850903891715)
(1.8870720316603153,25.30477000686348)
(1.8999092563654876,24.813613169724736)
(1.9127464810706596,24.322631978765575)
(1.9255837057758318,23.832078361119375)
(1.938420930481004,23.342204061005507)
(1.9512581551861763,22.85326011498795)
(1.9640953798913485,22.365496328578114)
(1.9769326045965208,21.879160756494734)
(1.989769829301693,21.394499188857676)
(2.002607054006865,20.911754645541038)
(2.0154442787120375,20.43116688084184)
(2.0282815034172095,19.952971900535953)
(2.041118728122382,19.477401493295833)
(2.053955952827554,19.004682778333844)
(2.0667931775327264,18.535037771016288)
(2.0796304022378984,18.068682968063882)
(2.092467626943071,17.605828953819415)
(2.105304851648243,17.14668002892447)
(2.1181420763534153,16.69143386260234)
(2.1309793010585873,16.240281169601293)
(2.14381652576376,15.793405412704264)
(2.156653750468932,15.350982531568322)
(2.1694909751741043,14.913180698511539)
(2.1823281998792763,14.480160101723937)
(2.1951654245844487,14.052072756239358)
(2.2080026492896208,13.629062342867964)
(2.2208398739947928,13.21126407515586)
(2.2336770986999652,12.798804594306008)
(2.2465143234051372,12.391801891868397)
(2.2593515481103097,11.990365259880438)
(2.2721887728154817,11.59459526801947)
(2.285025997520654,11.204583767208614)
(2.297863222225826,10.82041391900554)
(2.3107004469309986,10.442160249992087)
(2.3235376716361706,10.06988873027802)
(2.336374896341343,9.70365687513432)
(2.349212121046515,9.343513868676592)
(2.3620493457516876,8.989500708434681)
(2.3748865704568596,8.641650369567586)
(2.3877237951620316,8.299987987413557)
(2.400561019867204,7.964531057007224)
(2.4133982445723765,7.635289648147739)
(2.4262354692775485,7.312266634564534)
(2.4390726939827205,6.995457935702354)
(2.451909918687893,6.684852769632214)
(2.464747143393065,6.3804339155937)
(2.4775843680982375,6.082177984680854)
(2.4904215928034095,5.790055697205877)
(2.503258817508582,5.504032165301048)
(2.516096042213754,5.2240671793625495)
(2.5289332669189264,4.950115496982925)
(2.5417704916240984,4.682127133078087)
(2.554607716329271,4.4200476499735535)
(2.567444941034443,4.163818446283497)
(2.5802821657396153,3.913377043487983)
(2.5931193904447873,3.6686573691874513)
(2.6059566151499594,3.4295900360939338)
(2.618793839855132,3.1961026158954953)
(2.6316310645603043,2.9681199072125968)
(2.6444682892654763,2.7455641969440525)
(2.6573055139706483,2.5283555143808396)
(2.6701427386758207,2.3164118775437528)
(2.682979963380993,2.1096495312771775)
(2.695817188086165,1.9079831767061357)
(2.7086544127913372,1.7113261917343057)
(2.7214916374965097,1.5195908423277888)
(2.7343288622016817,1.332688484396801)
(2.747166086906854,1.1505297561449817)
(2.760003311612026,0.9730247608147823)
(2.7728405363171986,0.8000832398105363)
(2.7856777610223706,0.6316147362280353)
(2.798514985727543,0.46752874886579443)
(2.811352210432715,0.3077348768323226)
(2.8241894351378876,0.1521429549011002)
(2.8370266598430596,0.0006631797965242243)
(2.849863884548232,-0.0)
(2.9910733563051264,-0.0)
};
\end{scope}
\end{tikzpicture}

%% file: Figures/numericsfig3.tex
\begin{tikzpicture}[yscale=0.5,xscale= 0.8]
\draw[->] (-0.5, 0) -- (6, 0) node[below] {$S$};
        \draw[->] (0, -0.5) -- (0, 8) node[above] {$\lambda_S$};

        \draw[dashed] (0,7) -- (5, 7) ;
        \draw (50/140*5,-0.1) -- (50/140*5,0.1) node[below] {$50$};
        \draw (100/140*5,-0.1) -- (100/140*5,0.1) node[below] {$100$};
        \draw (150/140*5,-0.1) -- (150/140*5,0.1) node[below] {$150$};
        
        \draw (-0.1,7) -- (0.1,7) node[left] {$135$};
        \draw[blue, very thick, smooth] plot coordinates {
        ( 0.0 , 0.0 )
( 0.03571428571428571 , 0.10781918815822795 )
( 0.07142857142857142 , 0.14142524847024926 )
( 0.10714285714285715 , 0.18388537727454662 )
( 0.14285714285714285 , 0.23648309993296296 )
( 0.17857142857142855 , 0.300067604473632 )
( 0.2142857142857143 , 0.37470428521347243 )
( 0.25 , 0.4593463198268345 )
( 0.2857142857142857 , 0.5516640611760176 )
( 0.3214285714285714 , 0.6481683426899085 )
( 0.3571428571428571 , 0.7446741647484046 )
( 0.39285714285714285 , 0.8369966608622634 )
( 0.4285714285714286 , 0.9216470757734194 )
( 0.4642857142857143 , 0.9962964886769778 )
( 0.5 , 1.0598991836800855 )
( 0.5357142857142857 , 1.112522140829803 )
( 0.5714285714285714 , 1.1550167554483142 )
( 0.6071428571428571 , 1.18866957992521 )
( 0.6428571428571428 , 1.214917624819348 )
( 0.6785714285714285 , 1.2351573396211402 )
( 0.7142857142857142 , 1.2506395212756072 )
( 0.75 , 1.262427373102127 )
( 0.7857142857142857 , 1.2713940744900163 )
( 0.8214285714285714 , 1.2782414445364036 )
( 0.8571428571428572 , 1.2835275596779656 )
( 0.8928571428571429 , 1.2876963497787042 )
( 0.9285714285714286 , 1.2911057686144611 )
( 0.9642857142857143 , 1.2940533466380963 )
( 1.0 , 1.2967991960638907 )
( 1.0357142857142858 , 1.299587215130957 )
( 1.0714285714285714 , 1.3026655909962417 )
( 1.1071428571428572 , 1.3063078931006489 )
( 1.1428571428571428 , 1.3108361693324064 )
( 1.1785714285714286 , 1.316647534569272 )
( 1.2142857142857142 , 1.324245748080175 )
( 1.25 , 1.3342791183823206 )
( 1.2857142857142856 , 1.3475855614343022 )
( 1.3214285714285714 , 1.365244440623186 )
( 1.357142857142857 , 1.3886324160224583 )
( 1.3928571428571428 , 1.4194762020982594 )
( 1.4285714285714284 , 1.4598880632444682 )
( 1.4642857142857144 , 1.5123596032140434 )
( 1.5 , 1.579676892603166 )
( 1.5357142857142858 , 1.664709564405005 )
( 1.5714285714285714 , 1.7700282557059246 )
( 1.6071428571428572 , 1.8973348421309775 )
( 1.6428571428571428 , 2.0467636693070546 )
( 1.6785714285714284 , 2.2162248284618693 )
( 1.7142857142857144 , 2.401063045255674 )
( 1.75 , 2.59430438368113 )
( 1.7857142857142858 , 2.78758370227164 )
( 1.8214285714285714 , 2.972536660656027 )
( 1.8571428571428572 , 3.1421917391828864 )
( 1.8928571428571428 , 3.291897734181347 )
( 1.9285714285714286 , 3.419570526382129 )
( 1.9642857142857142 , 3.5253520587309812 )
( 2.0 , 3.610953723683717 )
( 2.0357142857142856 , 3.678957726220103 )
( 2.0714285714285716 , 3.7322474651771946 )
( 2.107142857142857 , 3.773625130227767 )
( 2.142857142857143 , 3.805600950680209 )
( 2.178571428571429 , 3.8303084587980263 )
( 2.2142857142857144 , 3.8494983753577916 )
( 2.25 , 3.8645741297122207 )
( 2.2857142857142856 , 3.876644525376162 )
( 2.3214285714285716 , 3.8865793175846672 )
( 2.357142857142857 , 3.895060515019224 )
( 2.392857142857143 , 3.902626514234091 )
( 2.4285714285714284 , 3.909708534437899 )
( 2.4642857142857144 , 3.9166599639613375 )
( 2.5 , 3.9237796753377134 )
( 2.5357142857142856 , 3.9313304434504173 )
( 2.571428571428571 , 3.939553504607166 )
( 2.607142857142857 , 3.9486801297312066 )
( 2.642857142857143 , 3.958940905515486 )
( 2.6785714285714284 , 3.9705732470664903 )
( 2.714285714285714 , 3.9838275121712248 )
( 2.75 , 3.998971950986221 )
( 2.7857142857142856 , 4.016296602536326 )
( 2.821428571428571 , 4.036116137039171 )
( 2.8571428571428568 , 4.058771537604031 )
( 2.8928571428571432 , 4.084630414953683 )
( 2.928571428571429 , 4.1140856558536925 )
( 2.9642857142857144 , 4.147552024885144 )
( 3.0 , 4.185460279526872 )
( 3.0357142857142856 , 4.228248334812408 )
( 3.0714285714285716 , 4.276349045844897 )
( 3.107142857142857 , 4.330174287928446 )
( 3.142857142857143 , 4.390095229634201 )
( 3.1785714285714284 , 4.456419033562609 )
( 3.2142857142857144 , 4.5293626892458745 )
( 3.25 , 4.609025264632849 )
( 3.2857142857142856 , 4.695360503891699 )
( 3.321428571428571 , 4.788152305083615 )
( 3.3571428571428568 , 4.886996048221025 )
( 3.3928571428571432 , 4.991288858617927 )
( 3.428571428571429 , 5.1002315437134165 )
( 3.4642857142857144 , 5.2128440581628395 )
( 3.5 , 5.327994966159344 )
( 3.535714285714286 , 5.44444365136216 )
( 3.5714285714285716 , 5.5608922646507 )
( 3.607142857142857 , 5.676042950383426 )
( 3.642857142857143 , 5.7886550720654375 )
( 3.678571428571429 , 5.897597158275074 )
( 3.7142857142857144 , 6.001889109362826 )
( 3.75 , 6.100731654848804 )
( 3.7857142857142856 , 6.193521811448651 )
( 3.821428571428571 , 6.279854810050122 )
( 3.857142857142857 , 6.359514345541981 )
( 3.892857142857143 , 6.432453886450526 )
( 3.9285714285714284 , 6.49877212762145 )
( 3.964285714285714 , 6.558685554195182 )
( 4.0 , 6.6125006473634205 )
( 4.035714285714286 , 6.660587655498697 )
( 4.071428571428571 , 6.703357211503531 )
( 4.107142857142857 , 6.741240493269817 )
( 4.142857142857143 , 6.774673151817832 )
( 4.178571428571429 , 6.804082888685101 )
( 4.214285714285714 , 6.829880343785669 )
( 4.25 , 6.852452837037271 )
( 4.285714285714286 , 6.872160466364434 )
( 4.321428571428571 , 6.889334076739581 )
( 4.357142857142858 , 6.904274658978506 )
( 4.392857142857142 , 6.917253797168758 )
( 4.428571428571429 , 6.928514848730168 )
( 4.464285714285714 , 6.938274604057916 )
( 4.5 , 6.946725229443671 )
( 4.535714285714286 , 6.954036345664811 )
( 4.571428571428571 , 6.960357134868425 )
( 4.607142857142857 , 6.965818400606946 )
( 4.642857142857143 , 6.9705345310058835 )
( 4.678571428571429 , 6.974605334136455 )
( 4.714285714285714 , 6.978117728811334 )
( 4.75 , 6.981147284218927 )
( 4.785714285714286 , 6.983759608931535 )
( 4.821428571428571 , 6.986011594594227 )
( 4.857142857142857 , 6.9879525226170385 )
( 4.892857142857142 , 6.989625043923613 )
( 4.928571428571429 , 6.991066042620133 )
( 4.964285714285714 , 6.992307394618931 )
};
\end{tikzpicture}

%% file: Tex/Benchmarks.tex
\section{Benchmark Analysis}\label{sec:benchmark}
In this section, we consider \revcolor{five} alternative operating policies that we use to benchmark the performance of the optimal \revcolor{seniority-based} policy derived earlier. 
\revcolor{
\begin{enumerate}
    \item \textbf{Fixed Intermediation Fee Policy}: This policy restricts the platform to a static fee structure where the intermediation fee is invariant to the platform's state, setting $F(k) = F$ for all $k \leq \size$. This serves as a benchmark for platforms that favor pricing simplicity over state-dependent optimization.

    \item \textbf{Uniform Assignment Policy}: This policy distributes incoming jobs to active servers uniformly at random, bypassing any tenure structures. This approach reflects standard ride-sharing practices where dispatches are dictated by local geographical proximity, which effectively randomizes job distribution relative to server seniority.

    \item \textbf{Fixed-Fee Uniform Policy}: This policy combines the operational constraints of both a static fee structure ($F(k) = F$) and randomized job matching. Because it lacks both dynamic fee adjustment capabilities and seniority-based assignment, this regime represents the intuitive baseline status quo of some decentralized marketplaces, allowing us to explicitly evaluate its degree of suboptimality.

    \item \textbf{Greedy Policy}: This policy forces the platform to expand as rapidly as possible toward the static market-clearing size $\mathcal{S}^*$. It benchmarks a myopic operational approach that focuses strictly on immediate capacity targets while ignoring the intertemporal costs of transitioning through intermediate states.
    
    \item \textbf{Maximum Growth Policy}: This policy maximizes the ultimate scale of the platform, subject only to basic economic solvency. It represents a strategic regime where market-share acquisition and long-term platform scale are prioritized over short-term profit maximization.
\end{enumerate}
Through a series of numerical experiments, we expose the structural shortcomings of these alternative regimes and quantify their performance gaps relative to our optimal seniority-based mechanism. To maintain consistency across the numerical experiments we use an MSL function (see \cref{def:MSL}) similar to the one in \Cref{subsec:Sensitivity} with the following parameters: $\Lambda_n=\,(25, 50, 60)$, $s_n=\,(10, 50, 100)$, $\xi_n=0.3\,(1, 1, 0.5)$, $m=1$, $r=0.01$, $p=1$, and $w=100\,p$. For a visual illustration of this function, see \Cref{Fig:fixedFComparison}. In \Cref{Table:benchmark_figure}, we compare the seniority-based assignment with these benchmarks. In this example, the maximum growth policy (\Cref{sec:maxSize}) can achieve the size \revcolor{$S^*=135$}, but no profit-maximizing policy grows further than \revcolor{$S^*=112$} as the demand is almost flat after that.}

\revcolor{\begin{table}[htbp]
\centering
\caption{\revcolor{Comparison of the optimal seniority-based policy against operational benchmarks.}}
\label{Table:benchmark_figure}
\begin{tabular}{lcc}
\toprule
\textbf{Policy} & \textbf{Target Size ($S^{\pi}$)} & \textbf{Normalized Platform Value ($J^{\pi}_0$)} \\
\midrule
Optimal             & 58    & 1 \\
Fixed Intermediation Fee       & 58 & 0.94 \\
Uniform Assignment         & 58 & 0.92 \\
Fixed-Fee Uniform Policy   & 60 & 0.79 \\
Greedy          & 112 & 0.80 \\
Maximum Growth        & 135    & 0.26 \\
\bottomrule
\end{tabular}
\end{table}}

 \vspace{-.3cm}
\subsection{Fixed Intermediation Fee}\label{sec:fixedF}
Suppose the platform restricts itself to charging a fixed intermediation fee $F(k) = F$ for all $k$. After some manipulations, the platform's optimization program \eqref{mainProgram} reduces to 
\begin{align*}\label{mainProgramAggFixed1}
\tag{$\mathcal{F}_{\texttt{Main}}$}
\;&\; J^*_{\mathcal{F}}= \sup_{\alpha_k,\gamma_s(k),\size,F} \;J_0 = \sum_{\ell = 0}^{\size} (\alpha_{\ell}-\alpha_{\ell+1}) \frac{F\lambda_{\ell}}{r}  \\
    \text{subject to} \;&\; v^{\pi}_s(k) = \sum_{\ell = k}^{\size} (\alpha_{\ell}-\alpha_{\ell+1}) \gamma_s(\ell) \geq \alpha_k w   && \mbox{for}\quad s \leq k \leq \size \\
    \;&\; \sum_{s=1}^k \gamma_s(k) = 
    \frac{(p-F)\lambda_k}{r} && \mbox{for} \quad k \leq \size\\
    \:&\; 0 \leq \alpha_{k+1} \leq \frac{m_k}{m_k+r}\alpha_k \quad \mbox{and}  \quad 0 \leq F \leq p  && \mbox{for} \quad  0 \leq k \leq \size.
\end{align*}
In the following example we show with a given $F$, growing as quickly as possible ($\alpha_{k+1} = \frac{m_k}{r+m_k}\alpha_k$ for all $k<\size$) is not the optimal approach \revcolor{anymore}. 

\begin{exm}\label{badexampl1} 
Consider a market with a maximum size of $\size=3$ and parameters $m_k = p = w = r = 1$. Let $\lambda_1 = 0.5, \lambda_2 = 5.8, \lambda_3 = 6$ and $F$ fixed at $0.5$. Given that $\imm(k)$ is strictly increasing, the platform is willing to grow to the maximum size as fast as possible. This means $\alpha_1 = \frac{1}{2}, \alpha_2 = \frac{1}{4}$ and $\alpha_3 = \frac{1}{8}$ which is not feasible since:
\begin{align*}
    \sum_{\ell = 1}^{3} (\alpha_{\ell}-\alpha_{\ell+1}) \left(\frac{(p-F)\lambda_{\ell}}{r} - w \ell \right) = \frac{1}{4} \left(\frac{0.5}{2} - 1\right)+ \frac{1}{8}\left(\frac{5.8}{2} - 2\right) + \frac{1}{8}\left(\frac{6}{2} - 3\right) < 0
\end{align*}
Growing to 2 servers is feasible, however, since:
\begin{align*}
    \sum_{\ell = 1}^{2} (\alpha_{\ell}-\alpha_{\ell+1}) \left(\frac{(p-F)\lambda_{\ell}}{r} - w \ell \right) = \frac{1}{4} \left(\frac{0.5}{2} - 1\right)+ \frac{1}{4}\left(\frac{5.8}{2} - 2\right) \geq 0
\end{align*}
This policy generates payoff of $J_{\mathcal{F}} = \frac{6.3}{8}$ for the platform. But if the platform decides to pause at size $2$ and then move to size $3$, by setting $\alpha_3 = \frac{1}{24}$ (or $u(3) = \revcolor{\frac{1}{5}}$) it can satisfy DCR constraints and also achieve $J^*_{\mathcal{F}} = \frac{6.33}{8}$.
\end{exm}

This phenomenon occurs because the platform's total revenue increases with demand, but the revenue per individual server does not. Specifically, while $\lambda_3$ is only slightly higher than $\lambda_2$, it must be divided among three servers instead of two. This makes it more challenging for the platform to manage its liabilities at this size. As a result, the platform may choose to remain at stage 2 to address its existing liabilities before progressing to the next stage. \revcolor{By remaining at stage 2, the platform operates where the available budget per server is higher. In contrast, growing prematurely to stage 3 makes the budget tighter and the liability constraints more binding.}

The above example demonstrates that pausing growth at specific milestones to better manage liabilities can be a prudent and necessary practice. It is also worth reiterating that this need for intermittent growth arises from the platform fixing its intermediation fee, highlighting the advantage of using dynamic fees—which not only maximize the platform’s payoff \revcolor{but also allow the platform to strategically lower fees in critical states to facilitate rapid growth.}

\revcolor{Finally, while \eqref{mainProgramAggFixed1} is not a convex program due to the structural coupling of the decision variables, it can be efficiently solved via a Fully Polynomial-Time Approximation Scheme (FPTAS). We achieve this by performing a line search over the bounded fee space $F \in [0,p]$. 

Crucially, fixing the fee $F$ does not immediately render the optimization program linear, as bilinear terms of the form $(\alpha_{\ell}-\alpha_{\ell+1})\gamma_s(\ell)$ persist in both the objective function and the participation constraints. To resolve this non-linearity, we introduce the linearized decision variables:
$x_s(\ell) \equiv (\alpha_{\ell}-\alpha_{\ell+1})\gamma_s(\ell).$
Substituting $x_s(\ell)$ back into \eqref{mainProgramAggFixed1} decouples the allocation rules from the transition dynamics, transforming the inner optimization into an LP that can be solved efficiently for any fixed value of $F$.

Let $J^*_{\mathcal{F}}(F)$ denote the platform's optimized payoff for a given fee $F$. To establish the theoretical performance guarantees of our grid-search scheme, we analyze the sensitivity of the objective function with respect to $F$. Specifically, for any two fees satisfying $F_2 \leq F_1$, the following structural relationship holds:
$$J^*_{\mathcal{F}}(F_2) \geq \frac{F_2}{F_1} J^*_{\mathcal{F}}(F_1).$$ 
This property arises because lowering the intermediation fee from $F_1$ to $F_2$ expands the instantaneous budget available for server compensation. This expansion relaxes DIR constraints, ensuring that the original assignment policy remains strictly feasible under the lower fee. Consequently, by evaluating a geometric grid of fees scaled by $(1+\epsilon)$, this structural monotonicity guarantees that we can find an $\epsilon$-optimal fee by solving at most $O\left(\frac{1}{\log(1+\epsilon)}\right) = O\left(\frac{1}{\epsilon}\right)$ linear programs. Because the inner optimization problem reduces to an LP, the overall grid search constitutes a Fully Polynomial-Time Approximation Scheme (FPTAS) that circumvents the non-convexity of the original formulation with minimal computational overhead.} In \cref{Fig:fixedFComparison}, we compare the fixed compensation policy for different values of $F$ in the range \revcolor{($0,0.4$)}. \revcolor{The optimal policy is achieved at $F^* \approx 0.15$ with a corresponding size of $\size^* = 58$. }

\begin{figure}[h!]
\centering
\input{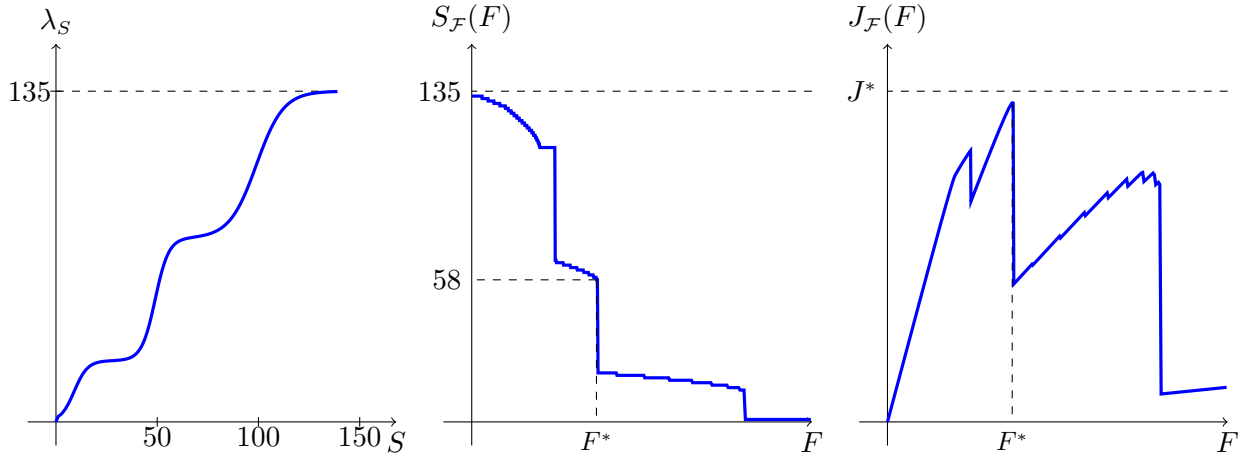}
\caption{\footnotesize \sf \revcolor{The maximum size (middle) and payoff (right) of the platform with a fixed compensation fee $F$. The range of x-axis is ($0,0.4$) and $F^* \approx 0.15$.}}  
\label{Fig:fixedFComparison}
\end{figure}

 \vspace{-.3cm}
\subsection{\revcolor{Uniform Assignment}}\label{sec:uniformBenchmark} 

Another benchmark we consider is a compensation scheme that does not reward seniority and instead provides equal pay to all servers. \revcolor{Under this policy, while the payout rate can vary dynamically across different market states, all servers present on the platform in any given state receive an identical payment regardless of their historical tenure. This policy can be enforced} by introducing an additional constraint, which requires $\gamma_s(k) = \gamma_{s+1}(k)$. Given the equal compensation policy, it is sufficient to satisfy the DIR constraint for server $s = 1$, as DIR for the other servers follows trivially. This simplifies the problem and reduces the number of constraints from $O(\size^2)$ to $O(\size)$:
\begin{align*}\label{UnifromProgram2}
\tag{\revcolor{$\mathcal{P}_{\texttt{Uniform}}$}}
     \sup_{\alpha_k,\gamma(k),\size} \;J_0 &= \sum_{\ell = 1}^{\size} (\alpha_{\ell}-\alpha_{\ell+1}) \left(\frac{p\lambda_{\ell}}{r} - \ell \gamma(\ell) \right) \\
    \text{subject to}\quad \;&\; v^{\pi}_1(k) = \sum_{\ell = k}^{\size} (\alpha_{\ell}-\alpha_{\ell+1}) \gamma(\ell) \geq \alpha_k w   && \mbox{for}\quad  k \leq \size \\
    \;&\; k \gamma(k) \leq 
    \frac{p\lambda_k}{r} && \mbox{for} \quad k \leq \size\\
    \:&\; 0 \leq \alpha_{k+1} \leq \frac{m_k}{m_k+r}\alpha_k && \mbox{for} \quad  0 \leq k \leq \size.
\end{align*}
\revcolor{As before, this problem can} 
be reformulated as a linear program with decision variables $\alpha_k$ and $(\alpha_k - \alpha_{k+1})\gamma_s(k)$. \revcolor{However, similar to the fixed intermediation fee case (\Cref{sec:fixedF}), one can construct an example\footnote{\revcolor{An example similar to \Cref{badexampl1} with $m_k = 10$, $\lambda_1 = 0.1$, $\lambda_2 = 3.8$, and $\lambda_3 = 4.9$ would work.}} that shows that rapid growth is not necessarily optimal under uniform assignment.}

While intuitively appealing, this cross-sectional fairness is only superficial because it violates inter-temporal equity. For instance, if demand $\lambda_s$ is a step function that remains zero until a server threshold $s^*$ is reached, early servers receive no immediate income. To incentivize their early participation, they must be compensated back later via higher payouts in mature states. Yet, a uniform compensation policy gives these exact same high payouts to servers who join after $s^*$, allowing late-comers to enjoy the upside of maturity without bearing any of the early-stage risk. This failure to differentiate by tenure ultimately results in a highly inequitable distribution. \revcolor{In \Cref{fig:uniform}, we illustrate this phenomenon by plotting the continuation values of servers at the moment they join the platform. We present this across two scenarios: first, under our baseline demand configuration introduced at the beginning of this section, and second, under an extreme case featuring a sharp single jump in demand ($\Lambda_1 = 100$, $\xi_1 = \infty$, $s_1 = 10$).}

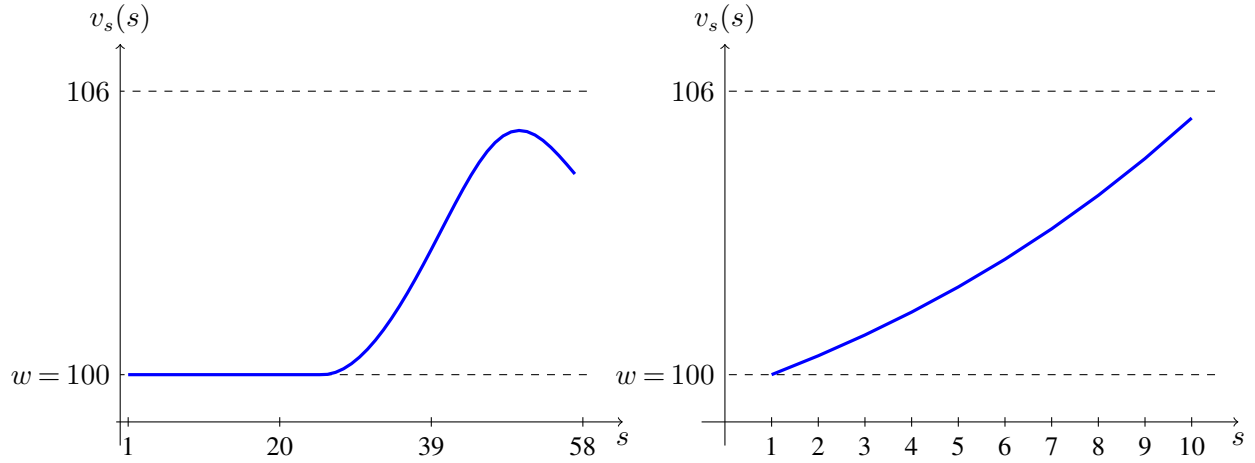
\begin{figure}[h]
\centering
\input{Figures/uniformBenchmark}
\caption{\footnotesize \sf \revcolor{Server continuation values $v_s(s)$ at the time of arrival under the optimal uniform compensation scheme. The left panel plots the baseline configuration where the platform optimally grows to a terminal size of $\size^* = 58$. The right panel illustrates the extreme step-demand scenario ($\Lambda_1 = 100$, $\xi_1 = \infty$, $s_1 = 10$), where early arrivals face uncompensated waiting periods. Consequently, a uniform policy forces an inequitable, strictly increasing value distribution that heavily over-compensates late-coming servers at maturity.}}
\label{fig:uniform}
\end{figure}

 \vspace{-.3cm}
\subsection{\revcolor{Fixed-Fee Uniform Policy}}
\revcolor{This policy is the most constrained operating regime, where the platform's only strategic levers are its target terminal size $\size$ and its state-by-state growth rates. Under this policy, the platform charges a fixed intermediation fee $F$, and the remaining revenues are distributed uniformly among all active servers. Similar to the fixed intermediation fee benchmark (\Cref{sec:fixedF}), finding the global optimum reduces to a one-dimensional line search over the possible platform sizes $\size$. For each candidate size, the corresponding sub-problem reduces to an even simpler LP than the fixed $F$ model, which can be solved efficiently to yield an FPTAS.}

 \vspace{-.3cm}
\subsection{\revcolor{Greedy}}
\revcolor{As an additional baseline, we examine a myopic ``greedy'' policy where the platform completely disregards the strategic pacing of its transitional growth phase. Instead, the platform immediately targets the static revenue-maximizing scale $\size^* = \argmax_{\size \in \mathbb{N}_0} \{R(\size)\}$ (see \eqref{eq:OptClearing}) and expands toward this size at the maximum possible speed. 

Mathematically, this implies that the platform's transition dynamics are frozen at their upper bound, forcing $\alpha_{k+1} = \frac{m_k}{r+m_k}\alpha_k$ for all $k < \size^*$. Under this regime, although the platform still implements seniority-based compensation and dynamic fee adjustments, it selects a suboptimal target size based purely on static parameters and forfeits the ability to utilize strategic growth delays. Consequently, this benchmark allows us to isolate and quantify the economic value of explicitly accounting for gradual growth during the network's initialization phase.

To evaluate the economic impact of this myopic framework, \Cref{fig:comparison_m} compares the terminal size and value generated by the optimal policy against the greedy benchmark across a spectrum of the server arrival rate $m$. Under the greedy policy, the platform locks into the optimal static scale ($112$), completely ignoring growth frictions. In contrast, the optimal policy dynamically scales intermediate targets ($\{1, 15, 58, 112\}$) to smooth the growth trajectory under server-scarce environments. Consequently, as illustrated in the right panel, the greedy policy incurs severe deficits when $m$ is low, ultimately failing to remain operational for sufficiently small values of $m$. As the arrival rate increases, these growth frictions disappear, causing the target sizes and total values of both regimes to converge.}

\begin{figure}[h]
    \centering
    \input{Figures/arrivalRate}
    \caption{\footnotesize \sf \revcolor{Comparative performance of the optimal operating policy (blue) versus the myopic greedy baseline (black) as a function of the potential server arrival rate $m$ (logarithmic scale). The left panel depicts the optimal terminal platform scale $S^*$, while the right panel displays the corresponding platform initialization value $J^*$.}}
    \label{fig:comparison_m}
\end{figure}
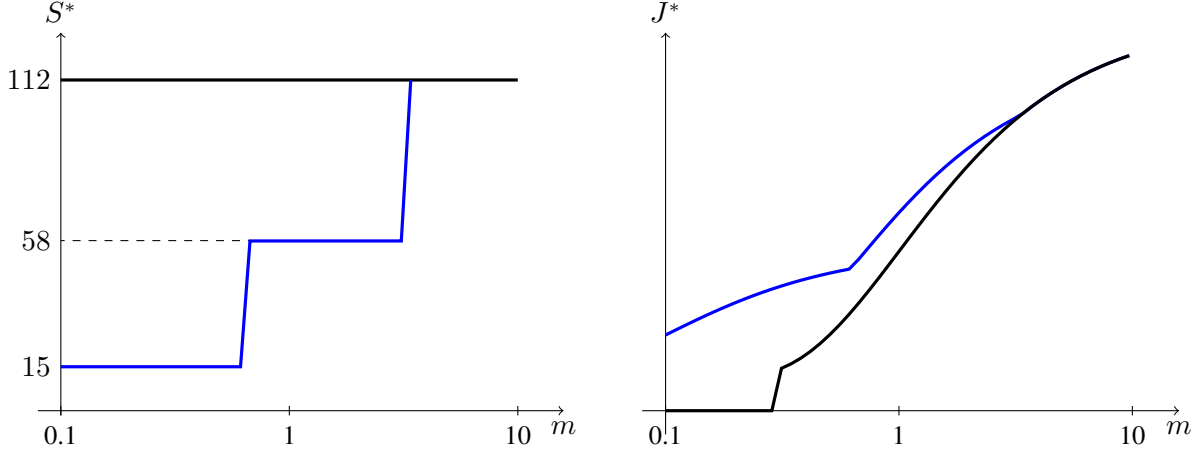

 \vspace{-.3cm}
\subsection{\revcolor{Maximum Growth}}\label{sec:maxSize}
\revcolor{Another} possible objective for the platform is to grow to its maximum size in order to establish market dominance, take advantage of potential economies of scale, and better exploit network effects. Setting aside profit considerations (other than the requirement to ensure a non-negative net present value) the maximum size is achieved when  
$\imm(k) = F(k) = 0$. This  simplifies the problem and a feasible solution only needs to meet a relaxed version of the DIR constraint. 
This problem is similar to a special case of the fixed intermediation fee \revcolor{(\Cref{sec:fixedF})} problem where $F = 0$. The key difference lies in the objective: here, we focus solely on maximizing the size rather than the platform's net revenue. 
\revcolor{To solve it, we can decouple $\size$ from the optimization process by treating it as a fixed parameter. This transforms the problem into a sequence of parametric feasibility checks. For any candidate size $\size$, we rewrite the objective function to maximize the terminal state transition parameter, solving $\sup \alpha_{\size}$. If this auxiliary linear program yields a strictly positive solution ($\alpha_{\size}^* > 0$), a valid growth trajectory exists that satisfies all server participation constraints up to state $\size$, confirming that this platform size is achievable. To determine the exact maximum sustainable platform size $\size^*$, we can perform a one-dimensional line search across the discrete range of potential scales up to a practical upper bound. This reduces the global optimization problem to evaluating a finite sequence of standard linear programs, preserving the computational tractability of our solution framework.}

%% file: Figures/uniformBenchmark.tex
\begin{tikzpicture}
\begin{scope}[shift={(0, 0)}, yscale = 5/8, xscale = 0.5*0.95/63*14]
    \draw[->] (-0.5 ,0) -- (63, 0) node[below] {$s$};
    \draw[->] (0,-0.5) -- (0, 8) node[above] {$v_s(s)$};
    
    \foreach \x in {1, 20, ..., 58} {
        \draw (\x, 0.1) -- (\x, -0.1) node[below, font=\small] {\x};
    }
    
    \draw[dashed] (58.5, 7) -- (0,7) node[left] {$106$};
    \draw[dashed] (58.5, 1) -- (0,1) node[left] {$w=100$};
    
    \draw[blue, very thick] plot coordinates {
        ( 1 , 0.9999999999999858 )
( 25 , 1.0 )
( 26 , 1.0066342444564214 )
( 27 , 1.0490122420495993 )
( 28 , 1.1251785454428784 )
( 29 , 1.2332268029961426 )
( 30 , 1.371301563426698 )
( 31 , 1.5375866991348346 )
( 32 , 1.7302825404644864 )
( 33 , 1.9475724959589655 )
( 34 , 2.187578882533913 )
( 35 , 2.4483068730784936 )
( 36 , 2.727574917394506 )
( 37 , 3.0229298300858005 )
( 38 , 3.331545230764519 )
( 39 , 3.650103629827697 )
( 40 , 3.9746658848685854 )
( 41 , 4.300537941795824 )
( 42 , 4.622154657681051 )
( 43 , 4.933014365817627 )
( 44 , 5.225713973945318 )
( 45 , 5.49214704056692 )
( 46 , 5.723925007202183 )
( 47 , 5.9130490543038405 )
( 48 , 6.052786378152817 )
( 49 , 6.138601411251258 )
( 50 , 6.168905525081115 )
( 51 , 6.145383723741446 )
( 52 , 6.07276907585684 )
( 53 , 5.958123456753626 )
( 54 , 5.809844329826603 )
( 55 , 5.636667585131249 )
( 56 , 5.446871850778194 )
( 57 , 5.247772322609194 )
    };
\end{scope}

\begin{scope}[shift={(8, 0)}, yscale = 5/8, xscale = 0.65*0.95]
    \draw[->] (-0.5 ,0) -- (11, 0) node[below] {$s$};
    \draw[->] (0,-0.5) -- (0, 8) node[above] {$v_s(s)$};

    \foreach \x in {1, 2, ..., 10} {
        \draw (\x, 0.1) -- (\x, -0.1) node[below, font=\small] {\x};
    }

    \draw[dashed] (10.5, 7) -- (0,7) node[left] {$106$};
    \draw[dashed] (10.5, 1) -- (0,1) node[left] {$w=100$};
    
    \draw[blue, very thick] plot coordinates {
        ( 1 , 1.0 )
        ( 2 , 1.4000 )
        ( 3 , 1.8400 )
        ( 4 , 2.3240 )
        ( 5 , 2.8564 )
        ( 6 , 3.4420 )
        ( 7 , 4.0862 )
        ( 8 , 4.7949 )
        ( 9 , 5.5744 )
        ( 10 , 6.4318 )
    };
\end{scope}
\end{tikzpicture}

%% file: Figures/arrivalRate.tex
\begin{tikzpicture}
\begin{scope}[shift={(0, 0)}, yscale = 5/8, xscale = 0.5*0.95*14/2.2]
    \draw[->] (-0.1 ,0) -- (2.2, 0) node[below] {$m$};
    \draw[->] (0,-0.1) -- (0, 8) node[above] {$S^*$};
    \foreach \power/\label in {0/0.1, 1/1, 2/10} {
    \draw (\power, 0.1) -- (\power, -0.1) node[below, font=\small] {\label};
    }
    
    \draw[dashed] (0.01, 7) -- (0,7) node[left] {$112$};
    \draw[dashed] (1-0.0066, 3.6) -- (0,3.6) node[left] {$58$};
    \draw[dashed] (1-0.2135,  0.93 ) -- (0, 0.93 ) node[left] {$15$};
    
    \draw[black, very thick] plot coordinates {
        (0, 7)
        (2, 7)
    };
    \draw[blue, very thick] plot coordinates {
        ( 0 , 0.93 )
        ( 0.7864610180062763 ,  0.93 )
( 0.8278537031645015 , 3.592920353982301 )
( 1.4901366656961026 , 3.592920353982301 )
( 1.5315293508543277 , 7)
    };
\end{scope}

\begin{scope}[shift={(8, 0)}, yscale = 5/8, xscale = 0.65*0.95*5]
    \draw[->] (-0.1 ,0) -- (2.2, 0) node[below] {$m$};
    \draw[->] (0,-0.5) -- (0, 8) node[above] {$J^*$};

    \foreach \power/\label in {0/0.1, 1/1, 2/10} {
    \draw (\power, 0.1) -- (\power, -0.1) node[below, font=\small] {\label};
    }

    \draw[blue, very thick] plot coordinates {
        ( 2.220446049250313e-16 , 1.6022590633230471 )
( 0.04139268515822514 , 1.70352013791602 )
( 0.08278537031645017 , 1.80364447496973 )
( 0.12417805547467542 , 1.9019827988261473 )
( 0.16557074063290034 , 1.9979651474173306 )
( 0.20696342579112548 , 2.091105310615759 )
( 0.2483561109493505 , 2.1810017843311074 )
( 0.28974879610757553 , 2.2673358399466355 )
( 0.3311414812658007 , 2.3498673363168794 )
( 0.3725341664240257 , 2.4284288814967185 )
( 0.41392685158225073 , 2.502918898156181 )
( 0.45531953674047576 , 2.5732940733578995 )
( 0.4967122218987008 , 2.6395615908706866 )
( 0.5381049070569258 , 2.701771460734688 )
( 0.5794975922151508 , 2.760009182167136 )
( 0.620890277373376 , 2.8143889057165756 )
( 0.6622829625316011 , 2.8650472006967997 )
( 0.703675647689826 , 2.9121374849005988 )
( 0.7450683328480512 , 2.9558251350366676 )
( 0.7864610180062763 , 2.996283267303205 )
( 0.8278537031645015 , 3.207479737856361 )
( 0.8692463883227265 , 3.4508095329554114 )
( 0.9106390734809515 , 3.691155557698743 )
( 0.9520317586391767 , 3.9266555146901405 )
( 0.9934244437974017 , 4.155752739972441 )
( 1.0348171289556267 , 4.37718619774218 )
( 1.0762098141138519 , 4.589972369741013 )
( 1.117602499272077 , 4.79338173775071 )
( 1.158995184430302 , 4.986912114144708 )
( 1.200387869588527 , 5.170260626407734 )
( 1.2417805547467522 , 5.34329574018912 )
( 1.2831732399049771 , 5.506030331269357 )
( 1.3245659250632023 , 5.658596498860781 )
( 1.3659586102214274 , 5.8012225527253625 )
( 1.4073512953796525 , 5.934212401682624 )
( 1.4487439805378775 , 6.057927415495655 )
( 1.4901366656961026 , 6.172770718836915 )
( 1.5315293508543277 , 6.293415090506224 )
( 1.5729220360125526 , 6.45052465407807 )
( 1.6143147211707778 , 6.597105476400705 )
( 1.655707406329003 , 6.733542306263272 )
( 1.6971000914872278 , 6.860266074587837 )
( 1.738492776645453 , 6.977739686167997 )
( 1.7798854618036781 , 7.08644595932789 )
( 1.821278146961903 , 7.186877549965957 )
( 1.8626708321201284 , 7.279528672071962 )
( 1.9040635172783533 , 7.364888416742798 )
( 1.9454562024365782 , 7.443435472089872 )
( 1.9868488875948036 , 7.515634054064354 )
    };

    \draw[black, very thick] plot coordinates {
    (0 , 0 )( 0.08278537031645017 , 0.0 )
( 0.45531953674047576 , 0.0 )
( 0.4967122218987008 , 0.8969310628478777 )
( 0.5381049070569258 , 0.9926047428869135 )
( 0.5794975922151508 , 1.112085618817526 )
( 0.620890277373376 , 1.2552443258194559 )
( 0.6622829625316011 , 1.4213117114425438 )
( 0.703675647689826 , 1.6089230410184472 )
( 0.7450683328480512 , 1.8161919124753432 )
( 0.7864610180062763 , 2.040806805209272 )
( 0.8278537031645015 , 2.2801415390063196 )
( 0.8692463883227265 , 2.5313706505883835 )
( 0.9106390734809515 , 2.791581547667732 )
( 0.9520317586391767 , 3.057876894300456 )
( 0.9934244437974017 , 3.3274626270282184 )
( 1.0348171289556267 , 3.5977189699822434 )
( 1.0762098141138519 , 3.8662535748581806 )
( 1.117602499272077 , 4.130937323737975 )
( 1.158995184430302 , 4.389924349464146 )
( 1.200387869588527 , 4.641658462514979 )
( 1.2417805547467522 , 4.884868476834938 )
( 1.2831732399049771 , 5.118554970249565 )
( 1.3245659250632023 , 5.341970871851772 )
( 1.3659586102214274 , 5.5545980068273995 )
( 1.4073512953796525 , 5.756121404879461 )
( 1.4487439805378775 , 5.9464028348308915 )
( 1.4901366656961026 , 6.125454695469604 )
( 1.5315293508543277 , 6.293415090506224 )
( 1.5729220360125526 , 6.450524654078089 )
( 1.6143147211707778 , 6.597105476400685 )
( 1.655707406329003 , 6.733542306263272 )
( 1.6971000914872278 , 6.860266074587866 )
( 1.738492776645453 , 6.9777396861680065 )
( 1.7798854618036781 , 7.086445959327929 )
( 1.821278146961903 , 7.186877549965976 )
( 1.8626708321201284 , 7.2795286720719234 )
( 1.9040635172783533 , 7.364888416742808 )
( 1.9454562024365782 , 7.443435472089901 )
( 1.9868488875948036 , 7.515634054064364 )
    };
\end{scope}
\end{tikzpicture}

%% file: Tex/Extensions.tex
\revcolor{\section{Extension: Financial Instruments and Server Impatience}\label{sec:extension}
In this section, we extend the base model in two directions. First, we allow the platform to access limited financing that can be used to subsidize early-stage growth. Second, we allow for heterogeneous patience levels, with servers being  more impatient than the platform.

We model the platform's financing option by allowing its total liability to grow up to a specified threshold $\Theta$. This captures settings in which the platform has limited access to external financing and can use it to subsidize early-stage growth by temporarily operating at a loss. Notably, the base model's state-by-state financial self-sufficiency requirement, $\Gamma(k) \geq 0$ in \eqref{eq:positiveCommision}, is stricter than imposing a debt capacity of $\Theta=0$. While $\Theta=0$ rules out net external debt along the growth path, it still allows the platform to use accumulated past revenues to absorb losses in some intermediate states. By varying $\Theta$, this extension lets us study how intertemporal revenue shifting and external debt capacity affect growth trajectories, recruitment timing, and compensation policies.

The difference in patience levels is modeled by assigning servers a higher discount rate $\tilde{r} \geq r$ (with $r$  defined as before for the platform). This setup reflects that individual service providers often have short-term liquidity needs (e.g., daily living expenses) that lead them to value immediate compensation more than future earnings. It also accounts for behavioral and risk considerations, such as myopic decision-making or the perception that distant rewards are more uncertain. In the original model, we had $\tilde{r} = r$, so this extension broadens the applicability to environments where the platform can plan over a long horizon while service providers prioritize near-term returns.
\begin{align*}\label{extendedProgram}
\tag{$\mathcal{P}_{\texttt{Extension}}$}
    \;&\; J^*=\sup_{\pi \in \Pi} \;J^{\pi}_0\\
    \text{subject to}\;&\; J^\pi_{k} = \frac{r}{r+\mu_{k}} \imm (k)+ \frac{\mu_{k}}{r+\mu_{k}}\,J^\pi_{k+1} &&\mbox{for}\quad k \leq \size \quad\mbox{and}\quad   J^\pi_{\size+1}=0 \\
    \;&\; v^{\pi}_s(k) = \frac{\tilde{r}}{\tilde{r}+\mu_{k}} \gamma_s (k) + \frac{\mu_{k}}{\tilde{r}+ \mu_{k} }\, v^{\pi}_s(k+1) &&\mbox{for}\quad k \leq \size \quad\mbox{and}\quad v^{\pi}_s(\size+1)=0 \\
    \;&\; r \imm(k) + \tilde{r} \sum_{s=1}^k \gamma_s(k) = p\demand{k} &&\mbox{for} \quad  k \leq \size\\
    \;&\; v^{\pi}_s(k) \geq \frac{r}{\tilde{r}}w &&\mbox{for} \quad  s \leq k \leq \size \\
    \:&\; 0 \leq \mu_{k} \leq m_k &&\mbox{for} \quad  0 \leq k \leq \size \\
    \;&\; \gamma_s(k) \geq 0 &&\mbox{for} \quad  s \leq k \leq \size. \\
    \;&\; \sum_{s = 0}^k \left(\frac{r}{r+\mu_s} \Gamma(s)\prod_{\ell = 0}^{s-1} \frac{\mu_{\ell}}{r+\mu_{\ell}}\right) \geq -\Theta &&\mbox{for} \quad  0 \leq k \leq \size.
\end{align*}
To account for asymmetric time preferences, the worker's reservation utility is expressed as $ \frac{r}{\tilde{r}}w$, where $wr$ represents the continuous flow payoff of external employment, discounted by $\tilde{r}$ to translate this flow into a present value.

Note that while the state-space reduction \eqref{mainProgramAgg} remains optimal and can still be implemented, the extended program \eqref{extendedProgram} is highly challenging to solve. Under asymmetric discounting, forcing rapid growth through the states is no longer necessarily optimal.\footnote{\revcolor{This can be illustrated by a simple counterexample where $\lambda_1=0,\lambda_2 = 3, \lambda_3 = 4$ with $m_k=1$, $r = 0.01$, $\tilde{r}=0.02$, $w=100$, and $p=1$.}} The fact that the servers' discount factor decays at a different rate introduces non-convexity into the policy transitions, complicating the global optimization.\footnote{\revcolor{We solved this specific instance in an ad hoc manner by leveraging the structure of the demand function to eliminate many of the potential solutions for $\mu_k$.}}
}

\begin{figure}[h]
    \centering
    \input{Figures/extensionFigure}
    \caption{\footnotesize \sf \revcolor{Optimal objective value $J^*$ as a function of the platform's maximum borrowing threshold $\Theta$ under symmetric discounting (a) and asymmetric discounting (b). The blue curves represent the optimal seniority-based compensation, while the black curves denote the uniform compensation policy.}}
    \label{fig:extension}
\end{figure}
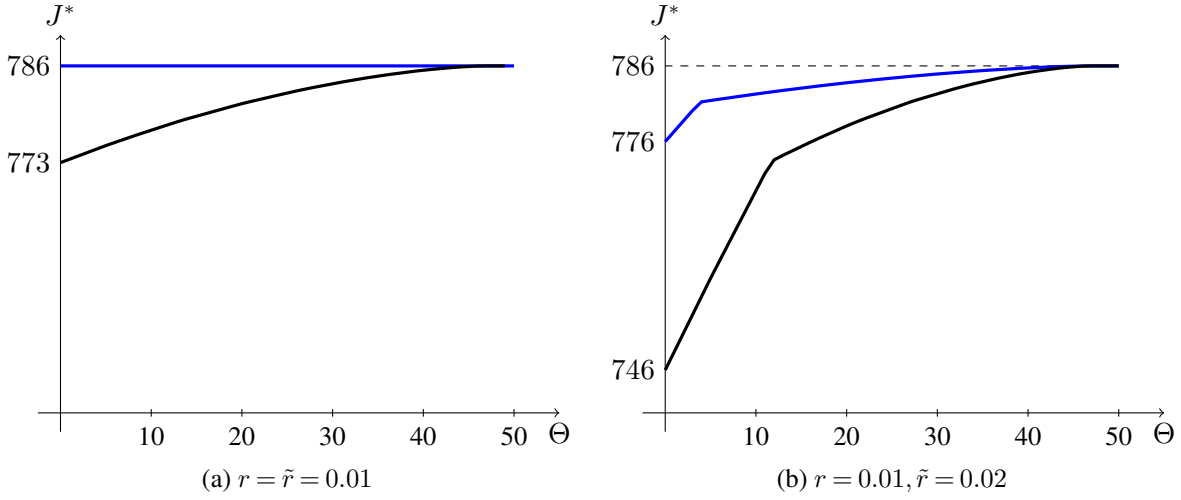
\vspace{-0.5cm}
\revcolor{We evaluate the structural impact of the platform's credit ceiling $\Theta$ under both symmetric and asymmetric discounting regimes in \Cref{fig:extension}, using the specific demand function defined earlier in \Cref{subsec:Sensitivity,sec:benchmark}. From \Cref{thm:original&relaxed}, we know that under symmetric discounting ($r = \tilde{r}$), the optimal compensation policy is invariant to the borrowing threshold $\Theta$. This theoretical guarantee is clearly reflected in the left panel, where the blue curve yields a perfectly horizontal trajectory, achieving full efficiency immediately and strictly outperforming the uniform compensation policy (black curve). In contrast, the right panel illustrates the friction introduced by asymmetric discounting ($r = 0.01, \tilde{r} = 0.02$). When the borrowing threshold $\Theta$ is tightly constrained, the platform lacks the liquid capital to match the servers' outside options continuously across all states. To satisfy the servers' DIR constraints, the platform is forced to rely on future compensations. However, because the platform and servers have different discount rates, shifting payouts into the future is inherently more expensive. The uniform compensation policy is impacted even more severely by this friction due to its rigid structural inability to optimize payouts across states. Conversely, as the borrowing capacity expands, this liquidity constraint binds less. With a sufficiently high credit ceiling, both policies ultimately converge to the global optimum because the platform has the financial freedom to subsidize servers stably at every state, entirely bypassing the need for expensive, time-shifted future compensations.

In sum, the results in this section extend the baseline mechanism in \Cref{sec:Main} to a setting with asymmetric discounting between the platform and servers. In the baseline model, seniority-based compensation can replicate the value of an unconstrained financing policy, offering a sustainable alternative to the subsidy-heavy growth strategies commonly emphasized in the platform literature. When servers are more impatient than the platform, however, future compensation becomes less valuable from the servers' perspective, and liquidity can recover part of this loss by shifting compensation earlier. Even in a restricted credit environment, the seniority-driven framework remains valuable because it gives the platform distributional flexibility: it can determine which cohorts receive scarce current liquidity and which cohorts are instead compensated through future demand. The comparison with uniform compensation shows that this flexibility is especially important when liquidity is scarce and market participants are impatient. In such settings, optimizing the temporal structure of compensation becomes central to platform viability.}

%% file: Figures/extensionFigure.tex
\begin{tikzpicture}
\begin{scope}[shift={(0, 0)}, yscale = 0.5, xscale = 0.12]

    \draw[->] (-2.5 ,0) -- (55, 0) node[below] {$\Theta$};
    \draw[->] (0,-0.5) -- (0, 10) node[above] {$J^*$};

    \node[below=0.6cm, align=center, font=\small] at (25, 0) {(a) $r = \tilde{r} = 0.01$};

    \draw[dashed] (0.01, 9.1801) -- (0,9.1801) node[left] {$786$};
    \draw[dashed] (0.01, 6.6213) -- (0,6.6213) node[left] {$773$};
    \draw[blue, very thick] plot coordinates {
        (0,9.1801)
        (50,9.1801)
    };

    \draw[black, very thick] plot coordinates {
    (0, 6.6213) (1, 6.7113) (2, 6.8013) (3, 6.8913) (4, 6.9813) (5, 7.0713) (6, 7.1549) (7, 7.2378) (8, 7.3207) (9, 7.4000) (10, 7.4780) (11, 7.5561) (12, 7.6341) (13, 7.7114) (14, 7.7823) (15, 7.8490) (16, 7.9157) (17, 7.9823) (18, 8.0490) (19, 8.1157) (20, 8.1809) (21, 8.2392) (22, 8.2950) (23, 8.3508) (24, 8.4066) (25, 8.4625) (26, 8.5159) (27, 8.5671) (28, 8.6151) (29, 8.6605) (30, 8.7060) (31, 8.7514) (32, 8.7927) (33, 8.8336) (34, 8.8732) (35, 8.9087) (36, 8.9443) (37, 8.9764) (38, 9.0075) (39, 9.0386) (40, 9.0650) (41, 9.0896) (42, 9.1114) (43, 9.1331) (44, 9.1508) (45, 9.1639) (46, 9.1758) (47, 9.1801) (48, 9.1801) (49, 9.1801)
};

    \foreach \x in {10, 20, ..., 50} {
        \draw (\x, 0.1) -- (\x, -0.1) node[below, font=\small] {\x};
    }
    
\end{scope}

\begin{scope}[shift={(8, 0)}, yscale = 0.5, xscale = 0.12]
    
    \draw[dashed] (50, 9.1801) -- (0,9.1801) node[left] {$786$};
    \draw[dashed] (0.01, 7.1771) -- (0,7.1771) node[left] {$776$};
    \draw[dashed] (0.01, 1.1331) -- (0,1.1331) node[left] {$746$};
    \draw[->] (-2.5 ,0) -- (55, 0) node[below] {$\Theta$};
    \draw[->] (0,-0.5) -- (0, 10) node[above] {$J^*$};

\draw[blue, very thick] plot coordinates {
    (0, 7.1771) (1, 7.4500) (2, 7.7230) (3, 7.9959) (4, 8.2287) (5, 8.2652) (6, 8.2996) (7, 8.3337) (8, 8.3679) (9, 8.4020) (10, 8.4362) (11, 8.4692) (12, 8.5010) (13, 8.5326) (14, 8.5622) (15, 8.5918) (16, 8.6214) (17, 8.6509) (18, 8.6788) (19, 8.7061) (20, 8.7334) (21, 8.7594) (22, 8.7845) (23, 8.8096) (24, 8.8343) (25, 8.8572) (26, 8.8801) (27, 8.9030) (28, 8.9246) (29, 8.9453) (30, 8.9661) (31, 8.9847) (32, 9.0032) (33, 9.0218) (34, 9.0397) (35, 9.0561) (36, 9.0718) (37, 9.0861) (38, 9.1004) (39, 9.1147) (40, 9.1270) (41, 9.1379) (42, 9.1480) (43, 9.1581) (44, 9.1665) (45, 9.1725) (46, 9.1781) (47, 9.1801) (50,9.1801)
};

\draw[black, very thick] plot coordinates {
    (0, 1.1331) (1, 1.6188) (2, 2.1046) (3, 2.5904) (4, 3.0761) (5, 3.5619) (6, 4.0265) (7, 4.4890) (8, 4.9515) (9, 5.4139) (10, 5.8764) (11, 6.3389) (12, 6.6931) (13, 6.8174) (14, 6.9321) (15, 7.0467) (16, 7.1614) (17, 7.2728) (18, 7.3789) (19, 7.4850) (20, 7.5911) (21, 7.6912) (22, 7.7849) (23, 7.8729) (24, 7.9608) (25, 8.0487) (26, 8.1366) (27, 8.2236) (28, 8.2984) (29, 8.3693) (30, 8.4401) (31, 8.5110) (32, 8.5816) (33, 8.6449) (34, 8.7060) (35, 8.7608) (36, 8.8157) (37, 8.8685) (38, 8.9161) (39, 8.9637) (40, 9.0041) (41, 9.0431) (42, 9.0760) (43, 9.1089) (44, 9.1357) (45, 9.1558) (46, 9.1736) (47, 9.1801) (50,9.1801)
};

    \foreach \x in {10, 20, ..., 50} {
        \draw (\x, 0.1) -- (\x, -0.1) node[below, font=\small] {\x};
    }

    \node[below=0.6cm, align=center, font=\small] at (25, 0) {(b) $r = 0.01, \tilde{r} = 0.02$};
    
\end{scope}
\end{tikzpicture}

%% file: Tex/Conclusion.tex
\section{Conclusion} 
\revcolor{ In this paper, we examined the platform’s optimal strategy for growth and job assignment to servers, demonstrating how different constraints and growth policies impact the platform's revenue. While some platforms prioritize maximizing profits and others focus on expanding their size, we showed that these objectives do not always align. In particular, while rapid expansion is generally beneficial, there are scenarios where strategic pauses in growth become necessary to ensure financial stability.

A central insight of our work is the importance of seniority-based assignment. By prioritizing early adopters, the platform creates strong intertemporal incentives that can encourage initial participation without relying on subsidies or external funding. In the baseline model, we show that even granting the platform unlimited access to external funds does not improve its optimal value once seniority-based compensation is used. \cref{sec:extension} qualifies this result by showing that, when servers are more impatient than the platform, liquidity can recover part of the value lost from costly future compensation. Thus, although seniority-based compensation cannot fully replace financial subsidies in all settings, it serves as a powerful complement. Moreover, we demonstrate that enforcing a static notion of fairness that ignores seniority leads to systematically inferior outcomes and ultimately harms all participants, including servers themselves, whereas an appropriately designed seniority-based assignment allows the platform to enforce intertemporal fairness among all servers. 

Our findings also  highlight the complexities of balancing platform growth, revenue, and fair compensation. Future research could extend this model to explore competition among platforms, or heterogeneous server types. Another promising direction is studying this problem under imperfect information, where both the platform and the servers operate with limited knowledge of current and future market conditions. Understanding these aspects would further enhance the practicality and robustness of job assignment and compensation schemes employed by platforms.}

\vspace{0.5cm}

%% file: Tex/Proofs.tex
\section{Proofs}\label{apx:proofs}
\subsection{Proof of \Cref{lem:aggregate identity}}
    We prove the lemma using backward induction. At the boundary $k = \size$, it is not hard to see from the definitions of $v^\pi_s(k)$ and $V^{\pi}_{k}$ in \eqref{individualBE} and \eqref{aggregateBE}, respectively, that the identity $V^{\pi}_{\size} = \sum_{s=1}^{\size} v_s(\size)$ holds trivially. Now assume this identity holds true for $k+1$ then from \eqref{aggregateBE} we have:
    \begin{align*}
        V^{\pi}_{k} =&
        \frac{r}{r+\mu_{k}}\left(\frac{p\, \lambda_k}{r} - \imm (k) \right) + \frac{\mu_{k}}{r+\mu_{k}}\,V^{\tpi}_{k+1} \\
        =& \frac{r}{r+\mu_{k}}\sum_{s=1}^k \gamma_s(k) + \frac{\mu_k}{r+\mu_k}
        \left(\sum_{s=1}^{k+1} v^\pi_s(k+1) + \sum_{s=k+2}^{\size} \frac{\alpha_s}{\alpha_{k+1}}\,v^\pi_s(s) \right) \\
        =& \sum_{s=1}^k \left(\frac{r}{r+\mu_{k}}\gamma_s(k) + \frac{\mu_k}{r+\mu_k} \,v^\pi_s(k+1)\right) + \frac{\alpha_{k+1}}{\alpha_k}\left( \sum_{s=k+1}^{\size} \frac{\alpha_s}{\alpha_{k+1}}\,v^\pi_s(s) \right) \\
        =& \sum_{s=1}^k v^\pi_s(k) + \sum_{s=k+1}^{\size} \frac{\alpha_s}{\alpha_k}\,v^\pi_s(s). ~~~\Box
    \end{align*}

\subsection{Proof of \Cref{thm:N&S}}
\noindent {\em Necessary condition:} This is a direct corollary of \Cref{lem:aggregate identity} and the fact that $v_s(k) \geq w$ for all $s \leq k$. Indeed, 
    \begin{align*}
        V^{\tpi}(k) = \sum_{s=1}^k v^\pi_s(k) + \sum_{s=k+1}^{\size} \frac{\alpha_s}{\alpha_k}\,v^\pi_s(s) \geq \sum_{s=1}^k w + \sum_{s=k+1}^{\size} \frac{\alpha_s}{\alpha_k}w = w \left(k + \sum_{s=k+1}^{\size} \frac{\alpha_s}{\alpha_k}\right)
    \end{align*}
\noindent  \textit{Sufficient condition:} 
    To establish the non-trivial direction of the theorem, we leverage Farkas' Lemma. For \revcolor{notational convenience}, let us define $A(k) =  \frac{p\lambda_k}{r}-\imm(k)$. Our goal is to demonstrate that if condition \eqref{iffCondition} holds or equivalently, if $\sum_{\ell = k}^{\size} (\alpha_{\ell}-\alpha_{\ell+1}) A({\ell}) \geq w(k\alpha_k + \sum_{\ell= k+1}^{\size} \alpha_{\ell})$ for all $k$, then there exists at least one feasible solution within the following polytope:
\begin{align}\label{primalFarkas}\tag{Primal}
    \left\{\displaystyle\sum_{\tempS=k}^{\size} (\alpha_{\tempS} - \alpha_{\tempS+1} )\gamma_s(\tempS) \geq \alpha_k w \quad  \forall s\leq k \leq \size,\quad  \displaystyle\sum_{s=1}^{k} \gamma_s(k)= A(k)  \quad  \forall  k \leq \size, \quad \gamma_s(k) \geq 0 \quad  \forall s\leq k \leq \size \right\}
    \end{align}
Using Farkas' Lemma, the existence of a feasible point in the polytope \eqref{primalFarkas} is equivalent to the infeasibility of the following alternative system:
\begin{align}\label{DualFarkas}\tag{Dual}
    \left\{\theta_k - (\alpha_{k}-\alpha_{k+1})\sum_{\tempS = s}^k \beta_{\tempS,s} \geq 0 \quad  \forall s\leq k \leq \size,\quad \displaystyle\sum_{\ell=1}^{\size} A({\ell}) \theta_{\ell} < w \sum_{s=1}^{\size} \sum_{\ell = s}^{\size} \beta_{\ell, s} \alpha_{\ell}, \quad \beta_{k,s} \geq 0 \quad s\leq k \leq \size \right\}
\end{align}
Here, $\beta$  represents the dual variable associated with the first set of constraints, while $\theta$ corresponds to the second set \revcolor{$(\sum_{s=1}^{k} \gamma_s(k)= A(k))$}. Define $\theta^{(\beta)}_k = \max_{s}\left\{(\alpha_{k}-\alpha_{k+1})\sum_{\tempS = s}^k \beta_{\tempS,s}\right\}$. \revcolor{Observe} that if $(\beta,\theta)$ is a feasible point to \eqref{DualFarkas}, then $(\beta,\theta^{(\beta)})$ is also feasible. This is because $\theta^{(\beta)}$ acts as a lower bound for $\theta$, and reducing $\theta$ only makes the strict inequality less restrictive. Consequently, we can focus on proving that under \eqref{iffCondition}, no $\beta$ exists such that $(\beta,\theta^{(\beta)})$ satisfies all the constraints. To do so, we leverage the following property of $\theta^{(\beta)}$: $$\frac{\theta^{(\beta)}_k}{\alpha_{k}-\alpha_{k+1}} = \max_{s}\left\{\sum_{\tempS = s}^k \beta_{\tempS,s}\right\} \geq \max_{s}\left\{\sum_{\tempS = s}^{k-1} \beta_{\tempS,s}\right\} = \frac{\theta^{(\beta)}_{k-1}}{\alpha_{k-1}-\alpha_{k}}.$$
Let us multiply \eqref{iffCondition} for each state $k$ by the term $\frac{\theta^{(\beta)}_k}{\alpha_{k}-\alpha_{k+1}} - \frac{\theta^{(\beta)}_{k-1}}{\alpha_{k-1}-\alpha_{k}}$ which is positive as established earlier. For notational consistency, let $\theta_0^{(\beta)} = 0$. Summing over all states $k$ we obtain:
\begin{align*}
\sum_{k=1}^{\size} \left(\left(\frac{\theta^{(\beta)}_k}{\alpha_{k}-\alpha_{k+1}} - \frac{\theta^{(\beta)}_{k-1}}{\alpha_{k-1}-\alpha_{k}}\right) \sum_{\ell = k}^{\size} (\alpha_{\ell}-\alpha_{\ell+1}) A({\ell})\right) \geq 
w\sum_{k=1}^{\size} \left(\frac{\theta^{(\beta)}_k}{\alpha_{k}-\alpha_{k+1}} - \frac{\theta^{(\beta)}_{k-1}}{\alpha_{k-1}-\alpha_{k}}\right)(k\alpha_k + \sum_{\ell= k+1}^{\size} \alpha_{\ell})
\end{align*}
Simplifying the left hand side we get:
\begin{align*}
\sum_{k=1}^{\size} \theta^{(\beta)}_kA(k) &\geq \revcolor{w}\sum_{k=1}^{\size} \left(\frac{\theta^{(\beta)}_k}{\alpha_{k}-\alpha_{k+1}} - \frac{\theta^{(\beta)}_{k-1}}{\alpha_{k-1}-\alpha_{k}}\right)(k\alpha_k + \sum_{\ell= k+1}^{\size} \alpha_{\ell})\\
&= w\sum_{k=1}^{\size} 
\frac{\theta^{(\beta)}_k}{\alpha_{k}-\alpha_{k+1}} \left(k\alpha_k - k\alpha_{k+1}\right)
\\
&= w\sum_{k=1}^{\size} k\theta^{(\beta)}_k
\\
&\geq w\sum_{k=1}^{\size}\sum_{s=1}^k (\alpha_k-\alpha_{k+1})\sum_{\tempS = s}^k \beta_{\tempS,s}\\
&= w \sum_{s=1}^{\size} \sum_{\ell = s}^{\size} \beta_{\ell, s} \alpha_{\ell},
\end{align*}
Where the last inequality follows directly from the definition of $\theta_k^{(\beta)}$. Observe that this inequality contradicts the strict inequality required for feasibility of the system in \eqref{DualFarkas}. Consequently, combining all the steps, we have shown that if  \eqref{iffCondition} holds, then the system in \eqref{DualFarkas} cannot be satisfied. By Farkas' Lemma, the infeasibility of \eqref{DualFarkas} guarantees the feasibility of the system in \eqref{primalFarkas}. $\Box$

\subsection{Proof of \Cref{lem:relax}} First, after a few straightforward manipulations, we can simplify \eqref{mainProgramAgg} to the following program:
\begin{align*}\label{mainProgramAgg2}
\tag{$\mathcal{P}_{\texttt{Agg,Rleaxed}}$}
    \;&\; \sup_{\alpha_k,\imm(k)} \;J^{\tpi}(0) = \sum_{k = 1}^{\size} (\alpha_{k}-\alpha_{k+1})\imm(k) \\
    \text{subject to}
    \:&\; \sum_{\ell = k}^{\size} (\alpha_{\ell}-\alpha_{\ell+1})(\frac{p\lambda_{\ell}}{r}-\imm({\ell})) \geq w \Bigg(k\alpha_k + \sum_{\ell=k+1}^{\size} \alpha_\ell \Bigg) && \mbox{for}\quad k  \leq \size\\
    \;&\; 0 \leq \alpha_{k+1} \leq \frac{m_k}{r+m_k}\alpha_k \quad 
    \mbox{and} \quad \imm(k) \leq \revcolor{\frac{p\lambda_{k}}{r}}
    && \mbox{for}\quad k \leq \size.
\end{align*}
A trivial upper bound for the objective can be obtained by evaluating the DCR constraint (i.e., the first constraint in \eqref{mainProgramAgg2}) at $k=1$. We get
\[
J^{\tpi}(0) = \sum_{\ell = 1}^{\size} (\alpha_{\ell}-\alpha_{\ell+1})\imm(\ell) 
\leq \sum_{\ell = 1}^{\size} \frac{p \lambda_{\ell}}{r} \;(\alpha_{\ell}-\alpha_{\ell+1}) - w \sum_{\ell=1}^{\size} \alpha_\ell
= \sum_{\ell=1}^{\size} \left(\frac{p \lambda_{\ell}}{r}-w\ell\right)(\alpha_{\ell}-\alpha_{\ell+1}).\]
Thus, this upper bound can be achieved by setting \revcolor{$\imm(k)=\hat{\imm}(k)= \frac{p \lambda_{k}}{r}-wk$}. With this choice of $\imm(k)$, all the DCR constraints in \eqref{mainProgramAgg2} become tight, thereby achieving the upper bound. This shows that $\hat{\imm}(k)$ is an optimal solution for the relaxed problem \eqref{mainProgramAgg2}. \revcolor{$\hat{\imm}(k)$ is not necessarily the optimal solution for \eqref{mainProgramAgg}, however, since it can be strictly negative for some $k$.} 

Plugging the optimal values for $\hat{\imm}(k)$, the relaxed program reduces to finding the values of $\alpha_k$ that maximize this upper bound:
\begin{equation} \label{mainProgramRelaxed}
\tag{$\mathcal{P}_{\texttt{relax}}$}
\sup_{\alpha_k} \;\widehat{J}^{\tpi}(0) = \sum_{\ell=1}^{\size} \left(\frac{p \lambda_{\ell}}{r}-w\ell\right)(\alpha_{\ell}-\alpha_{\ell+1}) \quad \text{subject to}\quad 0 \leq \alpha_{k+1} \leq \frac{m_k}{r+m_k}\alpha_k, ~~~k  \leq \size.
\end{equation}
The objective function \revcolor{can be rewritten} as: 
$$ \widehat{J}^{\tpi}(0;\alpha) = \frac{p}{r}\sum_{\ell=1}^{\size} \alpha_{\ell}\left(\lambda_{\ell} - \lambda_{\ell - 1} - \frac{rw}{p} \right).$$
For a feasible solution $\{{\alpha}_\ell\}$, we say that ${\alpha}_{\ell+1}$ is a `{\em slack}' variable if $0 < \alpha_{\ell+1} < \frac{m_\ell}{r+m_\ell}\alpha_\ell$. \revcolor{Because the program is linear in $\boldsymbol{\alpha}$, the fundamental theorem of linear programming guarantees the existence of an optimal solution at an extreme point where no such slack exists. Consequently, there must exist an optimal integer horizon $\hat{\size}$ such that $\hat{\alpha}_{k+1}=\frac{m_k}{r+m_k}\hat{\alpha}_k$ for $0 \leq k \leq \hat{\size}-1$, and $\hat{\alpha}_{\hat{\size}+1}=0$. Plugging this boundary-binding solution into the objective yields the simplified single-variable optimization problem over the horizon $\size$ as stated in \eqref{mainProgramAgg3}, which completes the proof. } $\Box$
\subsection {Proof of \Cref{thm:original&relaxed}} Let  $\hat{\pi}=(\hat{\alpha},\hat{\imm},\hat{\size})$ be the optimal solution to the relaxed problem
that ignores the non-negativity constraint $\imm(k) \geq 0$. For this relaxed formulation,  the platform's and servers' cumulative value functions satisfy (see proof of \Cref{lem:relax}): 
\begin{align*}
     & \widehat{J}(k) = \left(1-\frac{\hat{\alpha}_{k+1}}{\hat{\alpha}_k}\right)\hat{\imm}(k) + \frac{\hat{\alpha}_{k+1}}{\hat{\alpha}_k}\, \widehat{J}(k+1) ~~~~~\quad\mbox{and}\quad    \widehat{J}(\hat{\size}+1)=0 \\
    \:&\; \widehat{V}(k)=w \Big(k + \sum_{s=k+1}^{\hat{\size}} \frac{\hat{\alpha}_s}{\hat{\alpha}_k}\Big).
\end{align*}
The optimal objective value of the relaxed problem is $\widehat{J}=\widehat{J}(0)$. 

\revcolor{Note that if $\hat{\imm}(k) \geq 0$ for all $k \in [\hat{\size}]$, then $\hat{\pi} = (\hat{\alpha}, \hat{\imm}, \hat{\size})$ would be a feasible and optimal solution to \eqref{mainProgramAgg}, establishing the result. The core challenge of the proof lies in the case where $\hat{\imm}(k) < 0$ for some states. To resolve this, we demonstrate that the platform can systematically redistribute its revenues across time—moving current deficits into future periods—to achieve a fully non-negative cash-flow profile $\boldsymbol{\imm}^*$ that preserves the optimal relaxed objective value $\widehat{J}$.}

In what follows, we will prove two lemmas about the platform's value function $\widehat{J}(k)$ \revcolor{which will be useful in establishing the result.}
\begin{lem}\label{lem:Jmonot} \revcolor{The sequence $\{\widehat{J}(k)\colon k \in [\hat{\size}]\}$
is non-negative.}
\end{lem}
{\sc Proof of \Cref{lem:Jmonot}:} \revcolor{By the global optimality of the relaxed solution, the initial baseline satisfies $\widehat{J}(0)  = \widehat{J} \geq 0$. Suppose, by contradiction, that the sequence contains negative terms, and let $\ell+1 \leq \hat{\size}$ denote the first (smallest) state index where the value turns strictly negative, such that:
$$ \widehat{J}(\ell+1) < 0 \leq \widehat{J}(\ell). $$
Let us define an alternative policy $\tilde{\pi}=(\tilde{\alpha},\tilde{\imm},\tilde{S})$ such that $\tilde{S}=\ell$, $\tilde{\alpha}_k=\hat{\alpha}_k$ and $\tilde{\imm}(k)=\hat{\imm}(k)$ for $k \in [\tilde{S}]$. From the definitions of $\hat{\alpha}_k$ and $\hat{\imm}(k)$ in \Cref{lem:relax}, it is not hard to see that  $\tilde{\pi}$ is a feasible policy for the relaxed program with objective value $\widetilde{J}=\widetilde{J}(0)$, where $\widetilde{J}(0)$ is obtained from the recursion
$$\widetilde{J}(k) = \left(1-\frac{\hat{\alpha}_{k+1}}{\hat{\alpha}_k}\right) \hat{\imm}(k) + \frac{\hat{\alpha}_{k+1}}{\hat{\alpha}_k}\, \widetilde{J}(k+1) ~~~~~\quad\mbox{and}\quad    \widetilde{J}(\tilde{\size}+1)=0.$$
Because $\widehat{J}(\ell) \geq 0$ and $\widehat{J}(\ell+1) < 0$, it must hold that $\hat{\imm}(\ell) > 0 > \widehat{J}(\ell+1)$. Then by the definition of $\ell$ we get 
$$\widetilde{J}(\ell)=\hat{\imm}(\ell) > \left(1-\frac{\hat{\alpha}_{\ell+1}}{\hat{\alpha}_{\ell}}\right) \hat{\imm}(\ell) + \frac{\hat{\alpha}_{\ell+1}}{\hat{\alpha}_{\ell}}\, \widehat{J}(\ell+1) =\widehat{J}(\ell).$$
But this implies that $\widetilde{J}(k)>\widehat{J}(k)$ for all \revcolor{$k \leq \ell$}. In particular, $\widetilde{J}=\widetilde{J}(0)>\widehat{J}(0)=\widehat{J}$, which contradicts the optimality of $\hat{\pi}=(\hat{\alpha},\hat{\imm},\hat{\size})$ for the relaxed problem. This shows that  $\{\widehat{J}(k)\colon k \in [\hat{\size}]\}$ is non-negative for all $k$. $\blacksquare$}

\begin{lem}\label{lem:FeasibleDirection}
    If a vector $\imm(k)$ satisfies the DCR constraints, for any $x\in [0,1]$ and any $\sOne < \sTwo \leq \size$ the following updated vector 
    \begin{align*}
        \tilde{\imm}(k) = \begin{cases} 
                \imm(k) + x\imm(\sTwo) \frac{\alpha_{\sTwo}-\alpha_{\sTwo+1}}{\alpha_{\sOne}-\alpha_{\sOne+1}} & \text{if } k= \sOne, \\
                (1-x)\imm(k)
                & \text{if } k = \sTwo, \\
                \imm(k) & \text{o.w.}
            \end{cases}
    \end{align*}
    \revcolor{(i) satisfies the DCR constraints
    \\(ii) keeps $J_0$ fixed.}
\end{lem}
{\sc Proof \Cref{lem:FeasibleDirection}:}
Let us define $\tilde{J}$ the NPV of platform under the updated platform payoff $\tilde{\imm}(k)$:
\begin{itemize}
\item[--] For $k > \sTwo$.
$$\alpha_k\,\tilde{J}(k)=\sum_{\ell=k}^\size (\alpha_{\ell}-\alpha_{\ell+1})\tilde{\imm}(\ell)= \sum_{\ell=k}^\size (\alpha_{\ell}-\alpha_{\ell+1})\imm(\ell) = \alpha_k\,J(k).$$
\item[--] For $\sOne < k \leq \sTwo.$
$$\alpha_k\,\tilde{J}(k)=
\sum_{\ell=k}^\size (\alpha_{\ell}-\alpha_{\ell+1})\tilde{\imm}(\ell)=
\sum_{\ell=k}^\size (\alpha_{\ell}-\alpha_{\ell+1})\imm(\ell) - x(\alpha_{\sTwo}-\alpha_{\sTwo+1})\imm(\sTwo) \revcolor{\leq} \alpha_k\,J(k).$$
\item[--] For $k \leq \sOne.$
\begin{eqnarray*}
\alpha_k\,\tilde{J}(k) & = & \sum_{\ell=k}^\size (\alpha_{\ell}-\alpha_{\ell+1})\tilde{\imm}(\ell) \\
&=& \sum_{\ell=k}^{\size} (\alpha_{\ell}-\alpha_{\ell+1})\imm(\ell) + (\alpha_{\sOne}-\alpha_{\sOne+1})\left(\tilde{\imm}(\sOne)-\imm(\sOne)\right) + (\alpha_{\sTwo}-\alpha_{\sTwo+1}) \left(\tilde{\imm}(\sTwo)-\imm(\sTwo)\right) \\
&=& \sum_{\ell=k}^\size (\alpha_{\ell}-\alpha_{\ell+1})\imm(\ell) + x(\alpha_{\sOne}-\alpha_{\sOne+1}) \imm(\sTwo)\frac{\alpha_{\sTwo}-\alpha_{\sTwo+1}}{\alpha_{\sOne}-\alpha_{\sOne+1}}- x (\alpha_{\sTwo}-\alpha_{\sTwo+1})\imm(\sTwo) \\
&=& \sum_{\ell=k}^\size (\alpha_{\ell}-\alpha_{\ell+1})\imm(\ell)=\alpha_k\,J(k).
\end{eqnarray*} 
\end{itemize}
\revcolor{Combining these three cases confirms that the initial platform baseline value is preserved, $\tilde{J}(0) = J(0)$. Furthermore, because the total network surplus is invariant to the intertemporal cash-flow split, the updated server value function satisfies:
\[
\tilde{V}(k) = \frac{1}{\alpha_k} \sum_{\ell=k}^{\size} \frac{p\lambda_{\ell}}{r}(\alpha_{\ell}-\alpha_{\ell+1}) - \tilde{J}(k).
\]
And $\tilde{J}(k) \leq J(k)$ implies that $\tilde{V}(k) \geq V(k)$. Because the original profile was valid, the updated profile automatically satisfies $\tilde{V}(k) \geq w \left(k + \sum_{j=k+1}^{\size} \frac{\alpha_j}{\alpha_k}\right)$, which confirms that the DCR constraints hold true. $\blacksquare$}
\vspace{0.3cm}

\revcolor{\Cref{lem:FeasibleDirection} establishes a systematic method to offset local deficits by borrowing revenue from future, profitable states without violating the servers' dynamic participation boundaries. By consecutively applying this localized transformation across all downstream states $\sTwo \in \{\sOne+1, \dots, \size\}$, we can aggregate this intertemporal redistribution into a single uniform update rule.}
\begin{cor}\label{cor:FeasibleDirection}
    The following update also has \revcolor{the} properties proven in \Cref{lem:FeasibleDirection}:
    \begin{align*}
        \tilde{\imm}(k) = \begin{cases} 
                \imm(k) & \text{if } k < \sOne\\
                \imm(k) + x\sum_{\ell = \sOne+1}^{\size}\imm(\ell)\frac{\alpha_{\ell}-\alpha_{\ell+1}}{\alpha_{\sOne}-\alpha_{\sOne+1}} & \text{if } k= \sOne, \\
                (1-x)\imm(k)
                & \text{if } k > \sOne.
            \end{cases}
    \end{align*}
\end{cor}

Equipped with \Cref{lem:Jmonot} and \Cref{lem:FeasibleDirection}, we return to the proof of \Cref{thm:original&relaxed}. Let us assume in what follows that there exists an $\ell \in [\hat{\size}]$ such that $\hat{\imm}(\ell)<0$. Starting from the optimal policy $\hat{\pi}=(\hat{\alpha},\hat{\imm},\hat{\size})$ for the relaxed program, we will use a recursive method to construct an optimal solution $\pi^*=(\alpha^*,\imm^*,\size^*)$ to \eqref{mainProgramAgg} that satisfies the conditions in \Cref{thm:original&relaxed}.  
In particular, we will set $\size^*=\hat{\size}$ and $\alpha^*_k=\hat{\alpha}_k$ for all $k \leq \size^*$, and sequentially construct $\imm^*(k)$ in such a way that $\pi^*$ is feasible for \eqref{mainProgramAgg} and its objective value $J^*$ equals the optimal objective value of the relaxed program $\widehat{J}$. From this, we will conclude that $\pi^*$ is optimal for  \eqref{mainProgramAgg} proving \Cref{thm:original&relaxed}. \revcolor{To operationalize this revenue redistribution, we introduce the following sequential procedure:}

\vspace{0.2cm}

\noindent{\bf Recursive Procedure:}
\begin{enumerate} {
    \item {\sc Input:} A policy $\pi=(\alpha,\imm,\size)$ with the following characteristics:
    \begin{itemize}[leftmargin=0.8cm]
    \item $\imm({k}) \leq \frac{p\lambda_{k}}{r}$ for all $k \in [S]$.
    \item There exists $s \in [\size]$ such that $\imm(s)<0$. We assume $s$ is the largest index with this property. 
    \item The associated value functions $J(k)$ and $V(k)$ satisfy:
   $$ \hspace{-0.8cm}J(k)={1 \over \alpha_k}\,\sum_{\ell=k}^\size (\alpha_{\ell}-\alpha_{\ell+1})\imm(\ell)  \geq 0 \quad \mbox{and}\quad V(k)={1 \over \alpha_k}\,\sum_{\ell=k}^\size (\alpha_{\ell}-\alpha_{\ell+1})\left(\frac{p\lambda_{\ell}}{r}-\imm({\ell})\right) \geq w \Big(k + \revcolor{\sum_{\ell=k+1}^{\size} \frac{{\alpha}_\ell}{{\alpha}_k}}\Big).$$
\end{itemize}

\item{\sc Policy Update:} Let $x \in [0,1]$ be such that $0 \leq {\mathcal F}_{s}(x) \leq \frac{p\lambda_s}{r}$, where
$${\mathcal F}_{s}(x):=\imm(s)+x\sum_{\ell = s+1}^{\size}\imm(\ell)\frac{\alpha_{\ell}-\alpha_{\ell+1}}{\alpha_{s}-\alpha_{s+1}}.$$
Define the policy $\tilde{\pi}=(\alpha,\tilde{\imm},\size)$ such that
$\tilde{\imm}(k)=\imm(k)$ for $k<s$, $\tilde{\imm}(s)={\mathcal F}_s(x)$ and $\tilde{\imm}(k)=(1-x)\,\imm(k)$ for $k>s$. Let $\tilde{J}(k)$ and $\tilde{V}(k)$ be the associated value functions. $\blacksquare$}
\end{enumerate}

We claim that the policy $\tilde{\pi}=(\alpha,\tilde{\imm},\size)$ produced by \revcolor{the} previous procedure \revcolor{satisfies} the following properties:
\begin{enumerate}[leftmargin=1.5cm,label=(\alph*)]
\item $\tilde{\imm}(k) \leq \frac{p\lambda_k}{r}$ for all $k \in [S]$.
\item $|\{\tilde{\imm}(k)<0\}|=|\{\imm(k)<0\}|-1$, i.e, $\tilde{\imm}$ has one less negative component than $\imm$. In particular, $\imm(s) < 0 \leq \tilde{\imm}(s)$.
\item The value functions $\tilde{J}(k)$ and $\tilde{V}(k)$ satisfy:
$$\tilde{J}(k)  \geq 0, \qquad J(0)=\tilde{J}(0) \qquad \mbox{and}\qquad \tilde{V}(k) \geq w \Big(k + \revcolor{\sum_{\ell=k+1}^{\size} \frac{{\alpha}_{\ell}}{{\alpha}_k}}\Big).$$
\end{enumerate}
\revcolor{ Notice that (c) is a direct implication of \Cref{lem:FeasibleDirection} and \Cref{cor:FeasibleDirection} and the fact that $\tilde{J}(k) \in \{J(k),(1-x)J(k)\}$.} By the construction of $\tilde{\pi}$, to prove (a) and (b), it suffices to show that there exists a value of $x$ in Step 2 satisfying the conditions. 
\revcolor{Note that
$$J(s)={\alpha_s-\alpha_{s+1} \over \alpha_s}\,{\mathcal F}_s(1).$$
Since $J(s) \geq 0$, we have ${\mathcal F}_s(1) \geq 0$. Because ${\mathcal F}_s(0) = \imm(s) < 0$, this implies, first, that $s < \size$ (as $s=\size$ would force ${\mathcal F}_s(1) = \imm(s) < 0$), and second, that there exists an $x \in (0,1]$ such that ${\mathcal F}_s(x) \in [0,\frac{p\lambda_k}{r}]$, as required.}

Finally, to complete the proof of \Cref{thm:original&relaxed}, note that the optimal policy for the relaxed problem  $\hat{\pi}=(\hat{\alpha},\hat{\imm},\hat{\size})$  satisfies all the conditions in Step 1 of the recursive procedure. Thus, we can  sequentially apply this procedure, reducing by one in each iteration the number components of $\imm$ that violate the \revcolor{non-negativity constraint (by point (b) above)}. Since $\hat{\size}<\infty$, this  procedure terminates in finite time with a policy $\pi^*=(\hat{\alpha},\imm^*,\hat{\size})$, which is feasible for \eqref{mainProgramAgg} and whose objective value $J^*=\widehat{J}$. $\Box$



\subsection{Proof of \Cref{prop:SIRvsDIR}}
Let $D^*$ be the optimal payoff the platform can \revcolor{achieve} by relaxing the DIR constraints. Then:
\begin{align*}
    D^* & = \sup_{\size,\alpha_k} \sum_{\ell = 0}^{\size} \left(\alpha_{\ell}-\alpha_{\ell+1}\right) \imm(\ell) \\
    & = \sup_{\size,\alpha_k} \sum_{\ell = 0}^{\size} \left(\alpha_{\ell}-\alpha_{\ell+1}\right) \left(\frac{p\lambda_{\ell}}{r}- \sum_{s=1}^{\ell} \gamma_{s}(\ell)\right)\\
    & = \sup_{\size,\alpha_k} \left\{\sum_{\ell = 0}^{\size} \left(\alpha_{\ell}-\alpha_{\ell+1}\right) \frac{p\lambda_{\ell}}{r}- \sum_{s=1}^{\size} \sum_{\ell = s}^{\size} \left(\alpha_{\ell}-\alpha_{\ell+1}\right)  \gamma_{s}(\ell) \right\}\\
    & \leq \sup_{\size,\alpha_k} \left\{\sum_{\ell = 0}^{\size} \left(\alpha_{\ell}-\alpha_{\ell+1}\right) \frac{p\lambda_{\ell}}{r}- \sum_{s=1}^{\size} \alpha_s w
    \right\}\\
    & = \sup_{\size,\alpha_k} \sum_{\ell = 0}^{\size} \left(\alpha_{\ell}-\alpha_{\ell+1}\right) \left(\frac{p\lambda_{\ell}}{r}-  w\ell\right),
\end{align*}
where the inequality comes from the SIR constraint of the servers. Now similar to proof of \Cref{lem:relax} assume $0 < \alpha_{k+1} < \frac{m_k}{r+m_k}\alpha_k$, then the platform is better of pushing it in \revcolor{one} direction. Therefore:
\begin{align*}
    D^* & \leq \sup_{\size,\alpha_k} \sum_{\ell = 0}^{\size} \left(\alpha_{\ell}-\alpha_{\ell+1}\right) \left(\frac{p\lambda_{\ell}}{r}-  w\ell\right)\\
    & = \sup_{\size} \sum_{\ell = 0}^{\size} \left(\hat{\alpha}_{\ell}-\hat{\alpha}_{\ell+1}\right) \hat{\imm}(\ell) = \widehat{J} = J^*
\end{align*}
Finally, as this program is less constraints then $D^* \geq J^*$, consequently $D^* = J^*$. \revcolor{$\Box$}

\subsection{Proof of \Cref{lemma: microAssignment}:}
Under a policy $\pi$ generating the optimal aggregate policy $\tpi^*$, we have $v_s^{\pi}(s) = w$ for all $s \leq \size$. Hence, all feasible allocations satisfy the following inequality for all $k \geq s$:
\begin{align}\label{piNecCondition}
    v_s(k) \leq \frac{\alpha_s^*}{\alpha_k^*}v_s(s) = \frac{\alpha^*_s}{\alpha^*_k}w
\end{align}
\revcolor{First we} will show:
$$v_s^*(k) = \min\left(C_k,\frac{\alpha^*_s}{\alpha^*_k}w\right) \qquad \text{where $C_k$ is the solution to} \qquad \sum_{s=1}^k \min\left(C_k,\frac{\alpha^*_s}{\alpha^*_k}w\right) = V^{\tpi}_k - w \sum_{s=k+1}^{\size} \frac{\alpha^*_s}{\alpha^*_k}.$$
Assume for contradiction that this is not the case. Then, there must exist a policy $\pi'$ such that its $k-s+1^{\text{th}}$ lowest value is greater than $\min\left(C_k,\frac{\alpha^*_s}{\alpha^*_k}w\right)$. However, it cannot exceed $\frac{\alpha^*_s}{\alpha^*_k}w$, as this quantity serves as its upper bound. The only remaining possibility is that it is greater than $C_k$,  but this would imply that all other values with higher rank than $k-s+1$ receive more than $C_k$. This contradicts the constraint imposed by the fixed summation of values, completing the proof.

Now are left to prove 
\begin{align}\label{eq:proofLem3}
    \min\left(C_k,\frac{\alpha^*_s}{\alpha^*_k}w\right) \geq \frac{\alpha^*_{k+1}}{\alpha^*_k}\min\left(C_{k+1},\frac{\alpha^*_s}{\alpha^*_{k+1}}w\right) = \min\left(\frac{\alpha^*_{k+1}}{\alpha^*_k}C_{k+1},\frac{\alpha^*_s}{\alpha^*_{k}}w\right)
\end{align}
On the other hand, writing the summation constraint for $k+1$ we get:
\begin{align*}
    \frac{\alpha^*_{k+1}}{\alpha^*_k}V^{\tpi}_{k+1} =& \sum_{s=1}^{k+1} \min\left(\frac{\alpha^*_{k+1}}{\alpha^*_k}C_{k+1},\frac{\alpha^*_s}{\alpha^*_k}w\right) + w \sum_{s=k+2}^{\size} \frac{\alpha^*_s}{\alpha^*_k}\\
    =& \sum_{s=1}^{k} \min\left(\frac{\alpha^*_{k+1}}{\alpha^*_k}C_{k+1},\frac{\alpha^*_s}{\alpha^*_k}w\right) + w \sum_{s=k+1}^{\size} \frac{\alpha^*_s}{\alpha^*_k}
\end{align*}
Where the last equation comes \revcolor{from} the fact that $C_{k+1}\geq w$.
Now notice that $V_k^{\tpi} \geq \frac{\alpha^*_{k+1}}{\alpha^*_k}V_{k+1}^{\tpi}$ therefore,
\begin{align*}
    \sum_{s=1}^k \min\left(C_k,\frac{\alpha^*_s}{\alpha^*_k}w\right) + w \sum_{s=k+1}^{\size} \frac{\alpha^*_s}{\alpha^*_k} = V^{\tpi}_k \geq \frac{\alpha^*_{k+1}}{\alpha^*_k}V^{\tpi}_{k+1} 
    = \sum_{s=1}^{k} \min\left(\frac{\alpha^*_{k+1}}{\alpha^*_k}C_{k+1},\frac{\alpha^*_s}{\alpha^*_k}w\right) + w \sum_{s=k+1}^{\size} \frac{\alpha^*_s}{\alpha^*_k},\end{align*}
which means $C_k \revcolor{\geq} \frac{\alpha^*_{k+1}}{\alpha^*_k}C_{k+1}$ which directly leads to the desired result. $\Box$

\subsection{Proof of \Cref{prop:SmoothJ}} First we prove the feasibility of the program. \revcolor{Similar to proof of} \Cref{lem:Jmonot} we \revcolor{can prove that} the solution to \eqref{mainProgramAgg3} is non-decreasing. We will adjust the procedure in proof of \Cref{thm:original&relaxed} to maintain monotonicity of $J$ therefore the proof will be complete using induction. In the recursive procedure we didn't set $x$ to be any number and just used the fact that there exists an $x \leq 1$ which makes ${\mathcal F} (x)$ positive. Let us set $x$ to be the infimum of such numbers or equivalently the solution to ${\mathcal F} (x) =0$. Now we prove by using this update $J$ remains monotone increasing. We consider three cases:
\begin{itemize}
    \item[--] For $k \leq s$ we have $\tilde{J}_{k} = J_{k}$. so $\tilde{J}_{k} - \tilde{J}_{k-1} = J_{k} - J_{k-1} \geq 0$ for such $k$.
    \item[--] We have $\tilde{J}_{s} = 0 + \frac{\alpha^*_{s+1}}{\alpha^*_s}\tilde{J}_{s+1} < \tilde{J}_{s+1}$
    \item[--] For $k > s+1$: $\tilde{J}_{k} - \tilde{J}_{k-1} = (1-x)\left(J_{k} - J_{k-1}\right) \geq 0$.
\end{itemize}
Which completes the existence proof.

Now that we have established the existence of a solution to this program, let $J$ denote the solution to the quadratic programming problem in \eqref{eq:SmoothJ}, determined by $\imm(k)$. Similar to \Cref{cor:FeasibleDirection}, assume for $\sOne < \size$ we have $\imm( \sOne ) < \frac{p\lambda_k}{r}$ and $J_{\sOne} < J_{\sOne+1}$, then the following vector of commission fees 
    \begin{align}\label{direction}
        \tilde{\imm}(k) = \begin{cases} 
                (1-\epsilon)\imm(k)
                 & \text{if } k \geq \sOne+1, \\
                \imm(k) + \epsilon\sum_{\ell = \sOne+1}^{\size}\imm(\ell)\frac{\alpha^*_{\ell}-\alpha^*_{\ell+1}}{\alpha^*_{\sOne}-\alpha^*_{\sOne+1}}
                & \text{if } k = \sOne, \\
                \imm(k) & \text{o.w.}
            \end{cases}
    \end{align}
    is still feasible for sufficiently small $\epsilon$. Now let $\tilde{J}$ denote the updated payoff vector for the platform under $\tilde{\imm}$. We know for $\ell > \sOne$, $\tilde{J}_{\ell} = (1-\epsilon)J_{\ell}$ and $\tilde{J}_{\ell} =J_{\ell}$ for $\ell \leq \sOne$. Consequently, $\tilde{J}$ has lower quadratic variation than $J$ since
    \begin{align*}
        \sum_{\ell=1}^{\size} \big(\tilde{J}_{\ell} - \tilde{J}_{\ell-1}\big)^2 < \sum_{\ell=1}^{\sOne} \big(J_{\ell} - J_{\ell-1}\big)^2 + \big(J_{\sOne+1} - J_{\sOne}\big)^2 + (1-\epsilon)^2\sum_{\ell=\sOne+2}^{\size} \big(J_{\ell} - J_{\ell-1}\big)^2 < \sum_{\ell=1}^{\size} \big(J_{\ell} - J_{\ell-1}\big)^2.
    \end{align*}
    This shows at optimality, the update \eqref{direction} is not feasible, which means either $\imm(k) = \frac{p\lambda_k}{r}$ or $J_k = J_{k+1}$. Finally, consider a point like $\sOne$ where $J_{\sOne-1} = J_{\sOne} < J_{\sOne+1}$. Because of the strict inequality we know we should have $\imm(\sOne) = \frac{p\lambda_{\sOne}}{r}$, therefore:
    \begin{align*}
        J_{\sOne} = \left(1 - \frac{\alpha^*_{\sOne+1}}{\alpha^*_{\sOne}}\right)\frac{p\lambda_{\sOne}}{r} + \frac{\alpha^*_{\sOne+1}}{\alpha^*_{\sOne}}J_{\sOne + 1} < J_{\sOne+1}
    \end{align*}
    Given that $J_{\sOne}$ is written as strict convex combination ($\alpha^*_{\sOne} > 0$) of the two terms we get $\frac{p\lambda_{\sOne}}{r} < J_{\sOne}$. On the other hand,
    \begin{align*}
        J_{\sOne} = J_{\sOne - 1} = \left(1 - \frac{\alpha^*_{\sOne}}{\alpha^*_{\sOne-1}}\right)\imm(\sOne-1) + \frac{\alpha^*_{\sOne}}{\alpha^*_{\sOne-1}}J_{\sOne} = \imm(\sOne-1) \leq \frac{p\lambda_{\sOne-1}}{r},
    \end{align*}
    which contradicts with the fact that $\lambda_{\sOne - 1} < \lambda_{\sOne}$. Therefore we can't have a point $\sOne$ with $J_{\sOne-1} = J_{\sOne} < J_{\sOne+1}$. 
    Specifically, when $\imm(k) \neq \frac{p\lambda_k}{r}$ or equivalently $F(k) \neq p$, we have
\begin{align*}
    J_{k+1} = J_k = \left(1-\frac{\alpha^*_{k+1}}{\alpha^*_k}\right) \imm(k) + \frac{\alpha^*_{k+1}}{\alpha^*_k}\,J_{k+1}.
\end{align*}
This means $\imm(k) = J_k$, which can also be interpreted as $F(k) = \frac{rJ_k}{\lambda_k}$. It implies that $F_k$ decreases in regions where $J_k$ remains fixed. 
    
Bringing all the pieces together, $J$ begins as an increasing function until reaching a threshold $s^*$ and remains constant after that. In contrast, $F(k)$ behaves oppositely: it starts with $F(k)=p$ and begins to decrease once $k = s^*$.$\Box$

\subsection{Proof of \Cref{prop:SmoothV}}
    Let $\imm(k)$ represents the solution to \eqref{eq:SmoothV}. Notice that for all $s \leq \size$ with $\imm(s-1) > 0$, $\imm(s) < \frac{p\lambda_s}{r}$, and $V_{s} > w\left(s + \sum_{\ell=s+1}^{\size}\frac{\alpha^*_{\ell}}{\alpha^*_s}\right)$ the following vector
    \begin{align}\label{directionD}
        \tilde{\imm}(k) = \begin{cases} 
                (1-\epsilon)\imm(k)
                 & \text{if } k= s-1, \\
                \imm(k) + \epsilon \imm(s-1) \frac{\alpha^*_{s-1}-\alpha^*_{s}}{\alpha^*_{s}-\alpha^*_{s+1}}
                & \text{if } k = s, \\
                \imm(k) & \text{o.w.}
            \end{cases}
    \end{align}
    is feasible for sufficiently small $\epsilon$. This follows from a similar argument to the proof of \Cref{lem:FeasibleDirection}, where the only updated component of $V$ is $V_s$, which is decreasing. First, it does not violate any constraint, as  $V_s$ remains greater than $w\left(s + \sum_{\ell=s+1}^{\size}\frac{\alpha^*_{\ell}}{\alpha^*_s}\right)$ for sufficiently small $\epsilon$. Furthermore, as $V_s$ decreases, the solution to our optimization problem improves. This shows that under the optimal solution, no augmentation like \eqref{directionD} is feasible which means for each $s$ one of the three constraints are violated. 
    
    Let $s-1$ be the largest number in which $\imm(s-1) = \frac{p\lambda_{s-1}}{r}$, which means that the platform retains all the payoff. First, note that this cannot happen at $s-1 = \size$, as there is no way to satisfy DCR at $\size$. Also not that $\imm(s-1) > 0$ and $\imm(s) < \frac{p\lambda_s}{r}$ because of our assumption on $s-1$ being the largest such value. Finally since the servers don't receive any payment at stage $s-1$, 
    \begin{align*}
        \frac{\alpha^*_s}{\alpha^*_{s-1}}V_{s} = V_{s-1} \geq w\left(s-1 + \sum_{\ell = s}^\size \frac{\alpha^*_{\ell}}{\alpha^*_{s-1}}\right),
    \end{align*}
    therefore
    \begin{align*}
        V_{s} \geq w\left((s-1)\frac{\alpha^*_{s-1}}{\alpha^*_s} + \sum_{\ell = s}^\size \frac{\alpha^*_{\ell}}{\alpha^*_{s}}\right)>w\left(s + \sum_{\ell = s+1}^\size \frac{\alpha^*_{\ell}}{\alpha^*_{s}}\right).
    \end{align*}
    This means augmentation like \eqref{directionD} is possible and the contradiction shows that under the optimal solution we always have $\imm(s-1) < \frac{p\lambda_{s-1}}{r}$. Finally this shows for every $s\leq \size$ it is either $\imm(s-1)  = 0$, or $V_{s} = w\left(s + \sum_{\ell=s+1}^{\size}\frac{\alpha^*_{\ell}}{\alpha^*_s}\right)$.
$\Box$

\subsection{Proof of \Cref{prop:vertexprioritybased}}
\textbf{(i)} The polytope has $k$ dimensional with variables $\gamma_s(k)$. A vertex corresponds to a point where at least $k$ of the inequalities hold with equality. First, we demonstrate that the $k$ subsets that tighten the constraints must be nested, meaning each subset is contained within the next. Then, we establish that there are exactly $k!$ such vertices, which coincides with the number of distinct priority orderings that a priority-based policy can operate on.

First note that, as shown in the \cite{berenshtein1989convexity}, Poisson distributions is a distributions with convex ratio, that is $1-\stationary{k}{n}$ is concave in $n$. This means that none of the constraints in \eqref{eq:polytope} is redundant and potentially any of them might become tight in a vertex.

Assume for two sets like $\mathcal{I}_1,\mathcal{I}_2$ the constraints are tight. Then for their union we know:
\begin{align*}
    \left(1-\stationary{k}{|\mathcal{I}_1 \cup \mathcal{I}_2|}\right)\frac{p\lambda_k}{r} &\geq \sum_{s\in \mathcal{I}_1 \cup \mathcal{I}_2} \gamma_s(k) \\
    &= 
    \sum_{s\in \mathcal{I}_1} \gamma_s(k)  + \sum_{s\in \mathcal{I}_2} \gamma_s(k)  - \sum_{s\in  \mathcal{I}_1 \cap \mathcal{I}_2} \gamma_s(k) \\
    &\geq \left(1-\stationary{k}{|\mathcal{I}_1|}\right)\frac{p\lambda_k}{r} + \left(1-\stationary{k}{|\mathcal{I}_2|}\right)\frac{p\lambda_k}{r} - \left(1-\stationary{k}{|\mathcal{I}_1 \cap \mathcal{I}_2|}\right)\frac{p\lambda_k}{r}.
\end{align*}
This means:
\begin{align*}
    \stationary{k}{|\mathcal{I}_1|} + \stationary{k}{|\mathcal{I}_2|} \geq \stationary{k}{|\mathcal{I}_1 \cup \mathcal{I}_2|} + \stationary{k}{|\mathcal{I}_1 \cap \mathcal{I}_2|}.
\end{align*}
This is contradicting the convexity of $\stationary{k}{n}$ and concavity of $1-\stationary{k}{n}$ unless $\{\mathcal{I}_1,\mathcal{I}_2\}=\{\mathcal{I}_1\cap\mathcal{I}_2,\mathcal{I}_1\cup\mathcal{I}_2\}$. 
Therefore, each vertex of polytope with $\gamma_{\rho_1}(k) \geq \gamma_{\rho_2}(k) \geq \dots \geq \gamma_{\rho_k}(k)$ should have the constraint tight for the sets $\mathcal{I}^{\rho}_n = \{\rho_1,\rho_2,...,\rho_n\}$ for all $n \in [k]$. Finally, notice that there are $k!$ choices for permutations $\rho$ over $[k]$ and therefore there are $k!$ different vertices. $\Box$

\textbf{(ii)} Fix a permutation $\rho$ and look at the priority-based assignment with ranking $\rho$. Note that, for each $n$ the state of the servers in $\mathcal{I}^{\rho}_n = \{\rho_1,\rho_2,...,\rho_n\}$ follows a birth-death process with only $n$ vertices. Therefore, the throughput of the subsystem is equal to $\left(1-\stationary{k}{n}\right)\lambda_k$. Hence, for all $n \leq k$ we have:
\begin{align*}
    \sum_{s=1}^n \gamma_{\rho_s}(k) = \left(1-\stationary{k}{n}\right)\frac{p\lambda_k}{r},
\end{align*}
which is the same as the equations satisfied by the vertex.$\Box$
\subsection{Proof of \Cref{thm:implementability}}
Our proof has two steps. In the first step we show any feasible point can be projected to another feasible point that assigns all the revenue to the servers: 
\begin{align*}
    \sum_{s \in [k]} \gamma_s(k) = (1-\stationary{k}{k})\frac{p\lambda_k}{r}.
\end{align*}
Second, we demonstrate that any point satisfying the above condition can be implemented using an assignment policy along with a fixed service fee $f_s(k) = p$.

\textbf{1. Reduction:} Assume a feasible point like $\gamma(k)$ with $\sum_{s \in [k]} \gamma_s(k) < (1-\stationary{k}{k})\frac{p\lambda_k}{r}$. As shown in the proof of \Cref{prop:vertexprioritybased}, all the tight constraints are nested, therefore there is at least one element like $\gamma_{s'}(k)$ such that is not present in any of the tight constraints. We can inflate this number and define $\hat{\gamma}_{s'}(k)$ to be the maximum number satisfying all the constraints. If this adjustment results in $\hat{\gamma}_{s'}(k) + \sum_{s \neq s'} \gamma_s(k) = (1-\stationary{k}{k})\frac{p\lambda_k}{r}$, the process is complete. Otherwise, we continue iterating this adjustment until we obtain a new solution $\hat{\gamma}(k)$ that satisfies the constraint for $\mathcal{I} = [k]$ with an equation. Now, observe that if $\hat{\gamma}(k)$ can be implemented with $(\hat{q}_s(x),\hat{f}_s(x)=p)$, then to implement $\gamma(k)$ we can use $(\hat{q}_s(x),\frac{\gamma_s(k)}{\hat{\gamma}_s(k)}p)$, maintaining the same assignment rule but applying a proportionally reduced service fee.

\textbf{2. Implementation:}
Notice that after setting $f_s(k) = p$, a term $\frac{p\lambda_k}{r}$ from all the inequalities in \eqref{eq:polytope} can be simplified from both sides and they can be written as a new parameter $q_s$ the fraction of total jobs served by server $s$, in the following way:
\begin{align}\tag{$\Xi '$}\label{eq:probpolytope}
    \sum_{s\in \mathcal{I}} q_s \leq 1-\stationary{k}{|\mathcal{I}|} \quad \forall \mathcal{I} \subset [k] \qquad \text{and} \qquad \sum_{s\in[k]} q_s= 1-\stationary{k}{k}
\end{align}
We want to show if a set of $q_s$ is feasible in polytope \eqref{eq:probpolytope}, there exists an assignment policy $\pi \in \Pi_k$ such that $\e_{x \sim \psi_k^{\pi}}[q_s^{\pi}(k,x)] = q_s$. We will apply Kakutani fixed-point theorem to show the existance of such policy. 

First, observe that the space of policies is simply the following polytope:
\begin{align*}
    \Pi_k = \left\{\pi \middle| \sum_{\{s|x_s = 0\}} q_s^{\pi}(k,x) = 1 \quad \forall x \in \{0,1\}^k  \right\}.
\end{align*}
Therefore, it is easy to see the space of the policies is convex and bounded.

Let us define a function $\mathcal F: \Pi_k \rightarrow \reals^{2^k}$, that gets $\pi \in \Pi_k$ as an input and outputs the unique (see \cite{douc2018markov} chapter 1.4) stationary distributions $\psi_k^{\pi}$ this policy implements. Furthermore, define $\mathcal G: \text{range}(\mathcal F) \rightarrow \mathcal{P}(\Pi_k) $ as a set-value function that takes a stationary distribution as input and returns all policies that implement $q_s$ while preserving the stationary distribution, i.e. $\e_{x \sim \psi_k}[(\mathcal{G}(\psi_k))_s(k,x)] = q_s$. Therefore, if we prove $\mathcal F$ and $\mathcal G$ are well-defined and continuous ($\mathcal G$ has a closed graph) we can apply the Kakutani fixed-point theorem on $\mathcal F o \mathcal G$ and get the result. 

First of all, we know a linear system is continuous in its coefficients. Therefore $\mathcal F$ which solves a linear system to calculate the stationary distribution is indeed continuous and well-defined on $\Pi_k$.

For $\mathcal G$ however, we need to prove it is well-defined. We want to prove the following polytope is feasible
\begin{align*}
    \left\{ \sum_{\{x|x_s = 0\}} q_s(k,x) \psi_k(x) = q_s \quad \forall s, \qquad \sum_{\{s|x_s = 0\}}^k q_s(k,x) = 1 \quad \forall x, \qquad q_s(k,x) \geq 0,\right\}
\end{align*} 
To prove feasibility we use Farkas' lemma. If the program is not feasible then there exist a vector $y$ such that:
\begin{align*}
    y_s \psi_k(x) + z_x \geq 0 \quad \forall s,x:x_s = 0 \qquad \text{and} \qquad \sum_{s} y_sq_s + \sum_x z_x < 0.
\end{align*}
Assume $y_1 \geq y_2 \geq \dots \geq y_k$. For each $x$, define $s(x)$ the maximum $s$ such that $x_s = 0$. Summing all constraints $x,s(x)$ we get:
\begin{align*}
    \sum_{s=1}^k y_s \sum_{\{x|x_s = 0 \& x_j =1 \forall j > s\}} \psi_k(x) \geq - \sum_x z_x > \sum_s y_sq_s.
\end{align*}
This is equivalent to:
\begin{align}\label{ineqContradiction}
    y_k(1-\psi_k(\{1\}^k)) + \sum_{s=1}^{k-1} (y_s-y_{s+1})\sum_{s'=1}^s \sum_{\{x|x_{s'} = 0 \& x_j =1 \forall j > s\}} \psi_k(x) &\geq - \sum_x z_x \nonumber \\
    &> y_k\sum_{s'=1}^k q_{s'} + \sum_{s=1}^{k-1} (y_s-y_{s+1}) \sum_{s'=1}^sq_{s'}.
\end{align}
But we know for each $s$ we have:
\begin{align*}
    \sum_{s'=1}^s \sum_{\{x|x_{s'} = 0 \& x_j =1 \forall j > s\}} \psi_k(x) &= \sum_{\{x|x \neq \{1\}^k \& x_j =1 \forall j > s\}} \psi_k(x)\\
    &\leq \stationary{k}{k-s} - \stationary{k}{k} \\
    &= 1 - \stationary{k}{k} - (1-\stationary{k}{k-s}) \leq 
    \sum_{s'=1}^k q_{s'} - \sum_{s'=s+1}^kq_{s'}.
\end{align*}
The first inequality comes from the fact that the maximum chance of servers $\{s+1,...,k\}$ is to have them be the top priority in a priority-based policy.

Notice the last inequality contradicts with the strict inequality \eqref{ineqContradiction}. Therefore, $\mathcal{G}$ is well-defined. To prove $\mathcal{G}$ is closed graph, we know $\mathcal{G}(\psi_k^{\pi})$ is the solution to the following linear equation:
\begin{align*}
     \mathcal{G}(\psi_k^{\pi})^T \cdot \psi_k^{\pi} = q.
\end{align*}
Therefore all feasible pairs $(\psi_k^{\pi},\mathcal{G}(\psi_k^{\pi}))$ lie within the preimage of $q$ under multiplication, which is a closed set. Therefore, we conclude that $\mathcal{G}$ is a closed-graph function. Applying the Kakutani fixed-point theorem, we obtain the desired result. $\Box$




\subsection{Proof of \Cref{prop: decreasingGamma}}
Take the optimal solution with minimum $\sum_{k=1}^{\size} \sum_{s=1}^k \gamma_s(k)^2$. We will prove in this optimal solution, $\gamma_s(k)$ is decreasing in $s$.

Assume not; then there are $s < s'$ such that for one $k \geq s'$, $\gamma_{s'}(k) > \gamma_{s}(k)$. Let $\hat{k}$ be the largest among all. Since $(\alpha,\imm,\gamma,\size)$ is the optimal solution, we have $v_s(s') \geq w = v_{s'}(s')$, therefore there exists another $k \geq s'$ such that $\gamma_s(k) > \gamma_{s'}(k)$. Again let $\hat{k}'$ be the largest among them. We handle two cases:
\begin{enumerate}
    \item $\hat{k}' > \hat{k}:$ For all \(\hat{k}' \geq k > \hat{k}\), we have \( v_s(k) > v_{s'}(k) \geq w \). This implies that we can make small adjustments to the variables without violating any constraints.
Specifically, consider a sufficiently small \(\epsilon\). We define new variables as follows:  
\[
\hat{\gamma}_s(\hat{k}') = \gamma_s(\hat{k}') - \epsilon\alpha_{\hat{k}}, \quad \hat{\gamma}_{s'}(\hat{k}') = \gamma_{s'}(\hat{k}') + \epsilon\alpha_{\hat{k}}
\]
\[
\hat{\gamma}_s(\hat{k}) = \gamma_s(\hat{k}) + \epsilon\alpha_{\hat{k}'}, \quad \hat{\gamma}_{s'}(\hat{k}) = \gamma_{s'}(\hat{k}) - \epsilon\alpha_{\hat{k}'}
\]
Since the congestion constraints are convex, these adjustments remain feasible as long as \(\epsilon\) is sufficiently small. This ensures that the modifications do not violate any existing constraints, but the sum of squared values of $\gamma_s(k)$ decreases. This contradiction establishes the desired result.
\item $\hat{k}' < \hat{k}:$  Following the same reasoning as above, we can apply the same modification and arrive at a contradiction. $\Box$
\end{enumerate}

%% file: Tex/Refinments.tex
\renewcommand{\thefootnote}{\fnsymbol{footnote}}
\setcounter{footnote}{1}
\section{Policy Refinements}\label{sec:Policy Refinements}
\vspace{0.3cm}

\subsection{Smoothing Platform's Expected Payoffs.} A concrete approach for the platform to smooth its cash flows is to select a policy $\imm^* \in \Omega$ that minimizes the quadratic variation of its expected payoff stream, subject to the constraint that the policy ensures non-decreasing expected discounted payoffs:
\begin{align}\label{eq:SmoothJ}
    \min_{\imm^* \in \Omega} \;\sum_{k=1}^{\size} \big(J^{\tilde{\pi}^*}_k - J^{\tilde{\pi}^*}_{k-1}\big)^2 \quad \text{and} \quad  J^{\tilde{\pi}^*}_k \geq J^{\tilde{\pi}^*}_{k-1}\quad k \in [\size^*].
\end{align}

\begin{prop}\label{prop:SmoothJ} {\rm (Smooth NPV)} 
The optimization problem \eqref{eq:SmoothJ} is feasible, and there exists a threshold $s^*$ such that $J^{\tilde{\pi}^*}_k < J^{\tilde{\pi}^*}_{k+1}$ and $\imm^*(k) = \frac{p\lambda_k}{r}$ for $k < s^*$. For $k \geq s^*$, both $\imm^*(k)$ and $J^{\tilde{\pi}^*}_k$ remain constant, with $\imm^*(k) = J^{\tilde{\pi}^*}_k = J^{\tilde{\pi}^*}_{\size}$.
\end{prop}

The result in \Cref{prop:SmoothJ} is shown in \Cref{Fig:oscilation1}. The figure depicts three key components: (1) the average compensations received by servers, calculated as $\Bar{\gamma}(k) = {1 \over k}\,\sum_{s=1}^k \gamma_s(k)$ (left panel), (2) the platform's optimal compensation strategy (middle panel), and (3) the corresponding payoffs (right panel). All of these are under the smoothness condition defined in \eqref{eq:SmoothJ}.

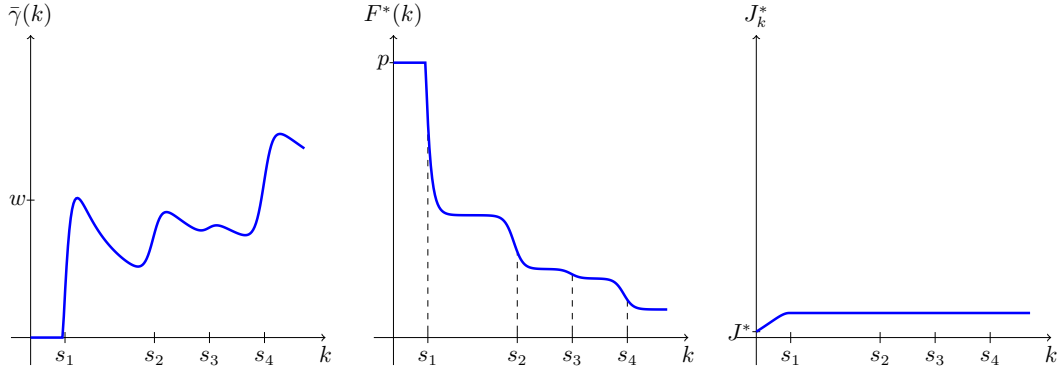
\begin{figure}[h]
\centering
\scalebox{0.8}{
\input{Figures/smoothJ}
}
\caption{\footnotesize \sf Average servers compensation $\Bar{\gamma}(k)$, $F^*(k)$ and $J^*_k$ for a platform with a smooth NPV, based on an MSL instance with $N = 4$. In this example, $s^* = 22$. Note that $F^*(k)$ decreases abruptly when the platform encounters a surge in demand to maintain a constant $\imm^*$ and $J^*_k$. These parameter choices are made to ensure that the cycles are visible in this figure and \cref{Fig:oscilation}. \\ \sc{Data}: $\Lambda_n=(70, 55, 20, 160)$, $s_n=(25, 90, 130, 170)$, $\xi_n=0.3\,(1, 1, 1, 1)$, $r=0.1,p=1,w=10, m =10$. }
\label{Fig:oscilation1}
\end{figure}

Interestingly, the middle panel reveals that the platform's optimal compensation mechanism follows a greedy strategy. At the start of the growth phase, the platform sets $F^*_{k} = p$, retaining all customer service fees for itself until it reaches the threshold $s^* = s_1$. Beyond this point, the platform gradually transitions to a compensation scheme that incrementally increases the fraction of the service fee allocated to servers in a piecewise constant manner. This gradual decrease in the intermediation fee charged by the platform is offset by a corresponding increase in growth and market demand, resulting in a constant expected payoff for the platform (right panel).\vspace{0.3cm}

\subsection{Smoothing Servers' Compensation.} While the compensation scheme in \cref{prop:SmoothJ} ensures that all servers are adequately compensated by satisfying their \eqref{eq:dynamicparticipation} constraints at all times, it may lead to dissatisfaction among senior servers, particularly during the growth phase, due to delayed or inconsistent payouts. The irregular cash flows (as shown in the left panel of \cref{Fig:oscilation1}) could create a sense of uncertainty and instability in the servers' experience with the platform, making this approach impractical.

To address this issue, a possibly more practical solution would involve providing individual servers with stable and consistent compensation throughout their tenure on the platform. This consideration leads to an alternative approach that prioritizes smoothing servers' compensation to create a more reliable environment for them while simultaneously ensuring the platform's growth and financial stability. 
To formalize this alternative criterion, we introduce the notion of the platform's liability process which is defined as follows:
\begin{align*}
    \liability{k} \coloneqq V_k - w\left(k + \sum_{\ell=k+1}^{\size^*}\frac{\alpha^*_{\ell}}{\alpha^*_k}\right), \qquad k=1,\dots,\size^*.
\end{align*} 
Note that $V_k$ represents the total payments to all servers from stage $k$ onward, while $w\left(k + \sum_{\ell=k+1}^{\size^*}\frac{\alpha^*_{\ell}}{\alpha^*_k}\right)$ corresponds to the minimum payments required to meet for future obligations (those of both currently active and future servers). Thus, the difference between these two quantities reflects the amount of compensation designated to address past commitments. By minimizing this objective, the platform aims to stabilize servers' cash flows over time. Ideally, closing this gap would ensure a constant value of $\gamma_s(k) = w$, resulting in fixed average compensation for servers and eliminating the need for them to wait for future periods with higher payouts. However, achieving this ideal solution is unlikely to be feasible, as demand is generally very low during the platform's early growth stages, making constant compensation impossible.

Therefore, a proxy for achieving smoother server compensation and reducing delays in reaching fair payouts is to minimize the redistribution of payments over time. To this end, the platform can formulate the following minimization problem:
\begin{align}\label{eq:SmoothV}
    \min_{\Gamma^* \in \Omega} \;\mathcal{C}\left(\liability{1},\dots,\liability{\size} \right)
\end{align}
Where $\mathcal{C}$ is a cost function, monotonically increasing in each argument, that takes the vector of liabilities as defined above as its input.

\begin{prop}\label{prop:SmoothV}
For any increasing function $\mathcal{C}$, the optimal solution to \eqref{eq:SmoothV} specifies a solution to
\eqref{mainProgramAgg} where the platform's cashflow rate $\imm(k)$, and liability follows a cycle pattern. Specifically, there exists an integer $n$ and an increasing sequence of thresholds $\{s^*_0, s^*_1, \dots, s^*_n\}$, with $s^*_0 = 0$ and $s^*_n = S^*$, such that $\imm (k) = 0$ for $k \in [s^*_{2i}, s^*_{2i+1})$ and 
$\liability{k+1} = 0$ for $k \in \revcolor{(s^*_{2i-1}, s^*_{2i})}$. 
\end{prop}

The solution proposed in \Cref{prop:SmoothV} is shown in \Cref{Fig:oscilation}. This solution is designed to minimize liabilities as outlined in \eqref{eq:SmoothV} with $\mathcal{C}(L_1,...,L_{\size}) = \sum_{s=1}^{\size}L_s^2$.

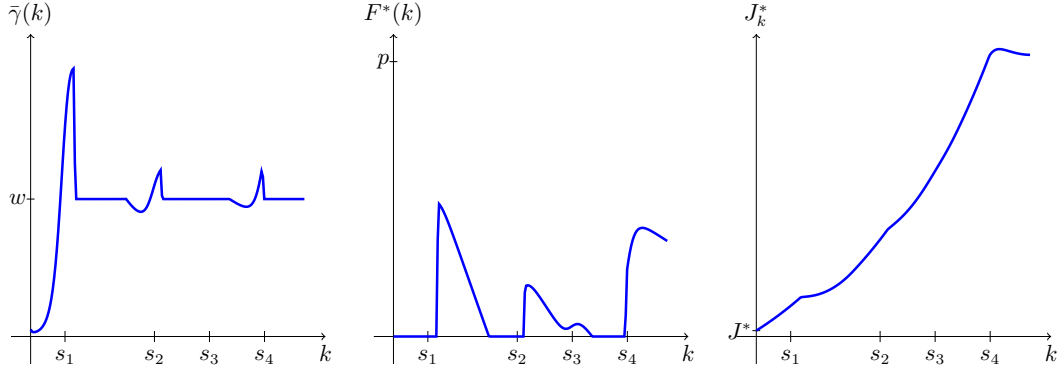
\begin{figure}[h!]
\centering
\scalebox{0.8}{
\input{Figures/smoothLiability}
}
\caption{\footnotesize \sf Average servers compensation $\Bar{\gamma}(k)$, $F^*(k)$ and $J^*_k$, based on an MSL instance  similar to \Cref{Fig:oscilation1}. In this example, $s^*_1 = 32, s^*_2 = 70, s^*_3 = 95$, $s^*_4 = 145, s^*_5 = 168$. Note that after $s_5^*$ the platform manages to keep liabilities zero and pay the servers exactly $w$.\\ \sc{Data}: $\Lambda_n=(70, 55, 20, 160)$, $s_n=(25, 90, 130, 170)$, $\xi_n=0.3\,(1, 1, 1, 1)$, $r=0.1,p=1,w=10, m =10$. }
\label{Fig:oscilation}
\end{figure}

The platform's optimal compensation mechanism reveals a distinctive pattern of growth cycles. In each cycle, $\imm(k)$ begins at zero, allowing the platform to allocate all available revenue toward settling prior commitments. Once these liabilities are cleared, the platform starts retaining the remaining revenue generated by demand. However, as the number of active servers increases while demand remains constant, this retained revenue gradually decreases (as shown in \Cref{Fig:oscilation}). This decline in $\imm(k)$ reflects the platform's deliberate effort to stabilize server payoffs while adjusting to evolving market conditions, resulting in diminishing platform payoffs over time.

On the servers' side, liabilities initially grow until a surge in demand occurs. After such a demand jump, liabilities begin to decrease but may rise again during periods where the number of active servers continues to grow while demand remains fixed.  Essentially, the platform’s optimal strategy is structured around a series of escalating economic growth levels. To move from one growth level to the next, the platform must achieve a critical size, defined by thresholds $s_i^*$, which ensure sufficient future cash flows to sustain continued growth.

In summary, these growth cycles capture the platform’s iterative approach to revenue optimization, balancing economic growth with the goal of maintaining smooth and predictable cashflows for servers.

%% file: Figures/smoothJ.tex
\def\xShift{6}
\begin{tikzpicture}
\begin{scope}[yscale = 1/1.1, xscale = 0.65]
    \draw[->] (-0.5 ,0) -- (7.5 ,0) node[below] {$k$};
    \draw[->] (0,-0.5) -- (0,5.5) node[above] {$\Bar{\gamma}(k)$};
    
    \draw (0.125*7,0.1) -- (0.125*7,-0.1) node[below] {$s_1$};
    \draw (0.45 *7,0.1) -- (0.45 *7,-0.1) node[below] {$s_2$};
    \draw (0.65*7,0.1) -- (0.65*7,-0.1) node[below] {$s_3$};
    \draw (0.85*7,0.1) -- (0.85*7,-0.1) node[below] {$s_4$};

    \draw (-0.1,5/2) -- (0.1,5/2) node[left] {$w$};
    

    \draw[blue, very thick] ( 0.0 , 1.7917651737309193e-10 ) -- ( 0.035175879396984924 , 8.980158630600421e-11 ) -- ( 0.07035175879396985 , 6.001680633952826e-11 ) -- ( 0.10552763819095477 , 4.513282109153138e-11 ) -- ( 0.1407035175879397 , 3.6212699505711043e-11 ) -- ( 0.17587939698492464 , 3.027856900391276e-11 ) -- ( 0.21105527638190955 , 2.605544748211587e-11 ) -- ( 0.24623115577889448 , 2.2907335750499414e-11 ) -- ( 0.2814070351758794 , 2.048278213706567e-11 ) -- ( 0.31658291457286436 , 1.857325404586163e-11 ) -- ( 0.35175879396984927 , 1.7048988483607403e-11 ) -- ( 0.3869346733668342 , 1.582747821693431e-11 ) -- ( 0.4221105527638191 , 1.485674831022047e-11 ) -- ( 0.45728643216080406 , 1.4106969560754935e-11 ) -- ( 0.49246231155778897 , 1.356655528657787e-11 ) -- ( 0.5276381909547738 , 1.3242462681972711e-11 ) -- ( 0.5628140703517588 , 1.3164502168111165e-11 ) -- ( 0.5979899497487436 , 1.3397061238151764e-11 ) -- ( 0.6331658291457287 , 1.4066876318745353e-11 ) -- ( 0.6683417085427136 , 1.543187799768475e-11 ) -- ( 0.7035175879396985 , 1.8075699667211407e-11 ) -- ( 0.7386934673366834 , 2.362716083387733e-11 ) -- ( 0.7738693467336684 , 3.888213422728966e-11 ) -- ( 0.8090452261306533 , 1.7384131171619782e-10 ) -- ( 0.8442211055276382 , 0.3851181032129418 ) -- ( 0.879396984924623 , 0.8713568868532149 ) -- ( 0.9145728643216081 , 1.300654239788657 ) -- ( 0.949748743718593 , 1.6622852059578384 ) -- ( 0.9849246231155779 , 1.9524018908469625 ) -- ( 1.0201005025125627 , 2.1734436640746795 ) -- ( 1.0552763819095476 , 2.3323827157257417 ) -- ( 1.0904522613065328 , 2.438622239340733 ) -- ( 1.1256281407035176 , 2.5022027839830163 ) -- ( 1.1608040201005025 , 2.5325734758601435 ) -- ( 1.1959798994974873 , 2.537953565918423 ) -- ( 1.2311557788944725 , 2.525084167953336 ) -- ( 1.2663316582914574 , 2.4992857828196966 ) -- ( 1.3015075376884422 , 2.4646320067915712 ) -- ( 1.3366834170854272 , 2.4241690973081647 ) -- ( 1.371859296482412 , 2.380132477260246 ) -- ( 1.407035175879397 , 2.334134981355804 ) -- ( 1.442211055276382 , 2.287321653416792 ) -- ( 1.4773869346733668 , 2.240491776606826 ) -- ( 1.5125628140703518 , 2.1941920595921607 ) -- ( 1.5477386934673367 , 2.148786872544478 ) -- ( 1.5829145728643217 , 2.104509345751145 ) -- ( 1.6180904522613067 , 2.06150051545626 ) -- ( 1.6532663316582914 , 2.019836723036432 ) -- ( 1.6884422110552764 , 1.9795486411567111 ) -- ( 1.7236180904522613 , 1.9406365998251487 ) -- ( 1.758793969849246 , 1.903080676582483 ) -- ( 1.7939698492462313 , 1.8668472083560792 ) -- ( 1.8291457286432162 , 1.8318933875316423 ) -- ( 1.864321608040201 , 1.7981716490804813 ) -- ( 1.899497487437186 , 1.765632764603259 ) -- ( 1.9346733668341707 , 1.7342271840660404 ) -- ( 1.9698492462311559 , 1.7039062569049044 ) -- ( 2.0050251256281406 , 1.674623156789936 ) -- ( 2.0402010050251254 , 1.6463334939717997 ) -- ( 2.0753768844221105 , 1.61899591477107 ) -- ( 2.1105527638190953 , 1.5925727458114185 ) -- ( 2.1457286432160805 , 1.5670302458975536 ) -- ( 2.1809045226130657 , 1.5423392885406417 ) -- ( 2.21608040201005 , 1.5184762225863797 ) -- ( 2.251256281407035 , 1.4954238560931723 ) -- ( 2.2864321608040203 , 1.4731728313666639 ) -- ( 2.321608040201005 , 1.4517234614082808 ) -- ( 2.35678391959799 , 1.4310882035240695 ) -- ( 2.3919597989949746 , 1.4112949583025294 ) -- ( 2.4271356783919598 , 1.392391454837004 ) -- ( 2.462311557788945 , 1.3744510458668286 ) -- ( 2.4974874371859297 , 1.3575803108599669 ) -- ( 2.532663316582915 , 1.3419289399113665 ) -- ( 2.567839195979899 , 1.3277024662957677 ) -- ( 2.6030150753768844 , 1.3151782240572718 ) -- ( 2.6381909547738696 , 1.3047252637056928 ) -- ( 2.6733668341708543 , 1.2968280210549175 ) -- ( 2.708542713567839 , 1.2921131300467434 ) -- ( 2.743718592964824 , 1.291377167757328 ) -- ( 2.778894472361809 , 1.2956104368887238 ) -- ( 2.814070351758794 , 1.3060078245017634 ) -- ( 2.849246231155779 , 1.3239519839137028 ) -- ( 2.884422110552764 , 1.3509473305300947 ) -- ( 2.9195979899497484 , 1.3884804924526857 ) -- ( 2.9547738693467336 , 1.4377808161073655 ) -- ( 2.9899497487437188 , 1.4994851625619885 ) -- ( 3.0251256281407035 , 1.5732454593196294 ) -- ( 3.0603015075376883 , 1.6573887543574062 ) -- ( 3.0954773869346734 , 1.7487883451867428 ) -- ( 3.130653266331658 , 1.8430898431667804 ) -- ( 3.1658291457286434 , 1.9353190866085426 ) -- ( 3.201005025125628 , 2.020718210907114 ) -- ( 3.2361809045226133 , 2.0955277748188585 ) -- ( 3.271356783919598 , 2.1574568704771373 ) -- ( 3.306532663316583 , 2.205743501008156 ) -- ( 3.341708542713568 , 2.240886450528169 ) -- ( 3.3768844221105527 , 2.2642211802094563 ) -- ( 3.4120603015075375 , 2.277498933043197 ) -- ( 3.4472361809045227 , 2.2825604564896436 ) -- ( 3.4824120603015074 , 2.2811285854602 ) -- ( 3.517587939698492 , 2.2747032768769637 ) -- ( 3.5527638190954773 , 2.2645283405450765 ) -- ( 3.5879396984924625 , 2.2516005971642623 ) -- ( 3.6231155778894473 , 2.2366997565558577 ) -- ( 3.6582914572864325 , 2.2204252794319785 ) -- ( 3.6934673366834168 , 2.203232662329715 ) -- ( 3.728643216080402 , 2.1854656487707684 ) -- ( 3.7638190954773867 , 2.1673832532959367 ) -- ( 3.798994974874372 , 2.149181728012213 ) -- ( 3.834170854271357 , 2.1310121604955343 ) -- ( 3.8693467336683414 , 2.1129945824556216 ) -- ( 3.9045226130653266 , 2.0952294680152 ) -- ( 3.9396984924623117 , 2.077807425853953 ) -- ( 3.9748743718592965 , 2.060817786821882 ) -- ( 4.010050251256281 , 2.044356678401223 ) -- ( 4.0452261306532655 , 2.0285350544093648 ) -- ( 4.080402010050251 , 2.0134869837104574 ) -- ( 4.115577889447236 , 1.9993782473189086 ) -- ( 4.150753768844221 , 1.9864148709039218 ) -- ( 4.185929648241206 , 1.9748505258799989 ) -- ( 4.221105527638191 , 1.964990658685506 ) -- ( 4.256281407035176 , 1.9571897107779581 ) -- ( 4.291457286432161 , 1.951836056522326 ) -- ( 4.326633165829146 , 1.9493180079131733 ) -- ( 4.361809045226131 , 1.9499649220002377 ) -- ( 4.396984924623116 , 1.953962425918541 ) -- ( 4.4321608040201 , 1.9612520339841775 ) -- ( 4.467336683417085 , 1.9714417887559752 ) -- ( 4.50251256281407 , 1.9837682340648857 ) -- ( 4.5376884422110555 , 1.997147335139017 ) -- ( 4.572864321608041 , 2.0103232913904625 ) -- ( 4.608040201005025 , 2.022078457941654 ) -- ( 4.64321608040201 , 2.03143380490809 ) -- ( 4.678391959798995 , 2.037773558219806 ) -- ( 4.71356783919598 , 2.0408670731541108 ) -- ( 4.748743718592965 , 2.04080641755014 ) -- ( 4.783919597989949 , 2.0379023009579504 ) -- ( 4.819095477386934 , 2.0325787087969345 ) -- ( 4.8542713567839195 , 2.0252900582544493 ) -- ( 4.889447236180905 , 2.0164677935921764 ) -- ( 4.92462311557789 , 2.006492866703344 ) -- ( 4.959798994974874 , 1.9956866770239692 ) -- ( 4.994974874371859 , 1.9843132387063536 ) -- ( 5.030150753768845 , 1.9725871647753948 ) -- ( 5.06532663316583 , 1.9606840443120888 ) -- ( 5.100502512562814 , 1.9487513657381133 ) -- ( 5.135678391959798 , 1.9369192406970153 ) -- ( 5.1708542713567835 , 1.925310795328528 ) -- ( 5.206030150753769 , 1.9140525605733962 ) -- ( 5.241206030150754 , 1.9032853806942598 ) -- ( 5.276381909547739 , 1.8931765396350793 ) -- ( 5.311557788944723 , 1.8839339286796477 ) -- ( 5.346733668341709 , 1.8758231911332244 ) -- ( 5.381909547738694 , 1.8691888664989729 ) -- ( 5.417085427135678 , 1.864480572501085 ) -- ( 5.452261306532663 , 1.862285109674423 ) -- ( 5.487437185929648 , 1.8633648719541434 ) -- ( 5.522613065326633 , 1.8687017956417225 ) -- ( 5.557788944723618 , 1.8795438182037907 ) -- ( 5.592964824120603 , 1.8974468187751412 ) -- ( 5.628140703517588 , 1.924298603117166 ) -- ( 5.663316582914573 , 1.962302382796247 ) -- ( 5.698492462311558 , 2.013886525301449 ) -- ( 5.733668341708543 , 2.0814990326777467 ) -- ( 5.768844221105528 , 2.167249070385417 ) -- ( 5.8040201005025125 , 2.2723874726925732 ) -- ( 5.839195979899497 , 2.396686748597231 ) -- ( 5.874371859296482 , 2.537881852184705 ) -- ( 5.909547738693467 , 2.6914189449172365 ) -- ( 5.944723618090452 , 2.850746503019921 ) -- ( 5.9798994974874375 , 3.0082104175547286 ) -- ( 6.015075376884422 , 3.156334567896101 ) -- ( 6.050251256281407 , 3.2890569796087172 ) -- ( 6.085427135678392 , 3.402511446498962 ) -- ( 6.1206030150753765 , 3.49518218391606 ) -- ( 6.155778894472362 , 3.5675379140335273 ) -- ( 6.190954773869347 , 3.6214032360034 ) -- ( 6.226130653266331 , 3.6593147278768363 ) -- ( 6.261306532663316 , 3.684009798951328 ) -- ( 6.296482412060302 , 3.6980930081126866 ) -- ( 6.331658291457287 , 3.703859712109143 ) -- ( 6.366834170854272 , 3.703232296363948 ) -- ( 6.402010050251256 , 3.6977647212929328 ) -- ( 6.437185929648241 , 3.6886817600286155 ) -- ( 6.472361809045227 , 3.6769311846256354 ) -- ( 6.507537688442211 , 3.663236597063845 ) -- ( 6.542713567839196 , 3.648144940890206 ) -- ( 6.57788944723618 , 3.6320665112147497 ) -- ( 6.613065326633166 , 3.6153072971065114 ) -- ( 6.648241206030151 , 3.5980944287307457 ) -- ( 6.683417085427136 , 3.5805958315973765 ) -- ( 6.718592964824121 , 3.5629352099516725 ) -- ( 6.7537688442211055 , 3.545203360566645 ) -- ( 6.788944723618091 , 3.5274666708558873 ) -- ( 6.824120603015075 , 3.5097734544603174 ) -- ( 6.85929648241206 , 3.4921586641620097 ) -- ( 6.894472361809045 , 3.4746474024954517 ) -- ( 6.92964824120603 , 3.4572572652305738 ) -- ( 6.964824120603015 , 3.439995716610587 );
\end{scope}

\begin{scope}[shift={(\xShift, 0)}, yscale = 1/1.1, xscale = 0.65]
    \draw[->] (-0.5 ,0) -- (7.5 ,0) node[below] {$k$};
    \draw[->] (0,-0.5) -- (0,5.5) node[above] {$F^*(k)$};

    \draw (0.125*7,0.1) -- (0.125*7,-0.1) node[below] {$s_1$};
    \draw (0.45 *7,0.1) -- (0.45 *7,-0.1) node[below] {$s_2$};
    \draw (0.65*7,0.1) -- (0.65*7,-0.1) node[below] {$s_3$};
    \draw (0.85*7,0.1) -- (0.85*7,-0.1) node[below] {$s_4$};

    \draw (-0.1,5) -- (0.1,5) node[left] {$p$};
    
    \draw[dashed] (0.125*7,0.1) -- (0.125*7, 3.8731793619524364) ;
    \draw[dashed] (0.45*7,0.1) -- (0.45*7,1.5330007093043487) ;
    \draw[dashed] (0.65*7,0.1) -- (0.65*7,1.1410727226987512) ;
    \draw[dashed] (0.85*7,0.1) -- (0.85*7,0.6922029315872698);
    
    \draw[blue, very thick]  ( 0.0 , 4.9999999931378944 ) -- ( 0.035175879396984924 , 4.999999994902992 ) -- ( 0.07035175879396985 , 4.999999996213307 ) -- ( 0.10552763819095477 , 4.99999999718592 ) -- ( 0.1407035175879397 , 4.999999997907792 ) -- ( 0.17587939698492464 , 4.999999998443505 ) -- ( 0.21105527638190955 , 4.99999999884102 ) -- ( 0.24623115577889448 , 4.999999999135952 ) -- ( 0.2814070351758794 , 4.9999999993547375 ) -- ( 0.31658291457286436 , 4.999999999517005 ) -- ( 0.35175879396984927 , 4.9999999996373194 ) -- ( 0.3869346733668342 , 4.999999999726487 ) -- ( 0.4221105527638191 , 4.999999999792525 ) -- ( 0.45728643216080406 , 4.999999999841366 ) -- ( 0.49246231155778897 , 4.999999999877404 ) -- ( 0.5276381909547738 , 4.999999999903869 ) -- ( 0.5628140703517588 , 4.999999999923121 ) -- ( 0.5979899497487436 , 4.999999999936846 ) -- ( 0.6331658291457287 , 4.999999999946167 ) -- ( 0.6683417085427136 , 4.999999999951662 ) -- ( 0.7035175879396985 , 4.999999999953147 ) -- ( 0.7386934673366834 , 4.99999999994862 ) -- ( 0.7738693467336684 , 4.9999999999278915 ) -- ( 0.8090452261306533 , 4.9999999997198845 ) -- ( 0.8442211055276382 , 4.44983128406235 ) -- ( 0.879396984924623 , 3.8731793619524364 ) -- ( 0.9145728643216081 , 3.445981515417599 ) -- ( 0.949748743718593 , 3.1295040680497506 ) -- ( 0.9849246231155779 , 2.8950523913892416 ) -- ( 1.0201005025125627 , 2.7213676074255866 ) -- ( 1.0552763819095476 , 2.5926968295466812 ) -- ( 1.0904522613065328 , 2.497374041932479 ) -- ( 1.1256281407035176 , 2.426756726162112 ) -- ( 1.1608040201005025 , 2.374445918605962 ) -- ( 1.1959798994974873 , 2.335689316381238 ) -- ( 1.2311557788944725 , 2.306976724408786 ) -- ( 1.2663316582914574 , 2.285706094401233 ) -- ( 1.3015075376884422 , 2.2699491375327487 ) -- ( 1.3366834170854272 , 2.258277194416705 ) -- ( 1.371859296482412 , 2.2496311133032063 ) -- ( 1.407035175879397 , 2.2432261247713616 ) -- ( 1.442211055276382 , 2.238480968159783 ) -- ( 1.4773869346733668 , 2.2349650997087287 ) -- ( 1.5125628140703518 , 2.2323597868459655 ) -- ( 1.5477386934673367 , 2.230428898710347 ) -- ( 1.5829145728643217 , 2.2289983836855036 ) -- ( 1.6180904522613067 , 2.2279388483097 ) -- ( 1.6532663316582914 , 2.227153934382149 ) -- ( 1.6884422110552764 , 2.226572764475133 ) -- ( 1.7236180904522613 , 2.226142061851241 ) -- ( 1.758793969849246 , 2.225821451133075 ) -- ( 1.7939698492462313 , 2.225580948904791 ) -- ( 1.8291457286432162 , 2.2253991814707152 ) -- ( 1.864321608040201 , 2.2252603450880457 ) -- ( 1.899497487437186 , 2.225151860331822 ) -- ( 1.9346733668341707 , 2.2250638721565754 ) -- ( 1.9698492462311559 , 2.2249884116794822 ) -- ( 2.0050251256281406 , 2.2249186329731994 ) -- ( 2.0402010050251254 , 2.224848291318197 ) -- ( 2.0753768844221105 , 2.2247710810155015 ) -- ( 2.1105527638190953 , 2.2246797992255734 ) -- ( 2.1457286432160805 , 2.2245661660830023 ) -- ( 2.1809045226130657 , 2.2244199696814158 ) -- ( 2.21608040201005 , 2.224227995976743 ) -- ( 2.251256281407035 , 2.223972914495342 ) -- ( 2.2864321608040203 , 2.2236317242759402 ) -- ( 2.321608040201005 , 2.223173722773233 ) -- ( 2.35678391959799 , 2.2225577956921088 ) -- ( 2.3919597989949746 , 2.221728836518515 ) -- ( 2.4271356783919598 , 2.2206130164848 ) -- ( 2.462311557788945 , 2.2191115701950532 ) -- ( 2.4974874371859297 , 2.217092704814583 ) -- ( 2.532663316582915 , 2.21438120718627 ) -- ( 2.567839195979899 , 2.2107452735114097 ) -- ( 2.6030150753768844 , 2.205880656657882 ) -- ( 2.6381909547738696 , 2.1993916550801518 ) -- ( 2.6733668341708543 , 2.1907705309569208 ) -- ( 2.708542713567839 , 2.1793777231371987 ) -- ( 2.743718592964824 , 2.1644279835411027 ) -- ( 2.778894472361809 , 2.1449914839140845 ) -- ( 2.814070351758794 , 2.120023489500073 ) -- ( 2.849246231155779 , 2.0884403453434737 ) -- ( 2.884422110552764 , 2.049259753743401 ) -- ( 2.9195979899497484 , 2.001809206862199 ) -- ( 2.9547738693467336 , 1.9459919656301128 ) -- ( 2.9899497487437188 , 1.8825360719764022 ) -- ( 3.0251256281407035 , 1.8131357878270906 ) -- ( 3.0603015075376883 , 1.7403757108745557 ) -- ( 3.0954773869346734 , 1.6673966238599294 ) -- ( 3.130653266331658 , 1.5973768776136588 ) -- ( 3.1658291457286434 , 1.5330007093043487 ) -- ( 3.201005025125628 , 1.4760923416878784 ) -- ( 3.2361809045226133 , 1.4275055337629114 ) -- ( 3.271356783919598 , 1.3872403556618762 ) -- ( 3.306532663316583 , 1.3546867381637098 ) -- ( 3.341708542713568 , 1.3288898554239783 ) -- ( 3.3768844221105527 , 1.3087697041386412 ) -- ( 3.4120603015075375 , 1.2932704142179658 ) -- ( 3.4472361809045227 , 1.2814438460190285 ) -- ( 3.4824120603015074 , 1.272484358363689 ) -- ( 3.517587939698492 , 1.2657328016325615 ) -- ( 3.5527638190954773 , 1.2606640192923804 ) -- ( 3.5879396984924625 , 1.2568674056265117 ) -- ( 3.6231155778894473 , 1.2540261132730992 ) -- ( 3.6582914572864325 , 1.2518977650528431 ) -- ( 3.6934673366834168 , 1.250297835679516 ) -- ( 3.728643216080402 , 1.2490859390887363 ) -- ( 3.7638190954773867 , 1.2481547902582748 ) -- ( 3.798994974874372 , 1.2474214244613326 ) -- ( 3.834170854271357 , 1.246820214886693 ) -- ( 3.8693467336683414 , 1.2462972448818046 ) -- ( 3.9045226130653266 , 1.2458056429025886 ) -- ( 3.9396984924623117 , 1.2453015372816751 ) -- ( 3.9748743718592965 , 1.2447403376821122 ) -- ( 4.010050251256281 , 1.2440730981204504 ) -- ( 4.0452261306532655 , 1.2432427705828935 ) -- ( 4.080402010050251 , 1.242180239916853 ) -- ( 4.115577889447236 , 1.2408001653784042 ) -- ( 4.150753768844221 , 1.238996886158314 ) -- ( 4.185929648241206 , 1.2366410311425 ) -- ( 4.221105527638191 , 1.2335780555978664 ) -- ( 4.256281407035176 , 1.2296307076733974 ) -- ( 4.291457286432161 , 1.2246082502988078 ) -- ( 4.326633165829146 , 1.2183256667487334 ) -- ( 4.361809045226131 , 1.2106351924357026 ) -- ( 4.396984924623116 , 1.2014691724009494 ) -- ( 4.4321608040201 , 1.1908868850207903 ) -- ( 4.467336683417085 , 1.1791101995388562 ) -- ( 4.50251256281407 , 1.1665286943054285 ) -- ( 4.5376884422110555 , 1.153660550258307 ) -- ( 4.572864321608041 , 1.1410727226987512 ) -- ( 4.608040201005025 , 1.1292843378095672 ) -- ( 4.64321608040201 , 1.1186865355068951 ) -- ( 4.678391959798995 , 1.1095028626361982 ) -- ( 4.71356783919598 , 1.1017939664085243 ) -- ( 4.748743718592965 , 1.0954930142412336 ) -- ( 4.783919597989949 , 1.0904524666420095 ) -- ( 4.819095477386934 , 1.086486993808948 ) -- ( 4.8542713567839195 , 1.0834050936272872 ) -- ( 4.889447236180905 , 1.0810283557956952 ) -- ( 4.92462311557789 , 1.0792006851703886 ) -- ( 4.959798994974874 , 1.077790707241417 ) -- ( 4.994974874371859 , 1.0766901809728162 ) -- ( 5.030150753768845 , 1.075810411756659 ) -- ( 5.06532663316583 , 1.075077880191211 ) -- ( 5.100502512562814 , 1.0744297209068596 ) -- ( 5.135678391959798 , 1.073809243448278 ) -- ( 5.1708542713567835 , 1.073161558357255 ) -- ( 5.206030150753769 , 1.0724290916665453 ) -- ( 5.241206030150754 , 1.071546788374982 ) -- ( 5.276381909547739 , 1.0704367132304335 ) -- ( 5.311557788944723 , 1.069001736558354 ) -- ( 5.346733668341709 , 1.0671179960992236 ) -- ( 5.381909547738694 , 1.0646258859958588 ) -- ( 5.417085427135678 , 1.061319493122457 ) -- ( 5.452261306532663 , 1.0569347692733742 ) -- ( 5.487437185929648 , 1.0511374269593379 ) -- ( 5.522613065326633 , 1.0435127526056451 ) -- ( 5.557788944723618 , 1.0335613990031214 ) -- ( 5.592964824120603 , 1.0207077444162498 ) -- ( 5.628140703517588 , 1.0043300919234825 ) -- ( 5.663316582914573 , 0.9838233559464267 ) -- ( 5.698492462311558 , 0.9587019379426147 ) -- ( 5.733668341708543 , 0.9287400082310169 ) -- ( 5.768844221105528 , 0.8941252536452371 ) -- ( 5.8040201005025125 , 0.8555769834100035 ) -- ( 5.839195979899497 , 0.8143650859914191 ) -- ( 5.874371859296482 , 0.7721843029209594 ) -- ( 5.909547738693467 , 0.7308942780053866 ) -- ( 5.944723618090452 , 0.6922029315872698 ) -- ( 5.9798994974874375 , 0.6574020835879419 ) -- ( 6.015075376884422 , 0.6272352674215781 ) -- ( 6.050251256281407 , 0.6019111555441778 ) -- ( 6.085427135678392 , 0.5812187458038613 ) -- ( 6.1206030150753765 , 0.5646808053852839 ) -- ( 6.155778894472362 , 0.5516953567137354 ) -- ( 6.190954773869347 , 0.5416404269431649 ) -- ( 6.226130653266331 , 0.5339383771383949 ) -- ( 6.261306532663316 , 0.5280873421446024 ) -- ( 6.296482412060302 , 0.5236704111188986 ) -- ( 6.331658291457287 , 0.5203519292086084 ) -- ( 6.366834170854272 , 0.5178676386017482 ) -- ( 6.402010050251256 , 0.5160128282950379 ) -- ( 6.437185929648241 , 0.514630774125225 ) -- ( 6.472361809045227 , 0.5136025209464725 ) -- ( 6.507537688442211 , 0.5128383513344413 ) -- ( 6.542713567839196 , 0.512270915280659 ) -- ( 6.57788944723618 , 0.511849827627028 ) -- ( 6.613065326633166 , 0.5115374906103082 ) -- ( 6.648241206030151 , 0.5113059012445178 ) -- ( 6.683417085427136 , 0.511134231530871 ) -- ( 6.718592964824121 , 0.5110070058417906 ) -- ( 6.7537688442211055 , 0.5109127381478182 ) -- ( 6.788944723618091 , 0.5108429009651695 ) -- ( 6.824120603015075 , 0.5107911680561379 ) -- ( 6.85929648241206 , 0.510752851495244 ) -- ( 6.894472361809045 , 0.5107244658258565 ) -- ( 6.92964824120603 , 0.5107037370190336 ) -- ( 6.964824120603015 , 0.5106953046318827 );

\end{scope}

\begin{scope}[shift={(2*\xShift, 0)}, yscale = 1/1.1, xscale = 0.65]
    \draw[->] (-0.5 ,0) -- (7.5 ,0) node[below] {$k$};
    \draw[->] (0,-0.5) -- (0,5.5) node[above] {$J^*_{k}$};
    
    \draw (0.125*7,0.1) -- (0.125*7,-0.1) node[below] {$s_1$};
    \draw (0.45 *7,0.1) -- (0.45 *7,-0.1) node[below] {$s_2$};
    \draw (0.65*7,0.1) -- (0.65*7,-0.1) node[below] {$s_3$};
    \draw (0.85*7,0.1) -- (0.85*7,-0.1) node[below] {$s_4$};

    \draw (-0.1,0.10686273222305485) -- (0.1,0.10686273222305485) node[left] {$J^*$};
    
    
    \draw[blue, very thick]( 0.0 , 0.10686256806700434 ) -- ( 0.035175879396984924 , 0.12289855498617719 ) -- ( 0.07035175879396985 , 0.1390834943123751 ) -- ( 0.10552763819095477 , 0.15541489400778552 ) -- ( 0.1407035175879397 , 0.17188885186941194 ) -- ( 0.17587939698492464 , 0.18849956312439264 ) -- ( 0.21105527638190955 , 0.20523866093520238 ) -- ( 0.24623115577889448 , 0.22209433559333003 ) -- ( 0.2814070351758794 , 0.23905016263346668 ) -- ( 0.31658291457286436 , 0.25608355177106334 ) -- ( 0.35175879396984927 , 0.2731637085472165 ) -- ( 0.3869346733668342 , 0.29024898174428193 ) -- ( 0.4221105527638191 , 0.3072834581838024 ) -- ( 0.45728643216080406 , 0.32419267429370424 ) -- ( 0.49246231155778897 , 0.3408783615014168 ) -- ( 0.5276381909547738 , 0.3572122631590666 ) -- ( 0.5628140703517588 , 0.3730293014556003 ) -- ( 0.5979899497487436 , 0.3881207889877458 ) -- ( 0.6331658291457287 , 0.40222901093609553 ) -- ( 0.6683417085427136 , 0.41504531693723634 ) -- ( 0.7035175879396985 , 0.4262146534137419 ) -- ( 0.7386934673366834 , 0.4353497591908828 ) -- ( 0.7738693467336684 , 0.4420572974756227 ) -- ( 0.8090452261306533 , 0.4459753337831034 ) -- ( 0.8442211055276382 , 0.4468169471467289 ) -- ( 0.879396984924623 , 0.4468171046463534 ) -- ( 0.9145728643216081 , 0.44681723518895966 ) -- ( 0.949748743718593 , 0.4468173600492318 ) -- ( 0.9849246231155779 , 0.4468174890646768 ) -- ( 1.0201005025125627 , 0.44681761799755826 ) -- ( 1.0552763819095476 , 0.4468177378015592 ) -- ( 1.0904522613065328 , 0.4468178646088554 ) -- ( 1.1256281407035176 , 0.4468180061429501 ) -- ( 1.1608040201005025 , 0.44681816335591806 ) -- ( 1.1959798994974873 , 0.44681830212355456 ) -- ( 1.2311557788944725 , 0.446818458452096 ) -- ( 1.2663316582914574 , 0.44681863691232127 ) -- ( 1.3015075376884422 , 0.4468188309654215 ) -- ( 1.3366834170854272 , 0.44681903094449416 ) -- ( 1.371859296482412 , 0.4468192260493766 ) -- ( 1.407035175879397 , 0.44681941130757785 ) -- ( 1.442211055276382 , 0.4468195872727163 ) -- ( 1.4773869346733668 , 0.4468197576612105 ) -- ( 1.5125628140703518 , 0.4468199279232503 ) -- ( 1.5477386934673367 , 0.44682010276581674 ) -- ( 1.5829145728643217 , 0.4468202862455337 ) -- ( 1.6180904522613067 , 0.4468204740485946 ) -- ( 1.6532663316582914 , 0.44682065922292313 ) -- ( 1.6884422110552764 , 0.4468208368024327 ) -- ( 1.7236180904522613 , 0.4468209983198044 ) -- ( 1.758793969849246 , 0.44682113866001827 ) -- ( 1.7939698492462313 , 0.44682126213793794 ) -- ( 1.8291457286432162 , 0.4468213794344136 ) -- ( 1.864321608040201 , 0.44682149686757944 ) -- ( 1.899497487437186 , 0.44682161479429006 ) -- ( 1.9346733668341707 , 0.4468217334268074 ) -- ( 1.9698492462311559 , 0.44682185309730116 ) -- ( 2.0050251256281406 , 0.4468219740485815 ) -- ( 2.0402010050251254 , 0.44682209657250094 ) -- ( 2.0753768844221105 , 0.4468222202868144 ) -- ( 2.1105527638190953 , 0.4468223442001147 ) -- ( 2.1457286432160805 , 0.4468224690293534 ) -- ( 2.1809045226130657 , 0.4468225955916818 ) -- ( 2.21608040201005 , 0.4468227241239382 ) -- ( 2.251256281407035 , 0.4468228548754203 ) -- ( 2.2864321608040203 , 0.44682298808504406 ) -- ( 2.321608040201005 , 0.4468231240232521 ) -- ( 2.35678391959799 , 0.446823262942428 ) -- ( 2.3919597989949746 , 0.4468234050934754 ) -- ( 2.4271356783919598 , 0.44682355072188606 ) -- ( 2.462311557788945 , 0.4468237000565045 ) -- ( 2.4974874371859297 , 0.4468238532968718 ) -- ( 2.532663316582915 , 0.44682401060605614 ) -- ( 2.567839195979899 , 0.44682417210679803 ) -- ( 2.6030150753768844 , 0.4468243385043995 ) -- ( 2.6381909547738696 , 0.44682450996733425 ) -- ( 2.6733668341708543 , 0.44682468650091467 ) -- ( 2.708542713567839 , 0.4468248680516549 ) -- ( 2.743718592964824 , 0.4468250544563155 ) -- ( 2.778894472361809 , 0.446825245631024 ) -- ( 2.814070351758794 , 0.44682544159670545 ) -- ( 2.849246231155779 , 0.44682564239497546 ) -- ( 2.884422110552764 , 0.4468258477183549 ) -- ( 2.9195979899497484 , 0.4468260518839344 ) -- ( 2.9547738693467336 , 0.44682626059176966 ) -- ( 2.9899497487437188 , 0.4468264716876924 ) -- ( 3.0251256281407035 , 0.44682668419251437 ) -- ( 3.0603015075376883 , 0.4468268980660095 ) -- ( 3.0954773869346734 , 0.4468271132878705 ) -- ( 3.130653266331658 , 0.44682732991752605 ) -- ( 3.1658291457286434 , 0.44682754803097335 ) -- ( 3.201005025125628 , 0.4468277676491037 ) -- ( 3.2361809045226133 , 0.4468279887829183 ) -- ( 3.271356783919598 , 0.44682821143255835 ) -- ( 3.306532663316583 , 0.4468284356139236 ) -- ( 3.341708542713568 , 0.44682866131998344 ) -- ( 3.3768844221105527 , 0.44682888849183056 ) -- ( 3.4120603015075375 , 0.446829117060922 ) -- ( 3.4472361809045227 , 0.44682934695276555 ) -- ( 3.4824120603015074 , 0.44682957809033863 ) -- ( 3.517587939698492 , 0.446829810397678 ) -- ( 3.5527638190954773 , 0.44683004380159175 ) -- ( 3.5879396984924625 , 0.4468302782336533 ) -- ( 3.6231155778894473 , 0.4468305136313717 ) -- ( 3.6582914572864325 , 0.44683074993836414 ) -- ( 3.6934673366834168 , 0.4468309871045941 ) -- ( 3.728643216080402 , 0.44683122509299356 ) -- ( 3.7638190954773867 , 0.4468314638712241 ) -- ( 3.798994974874372 , 0.4468317034044882 ) -- ( 3.834170854271357 , 0.44683194366323553 ) -- ( 3.8693467336683414 , 0.44683218462244656 ) -- ( 3.9045226130653266 , 0.4468324262611376 ) -- ( 3.9396984924623117 , 0.4468326685619761 ) -- ( 3.9748743718592965 , 0.4468329115105456 ) -- ( 4.010050251256281 , 0.4468331550954212 ) -- ( 4.0452261306532655 , 0.4468333993075866 ) -- ( 4.080402010050251 , 0.44683364414010573 ) -- ( 4.115577889447236 , 0.44683388958806836 ) -- ( 4.150753768844221 , 0.44683413564845176 ) -- ( 4.185929648241206 , 0.4468343823200016 ) -- ( 4.221105527638191 , 0.4468346296032948 ) -- ( 4.256281407035176 , 0.44683487750188045 ) -- ( 4.291457286432161 , 0.4468351260282972 ) -- ( 4.326633165829146 , 0.44683537518997407 ) -- ( 4.361809045226131 , 0.44683562499515106 ) -- ( 4.396984924623116 , 0.44683587545369674 ) -- ( 4.4321608040201 , 0.44683612657728644 ) -- ( 4.467336683417085 , 0.44683637837867884 ) -- ( 4.50251256281407 , 0.4468366308706142 ) -- ( 4.5376884422110555 , 0.44683688411341915 ) -- ( 4.572864321608041 , 0.4468371381349572 ) -- ( 4.608040201005025 , 0.4468373929397833 ) -- ( 4.64321608040201 , 0.4468376485301577 ) -- ( 4.678391959798995 , 0.44683790490647257 ) -- ( 4.71356783919598 , 0.44683816206771576 ) -- ( 4.748743718592965 , 0.4468384200124744 ) -- ( 4.783919597989949 , 0.44683867873915495 ) -- ( 4.819095477386934 , 0.4468389382465705 ) -- ( 4.8542713567839195 , 0.4468391985339417 ) -- ( 4.889447236180905 , 0.446839459601319 ) -- ( 4.92462311557789 , 0.4468397214455575 ) -- ( 4.959798994974874 , 0.44683998406062786 ) -- ( 4.994974874371859 , 0.446840247449223 ) -- ( 5.030150753768845 , 0.44684051160144134 ) -- ( 5.06532663316583 , 0.4468407764963356 ) -- ( 5.100502512562814 , 0.44684104213301135 ) -- ( 5.135678391959798 , 0.4468413083980092 ) -- ( 5.1708542713567835 , 0.4468415753019521 ) -- ( 5.206030150753769 , 0.4468418428554699 ) -- ( 5.241206030150754 , 0.4468421110691627 ) -- ( 5.276381909547739 , 0.44684237995325793 ) -- ( 5.311557788944723 , 0.4468426495178779 ) -- ( 5.346733668341709 , 0.44684291977355306 ) -- ( 5.381909547738694 , 0.44684319073101886 ) -- ( 5.417085427135678 , 0.44684346240139194 ) -- ( 5.452261306532663 , 0.446843734796464 ) -- ( 5.487437185929648 , 0.4468440079285545 ) -- ( 5.522613065326633 , 0.4468442818106073 ) -- ( 5.557788944723618 , 0.4468445564559551 ) -- ( 5.592964824120603 , 0.4468448318778355 ) -- ( 5.628140703517588 , 0.4468451080882543 ) -- ( 5.663316582914573 , 0.44684538505646265 ) -- ( 5.698492462311558 , 0.44684566232710154 ) -- ( 5.733668341708543 , 0.44684593989948596 ) -- ( 5.768844221105528 , 0.44684621776782185 ) -- ( 5.8040201005025125 , 0.44684649592014325 ) -- ( 5.839195979899497 , 0.4468467743390374 ) -- ( 5.874371859296482 , 0.4468470530031958 ) -- ( 5.909547738693467 , 0.4468473318919308 ) -- ( 5.944723618090452 , 0.446847610988403 ) -- ( 5.9798994974874375 , 0.44684789028012717 ) -- ( 6.015075376884422 , 0.4468481697496083 ) -- ( 6.050251256281407 , 0.44684844939487456 ) -- ( 6.085427135678392 , 0.4468487292205121 ) -- ( 6.1206030150753765 , 0.44684900921835924 ) -- ( 6.155778894472362 , 0.4468492893748064 ) -- ( 6.190954773869347 , 0.4468495696716468 ) -- ( 6.226130653266331 , 0.4468498500862199 ) -- ( 6.261306532663316 , 0.44685013059446055 ) -- ( 6.296482412060302 , 0.4468504111708562 ) -- ( 6.331658291457287 , 0.44685069178923387 ) -- ( 6.366834170854272 , 0.4468509724401091 ) -- ( 6.402010050251256 , 0.44685125310529417 ) -- ( 6.437185929648241 , 0.44685153376361897 ) -- ( 6.472361809045227 , 0.446851814383614 ) -- ( 6.507537688442211 , 0.4468520949325665 ) -- ( 6.542713567839196 , 0.4468523753928073 ) -- ( 6.57788944723618 , 0.4468526557507708 ) -- ( 6.613065326633166 , 0.44685293599833575 ) -- ( 6.648241206030151 , 0.4468532161343628 ) -- ( 6.683417085427136 , 0.4468534961664672 ) -- ( 6.718592964824121 , 0.44685377612221466 ) -- ( 6.7537688442211055 , 0.44685405606066153 ) -- ( 6.788944723618091 , 0.44685433592794244 ) -- ( 6.824120603015075 , 0.4468546157644593 ) -- ( 6.85929648241206 , 0.4468548957332188 ) -- ( 6.894472361809045 , 0.44685517598243574 ) -- ( 6.92964824120603 , 0.44685545704149154 ) -- ( 6.964824120603015 , 0.4468557279351104 );
\end{scope}
\end{tikzpicture}

%% file: Figures/smoothLiability.tex
\def\xShift{6}
\begin{tikzpicture}
\begin{scope}[yscale = 1/1.1, xscale = 0.65]
    \draw[->] (-0.5 ,0) -- (7.5 ,0) node[below] {$k$};
    \draw[->] (0,-0.5) -- (0,5.5) node[above] {$\Bar{\gamma}(k)$};
    
    \draw (0.125*7,0.1) -- (0.125*7,-0.1) node[below] {$s_1$};
    \draw (0.45 *7,0.1) -- (0.45 *7,-0.1) node[below] {$s_2$};
    \draw (0.65*7,0.1) -- (0.65*7,-0.1) node[below] {$s_3$};
    \draw (0.85*7,0.1) -- (0.85*7,-0.1) node[below] {$s_4$};
    
    \draw (-0.1,5/2) -- (0.1,5/2) node[left] {$w$};


    \draw[blue, very thick]( 0.0 , 0.13055449099904223 ) -- ( 0.035175879396984924 , 0.08809216649888735 ) -- ( 0.07035175879396985 , 0.0792468082931195 ) -- ( 0.10552763819095477 , 0.08019093551474919 ) -- ( 0.1407035175879397 , 0.08654169448337337 ) -- ( 0.17587939698492464 , 0.09726511513398753 ) -- ( 0.21105527638190955 , 0.11240674372505839 ) -- ( 0.24623115577889448 , 0.132558081688559 ) -- ( 0.2814070351758794 , 0.1587165924728566 ) -- ( 0.31658291457286436 , 0.1922714326876468 ) -- ( 0.35175879396984927 , 0.2350413542421209 ) -- ( 0.3869346733668342 , 0.28933773624636583 ) -- ( 0.4221105527638191 , 0.3580363988795188 ) -- ( 0.45728643216080406 , 0.44463981474356384 ) -- ( 0.49246231155778897 , 0.5533018039633945 ) -- ( 0.5276381909547738 , 0.6887710329693 ) -- ( 0.5628140703517588 , 0.8561894741949512 ) -- ( 0.5979899497487436 , 1.060663490715935 ) -- ( 0.6331658291457287 , 1.3065229167897487 ) -- ( 0.6683417085427136 , 1.5962232825381566 ) -- ( 0.7035175879396985 , 1.9289600845603139 ) -- ( 0.7386934673366834 , 2.299265263829609 ) -- ( 0.7738693467336684 , 2.696093261246004 ) -- ( 0.8090452261306533 , 3.1030232469875 ) -- ( 0.8442211055276382 , 3.4999999215780426 ) -- ( 0.879396984924623 , 3.866439922169214 ) -- ( 0.9145728643216081 , 4.184809270245885 ) -- ( 0.949748743718593 , 4.443434221039169 ) -- ( 0.9849246231155779 , 4.637649288847404 ) -- ( 1.0201005025125627 , 4.769183949945866 ) -- ( 1.0552763819095476 , 4.8443879157833845 ) -- ( 1.0904522613065328 , 4.872119831613794 ) -- ( 1.1256281407035176 , 3.1488632590176553 ) -- ( 1.1608040201005025 , 2.5000379870738976 ) -- ( 1.1959798994974873 , 2.5000994250476447 ) -- ( 1.2311557788944725 , 2.5000802705558822 ) -- ( 1.2663316582914574 , 2.500064356074421 ) -- ( 1.3015075376884422 , 2.500061685803166 ) -- ( 1.3366834170854272 , 2.5000617260030737 ) -- ( 1.371859296482412 , 2.500060216256826 ) -- ( 1.407035175879397 , 2.500057672023069 ) -- ( 1.442211055276382 , 2.500055276478037 ) -- ( 1.4773869346733668 , 2.5000535970562314 ) -- ( 1.5125628140703518 , 2.500052823204978 ) -- ( 1.5477386934673367 , 2.5000530827029848 ) -- ( 1.5829145728643217 , 2.500054445079199 ) -- ( 1.6180904522613067 , 2.5000568565793504 ) -- ( 1.6532663316582914 , 2.5000602025618868 ) -- ( 1.6884422110552764 , 2.500064159638668 ) -- ( 1.7236180904522613 , 2.50006820667581 ) -- ( 1.758793969849246 , 2.5000718799338824 ) -- ( 1.7939698492462313 , 2.5000751123491085 ) -- ( 1.8291457286432162 , 2.500082418712028 ) -- ( 1.864321608040201 , 2.500082633094517 ) -- ( 1.899497487437186 , 2.500086813382238 ) -- ( 1.9346733668341707 , 2.500080394692425 ) -- ( 1.9698492462311559 , 2.500054578899916 ) -- ( 2.0050251256281406 , 2.500025485706572 ) -- ( 2.0402010050251254 , 2.5001056791058414 ) -- ( 2.0753768844221105 , 2.5001583691182168 ) -- ( 2.1105527638190953 , 2.50020007491477 ) -- ( 2.1457286432160805 , 2.500196956092371 ) -- ( 2.1809045226130657 , 2.50019011911331 ) -- ( 2.21608040201005 , 2.5001737023101995 ) -- ( 2.251256281407035 , 2.5001822530711153 ) -- ( 2.2864321608040203 , 2.500382068109985 ) -- ( 2.321608040201005 , 2.500755564788421 ) -- ( 2.35678391959799 , 2.500756124914136 ) -- ( 2.3919597989949746 , 2.5011191957176115 ) -- ( 2.4271356783919598 , 2.5015125498974475 ) -- ( 2.462311557788945 , 2.471244476674523 ) -- ( 2.4974874371859297 , 2.4391403584543396 ) -- ( 2.532663316582915 , 2.4086729527179314 ) -- ( 2.567839195979899 , 2.380030841792898 ) -- ( 2.6030150753768844 , 2.353475372782918 ) -- ( 2.6381909547738696 , 2.329360430813377 ) -- ( 2.6733668341708543 , 2.308156068481961 ) -- ( 2.708542713567839 , 2.2904752768195995 ) -- ( 2.743718592964824 , 2.277101671080435 ) -- ( 2.778894472361809 , 2.2690132389590008 ) -- ( 2.814070351758794 , 2.2673931649392474 ) -- ( 2.849246231155779 , 2.273612988540145 ) -- ( 2.884422110552764 , 2.2891667890818947 ) -- ( 2.9195979899497484 , 2.315530550151122 ) -- ( 2.9547738693467336 , 2.353924385093805 ) -- ( 2.9899497487437188 , 2.4049759276396636 ) -- ( 3.0251256281407035 , 2.468328313380874 ) -- ( 3.0603015075376883 , 2.5423002398550003 ) -- ( 3.0954773869346734 , 2.62375702601453 ) -- ( 3.130653266331658 , 2.7083366606679786 ) -- ( 3.1658291457286434 , 2.791057693947862 ) -- ( 3.201005025125628 , 2.8671552776824276 ) -- ( 3.2361809045226133 , 2.932863248417272 ) -- ( 3.271356783919598 , 2.9858840293768987 ) -- ( 3.306532663316583 , 3.0254428652088827 ) -- ( 3.341708542713568 , 2.5605541907499747 ) -- ( 3.3768844221105527 , 2.500083614409425 ) -- ( 3.4120603015075375 , 2.5000882219215814 ) -- ( 3.4472361809045227 , 2.5000875527575217 ) -- ( 3.4824120603015074 , 2.5000376516662985 ) -- ( 3.517587939698492 , 2.4999831661944985 ) -- ( 3.5527638190954773 , 2.499952923557697 ) -- ( 3.5879396984924625 , 2.4999623597037117 ) -- ( 3.6231155778894473 , 2.4999945025864885 ) -- ( 3.6582914572864325 , 2.5000394206838905 ) -- ( 3.6934673366834168 , 2.5000655345521094 ) -- ( 3.728643216080402 , 2.5000706402002177 ) -- ( 3.7638190954773867 , 2.500067018475336 ) -- ( 3.798994974874372 , 2.500057893297793 ) -- ( 3.834170854271357 , 2.500045936425467 ) -- ( 3.8693467336683414 , 2.5000340213403716 ) -- ( 3.9045226130653266 , 2.500026219145885 ) -- ( 3.9396984924623117 , 2.5000234695660257 ) -- ( 3.9748743718592965 , 2.500023368944017 ) -- ( 4.010050251256281 , 2.5000236414929686 ) -- ( 4.0452261306532655 , 2.5000238707247573 ) -- ( 4.080402010050251 , 2.5000242705698112 ) -- ( 4.115577889447236 , 2.500024859294889 ) -- ( 4.150753768844221 , 2.500025590094159 ) -- ( 4.185929648241206 , 2.5000263784270476 ) -- ( 4.221105527638191 , 2.500026969138711 ) -- ( 4.256281407035176 , 2.500026968687746 ) -- ( 4.291457286432161 , 2.5000264336166227 ) -- ( 4.326633165829146 , 2.500025677164464 ) -- ( 4.361809045226131 , 2.500024762572091 ) -- ( 4.396984924623116 , 2.5000235759905527 ) -- ( 4.4321608040201 , 2.5000221405241776 ) -- ( 4.467336683417085 , 2.500020869981464 ) -- ( 4.50251256281407 , 2.500020934119073 ) -- ( 4.5376884422110555 , 2.5000244391433055 ) -- ( 4.572864321608041 , 2.500033864206376 ) -- ( 4.608040201005025 , 2.5000502870160024 ) -- ( 4.64321608040201 , 2.5000715780546208 ) -- ( 4.678391959798995 , 2.500092325221219 ) -- ( 4.71356783919598 , 2.5001051253424524 ) -- ( 4.748743718592965 , 2.500107032065097 ) -- ( 4.783919597989949 , 2.50010997612965 ) -- ( 4.819095477386934 , 2.500085282549197 ) -- ( 4.8542713567839195 , 2.500045038360555 ) -- ( 4.889447236180905 , 2.500088383837647 ) -- ( 4.92462311557789 , 2.500236514402676 ) -- ( 4.959798994974874 , 2.5003737966680806 ) -- ( 4.994974874371859 , 2.500275163970269 ) -- ( 5.030150753768845 , 2.5011503580847045 ) -- ( 5.06532663316583 , 2.497736232110617 ) -- ( 5.100502512562814 , 2.482125169372256 ) -- ( 5.135678391959798 , 2.4666647060815614 ) -- ( 5.1708542713567835 , 2.4514769539306607 ) -- ( 5.206030150753769 , 2.4366874566538415 ) -- ( 5.241206030150754 , 2.4224360976299533 ) -- ( 5.276381909547739 , 2.4088892253659373 ) -- ( 5.311557788944723 , 2.3962538203198007 ) -- ( 5.346733668341709 , 2.384794638770896 ) -- ( 5.381909547738694 , 2.374855356249161 ) -- ( 5.417085427135678 , 2.366884748764981 ) -- ( 5.452261306532663 , 2.3614687967453256 ) -- ( 5.487437185929648 , 2.359369094874742 ) -- ( 5.522613065326633 , 2.3615668004455905 ) -- ( 5.557788944723618 , 2.3693090914959813 ) -- ( 5.592964824120603 , 2.384151106705241 ) -- ( 5.628140703517588 , 2.4079799307204186 ) -- ( 5.663316582914573 , 2.442998081119882 ) -- ( 5.698492462311558 , 2.4916332264418224 ) -- ( 5.733668341708543 , 2.556332697584656 ) -- ( 5.768844221105528 , 2.639205003618807 ) -- ( 5.8040201005025125 , 2.7415003335110573 ) -- ( 5.839195979899497 , 2.862990546189291 ) -- ( 5.874371859296482 , 3.001409800007239 ) -- ( 5.909547738693467 , 2.904804804021875 ) -- ( 5.944723618090452 , 2.500064554847272 ) -- ( 5.9798994974874375 , 2.49997788607958 ) -- ( 6.015075376884422 , 2.500010781032115 ) -- ( 6.050251256281407 , 2.5000052675200872 ) -- ( 6.085427135678392 , 2.5000851323349593 ) -- ( 6.1206030150753765 , 2.5001066802697114 ) -- ( 6.155778894472362 , 2.500001798063875 ) -- ( 6.190954773869347 , 2.4999614018112144 ) -- ( 6.226130653266331 , 2.499966541686019 ) -- ( 6.261306532663316 , 2.500007898130656 ) -- ( 6.296482412060302 , 2.5000452352750213 ) -- ( 6.331658291457287 , 2.5000353220660845 ) -- ( 6.366834170854272 , 2.500021060013009 ) -- ( 6.402010050251256 , 2.500018082193675 ) -- ( 6.437185929648241 , 2.5000180934106666 ) -- ( 6.472361809045227 , 2.500018134703634 ) -- ( 6.507537688442211 , 2.5000196051370867 ) -- ( 6.542713567839196 , 2.5000239784592977 ) -- ( 6.57788944723618 , 2.5000223658065854 ) -- ( 6.613065326633166 , 2.500018219657229 ) -- ( 6.648241206030151 , 2.500014354308652 ) -- ( 6.683417085427136 , 2.5000132850041434 ) -- ( 6.718592964824121 , 2.5000146035969117 ) -- ( 6.7537688442211055 , 2.5000173967512156 ) -- ( 6.788944723618091 , 2.5000207837866 ) -- ( 6.824120603015075 , 2.500022236752803 ) -- ( 6.85929648241206 , 2.5000196065286158 ) -- ( 6.894472361809045 , 2.50002130822393 ) -- ( 6.92964824120603 , 2.500020093524454 ) -- ( 6.964824120603015 , 2.5000175391174233 );
\end{scope}

\begin{scope}[shift={(\xShift, 0)}, yscale = 1/1.1, xscale = 0.65]
    \draw[->] (-0.5 ,0) -- (7.5 ,0) node[below] {$k$};
    \draw[->] (0,-0.5) -- (0,5.5) node[above] {$F^*(k)$};

    \draw (0.125*7,0.1) -- (0.125*7,-0.1) node[below] {$s_1$};
    \draw (0.45 *7,0.1) -- (0.45 *7,-0.1) node[below] {$s_2$};
    \draw (0.65*7,0.1) -- (0.65*7,-0.1) node[below] {$s_3$};
    \draw (0.85*7,0.1) -- (0.85*7,-0.1) node[below] {$s_4$};

    \draw (-0.1,5) -- (0.1,5) node[left] {$p$};
    

    \draw[blue, very thick] ( 0.0 , 2.1265789701663647e-05 ) -- ( 0.035175879396984924 , 1.5920437329655552e-05 ) -- ( 0.07035175879396985 , 1.1926453115714482e-05 ) -- ( 0.10552763819095477 , 8.943252926006706e-06 ) -- ( 0.1407035175879397 , 6.715279502662894e-06 ) -- ( 0.17587939698492464 , 5.051110900417797e-06 ) -- ( 0.21105527638190955 , 3.8075961113655707e-06 ) -- ( 0.24623115577889448 , 2.877828969279163e-06 ) -- ( 0.2814070351758794 , 2.1820600902625647e-06 ) -- ( 0.31658291457286436 , 1.660849818169881e-06 ) -- ( 0.35175879396984927 , 1.2699238442937534e-06 ) -- ( 0.3869346733668342 , 9.763177113264523e-07 ) -- ( 0.4221105527638191 , 7.554937005090265e-07 ) -- ( 0.45728643216080406 , 5.891891603388984e-07 ) -- ( 0.49246231155778897 , 4.638136867713854e-07 ) -- ( 0.5276381909547738 , 3.692574037865215e-07 ) -- ( 0.5628140703517588 , 2.980069224016599e-07 ) -- ( 0.5979899497487436 , 2.4449181009395913e-07 ) -- ( 0.6331658291457287 , 2.0460457166793585e-07 ) -- ( 0.6683417085427136 , 1.7535291426516758e-07 ) -- ( 0.7035175879396985 , 1.546160209399898e-07 ) -- ( 0.7386934673366834 , 1.4098849485144687e-07 ) -- ( 0.7738693467336684 , 1.3370932771069314e-07 ) -- ( 0.8090452261306533 , 1.3269427829976158e-07 ) -- ( 0.8442211055276382 , 1.3873210256938566e-07 ) -- ( 0.879396984924623 , 1.5400668780337154e-07 ) -- ( 0.9145728643216081 , 1.8339185362288998e-07 ) -- ( 0.949748743718593 , 2.3789353237274852e-07 ) -- ( 0.9849246231155779 , 3.452670734227159e-07 ) -- ( 1.0201005025125627 , 5.916328907492469e-07 ) -- ( 1.0552763819095476 , 1.3677247730812635e-06 ) -- ( 1.0904522613065328 , 7.231661853518965e-06 ) -- ( 1.1256281407035176 , 1.7617367971254598 ) -- ( 1.1608040201005025 , 2.4081759115112487 ) -- ( 1.1959798994974873 , 2.3754281017142658 ) -- ( 1.2311557788944725 , 2.3336435890331106 ) -- ( 1.2663316582914574 , 2.28486054218192 ) -- ( 1.3015075376884422 , 2.230703998511587 ) -- ( 1.3366834170854272 , 2.172443020926271 ) -- ( 1.371859296482412 , 2.1110482717431553 ) -- ( 1.407035175879397 , 2.047260020585961 ) -- ( 1.442211055276382 , 1.9816443540708137 ) -- ( 1.4773869346733668 , 1.9146347598167068 ) -- ( 1.5125628140703518 , 1.846562862050706 ) -- ( 1.5477386934673367 , 1.7776819292714774 ) -- ( 1.5829145728643217 , 1.708185248891792 ) -- ( 1.6180904522613067 , 1.6382203462089653 ) -- ( 1.6532663316582914 , 1.5678997626797488 ) -- ( 1.6884422110552764 , 1.4973095953581277 ) -- ( 1.7236180904522613 , 1.4265158960591293 ) -- ( 1.758793969849246 , 1.3555692276835651 ) -- ( 1.7939698492462313 , 1.2845079180430967 ) -- ( 1.8291457286432162 , 1.2133545693421262 ) -- ( 1.864321608040201 , 1.1421489288014206 ) -- ( 1.899497487437186 , 1.0708923269887984 ) -- ( 1.9346733668341707 , 0.999622498432206 ) -- ( 1.9698492462311559 , 0.928367303208844 ) -- ( 2.0050251256281406 , 0.8571132171757565 ) -- ( 2.0402010050251254 , 0.7856822006838791 ) -- ( 2.0753768844221105 , 0.7143117257345782 ) -- ( 2.1105527638190953 , 0.6429897521865869 ) -- ( 2.1457286432160805 , 0.5717949659481489 ) -- ( 2.1809045226130657 , 0.5006801563248351 ) -- ( 2.21608040201005 , 0.42968663924992984 ) -- ( 2.251256281407035 , 0.35879167304721216 ) -- ( 2.2864321608040203 , 0.28773453916483516 ) -- ( 2.321608040201005 , 0.2166070537365345 ) -- ( 2.35678391959799 , 0.14655659454525166 ) -- ( 2.3919597989949746 , 0.07630400221221456 ) -- ( 2.4271356783919598 , 0.0066689966157704875 ) -- ( 2.462311557788945 , 4.14205952944202e-08 ) -- ( 2.4974874371859297 , 4.1821525526082786e-08 ) -- ( 2.532663316582915 , 4.227021618279153e-08 ) -- ( 2.567839195979899 , 4.279526384116158e-08 ) -- ( 2.6030150753768844 , 4.342704562913809e-08 ) -- ( 2.6381909547738696 , 4.4198138721414633e-08 ) -- ( 2.6733668341708543 , 4.51443412721093e-08 ) -- ( 2.708542713567839 , 4.6306211834875525e-08 ) -- ( 2.743718592964824 , 4.7731344731378535e-08 ) -- ( 2.778894472361809 , 4.947782488750774e-08 ) -- ( 2.814070351758794 , 5.161967340348757e-08 ) -- ( 2.849246231155779 , 5.425572665087685e-08 ) -- ( 2.884422110552764 , 5.752447482054752e-08 ) -- ( 2.9195979899497484 , 6.162926212488633e-08 ) -- ( 2.9547738693467336 , 6.688164887503293e-08 ) -- ( 2.9899497487437188 , 7.377745464554139e-08 ) -- ( 3.0251256281407035 , 8.313484701314006e-08 ) -- ( 3.0603015075376883 , 9.636014956890332e-08 ) -- ( 3.0954773869346734 , 1.160040099089481e-07 ) -- ( 3.130653266331658 , 1.4705654647279178e-07 ) -- ( 3.1658291457286434 , 2.003983935111869e-07 ) -- ( 3.201005025125628 , 3.038228896733302e-07 ) -- ( 3.2361809045226133 , 5.48259713386856e-07 ) -- ( 3.271356783919598 , 1.403680027543715e-06 ) -- ( 3.306532663316583 , 1.3830060858074397e-05 ) -- ( 3.341708542713568 , 0.8051973301979167 ) -- ( 3.3768844221105527 , 0.9242566669918052 ) -- ( 3.4120603015075375 , 0.9309965660975537 ) -- ( 3.4472361809045227 , 0.9270669355697345 ) -- ( 3.4824120603015074 , 0.9147724022818225 ) -- ( 3.517587939698492 , 0.8959020154891866 ) -- ( 3.5527638190954773 , 0.8719156878008512 ) -- ( 3.5879396984924625 , 0.84398289594558 ) -- ( 3.6231155778894473 , 0.8130660602967167 ) -- ( 3.6582914572864325 , 0.7799060265975761 ) -- ( 3.6934673366834168 , 0.7451152090581193 ) -- ( 3.728643216080402 , 0.7091294832875117 ) -- ( 3.7638190954773867 , 0.6722631066582154 ) -- ( 3.798994974874372 , 0.6347741255585813 ) -- ( 3.834170854271357 , 0.5968707995242852 ) -- ( 3.8693467336683414 , 0.5587275652697181 ) -- ( 3.9045226130653266 , 0.5204981278714192 ) -- ( 3.9396984924623117 , 0.48233673576298053 ) -- ( 3.9748743718592965 , 0.4444119357462781 ) -- ( 4.010050251256281 , 0.40691399420494634 ) -- ( 4.0452261306532655 , 0.37006632956847313 ) -- ( 4.080402010050251 , 0.3341438605566591 ) -- ( 4.115577889447236 , 0.29949220453222647 ) -- ( 4.150753768844221 , 0.26654569156071756 ) -- ( 4.185929648241206 , 0.23584363963884658 ) -- ( 4.221105527638191 , 0.20804010108618032 ) -- ( 4.256281407035176 , 0.18389818788598022 ) -- ( 4.291457286432161 , 0.16425616794464154 ) -- ( 4.326633165829146 , 0.14995352352837582 ) -- ( 4.361809045226131 , 0.1417044755798869 ) -- ( 4.396984924623116 , 0.13991850756273003 ) -- ( 4.4321608040201 , 0.14449598613603412 ) -- ( 4.467336683417085 , 0.15466076780289187 ) -- ( 4.50251256281407 , 0.16891219951466624 ) -- ( 4.5376884422110555 , 0.18516112636932677 ) -- ( 4.572864321608041 , 0.20104606354629462 ) -- ( 4.608040201005025 , 0.21433811620392007 ) -- ( 4.64321608040201 , 0.22329427882147976 ) -- ( 4.678391959798995 , 0.22684844192624665 ) -- ( 4.71356783919598 , 0.22461555070275765 ) -- ( 4.748743718592965 , 0.2167509432075003 ) -- ( 4.783919597989949 , 0.20374515220538114 ) -- ( 4.819095477386934 , 0.1863530659364987 ) -- ( 4.8542713567839195 , 0.16530305225290406 ) -- ( 4.889447236180905 , 0.14111977617537286 ) -- ( 4.92462311557789 , 0.11439797506574939 ) -- ( 4.959798994974874 , 0.08590729518418307 ) -- ( 4.994974874371859 , 0.05654958615836507 ) -- ( 5.030150753768845 , 0.024306976597067403 ) -- ( 5.06532663316583 , 1.7698732624680756e-08 ) -- ( 5.100502512562814 , 1.7864675620870012e-08 ) -- ( 5.135678391959798 , 1.8046844551525538e-08 ) -- ( 5.1708542713567835 , 1.8254954864553644e-08 ) -- ( 5.206030150753769 , 1.85012247264626e-08 ) -- ( 5.241206030150754 , 1.8799353564104807e-08 ) -- ( 5.276381909547739 , 1.9164946954441082e-08 ) -- ( 5.311557788944723 , 1.961623706824998e-08 ) -- ( 5.346733668341709 , 2.017514887215836e-08 ) -- ( 5.381909547738694 , 2.086885640102712e-08 ) -- ( 5.417085427135678 , 2.1732081002023477e-08 ) -- ( 5.452261306532663 , 2.2810567410168527e-08 ) -- ( 5.487437185929648 , 2.416651056519815e-08 ) -- ( 5.522613065326633 , 2.588735286480877e-08 ) -- ( 5.557788944723618 , 2.810067394818515e-08 ) -- ( 5.592964824120603 , 3.1000664219714086e-08 ) -- ( 5.628140703517588 , 3.489794989394313e-08 ) -- ( 5.663316582914573 , 4.0319915139459874e-08 ) -- ( 5.698492462311558 , 4.823018375452299e-08 ) -- ( 5.733668341708543 , 6.056281723647746e-08 ) -- ( 5.768844221105528 , 8.172364385068596e-08 ) -- ( 5.8040201005025125 , 1.2377752107011603e-07 ) -- ( 5.839195979899497 , 2.3131089624510782e-07 ) -- ( 5.874371859296482 , 7.073166366874994e-07 ) -- ( 5.909547738693467 , 0.392423452489086 ) -- ( 5.944723618090452 , 1.222122609356965 ) -- ( 5.9798994974874375 , 1.391077334280957 ) -- ( 6.015075376884422 , 1.5365018380640434 ) -- ( 6.050251256281407 , 1.6570219530009667 ) -- ( 6.085427135678392 , 1.7531837907486199 ) -- ( 6.1206030150753765 , 1.8274125455856252 ) -- ( 6.155778894472362 , 1.8827886697977123 ) -- ( 6.190954773869347 , 1.9222629677831464 ) -- ( 6.226130653266331 , 1.9488810172007263 ) -- ( 6.261306532663316 , 1.9653129132361857 ) -- ( 6.296482412060302 , 1.9738385606710016 ) -- ( 6.331658291457287 , 1.9763221941723332 ) -- ( 6.366834170854272 , 1.9741500396117155 ) -- ( 6.402010050251256 , 1.9684255612489974 ) -- ( 6.437185929648241 , 1.960023946270892 ) -- ( 6.472361809045227 , 1.9496097441202336 ) -- ( 6.507537688442211 , 1.9376828943905544 ) -- ( 6.542713567839196 , 1.9246204571329817 ) -- ( 6.57788944723618 , 1.910718518676463 ) -- ( 6.613065326633166 , 1.8961869269858513 ) -- ( 6.648241206030151 , 1.8811825533584412 ) -- ( 6.683417085427136 , 1.8658221749739643 ) -- ( 6.718592964824121 , 1.8501958695476888 ) -- ( 6.7537688442211055 , 1.8343716541070445 ) -- ( 6.788944723618091 , 1.8184005416705622 ) -- ( 6.824120603015075 , 1.8023229844011512 ) -- ( 6.85929648241206 , 1.786169538918207 ) -- ( 6.894472361809045 , 1.7699501002998772 ) -- ( 6.92964824120603 , 1.7536894126714218 ) -- ( 6.964824120603015 , 1.7373969616680787 );

\end{scope}

\begin{scope}[shift={(2*\xShift, 0)}, yscale = 1/1.1, xscale = 0.65]
    \draw[->] (-0.5 ,0) -- (7.5 ,0) node[below] {$k$};
    \draw[->] (0,-0.5) -- (0,5.5) node[above] {$J^*_{k}$};
    
    \draw (0.125*7,0.1) -- (0.125*7,-0.1) node[below] {$s_1$};
    \draw (0.45 *7,0.1) -- (0.45 *7,-0.1) node[below] {$s_2$};
    \draw (0.65*7,0.1) -- (0.65*7,-0.1) node[below] {$s_3$};
    \draw (0.85*7,0.1) -- (0.85*7,-0.1) node[below] {$s_4$};

    \draw (-0.1,0.10686256804749261) -- (0.1,0.10686256804749261) node[left] {$J^*$};

    
    \draw[blue, very thick] ( 0.0 , 0.10686256804749261 ) -- ( 0.035175879396984924 , 0.12293119358915039 ) -- ( 0.07035175879396985 , 0.1391605053847942 ) -- ( 0.10552763819095477 , 0.15555211029687244 ) -- ( 0.1407035175879397 , 0.17210763125640716 ) -- ( 0.17587939698492464 , 0.1888287074236834 ) -- ( 0.21105527638190955 , 0.2057169943505306 ) -- ( 0.24623115577889448 , 0.22277416414423623 ) -- ( 0.2814070351758794 , 0.2400019056330871 ) -- ( 0.31658291457286436 , 0.25740192453356947 ) -- ( 0.35175879396984927 , 0.27497594361923805 ) -- ( 0.3869346733668342 , 0.292725702891264 ) -- ( 0.4221105527638191 , 0.31065295975068485 ) -- ( 0.45728643216080406 , 0.32875948917237063 ) -- ( 0.49246231155778897 , 0.3470470838807105 ) -- ( 0.5276381909547738 , 0.36551755452704593 ) -- ( 0.5628140703517588 , 0.38417272986884976 ) -- ( 0.5979899497487436 , 0.4030144569506602 ) -- ( 0.6331658291457287 , 0.42204460128677557 ) -- ( 0.6683417085427136 , 0.4412650470456887 ) -- ( 0.7035175879396985 , 0.4606776972362432 ) -- ( 0.7386934673366834 , 0.4802844738954455 ) -- ( 0.7738693467336684 , 0.5000873182778132 ) -- ( 0.8090452261306533 , 0.5200881910460242 ) -- ( 0.8442211055276382 , 0.5402890724623801 ) -- ( 0.879396984924623 , 0.560691962580051 ) -- ( 0.9145728643216081 , 0.5812988814317572 ) -- ( 0.949748743718593 , 0.6021118692100038 ) -- ( 0.9849246231155779 , 0.6231329864222128 ) -- ( 1.0201005025125627 , 0.6443643139646551 ) -- ( 1.0552763819095476 , 0.6658079528718921 ) -- ( 1.0904522613065328 , 0.6874660221306348 ) -- ( 1.1256281407035176 , 0.7093406259782227 ) -- ( 1.1608040201005025 , 0.7173009583088295 ) -- ( 1.1959798994974873 , 0.7197293733713535 ) -- ( 1.2311557788944725 , 0.722127416679446 ) -- ( 1.2663316582914574 , 0.7246556496447392 ) -- ( 1.3015075376884422 , 0.7274414156286144 ) -- ( 1.3366834170854272 , 0.730584423404915 ) -- ( 1.371859296482412 , 0.7341622255640647 ) -- ( 1.407035175879397 , 0.7382351182069495 ) -- ( 1.442211055276382 , 0.7428501588208856 ) -- ( 1.4773869346733668 , 0.7480443379670443 ) -- ( 1.5125628140703518 , 0.7538470436021307 ) -- ( 1.5477386934673367 , 0.7602819580007312 ) -- ( 1.5829145728643217 , 0.7673685048974916 ) -- ( 1.6180904522613067 , 0.7751229435291704 ) -- ( 1.6532663316582914 , 0.7835591904066437 ) -- ( 1.6884422110552764 , 0.7926894350519396 ) -- ( 1.7236180904522613 , 0.8025245976536595 ) -- ( 1.758793969849246 , 0.8130746690725354 ) -- ( 1.7939698492462313 , 0.8243489656975798 ) -- ( 1.8291457286432162 , 0.8363563227813149 ) -- ( 1.864321608040201 , 0.8491052943113182 ) -- ( 1.899497487437186 , 0.8626040484963493 ) -- ( 1.9346733668341707 , 0.8768607346365211 ) -- ( 1.9698492462311559 , 0.8918831892679937 ) -- ( 2.0050251256281406 , 0.9076790371695282 ) -- ( 2.0402010050251254 , 0.924256064093329 ) -- ( 2.0753768844221105 , 0.9416236672143863 ) -- ( 2.1105527638190953 , 0.9597892424117611 ) -- ( 2.1457286432160805 , 0.978760352992126 ) -- ( 2.1809045226130657 , 0.9985439424019065 ) -- ( 2.21608040201005 , 1.0191474349839673 ) -- ( 2.251256281407035 , 1.040577972519583 ) -- ( 2.2864321608040203 , 1.0628429808566984 ) -- ( 2.321608040201005 , 1.0859522691284886 ) -- ( 2.35678391959799 , 1.1099149872308036 ) -- ( 2.3919597989949746 , 1.1347304012882997 ) -- ( 2.4271356783919598 , 1.1604090849737756 ) -- ( 2.462311557788945 , 1.186954708814703 ) -- ( 2.4974874371859297 , 1.2138242555394707 ) -- ( 2.532663316582915 , 1.2409624977276343 ) -- ( 2.567839195979899 , 1.268372122333286 ) -- ( 2.6030150753768844 , 1.2960558431797584 ) -- ( 2.6381909547738696 , 1.3240164012282896 ) -- ( 2.6733668341708543 , 1.3522565648493492 ) -- ( 2.708542713567839 , 1.380779130096672 ) -- ( 2.743718592964824 , 1.4095869209839926 ) -- ( 2.778894472361809 , 1.4386827897645103 ) -- ( 2.814070351758794 , 1.468069617213092 ) -- ( 2.849246231155779 , 1.4977503129112022 ) -- ( 2.884422110552764 , 1.5277278155345528 ) -- ( 2.9195979899497484 , 1.558005093143413 ) -- ( 2.9547738693467336 , 1.588585143475489 ) -- ( 2.9899497487437188 , 1.6194709942411474 ) -- ( 3.0251256281407035 , 1.6506657034205974 ) -- ( 3.0603015075376883 , 1.6821723595621658 ) -- ( 3.0954773869346734 , 1.713994082079891 ) -- ( 3.130653266331658 , 1.7461340215462595 ) -- ( 3.1658291457286434 , 1.7785953599694686 ) -- ( 3.201005025125628 , 1.811381311024241 ) -- ( 3.2361809045226133 , 1.8444951201273887 ) -- ( 3.271356783919598 , 1.8779400638515984 ) -- ( 3.306532663316583 , 1.9117194447913475 ) -- ( 3.341708542713568 , 1.9458364404889315 ) -- ( 3.3768844221105527 , 1.968498771217798 ) -- ( 3.4120603015075375 , 1.9894353626433228 ) -- ( 3.4472361809045227 , 2.010315092568817 ) -- ( 3.4824120603015074 , 2.0313339767951084 ) -- ( 3.517587939698492 , 2.0526519797445726 ) -- ( 3.5527638190954773 , 2.074398752046996 ) -- ( 3.5879396984924625 , 2.0966780004286636 ) -- ( 3.6231155778894473 , 2.119572028433563 ) -- ( 3.6582914572864325 , 2.1431453432855205 ) -- ( 3.6934673366834168 , 2.167448597593118 ) -- ( 3.728643216080402 , 2.1925210828001527 ) -- ( 3.7638190954773867 , 2.2183939036762332 ) -- ( 3.798994974874372 , 2.2450922292671893 ) -- ( 3.834170854271357 , 2.27263645028887 ) -- ( 3.8693467336683414 , 2.301043165152381 ) -- ( 3.9045226130653266 , 2.330325861025347 ) -- ( 3.9396984924623117 , 2.360495346705588 ) -- ( 3.9748743718592965 , 2.391559819670359 ) -- ( 4.010050251256281 , 2.4235246914272963 ) -- ( 4.0452261306532655 , 2.456392276191733 ) -- ( 4.080402010050251 , 2.49016128546626 ) -- ( 4.115577889447236 , 2.5248260229680577 ) -- ( 4.150753768844221 , 2.5603752654154643 ) -- ( 4.185929648241206 , 2.5967908422807175 ) -- ( 4.221105527638191 , 2.634045929739961 ) -- ( 4.256281407035176 , 2.672103139031295 ) -- ( 4.291457286432161 , 2.7109126065390576 ) -- ( 4.326633165829146 , 2.750410484601609 ) -- ( 4.361809045226131 , 2.790518423247642 ) -- ( 4.396984924623116 , 2.8311448717065826 ) -- ( 4.4321608040201 , 2.872189139019939 ) -- ( 4.467336683417085 , 2.9135488718284686 ) -- ( 4.50251256281407 , 2.9551307798540805 ) -- ( 4.5376884422110555 , 2.9968631241634034 ) -- ( 4.572864321608041 , 3.038707146465738 ) -- ( 4.608040201005025 , 3.0806641226134746 ) -- ( 4.64321608040201 , 3.122775715859998 ) -- ( 4.678391959798995 , 3.1651175588668234 ) -- ( 4.71356783919598 , 3.2077882912622355 ) -- ( 4.748743718592965 , 3.250897334760288 ) -- ( 4.783919597989949 , 3.294554403527947 ) -- ( 4.819095477386934 , 3.338862433157815 ) -- ( 4.8542713567839195 , 3.3839119051888362 ) -- ( 4.889447236180905 , 3.42978062888137 ) -- ( 4.92462311557789 , 3.476537023077012 ) -- ( 4.959798994974874 , 3.5242387222898994 ) -- ( 4.994974874371859 , 3.5729293658081707 ) -- ( 5.030150753768845 , 3.6226361599721013 ) -- ( 5.06532663316583 , 3.6734226563097145 ) -- ( 5.100502512562814 , 3.7251568825523123 ) -- ( 5.135678391959798 , 3.7774084510541366 ) -- ( 5.1708542713567835 , 3.8301825352374896 ) -- ( 5.206030150753769 , 3.883484360258702 ) -- ( 5.241206030150754 , 3.9373192035254307 ) -- ( 5.276381909547739 , 3.9916923952191334 ) -- ( 5.311557788944723 , 4.046609318822769 ) -- ( 5.346733668341709 , 4.102075411653756 ) -- ( 5.381909547738694 , 4.158096165402224 ) -- ( 5.417085427135678 , 4.21467712667463 ) -- ( 5.452261306532663 , 4.271823897542737 ) -- ( 5.487437185929648 , 4.329542136098007 ) -- ( 5.522613065326633 , 4.387837557011398 ) -- ( 5.557788944723618 , 4.446715932098548 ) -- ( 5.592964824120603 , 4.506183090890229 ) -- ( 5.628140703517588 , 4.566244921207849 ) -- ( 5.663316582914573 , 4.626907369743457 ) -- ( 5.698492462311558 , 4.68817644264303 ) -- ( 5.733668341708543 , 4.750058206090059 ) -- ( 5.768844221105528 , 4.812558786881446 ) -- ( 5.8040201005025125 , 4.875684372970856 ) -- ( 5.839195979899497 , 4.939441213884074 ) -- ( 5.874371859296482 , 5.003835620493203 ) -- ( 5.909547738693467 , 5.068873958865377 ) -- ( 5.944723618090452 , 5.124110056778752 ) -- ( 5.9798994974874375 , 5.15597897773103 ) -- ( 6.015075376884422 , 5.181343627354584 ) -- ( 6.050251256281407 , 5.200466938740693 ) -- ( 6.085427135678392 , 5.213876916911921 ) -- ( 6.1206030150753765 , 5.222291934120571 ) -- ( 6.155778894472362 , 5.226512090241468 ) -- ( 6.190954773869347 , 5.2273374094443374 ) -- ( 6.226130653266331 , 5.225518767093618 ) -- ( 6.261306532663316 , 5.221714742486468 ) -- ( 6.296482412060302 , 5.216479584120026 ) -- ( 6.331658291457287 , 5.210264006690202 ) -- ( 6.366834170854272 , 5.20342536684063 ) -- ( 6.402010050251256 , 5.1962452801958845 ) -- ( 6.437185929648241 , 5.18894259238627 ) -- ( 6.472361809045227 , 5.1816852549280705 ) -- ( 6.507537688442211 , 5.174601641323668 ) -- ( 6.542713567839196 , 5.167789827123172 ) -- ( 6.57788944723618 , 5.161324827004789 ) -- ( 6.613065326633166 , 5.155263754130732 ) -- ( 6.648241206030151 , 5.149650733538568 ) -- ( 6.683417085427136 , 5.144520185209656 ) -- ( 6.718592964824121 , 5.139899315455985 ) -- ( 6.7537688442211055 , 5.135809861682721 ) -- ( 6.788944723618091 , 5.132269476852205 ) -- ( 6.824120603015075 , 5.129292782575101 ) -- ( 6.85929648241206 , 5.126892072348723 ) -- ( 6.894472361809045 , 5.12507791054008 ) -- ( 6.92964824120603 , 5.123860083801826 ) -- ( 6.964824120603015 , 5.123247160963598 );
\end{scope}
\end{tikzpicture}

%% file: Tex/Congestion.tex
\section{Non-Negligible Service Times}\label{sec:congestion}
In \revcolor{\Cref{sec:Main}}, we analyzed the platform’s optimal operating policy in the limiting regime where service times are instantaneous. In practice, however, service times, and their impact on servers' effective utilization, can be a critical operational factor for some platforms. In this section, we investigate the challenges posed by reduced server availability when service times are positive and their implications for the optimal compensation scheme and assignment policy.

\revcolor{While limited server availability naturally reduces the platform's overall throughput and revenue by forcing it to turn away excess customers, its most critical implication in our model lies at the individual server level. Specifically, finite availability structurally restricts the set of feasible compensation schemes. Regardless of the overall demand rate or the platform’s assignment strategy, there is a hard upper bound on the number of customers any server (or cohort of servers) can process in a given period. This physical constraint fundamentally limits the platform's ability to design and deliver compensation payouts.}

Previously, a server's incentive compatibility could be expressed solely in terms of the platform size and compensations, which were functions of the platform size ${\cal S}_t$. However, in this setting, the functions $f_s,q_s$, and consequently the payoff-to-go $v_s$, also depend on ${\cal X}_t$, the state of the system:
\begin{equation*}\label{eq:congestionBE}
v_s(k,x) = \EX\left[\int_{t'}^\infty e^{-r\,(t-t')}\, f_s({\mathcal S}_{t-},{\mathcal X}_{t-})\,q_s({\mathcal S}_{t-},{\mathcal X}_{t-})\, \D {\mathcal A}_t \big| {\mathcal S}_{t'}=k,{\mathcal X}_{t'}=x\right],\end{equation*}
\revcolor{where ${\mathcal X}_t \in \{0,1\}^{{\mathcal S}_t}$ indicates their idleness. Specifically, ${\mathcal X}_t(s) = 0$ if server $s \leq {\mathcal S}_t$ is idle, and ${\mathcal X}_t(s) = 1$ if it is busy.}

Thus, servers' incentive compatibility constraints must be expressed in terms of the evolution of the state space $({\cal S}_t, {\cal X}_t)$, leading to a formulation that suffers from the curse of dimensionality. To make the analysis more tractable and to focus on the core issue of job assignment, we introduce the following assumptions in this section.

\begin{assumption}\label{assm:stationary} Conditional on ${\cal S}_t=k$, the servers' idleness  ${\cal X}_t$ has a stationary distribution.
\end{assumption}
The main implication of this assumption is that it allows us to formulate the platform's problem over a class of Markovian operating policies that depend only on ${\cal S}_t$. While the primary purpose of imposing \cref{assm:stationary} is to simplify the stochastic elements of the problem, it remains both reasonable and practical in settings where the time scale at which the platform grows its server pool is much slower than the time scale at which demand and service occur.


Let $\psi^{\pi}_k$ denote the stationary distribution of the idleness process ${\cal X}_t$, conditional on ${\cal S}_t = k$, under some policy $\pi$, and 
\begin{align}\label{Gamma definition congestion}
    \imm^{\pi}(k) = {\lambda_k \over r}\, \e_{x \sim \psi_k^{\pi}}[F(k,x)]\, \qquad \mbox{and} \qquad \gamma^{\pi}_s(k) = {\lambda_k \over r}\, \e_{x \sim \psi_k^{\pi}}[f_s(k,x)q_s(k,x)].
\end{align}

Given \cref{assm:stationary}, and with a slight abuse of notation, we define:
\begin{align*}
v_s^{\pi}(k) &\coloneqq \EX\left[\int_{t'}^\infty e^{-r\,(t-t')}\, f_s({\mathcal S}_{t},{\mathcal X}_{t-})\,q_s({\mathcal S}_{t},{\mathcal X}_{t-})\, \D {\mathcal A}_t \big| {\mathcal S}_{t'}=k\right]\\
&= \EX\left[\int_{t'}^\infty e^{-r\,(t-t')}\, \e[f_s({\mathcal S}_{t},{\mathcal X}_{t-})\,q_s({\mathcal S}_{t},{\mathcal X}_{t-})|{\cal S}_t]\, \lambda_{{\cal S}_t}\,\D t  \big| {\mathcal S}_{t'}=k\right]\\
&\revcolor{\stackrel{\mbox{\tiny A}}{=} \EX\left[\int_{t'}^\infty e^{-r\,(t-t')}\, r\, \gamma^{\pi}_s({\cal S}_t)\,\D t  \big| {\mathcal S}_{t'}=k\right]}\\
&= \revcolor{\EX\left[\int_{t'}^{\tau_{k+1}}e^{-r\,(t-t')}\, r\, \gamma^{\pi}_s({\cal S}_t)\,\D t  \big| {\mathcal S}_{t'}=k\right]+\EX\left[\int_{\tau_{k+1}}^\infty e^{-r\,(t-t')}\, r\, \gamma^{\pi}_s({\cal S}_t)\,\D t  \big| {\mathcal S}_{t'}=k\right]}\\
& =\left(1-\frac{\mu_k}{r+\mu_k} \right)\gamma^{\pi}_s(k) + \frac{\mu_k}{r+\mu_k} v_s^{\pi}(k+1),
\end{align*}
where $\stackrel{\mbox{\tiny A}}{=}$ denotes `approximation' under \cref{assm:stationary}. 

When service times are instantaneous, the only constraints we needed to impose on $\imm(k)$ and $\gamma_s(k)$ were their positivity and the requirement that $\imm(k) + \sum_{s=1}^k \gamma_s(k) = \frac{p\lambda_k}{r}$. With positive service times, however, additional requirements must be considered. First, total compensation can no longer equal $\frac{p\lambda_k}{r}$, as servers' limited availability results in lost demand. Second, $\gamma_s(k)$ can no longer be unbounded above, since each server has finite capacity, which limits the number of customers they can serve. As a result, a revised approach and a new set of constraints are necessary to properly define assignment policies under these additional conditions.

Next, in \ref{sec:fixedSize}, we analyze how server availability impacts idleness for a fixed platform size, highlighting its effects on both individual server payoffs and the platform’s overall payoff. Building on these new constraints, in \ref{sec:Congestion program}, we explore the implications of server utilization on the platform’s growth dynamics.

\subsection{Server Availability Constraints under Fixed Platform Size}\label{sec:fixedSize}
As mentioned earlier, positive service times impose additional constraints by reducing server idleness, which in turn narrows the feasible region for server compensations. In this section, we characterize the set of feasible server payoffs that a platform can implement when the server pool is held fixed at size $k$, and subsequently incorporate these constraints into the revised optimization problem presented in the next section.

Before presenting the main result of this section, we introduce a subset of policies that play a crucial role in characterizing the feasible region.

\begin{Defi}[priority-based policy]\label{def: det priority}
    An assignment policy is referred to as \textit{\rm Priority-Based} if jobs are assigned to the highest-ranked available server according to a predetermined ranking of servers.
\end{Defi}

Let us define $\stationary{k}{n}$ as the stationary probability that, under the priority-based policy, the top $n$ ranked servers are busy. Using the Erlang B formula, we obtain:
\begin{align*}
    \stationary{k}{n} = \frac{\lambda_k^n}{\eta^n n!}\stationary{k}{0} =  \frac{\frac{\lambda_k^n}{\eta^n n!}}{\sum_{j=0}^{\revcolor{n}} \frac{\lambda_k^j}{\eta^j j!}}.
\end{align*}

Note that the platform loses an arriving customer when all servers are busy, which is represented by $\stationary{k}{k}$. Therefore, we have the following {\em job preserving} constraint on the values of $\imm^{\pi}(k)$ and $\gamma^{\pi}_s(k)$:
\begin{align}\label{jobPreservingCondition}\tag{Job Preserving}
    \imm^{\pi}(k) + \sum_{s=1}^k \gamma^{\pi}_s(k) = (1-\stationary{k}{k})\frac{p\lambda_k}{r}.
\end{align}

Additionally, at the individual level, the platform cannot allocate all jobs to a single server, as that server may be busy when a new customer arrives. This limitation implies that it is infeasible to allocate the entire compensation for a given state to one server. Intuitively, a server $s$ achieves its maximum payoff when the platform sets $f_s(k,x) = p$ for its compensation, and all jobs are assigned to server $s$ whenever it is available. Consequently, by applying the PASTA property, a single server can accept at most a $1-\stationary{k}{1}$ fraction of the jobs. Then, for all $s \leq k$ we have $\e_{x \sim \psi_k^{\pi}}[q_s(k,x)] \leq 1-\stationary{k}{1}$ and using $f_s(k,x) \leq p$:
\begin{align}\label{singleServer}
    \gamma_s(k) \leq \left(1-\stationary{k}{1}\right)\frac{p\lambda_k}{r}
\end{align}

In general, for an arbitrary subset $\mathcal{I} \subseteq [k]$ of servers, the servers in $\mathcal{I}$ can accept at most $1-\stationary{k}{|\mathcal{I}|}$ of the jobs. This occurs if these servers are given priority over the others, and a job is assigned to them whenever any of them is idle. In other words, a necessary condition for the feasibility of an assignment policy is to have: 
\begin{align}\label{eq:polytope} 
    \sum_{s\in \mathcal{I}} \gamma_s(k) \leq 
    \left(1-\stationary{k}{|\mathcal{I}|}\right)\frac{p\lambda_k}{r}, \qquad \forall \;\mathcal{I} \subseteq [k].
\end{align}
Observe that when $|\mathcal{I}|=1$, we recover the same condition as in \eqref{singleServer} and when $\mathcal{I}=[k]$ , the resulting inequality is equivalent to \eqref{jobPreservingCondition}.

The inequalities in \eqref{eq:polytope} define a polytope for $\gamma_s(k)$. In \cref{prop:vertexprioritybased}, we show that this polytope has exactly $k!$ vertices and each of these vertices corresponds to a unique priority-based policy, determined by a distinct priority permutation.

\begin{prop}\label{prop:vertexprioritybased}
    The inequality in \eqref{eq:polytope} defines a polytope with exactly $k!$ vertices. Moreover, each of these $k!$ vertices can be implemented using a Priority-Based policy with $f_s(k) = p$ for all $s$.
\end{prop}

Finally, in \Cref{thm:implementability}, we show that every feasible point $\gamma(k)$ in \eqref{eq:polytope} corresponds to a feasible policy in which the server $s$ receives a payoff of $\gamma_s(k)$. The existence of such a policy confirms that the constraints defining the polytope are both necessary and sufficient for the feasibility of a policy implemented by the platform. 

A straightforward way to implement a point inside the polytope is to randomize over the vertices by choosing a randomization across priority-based policies that ensures a payoff of $\gamma_s(k)$ for server $s$. However, this approach is only feasible ex-ante, before the specific priority-based policy is realized. Identifying a policy that remains feasible ex-post is more challenging. In the following theorem, we apply Kakutani’s fixed point theorem to prove the existence of such a policy.\footnote{A fixed-point argument is needed since by changing the assignment policy,  $q_s^{\pi}(k,x)$, the stationary distribution, $\psi_k^{\pi}$, changes as well. Thus, a fixed-point argument is necessary to establish the existence of an assignment policy that induces a stationary distribution, ensuring that the expected payoff is $\gamma_s(k)$.}

\begin{thm}\label{thm:implementability}
    \revcolor{Let $\mathcal{P}(k) \subset \mathbb{R}_+^k$ denote the polytope of feasible compensation vectors defined by the linear inequalities in \eqref{eq:polytope}. Then, for every vector $\boldsymbol{\gamma}(k) = (\gamma_1(k), \dots, \gamma_k(k)) \in \mathcal{P}(k)$, there exists an admissible joint assignment and compensation policy $(q_s(k,x), f_s(k,x))_{s=1}^k$ such that$$\gamma_s(k) = \frac{\lambda_k}{r} \mathbb{E}_{x \sim \psi_k^{\pi}}[f_s(k,x)q_s(k,x)] \qquad \forall s \in [k],$$
    as in \eqref{Gamma definition congestion}.}
\end{thm} 

\subsection{Optimal Growth Path}\label{sec:Congestion program}

Equipped with \cref{thm:implementability}, we have that \eqref{eq:polytope} characterizes of all feasible servers' compensations, and the platform's optimization problem can be written as:
\begin{align*}\label{congestionProgram}
\tag{$\mathcal{P}_{\texttt{Congestion}}$}
    \;&\; J^*=\sup_{\pi \in \Pi} \;J^{\pi}(0)\\
    \text{subject to}\;&\; J^\pi_{k} = \frac{r}{r+\mu_k }\imm(k) + \frac{\mu_k}{r+\mu_k }\,J^\pi_{k+1} && \mbox{for} \quad k \leq \size \quad\mbox{and}\quad   J^\pi_{\size+1}=0 \\
    \;&\; v^{\pi}_s(k) = \frac{r}{r+\mu_k} \gamma_s(k) + \frac{\mu_k}{r+\mu_k}\,v^{\pi}_s(k+1)   && \mbox{for}\quad s \leq k \leq \size \quad\mbox{and}\quad   v^{\pi}_s(\size+1)=0  \\
    \;&\; \sum_{s\in \mathcal{I}} \gamma_s(k) \leq 
    \left(1-\stationary{k}{|\mathcal{I}|}\right)\frac{p\lambda_k}{r} && \mbox{for} \quad  k \leq \size \quad\mbox{and} \quad \mathcal{I}\subseteq [k]\\
    \;&\; \imm(k) + \sum_{s\in [k]} \gamma_s(k) =
    \left(1-\stationary{k}{k}\right)\frac{p\lambda_k}{r} && \mbox{for} \quad  k \leq \size \\
    \;&\; v^{\pi}_s(k) \geq w && \mbox{for} \quad  s \leq k \leq \size\\
    \:&\; 0 \leq \mu_k \leq m_k && \mbox{for} \quad  0 \leq k \leq \size.
\end{align*}
\begin{rem}
     As $\eta \rightarrow \infty$, service time converge to zero, and \eqref{congestionProgram} reduces to  \eqref{mainProgram}.
\end{rem}
The program \eqref{congestionProgram} includes an exponential number ($2^k - 1$) of constraints for each $k \leq \size$. Therefore, solving for an optimal $\gamma_s(k)$ could be in general computationally demanding. In the following proposition, we tackle this challenge by demonstrating that any feasible solution, $\gamma_s(k)$, can be transformed into another one, $\hat{\gamma}_s(k)$, with the same objective value such that $\hat{\gamma}_1(k) \geq \hat{\gamma}_2(k) \geq ... \geq \hat{\gamma}_k(k)$. Such transformation leads to a reformulation of \eqref{congestionProgram} with a polynomial number of constraints in the order $O(\size^2)$. 

\begin{prop}\label{prop: decreasingGamma}
    For every optimal solution $(\alpha,\imm,\gamma,\size)$ to \eqref{congestionProgram}, there is another feasible solution $(\alpha,\imm,\hat{\gamma},\size)$ such that $\hat{\gamma}_s(k)$ is decreasing in $s$.
\end{prop}

Using the above result, all congestion-related constraints in \eqref{eq:polytope} simplify to the following:
\begin{align*}
    \sum_{s=1}^n \gamma_s(k) \leq 
    \left(1-\stationary{k}{n}\right)\frac{p\lambda_k}{r} && \forall \;k, n\in[k].
\end{align*}
This simplification significantly reduces the complexity of \eqref{congestionProgram}, making it an optimization program with a polynomial number of constraints and variables. Surprisingly, by redefining the decision variables, $\tilde{\gamma}_s(k) \eqqcolon (\alpha_{k}-\alpha_{k+1})\gamma_s(k) $, we can further simplify it into the following linear program.
\begin{align*}
\;&\; J^*= \sup_{\alpha_k,\tilde{\gamma}_s(k),\size} \;J_0 = \sum_{\ell = 1}^{\size} \left((\alpha_{\ell}-\alpha_{\ell+1})\left(1-\stationary{\ell}{\ell}\right)\frac{p\lambda_{\ell}}{r} - \sum_{s=1}^{\ell} \tilde{\gamma}_s(\ell) \right) \\
    \text{subject to} \;&\; v^{\pi}_s(k) = \sum_{\ell = k}^{\size} \tilde{\gamma}_s(\ell) \geq \alpha_k w   && \mbox{for}\quad s \leq k \leq \size \\
    \;&\; \sum_{s=1}^n \tilde{\gamma}_s(k) \leq 
    (\alpha_{k}-\alpha_{k+1})\left(1-\stationary{k}{n}\right)\frac{p\lambda_k}{r} && \mbox{for} \quad n \leq k \leq \size\\
    \;&\; \tilde{\gamma}_s(k) \geq \tilde{\gamma}_{s+1}(k) && \mbox{for} \quad s < k \leq \size\\
    \:&\; 0 \leq \alpha_{k+1} \leq \frac{m_k}{m_k+r}\alpha_k && \mbox{for} \quad  0 \leq k \leq \size.
\end{align*}
\cref{fig:LPSolution} depicts heatmaps for the optimal value of the platform's size $\size^*$ (left panel) and its expected discounted payoff $J^*$ (right panel) as a function of the new servers' arrival rate, $m$,  and service rate, $\eta$, that characterize the two underlying time scales of the system (see discussion at the end of \cref{sec:model}).

\begin{figure}[h]
    \centering
    \begin{minipage}{0.9\textwidth}
        \centering
        \includegraphics[width=\linewidth]{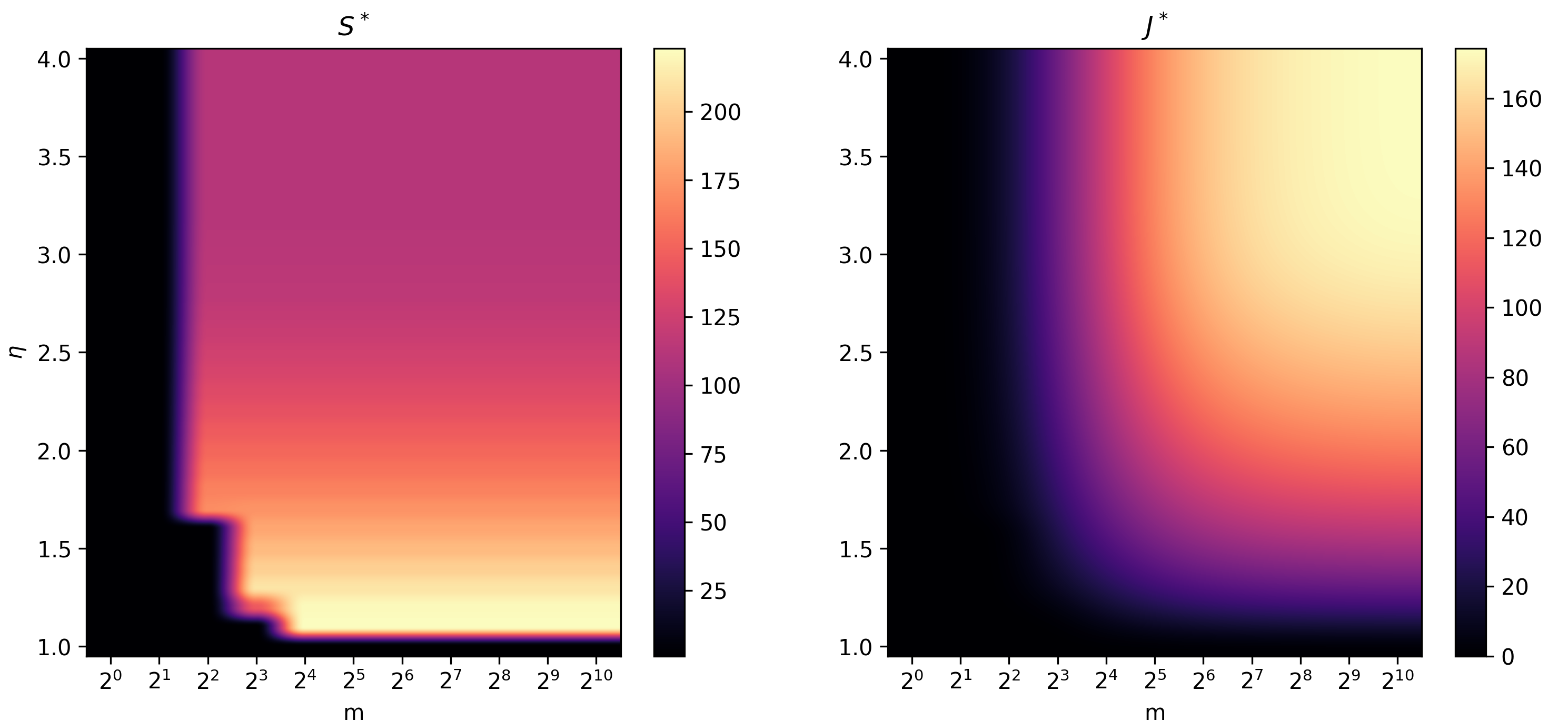}
    \end{minipage}    
    \caption{\footnotesize \sf The optimal size and revenue by changing time scales.\\
    \sc{Data}: $\Lambda_n=10(2, 5, 10, 12)$, $s_n=25(1, 2, 3, 4)$, $\xi_n=0.3\,(1, 1, 1, 1)$, $r=0.1,p=1,w=10$.}
    \label{fig:LPSolution}
\end{figure}

An interesting observation in \cref{fig:LPSolution} is that while $J^*$ increases with both $m$ and $\eta$ as expected, $S^*$ exhibits a non-monotonic behavior in $\eta$. This occurs because, when $\eta$ is very low, the platform cannot sustain growth and thus does not initiate the business. For sufficiently large $\eta$ , the business becomes profitable, and the platform has the economic incentive to expand. However, in an intermediate range where $\eta$ is high enough to justify expansion but still low enough for service capacity to remain a concern, the platform may overshoot the optimal $\mathcal{S}^*$ to meet the demand. When $\eta$ is large enough and capacity is no longer a bottleneck, the platform’s strategy aligns with \cref{cor:maxS}.

%% file: Tex/Extras.tex
\section{Seniority Based Algorithm}\label{sec:taskALG}
At state $\size$, the platform begins by distributing payoff (increasing $\gamma_s(\size)$) uniformly until all DIR constraints are satisfied at state $\size$. Once these constraints are met, the incentives of the last joined server are fully addressed, so there is no need to assign further payoff to this server. The platform then continues distributing payoff uniformly until either the total payoff is distributed, i.e., $\left(\sum_{s=1}^{\size} \gamma_s(\size) = \frac{p\lambda_{\size}}{r} - \imm(\size)\right)$ or the SIR of server $\size - 1$ is satisfied. In the first case, the platform moves on to allocating payoffs in state $\size - 1$. In the second case, it continues distributing payoff uniformly to the first $\size - 2$ servers, proceeding in a similar manner. 

\begin{algorithm}
\caption{Task Assignment Algorithm}\label{alg:taskAssignment}
\SetAlgoLined
\KwIn{$\size$, $p$, $\imm(k) \;\; \forall k$}
Initialize $\gamma_s(k) \gets 0 \;\; \forall k, s \leq \size$\;
Set $\ell \gets \size$\;
\For{$k \gets \size$ \KwTo $1$}{
    \While{$\sum_s \gamma_s(k) < \frac{p\lambda_k}{r} - \imm(k)$}{
        \For{$s \gets 1$ \KwTo $\ell$}{ 
            $\gamma_s(k) \gets \gamma_s(k) + \epsilon$, \qquad
            $v_s(k) \gets \left(1-\frac{\alpha_{k+1}}{\alpha_{k}}\right) \gamma_s(k)  + \frac{\alpha_{k+1}}{\alpha_{k}} v_s(k+1)$\;
        }
        \lIf{$v_{\ell}(k) = \frac{\alpha_{\ell}}{\alpha_k}w$}{$\ell \gets \ell-1$}
    }
}
\end{algorithm}

\section{Alternative Sensitivity Analysis: Logistic Demand Function}\label{apx:numericalMSL}

\revcolor{To complement the insights from \Cref{subsec:Sensitivity}, we evaluate an alternative market structure where consumer adoption behaves like a smooth threshold function derived from a logistic function.}

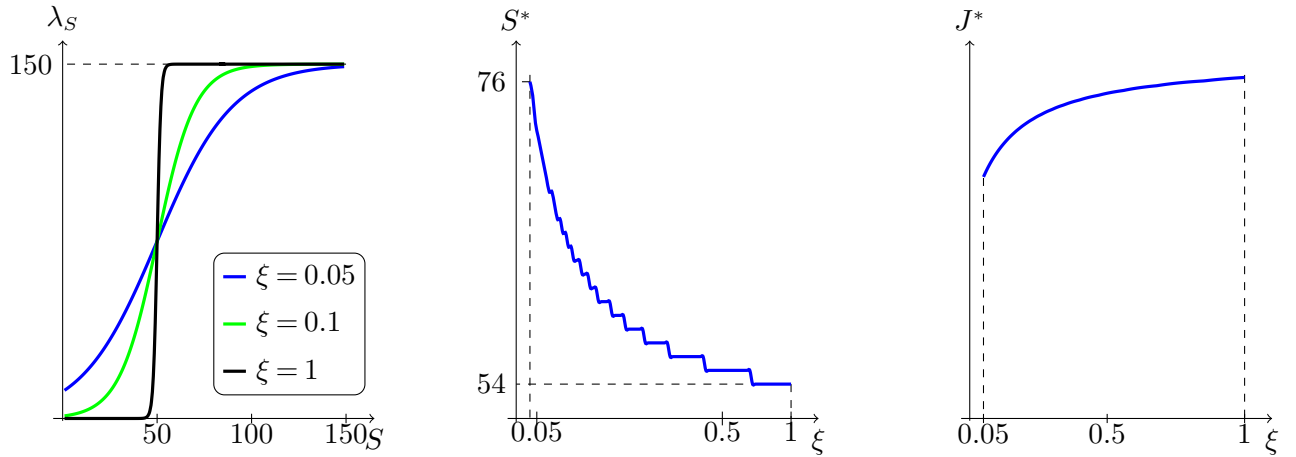
\begin{figure}[!h]
\centering
\input{Figures/numericsfig1}
\caption{\footnotesize \sf \revcolor{Optimal platform size (middle) and expected payoff (right) as functions of the sensitivity parameter $\xi$. For reference, the logistic demand function $\lambda(s) = 150 \frac{\exp(\xi(s-50))}{1+\exp(\xi(s-50))}$ is illustrated for $\xi \in \{0.05, 0.1, 1\}$ (left). The horizontal axes for the middle and right panels range from $0.05$ to $1$.}} 
\label{Fig:sensitivity1}
\end{figure}

\revcolor{As illustrated in \Cref{Fig:sensitivity1}, a larger value of $\xi$ characterizes a sharper, more localized threshold in market adoption around a capacity of $50$. This abrupt step-like increase in demand grants the platform the strategic flexibility to truncate its growth immediately after capturing the bulk of the consumer mass. We observe this structural response directly in the middle panel of \Cref{Fig:sensitivity1}: under highly sensitive demand ($\xi = 1$), the platform halts expansion early at an optimal capacity of $S^* = 54$, precisely because it has already unlocked the vast majority of the $150$ potential customers. Conversely, in more gradual adoption regimes (e.g., $\xi = 0.05$), the platform is forced to expand to a significantly larger scale ($S^* = 76$) to trigger comparable market demand. This broader expansion, however, yields a lower overall payoff due to the substantial long-term compensation liabilities required to sustain the larger pool of active service providers.}

%% file: Figures/numericsfig1.tex
\begin{tikzpicture}
\begin{scope}[shift={(0, 0)},yscale = 1/32,xscale = 1/40]
        \draw[->] (-5,0) -- (165,0) node[below] {$S$};
        \draw[->] (0,-1) -- (0,160) node[above] {$\lambda_S$};
    
        \draw[dashed] (150, 150) -- (0,150) node[left] {$150$};
    
        \draw (50,-2) -- (50,2) node[below] {$50$};
        \draw (100,-2) -- (100,2) node[below] {$100$};
        \draw (150,-2) -- (150,2) node[below] {$150$};
        
        \draw[blue, very thick, smooth] plot coordinates {(1,11.915782377596754)
(2,12.47590447408835)
(3,13.059865866040688)
(4,13.668444152228416)
(5,14.302419734866426)
(6,14.96257336795277)
(7,15.649683467850197)
(8,16.36452317934194)
(9,17.10785719196862)
(10,17.880438303317632)
(11,18.683003728112457)
(12,19.516271154449676)
(13,20.38093455136414)
(14,21.27765973507317)
(15,22.207079704753426)
(16,23.1697897625302)
(17,24.16634243648778)
(18,25.197242229911325)
(19,26.26294022460597)
(20,27.363828570953448)
(21,28.500234902296942)
(22,29.672416716212734)
(23,30.8805557701421)
(24,32.12475254361621)
(25,33.40502082379633)
(26,34.72128247514735)
(27,36.07336245763333)
(28,37.46098416073236)
(29,38.8837651226769)
(30,40.341213205499265)
(31,41.832723296570535)
(32,43.3575746062494)
(33,44.91492862890406)
(34,46.503827830858135)
(35,48.12319512369105)
(36,49.77183417477509)
(37,51.448430598975186)
(38,53.15155406613067)
(39,54.8796613483799)
(40,56.63110031972182)
(41,58.40411490761671)
(42,60.1968509831322)
(43,62.007363162400495)
(44,63.83362247825115)
(45,65.67352486713028)
(46,67.52490040312831)
(47,69.38552319843757)
(48,71.25312187815899)
(49,73.12539052736845)
(50,75.0)
(51,76.87460947263155)
(52,78.74687812184101)
(53,80.61447680156242)
(54,82.47509959687169)
(55,84.32647513286972)
(56,86.16637752174886)
(57,87.99263683759952)
(58,89.8031490168678)
(59,91.59588509238331)
(60,93.36889968027819)
(61,95.12033865162012)
(62,96.84844593386933)
(63,98.55156940102482)
(64,100.22816582522492)
(65,101.87680487630895)
(66,103.49617216914187)
(67,105.08507137109595)
(68,106.6424253937506)
(69,108.16727670342947)
(70,109.65878679450073)
(71,111.1162348773231)
(72,112.53901583926763)
(73,113.92663754236665)
(74,115.27871752485265)
(75,116.59497917620368)
(76,117.87524745638379)
(77,119.11944422985789)
(78,120.32758328378728)
(79,121.49976509770306)
(80,122.63617142904654)
(81,123.73705977539403)
(82,124.80275777008868)
(83,125.83365756351222)
(84,126.8302102374698)
(85,127.79292029524659)
(86,128.72234026492683)
(87,129.61906544863587)
(88,130.48372884555033)
(89,131.31699627188752)
(90,132.11956169668235)
(91,132.8921428080314)
(92,133.63547682065806)
(93,134.3503165321498)
(94,135.03742663204724)
(95,135.69758026513358)
(96,136.3315558477716)
(97,136.9401341339593)
(98,137.52409552591166)
(99,138.08421762240323)
(100,138.62127299681347)
(101,139.13602719577233)
(102,139.62923694849798)
(103,140.1016485763172)
(104,140.55399659145053)
(105,140.98700247389888)
(106,141.4013736151697)
(107,141.7978024176089)
(108,142.17696553823734)
(109,142.5395232662165)
(110,142.88611902336498)
(111,143.21737897750688)
(112,143.5339117588371)
(113,143.83630826993024)
(114,144.12514158048535)
(115,144.40096689838057)
(116,144.66432160910458)
(117,144.91572537612328)
(118,145.15568029523243)
(119,145.3846710964278)
(120,145.60316538729654)
(121,145.81161393238904)
(122,146.01045096347013)
(123,146.20009451596923)
(124,146.38094678734961)
(125,146.55339451349616)
(126,146.71780935958043)
(127,146.87454832219794)
(128,147.0239541398884)
(129,147.16635570944413)
(130,147.30206850568626)
(131,147.43139500264084)
(132,147.55462509428386)
(133,147.67203651323865)
(134,147.78389524600905)
(135,147.8904559435132)
(136,147.9919623258503)
(137,148.08864758038325)
(138,148.18073475235886)
(139,148.26843712741166)
(140,148.35195860541103)
(141,148.4314940652123)
(142,148.50722971996436)
(143,148.5793434627081)
(144,148.64800520207282)
(145,148.7133771879432)
(146,148.775614327026)
(147,148.83486448829777)
(148,148.89126879835771)
(149,148.94496192675072)
};

\draw[green, very thick, smooth] plot coordinates {
(1,1.1087312016422957)
(2,1.2243856729739835)
(3,1.3519947979271734)
(4,1.4927702800356473)
(5,1.648041394588977)
(6,1.819265247641135)
(7,2.008037674149717)
(8,2.216104753990958)
(9,2.4453749057161396)
(10,2.697931494313733)
(11,2.976045860111625)
(12,3.2821906404195706)
(13,3.61905321265038)
(14,3.9895490365298776)
(15,4.396834612703447)
(16,4.8443197047675755)
(17,5.335678390895424)
(18,5.874858419514654)
(19,6.466088241162917)
(20,7.113880976635018)
(21,7.823034461762658)
(22,8.598626384830311)
(23,9.446003408549473)
(24,10.370763051502022)
(25,11.378727003186533)
(26,12.47590447408835)
(27,13.668444152228416)
(28,14.96257336795277)
(29,16.36452317934194)
(30,17.880438303317632)
(31,19.516271154449676)
(32,21.27765973507317)
(33,23.1697897625302)
(34,25.197242229911325)
(35,27.363828570953448)
(36,29.672416716212734)
(37,32.12475254361621)
(38,34.72128247514735)
(39,37.46098416073236)
(40,40.341213205499265)
(41,43.3575746062494)
(42,46.503827830858135)
(43,49.77183417477509)
(44,53.15155406613067)
(45,56.63110031972182)
(46,60.1968509831322)
(47,63.83362247825115)
(48,67.52490040312831)
(49,71.25312187815899)
(50,75.0)
(51,78.74687812184101)
(52,82.47509959687169)
(53,86.16637752174886)
(54,89.8031490168678)
(55,93.36889968027819)
(56,96.84844593386933)
(57,100.22816582522492)
(58,103.49617216914187)
(59,106.6424253937506)
(60,109.65878679450073)
(61,112.53901583926763)
(62,115.27871752485265)
(63,117.87524745638379)
(64,120.32758328378728)
(65,122.63617142904654)
(66,124.80275777008868)
(67,126.8302102374698)
(68,128.72234026492683)
(69,130.48372884555033)
(70,132.11956169668235)
(71,133.63547682065806)
(72,135.03742663204724)
(73,136.3315558477716)
(74,137.52409552591166)
(75,138.62127299681347)
(76,139.62923694849798)
(77,140.55399659145053)
(78,141.4013736151697)
(79,142.17696553823734)
(80,142.88611902336498)
(81,143.5339117588371)
(82,144.12514158048535)
(83,144.66432160910458)
(84,145.15568029523243)
(85,145.60316538729654)
(86,146.01045096347013)
(87,146.38094678734961)
(88,146.71780935958043)
(89,147.0239541398884)
(90,147.30206850568626)
(91,147.55462509428386)
(92,147.78389524600905)
(93,147.9919623258503)
(94,148.18073475235886)
(95,148.35195860541103)
(96,148.50722971996436)
(97,148.64800520207282)
(98,148.775614327026)
(99,148.89126879835771)
(100,148.99607236135728)
(101,149.0910297762624)
(102,149.17705516508244)
(103,149.25497975249147)
(104,149.32555902585884)
(105,149.38947934261557)
(106,149.4473640150846)
(107,149.499778903888)
(108,149.54723755129373)
(109,149.59020588554833)
(110,149.62910652650478)
(111,149.66432272180856)
(112,149.69620194164253)
(113,149.72505915862607)
(114,149.75117983797384)
(115,149.77482266148945)
(116,149.79622200743566)
(117,149.81559020680837)
(118,149.83311959507094)
(119,149.84898437698715)
(120,149.8633423208399)
(121,149.87633629704024)
(122,149.8880956749245)
(123,149.89873759040506)
(124,149.90836809608484)
(125,149.91708320446145)
(126,149.92496983393806)
(127,149.93210666651393)
(128,149.93856492525208)
(129,149.94440907890603)
(130,149.94969748043002)
(131,149.95448294549567)
(132,149.9588132765848)
(133,149.96273173772408)
(134,149.966277484465)
(135,149.96948595329172)
(136,149.9723892142555)
(137,149.97501629028338)
(138,149.97739344629102)
(139,149.97954445093808)
(140,149.98149081360208)
(141,149.98325199890556)
(142,149.98484562091383)
(143,149.98628761892218)
(144,149.98759241657226)
(145,149.9887730658734)
(146,149.98984137755704)
(147,149.99080803905753)
(148,149.99168272129157)
(149,149.99247417529745)
};

\draw[black, very thick, smooth] plot coordinates {
(1,7.864328495045196e-20)
(2,2.1377461241114027e-19)
(3,5.81099644303078e-19)
(4,1.5795926036330718e-18)
(5,4.29377787082409e-18)
(6,1.1671698361700696e-17)
(7,3.1726965563866205e-17)
(8,8.62428339644034e-17)
(9,2.3443232840024833e-16)
(10,6.372531382937383e-16)
(11,1.732233625952368e-15)
(12,4.708699188072045e-15)
(13,1.27995714386161e-14)
(14,3.479284245365354e-14)
(15,9.457675140220478e-14)
(16,2.570862647313015e-13)
(17,6.988329217655063e-13)
(18,1.899624832364102e-12)
(19,5.1637156627047874e-12)
(20,1.403643445325895e-11)
(21,3.815498471064413e-11)
(22,1.0371600160403133e-10)
(23,2.8192932248033255e-10)
(24,7.663633542055835e-10)
(25,2.0831915797156717e-09)
(26,5.662701816204871e-09)
(27,1.5392819445973242e-08)
(28,4.184202138136217e-08)
(29,1.1373840633243576e-07)
(30,3.091730427285305e-07)
(31,8.40419460921891e-07)
(32,2.2844969269140523e-06)
(33,6.209906321091496e-06)
(34,1.688027430826425e-05)
(35,4.588533403884371e-05)
(36,0.0001247292041496198)
(37,0.0003390486446855361)
(38,0.0009216261903322077)
(39,0.002505213277214277)
(40,0.006809680305365159)
(41,0.01850918639793476)
(42,0.050302519569971726)
(43,0.1366576791600968)
(44,0.37089347349521623)
(45,1.0039276386427285)
(46,2.697931494313733)
(47,7.113880976635018)
(48,17.880438303317632)
(49,40.341213205499265)
(50,75.0)
(51,109.65878679450073)
(52,132.11956169668235)
(53,142.88611902336498)
(54,147.30206850568626)
(55,148.99607236135728)
(56,149.62910652650478)
(57,149.8633423208399)
(58,149.94969748043002)
(59,149.98149081360208)
(60,149.99319031969463)
(61,149.9974947867228)
(62,149.99907837380968)
(63,149.9996609513553)
(64,149.99987527079585)
(65,149.99995411466597)
(66,149.9999831197257)
(67,149.99999379009367)
(68,149.99999771550307)
(69,149.99999915958054)
(70,149.99999969082697)
(71,149.9999998862616)
(72,149.99999995815799)
(73,149.99999998460717)
(74,149.9999999943373)
(75,149.9999999979168)
(76,149.99999999923364)
(77,149.99999999971809)
(78,149.9999999998963)
(79,149.99999999996186)
(80,149.99999999998596)
(81,149.99999999999483)
(82,149.9999999999981)
(83,149.9999999999993)
(84,149.99999999999974)
(85,149.9999999999999)
(86,149.99999999999997)
(87,150.0)
(149,150.0)
};

\draw[rounded corners] (80,10) rectangle (160,70);

\draw[green, very thick] (85,40) -- (97,40);
\node[right] at (97,40) {$\xi=0.1$};

\draw[blue, very thick] (85,60) -- (97,60);
\node[right] at (97,60) {$\xi=0.05$};

\draw[black, very thick] (85,20) -- (97,20);
\node[right] at (97,20) {$\xi=1$};

\end{scope}

\begin{scope}[shift={(6, 0)}, yscale = 1/4/5.5, xscale = 4/110]
        \draw[->] (-5/1.6,0) -- (110,0) node[below] {$\xi$};
        \draw[->] (0,-1) -- (0,110) node[above] {$S^*$};
        \draw[dashed] (5, 100) -- (5,0);
        \draw[dashed] (100, 10) -- (100,0);
        
        \draw[dashed] (5, 98) -- (0,98) node[left] {$76$};
        \draw[dashed] (100, 10) -- (0,10) node[left] {$54$};

        \draw (7.5,-2) -- (7.5,2) node[below] {$0.05$};
        
        \draw (75,-2) -- (75,2) node[below] {$0.5$};
    
        \draw (100,-2) -- (100,2) node[below] {$1$};

        \draw[blue, very thick, smooth] plot coordinates {
        (5,98)
        (6,94)
        (7,86)
        (8,82)
        (9,78)
        (10,74)
        (11,70)
        (12,66)
        (13,66)
        (14,62)
        (15,58)
        (16,58)
        (17,54)
        (18,54)
        (19,50)
        (20,50)
        (21,46)
        (22,46)
        (23,46)
        (24,42)
        (25,42)
        (26,42)
        (27,38)
        (28,38)
        (29,38)
        (30,34)
        (31,34)
        (32,34)
        (33,34)
        (34,34)
        (35,30)
        (36,30)
        (37,30)
        (38,30)
        (39,30)
        (40,26)
        (41,26)
        (42,26)
        (43,26)
        (44,26)
        (45,26)
        (46,26)
        (47,22)
        (48,22)
        (49,22)
        (50,22)
        (51,22)
        (52,22)
        (53,22)
        (54,22)
        (55,22)
        (56,18)
        (57,18)
        (58,18)
        (59,18)
        (60,18)
        (61,18)
        (62,18)
        (63,18)
        (64,18)
        (65,18)
        (66,18)
        (67,18)
        (68,18)
        (69,14)
        (70,14)
        (71,14)
        (72,14)
        (73,14)
        (74,14)
        (75,14)
        (76,14)
        (77,14)
        (78,14)
        (79,14)
        (80,14)
        (81,14)
        (82,14)
        (83,14)
        (84,14)
        (85,14)
        (86,10)
        (87,10)
        (88,10)
        (89,10)
        (90,10)
        (91,10)
        (92,10)
        (93,10)
        (94,10)
        (95,10)
        (96,10)
        (97,10)
        (98,10)
        (99,10)
        (100,10)
        };
\end{scope}
\begin{scope}[shift={(12, 0)},yscale = 1/11,xscale = 4/110]
        \draw[->] (-5/1.6,0) -- (110,0) node[below] {$\xi$};
        \draw[->] (0,-5/15) -- (0,55) node[above] {$J^*$};
        
        \draw[dashed] (5, 35) -- (5,0);
        \draw[dashed] (100, 50) -- (100,0);
    
        \draw (5,-2*5/15) -- (5,2*5/15) node[below] {$0.05$};
        
        \draw (50,-2*5/15) -- (50,2*5/15) node[below] {$0.5$};
    
        \draw (100,-2*5/15) -- (100,2*5/15) node[below] {$1$};

        \draw[blue, very thick, smooth] plot coordinates {
            (5,35.158933468055494)
            (6,36.02183653043841)
            (7,36.805697003452146)
            (8,37.518496510662615)
            (9,38.168899024592854)
            (10,38.765809244151995)
            (11,39.31510352434395)
            (12,39.81997260559689)
            (13,40.28634246718259)
            (14,40.73000994103867)
            (15,41.13463934903292)
            (16,41.51095014279427)
            (17,41.87061687313222)
            (18,42.19523804137556)
            (19,42.514397971808286)
            (20,42.80306884753488)
            (21,43.08029187158033)
            (22,43.345098881790875)
            (23,43.577301732579606)
            (24,43.82395270639853)
            (25,44.04574170792356)
            (26,44.24180460641901)
            (27,44.45296915325716)
            (28,44.64800734627719)
            (29,44.82197345293809)
            (30,44.984371510157125)
            (31,45.1644627127136)
            (32,45.32662893879036)
            (33,45.47255199426637)
            (34,45.60377419882654)
            (35,45.7502746336251)
            (36,45.89442361256437)
            (37,46.025363788082444)
            (38,46.14424710130815)
            (39,46.25213451080714)
            (40,46.357389382968215)
            (41,46.48290627085463)
            (42,46.598002284759495)
            (43,46.70350416953826)
            (44,46.800178796257484)
            (45,46.888736552094315)
            (46,46.96983470935422)
            (47,47.05175587750905)
            (48,47.15353135151769)
            (49,47.247727560194164)
            (50,47.33488703224962)
            (51,47.415516389329575)
            (52,47.49008829052426)
            (53,47.559043340258924)
            (54,47.622791947115864)
            (55,47.68171612408535)
            (56,47.74086757078164)
            (57,47.821996615254704)
            (58,47.89779431254008)
            (59,47.968599135265784)
            (60,48.03472947607465)
            (61,48.096484656054045)
            (62,48.154145907796995)
            (63,48.20797732967293)
            (64,48.25822680871527)
            (65,48.30512691023516)
            (66,48.34889573287078)
            (67,48.389737728287855)
            (68,48.42784448517079)
            (69,48.48883803865695)
            (70,48.54851970661097)
            (71,48.604806832260444)
            (72,48.6578870699373)
            (73,48.70793825679985)
            (74,48.75512886422875)
            (75,48.799618435783465)
            (76,48.84155801115744)
            (77,48.88109053573829)
            (78,48.918351255524264)
            (79,48.95346809727775)
            (80,48.98656203389923)
            (81,49.01774743510123)
            (82,49.047132403534604)
            (83,49.07481909658738)
            (84,49.100904034127225)
            (85,49.12547839250319)
            (86,49.163901183186525)
            (87,49.209095175621286)
            (88,49.25213131496026)
            (89,49.2931101134694)
            (90,49.3321276219486)
            (91,49.36927560705245)
            (92,49.40464172359477)
            (93,49.43830968176147)
            (94,49.47035940918564)
            (95,49.50086720786639)
            (96,49.529905905937824)
            (97,49.55754500431554)
            (98,49.58385081826564)
            (99,49.60888661395732)
            (100,49.632712740073984)
            };
            \end{scope}
\end{tikzpicture}

%% file: main.bbl
\begin{thebibliography}{70}
\providecommand{\natexlab}[1]{#1}
\providecommand{\url}[1]{\texttt{#1}}
\expandafter\ifx\csname urlstyle\endcsname\relax
  \providecommand{\doi}[1]{doi: #1}\else
  \providecommand{\doi}{doi: \begingroup \urlstyle{rm}\Url}\fi

\bibitem[Af{\`e}che et~al.(2022)Af{\`e}che, Caldentey, and Gupta]{Afecheetal2022}
Philipp Af{\`e}che, Ren{\'e} Caldentey, and Varun Gupta.
\newblock On the optimal design of a bipartite matching queueing system.
\newblock \emph{Operations Research}, 70\penalty0 (1):\penalty0 363--401, 2022.

\bibitem[Amokeoja(2026)]{forbesafrica2026seriesa}
Oluwatomisin Amokeoja.
\newblock Deal sizes grow as fewer than 10\% of african startups reach series a, new report shows.
\newblock \url{https://www.forbesafrica.com/current-affairs/2026/05/07/deal-sizes-grow-as-fewer-than-10-of-african-startups-reach-series-a-new-report-shows}, 2026.

\bibitem[Anderson et~al.(2021)Anderson, McClain, Faverio, and Gelles-Watnick]{andersonEtAl2021gig}
Monica Anderson, Colleen McClain, Michelle Faverio, and Risa Gelles-Watnick.
\newblock The state of gig work in 2021.
\newblock \url{https://www.pewresearch.org/internet/2021/12/08/the-state-of-gig-work-in-2021/}, December 2021.

\bibitem[Angrist et~al.(2021)Angrist, Caldwell, and Hall]{angrist2021uber}
Joshua~D Angrist, Sydnee Caldwell, and Jonathan~V Hall.
\newblock Uber versus taxi: A driver’s eye view.
\newblock \emph{American Economic Journal: Applied Economics}, 13\penalty0 (3):\penalty0 272--308, 2021.

\bibitem[Aouad and Saban(2023)]{aouad2023online}
Ali Aouad and Daniela Saban.
\newblock Online assortment optimization for two-sided matching platforms.
\newblock \emph{Management Science}, 69\penalty0 (4):\penalty0 2069--2087, 2023.

\bibitem[Armony and Ward(2010)]{armony2010fair}
Mor Armony and Amy~R Ward.
\newblock Fair dynamic routing in large-scale heterogeneous-server systems.
\newblock \emph{Operations Research}, 58\penalty0 (3):\penalty0 624--637, 2010.

\bibitem[Armstrong(2006)]{armstrong2006competition}
Mark Armstrong.
\newblock Competition in two-sided markets.
\newblock \emph{The RAND journal of economics}, 37\penalty0 (3):\penalty0 668--691, 2006.

\bibitem[Asadpour et~al.(2023)Asadpour, Lobel, and van Ryzin]{asadpour2023minimum}
Arash Asadpour, Ilan Lobel, and Garrett van Ryzin.
\newblock Minimum earnings regulation and the stability of marketplaces.
\newblock \emph{Manufacturing \& Service Operations Management}, 25\penalty0 (1):\penalty0 254--265, 2023.

\bibitem[Aveklouris et~al.(2024)Aveklouris, DeValve, Stock, and Ward]{aveklouris2024matching}
Angelos Aveklouris, Levi DeValve, Maximiliano Stock, and Amy Ward.
\newblock Matching impatient and heterogeneous demand and supply.
\newblock \emph{Operations Research}, 2024.

\bibitem[Azariadis(1975)]{azariadis1975implicit}
Costas Azariadis.
\newblock Implicit contracts and underemployment equilibria.
\newblock \emph{Journal of Political Economy}, 83\penalty0 (6):\penalty0 1183--1202, 1975.

\bibitem[Bai et~al.(2019)Bai, So, Tang, Chen, and Wang]{bai2019coordinating}
Jiaru Bai, Kut~C So, Christopher~S Tang, Xiqun Chen, and Hai Wang.
\newblock Coordinating supply and demand on an on-demand service platform with impatient customers.
\newblock \emph{Manufacturing \& Service Operations Management}, 21\penalty0 (3):\penalty0 556--570, 2019.

\bibitem[Baily(1974)]{baily1974wages}
Martin~Neil Baily.
\newblock Wages and employment under uncertain demand.
\newblock \emph{The Review of Economic Studies}, 41\penalty0 (1):\penalty0 37--50, 1974.

\bibitem[Baker et~al.(2002)Baker, Gibbons, and Murphy]{baker2002relational}
George Baker, Robert Gibbons, and Kevin~J. Murphy.
\newblock Relational contracts and the theory of the firm.
\newblock \emph{The Quarterly Journal of Economics}, 117\penalty0 (1):\penalty0 39--84, 2002.

\bibitem[Benjaafar et~al.(2022)Benjaafar, Ding, Kong, and Taylor]{benjaafar2022labor}
Saif Benjaafar, Jian-Ya Ding, Guangwen Kong, and Terry Taylor.
\newblock Labor welfare in on-demand service platforms.
\newblock \emph{Manufacturing \& Service Operations Management}, 24\penalty0 (1):\penalty0 110--124, 2022.

\bibitem[Berenshtein et~al.(1989)Berenshtein, Vainshtein, and Kreinin]{berenshtein1989convexity}
AD~Berenshtein, AD~Vainshtein, and A~Ya Kreinin.
\newblock A convexity property of the poisson distribution and its application in queueing theory.
\newblock \emph{Journal of Soviet Mathematics}, 47\penalty0 (1):\penalty0 2288--2292, 1989.

\bibitem[Besbes et~al.(2023)Besbes, Fonseca, Lobel, and Zheng]{besbes2023signaling}
Omar Besbes, Yuri Fonseca, Ilan Lobel, and Fanyin Zheng.
\newblock Signaling competition in two-sided markets.
\newblock In \emph{Proceedings of the 24th ACM Conference on Economics and Computation}, pages 293--293, 2023.

\bibitem[Besbes et~al.(2024)Besbes, Goyal, Iyengar, and Singal]{besbes2024workforce}
Omar Besbes, Vineet Goyal, Garud Iyengar, and Raghav Singal.
\newblock Workforce scheduling with heterogeneous time preferences: Effective wages and workers’ supply.
\newblock \emph{Manufacturing \& Service Operations Management}, 26\penalty0 (5):\penalty0 1768--1786, 2024.

\bibitem[Bimpikis et~al.(2024)Bimpikis, Papanastasiou, and Zhang]{bimpikis2024information}
Kostas Bimpikis, Yiangos Papanastasiou, and Wenchang Zhang.
\newblock Information provision in two-sided platforms: Optimizing for supply.
\newblock \emph{Management Science}, 70\penalty0 (7):\penalty0 4533--4547, 2024.

\bibitem[{Board of Governors of the Federal Reserve System}(2025)]{federalReserve2025economicWellBeing}
{Board of Governors of the Federal Reserve System}.
\newblock Report on the economic well-being of u.s. households in 2024: Employment and gig work.
\newblock \url{https://www.federalreserve.gov/publications/2025-economic-well-being-of-us-households-in-2024-employment-and-gig-work.htm}, May 2025.

\bibitem[{Branch} and {Marqeta}(2023)]{branchMarqeta2023gig}
{Branch} and {Marqeta}.
\newblock {Branch x Marqeta 2023 Gig Payments Report}.
\newblock \url{https://branchapp.com/resources/2023-gig-payments-report}, 2023.

\bibitem[Br{\'e}maud(1981)]{bremaud1981point}
Pierre Br{\'e}maud.
\newblock Point processes and queues.
\newblock \emph{Springer}, 1981.

\bibitem[Cabral(2019)]{cabral2019towards}
Lu{\'\i}s Cabral.
\newblock Towards a theory of platform dynamics.
\newblock \emph{Journal of Economics \& Management Strategy}, 28\penalty0 (1):\penalty0 60--72, 2019.

\bibitem[Caillaud and Jullien(2003)]{caillaud2003chicken}
Bernard Caillaud and Bruno Jullien.
\newblock Chicken \& egg: Competition among intermediation service providers.
\newblock \emph{RAND journal of Economics}, pages 309--328, 2003.

\bibitem[Caldentey et~al.(2025)Caldentey, Hillas, and Gupta]{caldentey2024designing}
Ren{\'e} Caldentey, Lisa~Aoki Hillas, and Varun Gupta.
\newblock Designing service menus for bipartite queueing systems.
\newblock \emph{Operations Research}, 73\penalty0 (3):\penalty0 1496--1534, 2025.

\bibitem[Castro and Frazelle(2024)]{castro2024getting}
Francisco Castro and Andrew~E Frazelle.
\newblock Getting out of your own way: Introducing autonomous vehicles on a ride-hailing platform.
\newblock \emph{Production and Operations Management}, 33\penalty0 (10):\penalty0 2014--2030, 2024.

\bibitem[Castro et~al.(2022)Castro, Gao, and Martin]{castro2022supply}
Francisco Castro, Jian Gao, and Sebastien Martin.
\newblock Supply prioritization in hybrid marketplaces.
\newblock \emph{Available at SSRN 4119096}, 2022.

\bibitem[Castro et~al.(2024)Castro, Gao, and Martin]{castro2024autonomous}
Francisco Castro, Jian Gao, and Sebastien Martin.
\newblock Autonomous vehicles in ride-hailing and the threat of spatial inequalities.
\newblock \emph{Available at SSRN 4332493}, 2024.

\bibitem[Damani(2025)]{damani2025fundingwinter}
Anirudh Damani.
\newblock India's startup funding winter nearing a turning point, says artha vc's anirudh damani.
\newblock \url{https://www.fortuneindia.com/startups-news/indias-startup-funding-winter-nearing-a-turning-point-says-artha-vcs-anirudh-damani/131490}, 2025.

\bibitem[Dong and Ibrahim(2020)]{dong2020managing}
Jing Dong and Rouba Ibrahim.
\newblock Managing supply in the on-demand economy: Flexible workers, full-time employees, or both?
\newblock \emph{Operations Research}, 68\penalty0 (4):\penalty0 1238--1264, 2020.

\bibitem[Douc et~al.(2018)Douc, Moulines, Priouret, and Soulier]{douc2018markov}
Randal Douc, Eric Moulines, Pierre Priouret, and Philippe Soulier.
\newblock \emph{Markov chains}, volume~4.
\newblock Springer, 2018.

\bibitem[Dua et~al.(2022)Dua, Ellingrud, Hancock, Luby, Madgavkar, and Pemberton]{duaEtAl2022freelance}
Andr{\'e} Dua, Kweilin Ellingrud, Bryan Hancock, Ryan Luby, Anu Madgavkar, and Sarah Pemberton.
\newblock Freelance, side hustles, and gigs: Many more americans have become independent workers.
\newblock \url{https://www.mckinsey.com/featured-insights/sustainable-inclusive-growth/future-of-america/freelance-side-hustles-and-gigs-many-more-americans-have-become-independent-workers}, August 2022.

\bibitem[Eisenmann(2006)]{eisenmann2006internet}
Thomas~R Eisenmann.
\newblock Internet companies' growth strategies: determinants of investment intensity and long-term performance.
\newblock \emph{Strategic Management Journal}, 27\penalty0 (12):\penalty0 1183--1204, 2006.

\bibitem[Esmaeili et~al.(2023)Esmaeili, Duppala, Cheng, Nanda, Srinivasan, and Dickerson]{esmaeili2023rawlsian}
Seyed Esmaeili, Sharmila Duppala, Davidson Cheng, Vedant Nanda, Aravind Srinivasan, and John~P Dickerson.
\newblock Rawlsian fairness in online bipartite matching: Two-sided, group, and individual.
\newblock In \emph{Proceedings of the AAAI Conference on Artificial Intelligence}, volume~37, pages 5624--5632, 2023.

\bibitem[Farrell and Greig(2017)]{farrell2017online}
Diana Farrell and Fiona Greig.
\newblock The online platform economy: Has growth peaked?
\newblock \emph{Available at SSRN 2911194}, 2017.

\bibitem[Feng et~al.(2023)Feng, Niazadeh, and Saberi]{feng2024two}
Yiding Feng, Rad Niazadeh, and Amin Saberi.
\newblock Two-stage stochastic matching and pricing with applications to ride hailing.
\newblock \emph{Operations Research}, 72\penalty0 (4):\penalty0 1574--1594, 2023.

\bibitem[Freund et~al.(2026)Freund, Lobel, and Zhao]{freund2024supply}
Daniel Freund, Ilan Lobel, and Jiayu Zhao.
\newblock On the supply of autonomous vehicles in platforms.
\newblock \emph{Manufacturing \& Service Operations Management}, 28\penalty0 (2):\penalty0 496--516, 2026.

\bibitem[Guiso et~al.(2013)Guiso, Pistaferri, and Schivardi]{guiso2013credit}
Luigi Guiso, Luigi Pistaferri, and Fabiano Schivardi.
\newblock Credit within the firm.
\newblock \emph{The Review of Economic Studies}, 80\penalty0 (1):\penalty0 211--247, 2013.

\bibitem[Gurvich et~al.(2019)Gurvich, Lariviere, and Moreno]{gurvich2019operations}
Itai Gurvich, Martin Lariviere, and Antonio Moreno.
\newblock Operations in the on-demand economy: Staffing services with self-scheduling capacity.
\newblock In \emph{Sharing Economy: Making Supply Meet Demand}, pages 249--278. Springer, 2019.

\bibitem[Hagiu(2006)]{hagiu2006pricing}
Andrei Hagiu.
\newblock Pricing and commitment by two-sided platforms.
\newblock \emph{The RAND Journal of Economics}, 37\penalty0 (3):\penalty0 720--737, 2006.

\bibitem[He and Goh(2022)]{he2022profit}
Eryn~Juan He and Joel Goh.
\newblock Profit or growth? dynamic order allocation in a hybrid workforce.
\newblock \emph{Management Science}, 68\penalty0 (8):\penalty0 5891--5906, 2022.

\bibitem[Hu(2019)]{hu2019sharing}
Ming Hu.
\newblock \emph{Sharing economy: making supply meet demand}, volume~6.
\newblock Springer, 2019.

\bibitem[Iyengar and Singal(2024)]{iyengar2024online}
Garud Iyengar and Raghav Singal.
\newblock Online matching with heterogeneous supply and minimum allocation guarantees.
\newblock \emph{Available at SSRN 4899227}, 2024.

\bibitem[Jacobs and Farmand(2025)]{jacobs2025chicago}
Ken Jacobs and Aida Farmand.
\newblock Chicago transportation network drivers’ earnings and projections under proposed city ordinance.
\newblock 2025.

\bibitem[Kanoria and Saban(2021)]{kanoria2017facilitating}
Yash Kanoria and Daniela Saban.
\newblock Facilitating the search for partners on matching platforms.
\newblock \emph{Management Science}, 67\penalty0 (10):\penalty0 5990--6029, 2021.

\bibitem[Kerimov et~al.(2024)Kerimov, Ashlagi, and Gurvich]{kerimov2024dynamic}
S{\"u}leyman Kerimov, Itai Ashlagi, and Itai Gurvich.
\newblock Dynamic matching: Characterizing and achieving constant regret.
\newblock \emph{Management Science}, 70\penalty0 (5):\penalty0 2799--2822, 2024.

\bibitem[Kohlenberg and Gurvich(2025)]{kohlenberg2024cost}
Angela Kohlenberg and Itai Gurvich.
\newblock The cost of impatience in dynamic matching: Scaling laws and operating regimes.
\newblock \emph{Management Science}, 71\penalty0 (4):\penalty0 3303--3319, 2025.

\bibitem[Kumar et~al.(2010)Kumar, Lifshits, and Tomkins]{kumar2010evolution}
Ravi Kumar, Yury Lifshits, and Andrew Tomkins.
\newblock Evolution of two-sided markets.
\newblock In \emph{Proceedings of the third ACM international conference on Web search and data mining}, pages 311--320, 2010.

\bibitem[Lazear(1979)]{lazear1979mandatory}
Edward~P. Lazear.
\newblock Why is there mandatory retirement?
\newblock \emph{Journal of Political Economy}, 87\penalty0 (6):\penalty0 1261--1284, 1979.

\bibitem[Levin(2003)]{levin2003relational}
Jonathan Levin.
\newblock Relational incentive contracts.
\newblock \emph{American Economic Review}, 93\penalty0 (3):\penalty0 835--857, 2003.

\bibitem[Lian and Van~Ryzin(2021)]{lian2021optimal}
Zhen Lian and Garrett Van~Ryzin.
\newblock Optimal growth in two-sided markets.
\newblock \emph{Management Science}, 67\penalty0 (11):\penalty0 6862--6879, 2021.

\bibitem[Lian and van Ryzin(2025)]{lian2024capturing}
Zhen Lian and Garrett van Ryzin.
\newblock Capturing the benefits of autonomous vehicles in ride hailing: The role of market configuration.
\newblock \emph{Management Science}, 71\penalty0 (7):\penalty0 5491--5510, 2025.

\bibitem[Lian et~al.(2022)Lian, Martin, and van Ryzin]{lian2022labor}
Zhen Lian, Sebastien Martin, and Garrett van Ryzin.
\newblock Labor cost free-riding in the gig economy.
\newblock \emph{Available at SSRN 3775888}, 2022.

\bibitem[Lobel et~al.(2024)Lobel, Martin, and Song]{lobel2024frontiers}
Ilan Lobel, S{\'e}bastien Martin, and Haotian Song.
\newblock Frontiers in operations: Employees vs. contractors: An operational perspective.
\newblock \emph{Manufacturing \& Service Operations Management}, 2024.

\bibitem[Ma et~al.(2022)Ma, Xu, and Xu]{ma2020group}
Will Ma, Pan Xu, and Yifan Xu.
\newblock Group-level fairness maximization in online bipartite matching.
\newblock In \emph{AAMAS Conference proceedings}, 2022.

\bibitem[Michelacci and Quadrini(2009)]{michelacci2009financial}
Claudio Michelacci and Vincenzo Quadrini.
\newblock Financial markets and wages.
\newblock \emph{The Review of Economic Studies}, 76\penalty0 (2):\penalty0 795--827, 2009.

\bibitem[{\"O}zkan and Ward(2020)]{dynamic_matching_ride_sharing}
Erhun {\"O}zkan and Amy~R Ward.
\newblock Dynamic matching for real-time ride sharing.
\newblock \emph{Stochastic Systems}, 10\penalty0 (1):\penalty0 29--70, 2020.

\bibitem[Papanastasiou and Sekar(2024)]{papanastasiou2024designing}
Yiangos Papanastasiou and Shreyas Sekar.
\newblock Designing pandora's box: Product rankings and marketplace participation.
\newblock \emph{Available at SSRN}, 2024.

\bibitem[Rawls(1971)]{rawls1971atheory}
John Rawls.
\newblock A theory of justice.
\newblock \emph{Cambridge (Mass.)}, 1971.

\bibitem[Rochet and Tirole(2003)]{rochet2003platform}
Jean-Charles Rochet and Jean Tirole.
\newblock Platform competition in two-sided markets.
\newblock \emph{Journal of the european economic association}, 1\penalty0 (4):\penalty0 990--1029, 2003.

\bibitem[Rothman(2016{\natexlab{a}})]{rothman2016steroid}
Simon Rothman.
\newblock The steroid era of startups is over.
\newblock \url{https://medium.com/@simonrothman/the-steroid-era-of-startups-is-over-dcee6e565206}, May 2016{\natexlab{a}}.

\bibitem[Rothman(2016{\natexlab{b}})]{rothman2016uber}
Simon Rothman.
\newblock Why uber won: The startup steroid era and the use of capital as a performance enhancing drug.
\newblock \url{https://medium.com/@simonrothman/why-uber-won-5598a2a66561}, April 2016{\natexlab{b}}.

\bibitem[Sen(2017)]{sen2017collective}
Amartya Sen.
\newblock \emph{Collective choice and social welfare: Expanded edition}.
\newblock Penguin UK, 2017.

\bibitem[Shi(2022)]{shi2022optimal}
Peng Shi.
\newblock Optimal priority-based allocation mechanisms.
\newblock \emph{Management Science}, 68\penalty0 (1):\penalty0 171--188, 2022.

\bibitem[Shi(2025)]{shi2024optimal}
Peng Shi.
\newblock Optimal match recommendations in two-sided marketplaces with endogenous prices.
\newblock \emph{Management Science}, 71\penalty0 (9):\penalty0 7431--7448, 2025.

\bibitem[Siddiq and Taylor(2022)]{siddiq2022ride}
Auyon Siddiq and Terry~A Taylor.
\newblock Ride-hailing platforms: Competition and autonomous vehicles.
\newblock \emph{Manufacturing \& Service Operations Management}, 24\penalty0 (3):\penalty0 1511--1528, 2022.

\bibitem[Stein and Shakarchi(2009)]{stein2009real}
Elias~M Stein and Rami Shakarchi.
\newblock \emph{Real analysis: measure theory, integration, and Hilbert spaces}.
\newblock Princeton University Press, 2009.

\bibitem[Taylor(2018)]{taylor2018demand}
Terry~A Taylor.
\newblock On-demand service platforms.
\newblock \emph{Manufacturing \& Service Operations Management}, 20\penalty0 (4):\penalty0 704--720, 2018.

\bibitem[Teare(2024)]{teare2024globalStartupFunding}
Gen{\'e} Teare.
\newblock Global startup funding in 2023 clocks in at lowest level in 5 years.
\newblock \url{https://news.crunchbase.com/venture/global-funding-data-analysis-ai-eoy-2023/}, January 2024.

\bibitem[Zhang et~al.(2022)Zhang, Hu, and Wang]{hu2022implications}
Zhoupeng~Jack Zhang, Ming Hu, and Jianfu Wang.
\newblock Implications of worker classification in on-demand economy.
\newblock In \emph{INFORMS International Conference on Service Science}, pages 407--423. Springer, 2022.

\bibitem[Zhu et~al.(2026)Zhu, Chen, Zhang, and Zhou]{zhu2024optimal}
Yixin Zhu, Hongfan Chen, Renyu Zhang, and Sean~X Zhou.
\newblock Optimal growth of a two-sided platform with heterogeneous agents.
\newblock \emph{Production and Operations Management}, 35\penalty0 (4):\penalty0 1267--1284, 2026.

\end{thebibliography}
